\documentclass[12pt,a4paper]{article}
\usepackage[english]{babel}
\usepackage[tbtags]{amsmath}  
\usepackage{amsfonts}
\usepackage{amssymb}
\usepackage{amsthm}
\usepackage{latexsym}
\usepackage{mathtools}   
\usepackage{tabularx}
\usepackage{enumerate}

\usepackage[dvips]{graphicx}

\newtheorem{theorem}{Theorem}[section]
\newtheorem{prop}[theorem]{Proposition} 
\newtheorem{lemma}[theorem]{Lemma}
\newtheorem{cor}[theorem]{Corollary}

\newtheorem{defin}[theorem]{Definition}
\newtheorem{remark}[theorem]{Remark}

\DeclareMathOperator{\tr}{Tr}
\DeclareMathOperator{\sign}{sign}

\DeclareMathOperator{\erfc}{erfc}
\DeclareMathOperator{\erf}{erf}

\newcommand{\nn}{\nonumber}

\newcommand{\dd}{\mathrm{d}} 

\newcommand{\ve}[1]{\underline{#1}}
\newcommand{\ip}[2]{\left\langle #1,#2 \right\rangle} 
\newcommand{\lO}{\mathcal{O}}

\newcommand{\cc}[1]{\overline{#1}} 

\newcommand{\R}{\mathbb{R}}
\newcommand{\C}{\mathbb{C}}
\newcommand{\Cp}{\C^{+}}
\newcommand{\Cpp}{\C^{++}}
\newcommand{\N}{\mathbb{N}}

\newcommand{\dist}{\mathrm{dist}}
\renewcommand{\Re}{\mathrm{Re}}
\renewcommand{\Im}{\mathrm{Im}}
\newcommand{\e}{\mathrm{e}}

\newcommand{\n}{n} 
\newcommand{\op}[1]{p_{#1,\n}} 
\newcommand{\dop}[1]{p^{\prime}_{#1,\n}} 
\newcommand{\orthop}[1]{q_{#1,\n}} 
\newcommand{\tildeop}[1]{\tilde{p}_{#1,\n}} 
\newcommand{\hn}[1]{h_{#1,\n}}
\newcommand{\rn}[1]{r_{#1,\n}}
\newcommand{\Kn}{\tilde{K}_\n} 
\newcommand{\tKn}{K_\n} 
\newcommand{\nK}{K_\n} 
\newcommand{\F}{F} 
\newcommand{\Fz}{\tilde{z}} 
\newcommand{\Fx}{\tilde{x}} 
\newcommand{\Fy}{\tilde{y}} 
\newcommand{\Fu}{\tilde{u}} 
\newcommand{\Fv}{\tilde{v}} 
\newcommand{\Xdelta}[3]{#1_{#2,#3}}
\newcommand{\Adelta}{A_{\delta,\F}} 
\newcommand{\Bdelta}{B_{\delta,\F}}
\newcommand{\Cdelta}{C_{\delta,\F}}

\renewcommand{\t}{t} 
\newcommand{\f}{f}
\newcommand{\gp}{g_+}
\newcommand{\gm}{g_-}
\newcommand{\gpm}{g_\pm}
\newcommand{\gw}{g_w}
\newcommand{\gz}{g_z}
\newcommand{\hu}{h_1}
\newcommand{\huu}{h_2}
\newcommand{\huv}{h_{11}}
\newcommand{\fmax}{\f_{\max_{}}}
\newcommand{\fmaxc}{\fmax^{\textrm{cont}}}
\newcommand{\gmax}{g_{\max_{}}}
\newcommand{\zmax}{z_{\max_{}}}
\newcommand{\HM}{k_\n} 
\newcommand{\KM}{\tilde{k}_\n} 
\newcommand{\rhon}{\rho_\n}
\newcommand{\Rn}[1]{R_\n^{(#1)}} 
\newcommand{\Et}{E_\t} 
\newcommand{\pEt}{\partial \Et}
\newcommand{\Edelta}[1]{E_{\t,#1}} 
\newcommand{\Omegadelta}[1]{\Omega_{\t,#1}}
\newcommand{\mas}{a_\t} 
\newcommand{\mis}{b_\t} 
\newcommand{\nphase}{\theta_{\n,\t}} 
\newcommand{\acorrdet}[2]{S_{\n,\t,z_0,#2}^{#1}} 
\newcommand{\corrdet}[1]{\acorrdet{#1}{\ve{a}}} 
\newcommand{\Rtz}{R_{\ve{\xi},\n}(\t,z)} 
\newcommand{\Rsum}[1]{R_{\Sigma,\ve{\xi},\n} #1 } 
\newcommand{\Rxi}[2]{R_{#1,\ve{\xi},\n} #2 } 
\newcommand{\Rzero}{R_{0,\n}(\t,z)} 
\newcommand{\Rzeroz}[1]{R_{0,\n}(\t,#1)} 
\newcommand{\Rzerotz}[1]{R_{0,\n} #1} 
\newcommand{\Rone}{R_{1,\xi_1,\n}} 
\newcommand{\Rtwo}{R_{2,\xi_2,\n}(\t,z)} 
\newcommand{\Rtwoz}[1]{R_{2,\xi_2,\n}(\t,#1)} 
\newcommand{\Rtwotz}[1]{R_{2,\xi_2,\n} #1} 
\newcommand{\Rthree}{R_{3,\xi_3,\n}(\t,z)} 
\newcommand{\Rthreez}[1]{R_{3,\xi_3,\n}(\t,#1)} 
\newcommand{\Rthreetz}[1]{R_{3,\xi_3,\n} #1} 
\newcommand{\Rfour}{R_{4,\xi_4}(\t,z_0,v)} 
\newcommand{\Rfourz}[1]{R_{4,\xi_4}(\t,#1)} 
\newcommand{\Rfourtz}[1]{R_{4,\xi_4} #1 } 
\newcommand{\Rfive}{R_{5,\xi_5}(\t,z_0,v)} 
\newcommand{\Rfivez}[1]{R_{5,\xi_5}(\t,#1)} 
\newcommand{\Rsix}{R_{6,\xi_2,\xi_3,\xi_6,\n}(\t,z_0,v)} 
\newcommand{\Rsixz}[1]{R_{6,\xi_2,\xi_3,\xi_6,\n}(\t,#1)} 
\newcommand{\ftwo}{f_{2,z}} 
\newcommand{\fthree}[1]{f_{3 #1 ,z}} 
\newcommand{\ffour}{f_{4,z_0}} 
\newcommand{\ffourz}[1]{f_{4,#1}} 
\newcommand{\ffive}{f_{5,z_0}} 
\newcommand{\ffivez}[1]{f_{5,#1}} 
\newcommand{\fsix}{f_{6,z_0,\xi_2,\xi_3,\n}} 
\newcommand{\fsixz}[1]{f_{6,#1,\xi_2,\xi_3,\n}} 
\newcommand{\U}[1]{U_{#1}} 
\newcommand{\W}[1]{W_{#1}} 
\newcommand{\T}[1]{T_{#1}} 

\newcommand{\chapt}{t}
\newcommand{\chapw}{w}
\newcommand{\chapz}{z}

\title{Universality in Gaussian Random Normal Matrices}
\author{Roman Riser}
\date{}

\begin{document}
\maketitle
\renewcommand{\footnoterule}{%
  \kern -3pt
  \hrule width 0.9in
  \kern 2.6pt
}
\let\thefootnote\relax\footnote{This work is based on my PhD thesis under the supervision of Prof.~Giovanni Felder at the ETH Zurich}
\abstract{We prove that for Gaussian random normal matrices the correlation function has universal behavior. 

Using the technique of orthogonal polynomials and identities similar to the Christoffel-Darboux formula, we find that in the limit, as the dimension of the matrix tends to infinity, the density of eigenvalues converges to a constant inside of an ellipse and to zero outside. The convergence holds locally uniformly. At the boundary, in scaled coordinates holding the distance between eigenvalues constant, we show that the density is falling off to zero like the complementary error function. The convergence is uniform on the boundary. Further we give an explicit expression for the the correlation function.}
\tableofcontents
\newpage
\renewcommand{\footnoterule}{%
  \kern -3pt
  \hrule width 2in
  \kern 2.6pt
}






\numberwithin{equation}{section}


\section{Introduction}

In 1999, P. Wiegmann and A. Zabrodin have found an interesting connection between normal random matrices and conformal mappings \cite{WZ1}, \cite{WZ2}. They considered the following probability distribution on all normal $\n\times \n$ matrices $M$
\begin{align}
P_\n(M)&=\tfrac{1}{Z_\n}\e^{-\n \tr V(M)}, \qquad Z_\n=\int \e^{-\n \tr V(M)} \dd M, \\
\intertext{with the potential}
V(M)&=\tfrac{1}{\t_0}\left(M^* M - \sum_{k\in \N} \left(\cc{t}_k M^{*k} + t_k M^k\right) \right). \label{eq_intro1}
\end{align} 
At a level of formal manipulation, they showed that in the limit $\n\rightarrow\infty$ the eigenvalues fill a simply connected domain with constant density. 

Since already the initial integral diverges, this model needs a regularization. P. Elbau and G. Felder have put it in a rigorous mathematical context by introducing a cut-off $D$ and restricting the probability distribution to normal matrices with spectrum in $D$ \cite{felder-elbau}. 
A generalization of this model for more general potential was considered by P. Etingof and X. Ma \cite{etingof-ma}.

Motivated by these works, we wanted to study the kernel and correlation functions of this model for the potential \eqref{eq_intro1} and look for universalities. 

Universality properties in Random Matrix Models are long known for Hermitian Matrices. Among them we have the famous Wigner Semi-Circle Law, the sine kernel and different universalities for the correlation functions in the bulk of the spectrum and at the edge. 
Many of these universality proofs are based on a Riemann-Hilbert approach (see for example \cite{deift}). 

In \cite{its-takhtajan} A. Its and L. Takhtajan showed some ideas of an analogous approach (which they called $\cc{\partial}$-problem) for the potential \eqref{eq_intro1} of Random Normal Matrices. Unfortunately this has not yet led to a satisfying result. 

We have considered a quite different approach where we approximate the potential by a Gaussian potential. In a future paper the aim is to show a new way to prove universalities for all potentials as in \eqref{eq_intro1}. As an important step we need universal correlation functions for all Gaussian potentials which we are going to find in this thesis. We include results which also hold at the boundary of the support of the density. 

P. Bleher and A. Kuijlaars have found a way to avoid the cut-off when they considered the cubic case of the potential \eqref{eq_intro1} in \cite{bleher-kuijlaars}. 
Further results on Random Normal Matrices have also been found by Y. Ameur, H. Hedenmalm and N. Makarov \cite{makarov}, where they consider quite general potentials which don't need cut-offs. 

The result in this thesis is based on the method of orthogonal polynomials. Our work is organized as following. In Chapter \ref{ch_definition} we will review some known facts about Normal Matrix Models and introduce basic definitions such as the orthonormal polynomials and the reproducing kernel. In Chapter \ref{sec_gaussian} we will consider the simplest potential ($t_2=0$), where we will compute the kernel and the correlation functions. This will illustrate the techniques we are going to use in the more complex case. 
In Chapter \ref{ch_on_poly} we will review the recursion relation and a relation for the derivative of the orthonormal polynomials. With the help of these, we will show how the orthonormal polynomial in our case are related to the Hermite polynomials. An important point of our proof are the identities for the reproducing kernel $\nK(z,w)$ which we will find in Chapter \ref{sec_ident}. They have a similar function as the Christoffel-Darboux formula in the Riemann-Hilbert approach, but while the latter reveals directly the kernel, we get some relation which involves derivatives of the kernel.

In Chapter \ref{ch_asymp} we will review the Plancherel-Rotach asymptotics, 
which we need to approximate the Hermite polynomials. Indirectly our works so gets some connection to the Riemann-Hilbert approach as it could be used to find the Plancherel-Rotach asymptotics. 
With the help of the Plancherel-Rotach asymptotics, we are going to study the asymptotics of the identities evaluated at $w=\cc{z}$. 

In Chapter \ref{sec_density} we find the asymptotics of the kernel at $\cc{w}=z$, i.e.\ of the density, by integration. Then we analyze the asymptotics of the identities evaluated at $\cc{z}=\cc{z}_0+u$ and $w=z_0+v$ where $u$ and $v$ are of order $\lO(\n^{-1/2})$ in Chapter \ref{ch_kernel_corr}. Again by integration we will get the asymptotics of the reproducing kernel evaluated at $z$ and $w$. We will find that the kernel is universal apart from a nonrelevant phase which finally disappears when we compute the correlation functions. 

In Chapter \ref{sec_verification} we will review all approximation we have done in the previous chapters and give explicit estimations of them. Additionally we find where the universalities hold uniformly. In the result we get an additional term at the boundary of the support compared to the result inside of the ellipse, where the density is constant. But in both cases the result is universal and holds (locally) uniformly. This is our main result that we state in Theorem \ref{th_corr_uniform_ren} in Chapter \ref{sec_t}, where we additionally show that our initial assumption (that $t_2$ is real) can be removed and that our result holds for all Gaussian potentials with $t_0\in (0,\infty)$, $t_1,t_2\in \C$ with $2 \vert t_2\vert<1$.





\section{Definitions of the Normal Matrix Model and Known Facts}\label{ch_definition}
\sectionmark{Definitions and Known Facts}

In this chapter we state all the definitions for the normal random matrix model and some known facts. They are analogous to the Hermitian random matrix model with some modification. They can be found in 
\cite{deift}, \cite{bi} for the Hermitian case and in 
\cite{felder-elbau} for the normal matrix case.


\subsection{Potential and Orthogonal Polynomials}
We want to consider the following potential $V$ as in \cite{felder-elbau}. 
\begin{defin}\label{def_potential} 
Let $d\in \N$, $t_0>0$ and $t_1, \ldots, t_d\in \C$ with $t_d\neq 0$ if $d>2$. Then we define the potential $V$ as
\begin{equation}
V(z)=\frac{1}{t_0}\left(\vert z\vert^2 - 2\Re \sum_{k=1}^{d} t_k z^k\right), \qquad z\in \C.
\end{equation}
\end{defin}
We will concentrate on the case when $V$ is Gaussian, i.e.\ $d=2$ where we will find an explicit computation for the kernel and for the correlation functions.

\begin{defin}\label{def_inner_product} 
Let $D$ be a closed subset of $\C$, $\n\in\N$, $V$ as in Definition \ref{def_potential} and $f,g\in C^0(D)$ (and if $D$ is not compact $\vert f(z)\vert ,\vert g(z)\vert$ shall increase not faster than polynomial when $\vert z\vert \rightarrow \infty)$. We define the inner product
\begin{equation}
\ip{f}{g}=\int_D \cc{f(z)}g(z)\e^{-\n V(z)}\dd^2 z,
\end{equation}
where $D$ has to be such that the integral for $\ip{1}{1}$ exists and is finite. We will call $D$ the cut-off of the integral.
\end{defin}
Note that the condition for $D$ is fulfilled if $D$ is compact. If $V$ is Gaussian and $\vert t_2\vert<\frac{1}{2}$ we can choose $D=\C$. 

\begin{remark}
We want to remark that the inner product depends on the parameter $\n\in\N$ which stands for the dimension of the matrices in the random matrix model.
\end{remark}

\begin{defin}\label{def_ortho_poly}
Let $\n\in\N$. The polynomials $(\orthop{k})_{k\in \N_0}$ of degree $k$ with leading coefficient one defined by 
\begin{equation}
\ip{\orthop{k}}{\orthop{l}}=\delta_{kl} \hn{k}, \qquad \forall k,l\in \N_0,
\end{equation}
we will call the orthogonal polynomials associated with the potential $V$. $\hn{k}>0$ are the normalization constants. 
\end{defin}
We note that the so defined polynomials $\orthop{k}$ are unique. Their existence and uniqueness can be seen by using the Gram-Schmidt process on the monic polynomials $(z^k)_{k\in\N_0}$.  

\begin{defin}\label{def_op}
Let $\n\in \N$ and $(\orthop{k})_{k\in \N_0}$ the orthogonal polynomials with leading coefficient one as in Definition \ref{def_ortho_poly}. Then we define the orthonormal polynomials $(\op{k})_{k\in \N_0}$ for the potential $V$ by
\begin{equation}
\op{k}(z)=\frac{\orthop{k}(z)}{\sqrt{\hn{k}}}, \qquad z\in \C,
\end{equation}
where $\hn{k}>0$ are the normalization constants from Definition \ref{def_ortho_poly}.
\end{defin}

\subsection{Kernel}
\begin{defin}\label{def_kernel}
Let $\n\in\N$. With the help of the orthonormal polynomials $(\op{k})_{k=0}^{\n-1}$ we can define the (unnormalized) reproducing kernel $\Kn$ for the potential $V$ and $D\subset \C$ as in Definition \ref{def_inner_product} by
\begin{equation}
\Kn(w,z)=\sum_{k=0}^{\n-1} \cc{\op{k}(w)}\op{k}(z), \qquad w,z\in D.
\end{equation}
\end{defin}

\begin{prop}
The reproducing kernel $\Kn$ for the potential $V$ with $D\subset \C$ as in Definition \ref{def_inner_product} fulfills
\begin{equation}
\int_D \Kn(z,z) \e^{-\n V(z)}\dd^2 z=\n.
\end{equation}
\end{prop}

\begin{defin}\label{def_norm_kernel}
Let $\n\in\N$. We define the normalized reproducing kernel $\tKn$ for the potential $V$ and $D\subset \C$ as in Definition \ref{def_inner_product} by
\begin{equation}
\tKn(w,z)=\e^{-\frac{\n}{2} (V(w)+V(z))} \Kn(w,z) =\e^{-\frac{\n}{2} (V(w)+V(z))} \sum_{k=0}^{\n-1} \cc{\op{k}(w)}\op{k}(z), \quad w,z\in D.
\end{equation}
\end{defin}
\begin{remark}
If each orthonormal polynomial $\op{k}(z)$ is either an even or odd, function then $\cc{\op{k}(w)}\op{k}(z)$ and therefore $\tKn(w,z)$ will not change when we map simultaneously the variables $w$ and $z$ to $-w$ and $-z$.
\end{remark}

\begin{cor}\label{cor_norm}
The reproducing kernel $\tKn$ for the potential $V$ with $D\subset \C$ as in Definition \ref{def_inner_product} fulfills
\begin{equation}
\int_D \tKn(z,z) \dd^2 z=\n.
\end{equation}
\end{cor}

\begin{prop}\label{prop_kernel_const}
Let $V$ be a potential as in Definition \ref{def_potential} and $V^*(z)=V(z)+C$ for $z\in \C$ where $C$ is a real constant. $\tKn$ shall be the normalized reproducing kernel for the potential $V$ and $\tKn^*$ the analogous kernel for the potential $V^*$. Then 
\begin{equation}
\tKn(w,z)=\tKn^*(w,z), \qquad \forall w,z\in\C.
\end{equation}
\end{prop}
\begin{proof}
For $m\in\N_0$ we call $\op{m}$ the orthonormal polynomials for the potential $V$ and $\op{m}^{*}$ the analogous polynomials for $V^*$. Obviously from Definition \ref{def_op} and \ref{def_inner_product} we see that $\op{m}^{*}=\op{m}\e^{-\n C/2}$. It follows from the definition of the normalized reproducing kernel that $\tKn=\tKn^*$.
\end{proof}


\subsection{Density and Correlation Functions}
\begin{defin}\label{def_rho}
The density of eigenvalues at $z\in D$ is given by
\begin{equation}
\rho_\n(z)=\frac{\tKn(z,z)}{\n}.
\end{equation}
\end{defin}
\begin{remark}
By Corollary \ref{cor_norm} follows that the density is normalized to one.
\end{remark}

\begin{defin}\label{def_rho_infty}
We define the density of eigenvalues for $\n \rightarrow\infty$ at $z\in D$ as
\begin{equation}
\rho(z)=\lim_{\n\rightarrow\infty} \rho_\n(z).
\end{equation}
\end{defin}

\begin{defin}\label{def_correlation}
Let $1\le m\le \n$. We define the $m$-point correlation function $\Rn{m}$ for the potential $V$ and $D$ as in Definition \ref{def_inner_product} at $z_1,\ldots,z_m\in D\subset \C$ as
\begin{equation}
\Rn{m} \left(z_1,\ldots,z_m \right)=\det\left(\tKn(z_k,z_l)\right)_{k,l=1}^m,
\end{equation}
\end{defin}



\section{Pure Gaussian Potential}\label{sec_gaussian}
In this chapter we will look at the case when $d=2$ and $t_1=t_2=0$, i.e.\ $V(z)=\vert z\vert^2$ is pure Gaussian. This special case we will only consider here as in the following chapters we will set $\vert t_2\vert>0$. It is also an easy example which will illustrate the techniques we are going to use for the more general case. Here the orthonormal polynomials will be just monomials. Most of the ideas we will need later can be found in this chapter. But later the calculations will get much more complicate as the orthonormal polynomials are no longer monomials and we will need an approximation to find their asymptotics. 

\subsection{Orthonormal Polynomials}

\begin{prop}\label{prop_gauss_onp}
Let $\n\in \N$. For the pure Gaussian potential $V(z)=\vert z \vert^2$ the orthonormal polynomials respect to the inner product as in Definition \ref{def_inner_product} with $D=\C$ are
$$\op{m}(z)=\sqrt{\frac{\n^{m+1}}{ \pi m!}}z^m, \qquad \forall m\in \N_0.$$
\end{prop}
\begin{proof}
Obviously for any rotationally symmetric potential $V$ and rotationally symmetric domain $D$ the monomials $(z^m)_{m\in\N_0}$ are orthogonal. We only are left to calculate the normalization factors. In the pure Gaussian case we can easily get this by $m$-times partial integration.
\begin{equation}\begin{split}
\ip{z^m}{z^m}&=\int_\C \vert z\vert^{2m} \e^{-\n\vert z\vert^2}\dd^2z=2\pi\int_0^\infty r^{2m}r\e^{-\n r^2}\dd r\\
&=2\pi\left(r^{2m}\left(-\tfrac{1}{2\n}\right)\e^{-\n r^{2}}\Big\vert_{r=0}^\infty+\int_0^\infty 2m r^{2m-1}\tfrac{1}{2\n} \e^{-\n r^{2}}\dd r \right)\\
&=\ldots=\pi m! \left(\frac{1}{\n}\right)^{m+1}.
\end{split}\end{equation}

Thus the orthonormal polynomials are
$$p_m^\n(z)=\sqrt{\frac{\n^{m+1}}{ \pi m!}}z^m.$$
\end{proof}

\subsection{Identities for the Kernel}

For the normalized reproducing kernel $\nK$ we define a ''pre-kernel'' $\HM$.
\begin{defin}\label{def_hm_0}
Let $\n\in\N$. Then we define the ``pre-kernel'' $\HM$ as following.
\begin{equation}
\HM(w,z)=\frac{1}{\n}\e^{-\n w z}\sum_{m=0}^{\n-1} \op{m}(w)\op{m}(z)=\frac{1}{\pi}\e^{-\n w z}\sum_{m=0}^{\n-1} \frac{(\n w  z)^m}{m!} , \qquad \forall w,z\in\C.
\end{equation}
\end{defin}
$\HM$ is related to $\nK$ by
\begin{align}\label{eq_HM_tKn_0}
\HM(\cc{w},z)&=\frac{\nK(w,z)}{\n} e^{\n\left(\frac{\vert w \vert^2+\vert z \vert^2}{2}-\cc{w}z \right)} = \frac{\nK(w,z)}{\n}e^{\frac{\n}{2}\left(\vert w-z\vert^2 +w\cc{z}-\cc{w}z\right)}, 
\end{align}
and at $w=z$ we have $\n \HM(\cc{z},z)=\n\rhon(z)=\nK(z,z)$.

\begin{prop}\label{prop_id_0}
Let $\n\in\N$. Then $\HM$ fulfills the identities
\begin{align}\label{eq_id_0}
\frac{\partial \HM}{\partial w} (w,z)\pm \frac{\partial \HM}{\partial z} (w,z) =\mp \bigl(\op{\n}(w) \op{\n-1}(z) \pm \op{\n-1}(w) \op{\n}(z)\bigr)\e^{-\n w z}, \qquad \forall w,z\in\C.
\end{align}
\end{prop}
\begin{proof}
This follows directly by differentiating
\begin{align}
\frac{\partial \HM}{\partial w} (w,z)&=-\tfrac{\n z}{\pi} \e^{-\n w z} \sum_{m=0}^{\n-1} \tfrac{(\n w  z)^m}{m!}+\tfrac{1}{\pi} \e^{-\n w z} \sum_{m=1}^{\n-1} \tfrac{(\n w)^m m z^{m-1}} {m!}= - \e^{-\n w z} \tfrac{\n^{\n+1} w^{\n-1}  z^{\n}}{\pi \n !} \nn \\
&= - \op{\n-1}(w) \op{\n}(z), \label{eq_id_w_0} \\
\intertext{and analogously}
\frac{\partial \HM}{\partial z} (w,z)&= - \op{\n-1}(z) \op{\n}(w). \label{eq_id_z_0}
\end{align}
\end{proof}

\subsection{Asymptotics of Identities}
Next we study the asymptotics of identities \eqref{eq_id_0} evaluated at $w=\cc{z}$. Using the Stirling approximation $\n ! = \sqrt{2 \pi \n}\n^\n \e^{-\n} \left(1+\lO(\n^{-1})\right)$, we see from \eqref{eq_id_w_0} and \eqref{eq_id_z_0} that the asymptotics of \eqref{eq_id_0} as $\n\rightarrow \infty$ is
\begin{align}
\frac{\partial \HM}{\partial w} (w,z)\pm \frac{\partial \HM}{\partial z} (w,z) &=\mp \sqrt{\tfrac{\n}{2\pi^3 }} \e^{\n\left( \log(wz) - wz + 1 \right)} \left(\tfrac{1}{w} \pm \tfrac{1}{z} \right)\left( 1+\lO(\n^{-1}) \right), \quad \forall w,z\in\C\setminus\{0\}, \label{eq_asym_id_0} \\
\intertext{and evaluated at $w=\cc{z}$}
\frac{\partial \HM}{\partial w} (\cc{z},z)\pm \frac{\partial \HM}{\partial z} (\cc{z},z) &=\mp \sqrt{\tfrac{\n}{2\pi^3 }} \e^{\n\left( \log\left(\vert z \vert^2\right) - \vert z\vert^2 + 1 \right)} \left(\tfrac{1}{\cc{z}} \pm \tfrac{1}{z} \right) \left( 1+\lO(\n^{-1}) \right) \nn \\
&= \sqrt{\tfrac{\n}{2\pi^3 }} g_{0,\pm}(z)\e^{\n \f_0(z)} \left( 1+\lO(\n^{-1}) \right), \qquad \forall z\in\C\setminus\{0\}, \label{eq_asym_id_eval_0}
\end{align}
where 
\begin{align}\label{eq_f_g_0}
\f_0(z)=\log\left(\vert z \vert^2\right) - \vert z\vert^2 + 1, \qquad \text{and} \qquad g_{0,\pm}(z)=\mp \left( \tfrac{1}{\cc{z}} \pm \tfrac{1}{z}  \right), \qquad \forall z\in\C\setminus\{0\}. 
\end{align}
$\f_0$ and $g_{0,\pm}$ are continuous functions on $\C\setminus\{0\}$. $\f_0(z)=0$ if and only if $\vert z \vert=1$ and $\f_0(z)<0$ if $\vert z \vert\neq 1$. 

\begin{remark}
Note that the $\lO(\n^{-1})$ term in \eqref{eq_asym_id_0} and \eqref{eq_asym_id_eval_0} holds uniformly for $z\in \C\setminus\{0\}$ since the only approximation we have used, the Stirling approximation, is independent of $z$.
\end{remark}

\subsection{Density}
\begin{prop}\label{prop_density_0}
The density of eigenvalues for the potential $V(z)=\vert z\vert^2$  is given by
\begin{equation}
\rho(z)=\lim_{\n\rightarrow\infty} \rho_\n(z)=\left\{\begin{array}{lll} \tfrac{1}{\pi} & \text{if $\vert z \vert<1$,}\\ \tfrac{1}{2 \pi} & \text{if $\vert z\vert=1$,}\\ 0 & \text{if $\vert z \vert>1$.}\end{array}\right. 
\end{equation}
\end{prop}
\begin{proof}
It is obvious that the density at zero is $\tfrac{1}{\pi}$ since
\begin{align}
\rhon(0)=\tfrac{1}{\n} \e^{-\n V(0)} \sum_{k=0}^{\n-1} \vert \op{k}(0) \vert^2=\tfrac{1}{\n} \vert \op{0} \vert^2= \tfrac{1}{\pi}, \qquad \forall \n\in\N.
\end{align}
For the density $\rhon(x,y)=\rhon(x+i y)$ we can find the derivatives by
\begin{align}
\frac{\partial \rhon}{\partial x}(x,y)&=\frac{\partial \HM}{\partial w}(x-i y,x+i y)+\frac{\partial \HM}{\partial z}(x-i y,x+i y), \\
\intertext{and}
\frac{\partial \rhon}{\partial y}(x,y)&=-i \left( \frac{\partial \HM}{\partial w}(x-i y,x+i y) - \frac{\partial \HM}{\partial z}(x-i y,x+i y) \right),
\end{align}
which asymptotics we know from \eqref{eq_asym_id_eval_0}. The density $\rhon$ in the limit $\n\rightarrow \infty$ we can get by integration of $\tfrac{\partial \rhon}{\partial x}$ along a path parallel to the real axis or $\tfrac{\partial \rhon}{\partial y}$ along a path parallel to the imaginary axis. Any path which does not cross (or touch) the unit circle does not contribute to the integral because \eqref{eq_asym_id_eval_0} holds uniformly and $\e^{\n \f_0(z)}$ converges uniformly to zero where the supremum of $\f_0$ on this path is strictly negative. Therefore the density inside and outside of the unit circle is constant.

To find the density on the circle and the jump of the density across the circle we consider a short path in the neighborhood of the unit circle either parallel to the real axis or parallel to the imaginary axis. We will approximate $\f_0(x,y)=\f_0(x+i y)$ by its Taylor expansion up to order two either as a function of $x$ or as a function of $y$. The neighborhood of the unit circle in which the path lies shall be so small that we can neglect the difference of $\f_0$ and its Taylor expansion and the variation of the continuous functions $g_{0,\pm}$ along the path. Let $z_1=x_1+i y_1$ be the starting point and $z_2=x_2+i y_2$ the ending point of the path. $z_2$ can lie inside, outside or on the unit circle. We will call $z_E=x_E+i y_E$ the point of intersection of the line through $z_1$ and $z_2$ and the circle. We know that the density at $z_1$ is $\tfrac{1}{\pi}$. For the density at $z_2$ we get 
\begin{align}\label{eq_proof_density_0}
\rho(z_2)&\approx \tfrac{1}{\pi}+\lim_{\n\rightarrow\infty}\sqrt{\tfrac{\n}{2\pi^3}} g_{0,+}(z_E) \int_{x_1}^{x_2} \e^{\n \frac{\partial^2\f_0}{\partial x^2}(x_E,y_E)\frac{(x-x_E)^2}{2}} \dd x\\ 
&= \tfrac{1}{\pi}+\lim_{\n\rightarrow\infty}\sqrt{\tfrac{\n}{2\pi^3}} (-2x_E) \int_{x_1-x_E}^{x_2-x_E} \e^{-4 x_E^2 \n x^2} \dd x\\ 
&= \tfrac{1}{\pi}-\tfrac{1}{\pi^{3/2}} \int_{-\infty}^{\infty (\vert x_2\vert-\vert x_E\vert)} \e^{-x^2} \dd x= \left\{ \begin{array}{lll} \frac{1}{\pi} &\text{if $\vert x_2\vert < \vert x_E \vert$,}\\ \frac{1}{2\pi} &\text{if $x_2=x_E$,}\\ 0 &\text{if $\vert x_2 \vert > \vert x_E \vert$,} \end{array}\right.
\end{align}
for a path parallel to the real axis. We have used that 
\begin{align}
\frac{\partial^2\f_0}{\partial x^2}(x,y)\Big\vert_{x^2+y^2=1}=-4 x^2, \qquad \text{and}\qquad  g_{0,+}(x,y)\Big\vert_{x^2+y^2=1}=-2x.
\end{align}


An analogous computation can also be done with a path parallel to the imaginary axis. More details can be found in the proof for the case when $\t_2\neq 0$ in Proposition \ref{prop_density} and Proposition \ref{prop_veri_density}. 
\end{proof}

\subsection{Kernel}
As we have seen in the proof of Proposition \ref{prop_id_0}, $\HM$ also fulfills the identities
\begin{align}
\frac{\partial \HM}{\partial w} (w,z)&= - \op{\n-1}(w) \op{\n}(z), \label{eq_id_w2_0} \\
\intertext{and}
\frac{\partial \HM}{\partial z} (w,z)&= - \op{\n-1}(z) \op{\n}(w). \label{eq_id_z2_0}
\end{align}
Now we want to find the kernel at $w=\cc{z}_0+\tfrac{\cc{a}}{\sqrt{\n}}$ and $z=z_0+\tfrac{b}{\sqrt{\n}}$ for $z_0,a,b\in\C$. Evaluating the identities at $w=\cc{z}_0+u$ and $z=z_0+v$ with $u=\lO(\n^{-1/2})$ and $v=\lO(\n^{-1/2})$, we get their asymptotics similar as in \eqref{eq_asym_id_eval_0} 
\begin{align}
\frac{\partial \HM}{\partial w} (\cc{z}_0+u,z_0+v) &= - \sqrt{\tfrac{\n}{2\pi^3 }} \tfrac{1}{\cc{z}_0} \e^{\n \left(\f_0(z_0)+u\left(\frac{1}{\cc{z}_0}-z_0\right)-\frac{u^2}{2 \cc{z}_0^2} +v\left(\frac{1}{z_0}-\cc{z}_0\right)-\frac{v^2}{2 z_0^2} -uv \right)} \left( 1+\lO(\n^{-1/2}) \right),  \label{eq_asym_id_kernelw_0} \\
\intertext{and}
\frac{\partial \HM}{\partial z} (\cc{z}_0+u,z_0+v)&= - \sqrt{\tfrac{\n}{2\pi^3 }} \tfrac{1}{z_0} \e^{\n \left(\f_0(z_0)+u\left(\frac{1}{\cc{z}_0}-z_0\right)-\frac{u^2}{2 \cc{z}_0^2} +v\left(\frac{1}{z_0}-\cc{z}_0\right)-\frac{v^2}{2 z_0^2} -uv \right)} \left( 1+\lO(\n^{-1/2}) \right),  \label{eq_asym_id_kernelz_0}
\end{align}
where $\f_0(z_0)=\log\left(\vert z_0 \vert^2\right) - \vert z_0\vert^2 + 1$. We have used that
\begin{align}
\tfrac{1}{\cc{z}_0+u}&=\tfrac{1}{\cc{z}_0}+\lO(\n^{-1/2}), \qquad \tfrac{1}{z_0+v}=\tfrac{1}{z_0}+\lO(\n^{-1/2}), \\
\intertext{and} 
\log((\cc{z}_0+u)(z_0+v))&=\log(\vert z_0 \vert^2)+\tfrac{u}{\cc{z}_0}+\tfrac{v}{z_0}-\tfrac{u^2}{2\cc{z}_0^2}-\tfrac{v^2}{2 z_0^2}+\lO\left(\n^{-3/2}\right).
\end{align}

\begin{prop}\label{prop_kernel_0}
Let $a,b,z_0\in\C$. Then the asymptotics as $\n\rightarrow\infty$ of the kernel $\tKn$ for the potential $V(z)=\vert z\vert^2$ is given by
\begin{align}
\tfrac{\tKn}{\n}\left(z_0+\tfrac{a\e^{i\psi}}{\sqrt{\n}},z_0+\tfrac{b\e^{i\psi}}{\sqrt{\n}}\right) =\e^{i\nphase(z_0,a\e^{i\psi},b\e^{i\psi})} \e^{-\frac{\vert a-b\vert^2}{2}} \left\{\begin{array}{lll}\frac{1}{\pi} & \text{if $\vert z_0 \vert < 1$,}\\ \frac{1}{2\pi}\erfc\left(\frac{\cc{a}+b}{\sqrt{2}}\right) & \text{if $ \vert z_0 \vert=1$, }\\ 0 & \text{if $\vert z_0 \vert>1$,}\end{array}\right. 
\end{align}
with
\begin{equation}
\begin{split}
\nphase(z_0,a\e^{i\psi},b\e^{i\psi})=&\sqrt{\n}\Im\left(\e^{-i\psi} z_0 (\cc{a}-\cc{b}) \right) + \Im\left( \cc{a}b \right),
\end{split}
\end{equation}
where $\psi\in(-\pi,\pi]$ can be arbitrary if $\vert z_0 \vert\neq 1$ and $\psi=\arg z_0$ if $\vert z_0\vert=1$.
\end{prop}
\begin{proof}
First we compute $\HM$ at $w=\cc{z}_0+\tfrac{\cc{a}}{\sqrt{\n}}$ and $z=z_0+\tfrac{b}{\sqrt{\n}}$ by integration
\begin{align}
\HM\left(\cc{z}_0+\tfrac{\cc{a}}{\sqrt{\n}},z_0+\tfrac{b}{\sqrt{\n}} \right) =\HM(\cc{z}_0,z_0)&+\int_0^{\frac{\cc{a}}{\sqrt{\n}}} \frac{\partial \HM}{\partial w}\left(\cc{z}_0+u,z_0\right)\dd u \nn \\
&+ \int_0^{\frac{b}{\sqrt{\n}}} \frac{\partial \HM}{\partial z}\left(\cc{z}_0+\tfrac{\cc{a}}{\sqrt{\n}},z_0+v\right)\dd v.
\end{align}
If $\vert z_0\vert \neq 1$, we know that $\f_0$ is strictly negative and thus we see from \eqref{eq_asym_id_kernelw_0} and \eqref{eq_asym_id_kernelz_0} that the integrals become zero in the limit $\n\rightarrow \infty$. Therefore then 
\begin{align}\HM\left(\cc{z}_0+\tfrac{\cc{a}}{\sqrt{\n}},z_0+\tfrac{b}{\sqrt{\n}} \right) \approx \HM(\cc{z}_0,z_0)\approx\left\{\begin{array}{ll} \frac{1}{\pi}, & \text{if $\vert z_0 \vert<1$,} \\ 0, & \text{if $\vert z_0\vert>1$.} \end{array}\right. 
\end{align}

If $\vert z_0 \vert=1$, on the other hand, $\f_0$ is zero and we get as $\n\rightarrow\infty$ in a first step
\begin{align}
\HM\left(\cc{z}_0+\tfrac{\cc{a}}{\sqrt{\n}},z_0\right) &\approx \tfrac{1}{2\pi}-\sqrt{\tfrac{\n}{2 \pi^3}} \tfrac{1}{\cc{z}_0} \int_0^{\frac{\cc{a}}{\sqrt{\n}}} \e^{-\n \frac{u^2}{2 \cc{z}_0^2} } \dd u \nn \\
& = \tfrac{1}{2\pi}-\pi^{-3/2} \int_0^{\frac{\cc{a}}{\sqrt{2}\cc{z}_0 } } \e^{-u^2} \dd u = \tfrac{1}{2\pi} \erfc\left(\tfrac{\cc{a} z_0}{\sqrt{2}}\right),
\end{align}
where we have used that according to Proposition \ref{prop_density_0} $\HM(\cc{z}_0,z_0)\approx \tfrac{1}{2\pi}$ as $\n\rightarrow\infty$ and
\begin{align}
\f_0(z_0)+u\left(\tfrac{1}{\cc{z}_0}-z_0\right)-\tfrac{u^2}{2 \cc{z}_0^2} +v\left(\tfrac{1}{z_0}-\cc{z}_0\right)-\tfrac{v^2}{2 z_0^2} -uv = -\tfrac{u^2}{2 \cc{z}_0^2},
\end{align}
when we set $v=0$ and if $\vert z_0 \vert=1$. With the help of an analogous computation we find in a second step that
\begin{align}
\HM\left(\cc{z}_0+\tfrac{\cc{a}}{\sqrt{\n}},z_0+\tfrac{b}{\sqrt{\n}} \right) &\approx \tfrac{1}{2\pi} \erfc\left(\tfrac{\cc{a} z_0 + b \cc{z}_0}{\sqrt{2}}\right).
\end{align}
More details can be found in the proofs of the general case in Lemma \ref{lemma_hm_w_ellipse}, Proposition \ref{prop_hm_wz_ellipse} and Proposition \ref{prop_veri_kernel}.

When we rotate the coordinate system such that the real axis is normal on the unit circle, we get 
\begin{align}
\HM\left(\cc{z}_0\bigl(1+\tfrac{\cc{a}}{\sqrt{\n}}\bigr),z_0\bigl(1+\tfrac{b}{\sqrt{\n}}\bigr) \right) &\approx \tfrac{1}{2\pi} \erfc\left(\tfrac{\cc{a}+b}{\sqrt{2}}\right).
\end{align}

Using the relation \eqref{eq_HM_tKn_0} evaluated at $w=\cc{z}_0+\tfrac{\cc{a}}{\sqrt{\n}}$ and $z=z_0+\tfrac{b}{\sqrt{\n}}$, we get back to the reproducing kernel $\nK$. For the general case we will do this computation with more details in Theorem \ref{th_km}.
\end{proof}

\subsection{Correlation Functions}
\begin{prop}\label{prop_corr_0}
Let $\Rn{m}$ be the $m$-point correlation function and $a_1,\ldots,a_m\in\C$. Then
\begin{align}
\lim_{\n\rightarrow\infty} \frac{\Rn{m}}{\n^m} \left(z_0+\tfrac{a_1}{\sqrt{\n}},\ldots,z_0+\tfrac{a_m}{\sqrt{\n}} \right) &= \det \left( \tfrac{1}{\pi}\e^{-\frac{1}{2}\vert a_k-a_l \vert^2+i\Im \left(\cc{a}_k a_l\right) }  \right)_{k,l=1}^m, \quad \!\! \text{if $\vert z_0 \vert<1$}, \\ 
\lim_{\n\rightarrow\infty} \frac{\Rn{m}}{\n^m} \left(z_0+\tfrac{a_1}{\sqrt{\n}},\ldots,z_0+\tfrac{a_m}{\sqrt{\n}} \right)&=0, \qquad \text{if $\vert z_0 \vert>1$}, \\
\lim_{\n\rightarrow\infty} \frac{\Rn{m}}{\n^m} \left(z_0\bigl(1+\tfrac{a_1}{\sqrt{\n}}\bigr),\ldots,z_0\bigl(1+\tfrac{a_m}{\sqrt{\n}}\bigr)\right)&=\det \left( \tfrac{\erfc\left(\frac{\cc{a}_k+a_l}{\sqrt{2}}\right)}{2\pi} \e^{-\frac{1}{2}\vert a_k-a_l \vert^2+ i\Im \left(\cc{a}_k a_l\right) }  \right)_{k,l=1}^m, \nn \\
& \qquad \text{if $\vert z_0 \vert=1$}.
\end{align}
\end{prop}
\begin{proof}
The proposition follows from a straightforward computation from Definition \ref{def_correlation} and Proposition \ref{prop_kernel_0}. The oscillating term $\sqrt{\n}\Im\left(\e^{-i\psi} z_0 (\cc{a}-\cc{b}) \right)$ in the phase $\nphase$ is canceled when we evaluate the product 
\begin{align}
\prod_{k=1}^m \frac{\tKn}{\n}\left(z_0+\tfrac{a_k}{\sqrt{\n}},z_0+\tfrac{a_{\sigma(k)}}{\sqrt{\n}}\right), \qquad a=(a_1,\ldots,a_m)
\end{align} 
in the determinant. We will show all details of this computation for the general case in Theorem \ref{th_corr} .
\end{proof}


\section{Orthonormal Polynomials for General Gaussian Potentials}\label{ch_on_poly}
\sectionmark{Orthonormal Polynomials}

Now we consider all potentials as in Definition \ref{def_potential} with $d=2$ and $\vert \t_2\vert <\tfrac{1}{2}$. The case when $t_2=0$ we have already discussed in Chapter \ref{sec_gaussian}. Therefore we can now always assume $\vert t_2\vert>0$. For simplicity we set $t_0=1$, $t_1=0$ and $2 t_2=\t$. We will discuss in Chapter \ref{sec_t} what happens if $t_0\neq 1$ and $t_1\neq 0$. The potential $V$ now reads
\begin{equation}
V(z)=\vert z\vert^2-\Re(\t z^2), \qquad z,t\in \C, \text{ with } 0<\vert \t \vert<1.
\end{equation}
In Definition \ref{def_inner_product} we set $D=\C$ because we don't need a cut-off.

\subsection{Recursion Relation}


For the Gaussian potential $V$ with $\vert t\vert <1$ we have the following recursion relation for the orthonormal polynomials (see \cite{elbau}, Proposition 5.10, where the error term $o(\e^{-\n \epsilon})$ has disappeared when $D=\C$).
\begin{equation}\label{eq_rec}
z \op{m}(z)=\rn{m+1}\op{m+1}(z)+\cc{\t} \, \rn{m} \op{m-1}(z), \qquad \forall m\in \N_0, z\in\C,
\end{equation}
with the real coefficients
\begin{equation}\label{eq_rm}
\rn{m}=\sqrt{\frac{m }{ \n ( 1-\vert \t\vert^2 ) }}, \qquad m\in \N_0,
\end{equation}
and $\op{-1}=0$.

\subsection{Orthonormal Polynomials $\op{m}$}
We are interested in the orthonormal polynomials for the weight 
\begin{equation}
\e^{-\n V(z)}=\e^{\n\left(-|z|^2+\frac{\t z^2}{2}+\frac{\cc{\t}\cc{z}^2}{2}\right)}.
\end{equation}
We will find that the orthonormal polynomials happen to be rescaled Hermite polynomials. 

\begin{prop}\label{prop_op}
Let $\t\in \C$ with $0<\vert \t\vert<1$. The orthonormal polynomial of order $m$ for the potential $V(z)=\vert z\vert^2-\Re(\t z^2)$ is
\begin{equation}\label{eq_op_herm}
\op{m}(z)=H_m(sz) \frac{\left(\n(1-\vert\t\vert^2)\right)^{\frac{m}{2}} \op{0}}{2^m s^m \sqrt{m!}}=H_m(sz) \frac{\cc{t}^{\frac{m}{2}} \op{0}}{2^{\frac{m}{2}}\sqrt{m!} }, \qquad z\in\C,
\end{equation}
where $s=\sqrt{\frac{\n(1-\vert \t\vert^2)}{2\cc{\t}}}$ and the zeroth order orthonormal polynomial is 
\begin{equation}
\op{0}=\sqrt{\tfrac{\n}{\pi}\sqrt{1-\vert\t\vert^2}}.
\end{equation} 
$H_m$ is the Hermite polynomial of order $m$. 
\end{prop}
For the Hermite polynomial we use the definition as in \cite[p. 91]{courant-hilbert} or \cite[p. 431]{hille} where $H_m$ is given by the generating function
\begin{equation}
\e^{2z u-u^2}=\sum_{m=0}^\infty H_m(z) \frac{u^m}{m!}.
\end{equation}
For more information regarding Hermite polynomials see also \cite{sze}.

\begin{remark}
In the limit $\t\rightarrow 0$ the orthonormal polynomials $\op{m}$ converge to the monomials of Proposition \ref{prop_gauss_onp}.
\end{remark}


\subsection{Proof of Orthonormality}
\begin{proof}[Proof of Proposition \ref{prop_op}]
We begin with the normalization of $\op{0}$. For this we have to calculate $\hn{0}=\ip{1}{1}$. We write $z=x+i y$ and as $\Re\, t<1$, we get
\begin{equation}
\begin{split}
V(z)&=x^2 (1-\Re\, t)+y^2(1+\Re\, t)+2xy \Im\, t\\
&=(1-\Re\, \t)\left(x+\frac{y \Im\, \t}{1-\Re\, t}\right)^2+y^2 \frac{1-\vert t\vert^2}{1-\Re\, t}.
\end{split}
\end{equation}
So we find
\begin{equation}
\begin{split}
\ip{1}{1}&=\int_\C \e^{-\n V(z)}\dd^2 z=\int_{-\infty}^\infty \int_{-\infty}^\infty \e^{-\n (1-\Re\, \t)\left(x+\frac{y \Im\, \t}{1-\Re\, \t}\right)^2 } \dd x \e^{-\n y^2 \frac{1-\vert \t\vert^2}{1-\Re\, \t}}  \dd y \\
&=\frac{1}{\sqrt{\n (1-\Re\, \t)}}\int_{-\infty}^\infty \e^{-x^2 } \dd x \sqrt{\frac{1- \Re\, \t }{\n(1-\vert t\vert^2)}}\int_{-\infty}^\infty \e^{-y^2 } \dd y= \frac{\pi}{\n \sqrt{1-\vert t\vert^2}}.
\end{split}
\end{equation}
Thus we have $\op{0}=\sqrt{\tfrac{\n}{\pi}\sqrt{1-\vert\t\vert^2}}$.

Now we define the polynomials
\begin{equation}\label{eq_tilde_p}
\tilde{p}_m(z)=H_m(sz) \frac{\left(\n(1-\vert\t\vert^2)\right)^{\frac{m}{2}} \op{0}}{2^m s^m \sqrt{m!}}, \qquad m\in\N,z\in \C,
\end{equation}
where $s$ is defined as in the proposition. We set $\tilde{p}_{-1}=0$. $\tilde{p}_0$ equals $\op{0}$ since $H_0=1$. We are going to check that $\tilde{p}_m$ fulfill the same recursion relation \eqref{eq_rec} as the orthonormal polynomials $\op{m}$. Together with $\tilde{p}_0=\op{0}$, this will prove the proposition.

For the Hermite polynomials we have the recursion relation (see \cite[p. 92]{courant-hilbert})
\begin{equation}\label{eq_hermite_rec}
H_{m+1}(z)-2z H_m(z)+2m H_{m-1}(z)=0, \qquad\forall m\in \N_0,
\end{equation}
where we have set $H_{-1}=0$. We fix $m\in\N_0$ and use \eqref{eq_tilde_p} for $\tilde{p}_m$ to get
\begin{equation}
\begin{split}
z \tilde{p}_m(z)&=2s z H_m(s z) \frac{\left(\n(1-\vert\t\vert^2)\right)^{\frac{m}{2}} \op{0}}{2^{m+1} s^{m+1} \sqrt{m!}}\\
&=\bigl(H_{m+1}(sz)+2m H_{m-1}(sz) \bigr) \frac{\left(\n(1-\vert\t\vert^2)\right)^{\frac{m}{2}} \op{0}}{2^{m+1} s^{m+1} \sqrt{m!}} \\
&=\frac{\sqrt{m+1}}{\sqrt{\n(1-\vert\t\vert^2)}} H_{m+1}(sz) \frac{\left(\n(1-\vert\t\vert^2)\right)^{\frac{m+1}{2}} \op{0}}{2^{m+1} s^{m+1} \sqrt{(m+1)!}} \\
&\phantom{=}+\frac{\sqrt{m \n(1-\vert\t\vert^2)}}{2s^2} H_{m-1}(sz) \frac{\left(\n(1-\vert\t\vert^2)\right)^{\frac{m-1}{2}} \op{0}}{2^{m-1} s^{m-1} \sqrt{(m-1)!}}\\
&=\rn{m+1} \tilde{p}_{m+1}(z)+\cc{t} \rn{m} \tilde{p}_{m-1}(z),
\end{split}
\end{equation}
where we have used \eqref{eq_hermite_rec} in the second line. As $m$ was arbitrary, $\tilde{p}_m$ fulfills the recursion relation \eqref{eq_rec} for all $m\in \N_0$, which proves the proposition.
\end{proof}


A more direct proof, without using the recursion relation, for a similar problem can be found in \cite{eijndhoven}. 

\subsection{Relation for the Derivative of $\op{m}$}
\begin{prop}\label{prop_d_op} 
The orthonormal polynomials $\op{m}$ fulfills the following relation.
\begin{equation}\label{eq_op_der}
\frac{\dd \op{m}}{\dd z}(z)=\n \rn{m}(1-\vert \t\vert^2) \op{m-1}(z)=\sqrt{\n m(1-\vert t\vert^2)}\op{m-1}(z), \qquad \forall m\in \N_0.
\end{equation}
\end{prop}
\begin{proof}
Let $m\in\N_0$. For the Hermite polynomials we have the relation (see \cite[p. 431]{hille}) 
\begin{equation}
H_m'(z)=2 m H_{m-1}(z).
\end{equation}
Using this relation for the $\op{m}$ from Proposition \ref{prop_op}, we get 
\begin{equation}
\frac{\dd \op{m}}{\dd z}(z)=H_{m-1}(sz) \frac{\sqrt{m}\left(\n(1-\vert\t\vert^2)\right)^{\frac{m}{2}} \op{0}}{2^{m-1} s^{m-1} \sqrt{(m-1)!}}=\sqrt{\n m(1-\vert t\vert^2)}\op{m-1}(z).
\end{equation}
\end{proof}
\begin{remark}
Another proof of such a relation can be found in \cite{elbau} for orthonormal polynomial for more general potentials $V$.
\end{remark}

\subsection{Zeros of the Orthonormal Polynomials}\label{sec_zeros}
\begin{prop}\label{prop_zeros}
Let $m\in \N$, $\F=\sqrt{\frac{4 \cc{\t}}{1-\vert \t\vert^2}}$ and $\F_m=\F \sqrt{\frac{m+\frac{1}{2}}{\n}}$. Then all of the $m$ zeros of $\op{m}$ are simple and lie in the interval $[-\F_m,\F_m]$ in the complex plane.
\end{prop}
\begin{proof}
$\op{m}(z)$ can only be zero when $H_m(sz)$ is zero, where $s\neq 0$ is like in Proposition \ref{prop_op}. The zeros of the Hermite polynomials are all simple and lie on the real axis (see \cite{sze}). Let $x_1<x_2< \cdots <x_m\in\R$ be the $m$ zeros of $H_m$. Because $H_m$ is a even or odd function we have that $x_k=-x_{m-k}$ for $k=1,\ldots,m$. There is an upper bound for the largest positive zero $x_m<\sqrt{2m+1}$ (see \cite[p. 129]{sze}). 

It follows that the zeros of $\op{m}$ lie at
\begin{align}
z_k&=s^{-1} x_k=\sqrt{\tfrac{2\cc{\t}}{\n (1-\vert t\vert^2)} } x_k, \qquad k=1,\ldots,m, 
\intertext{so they all lie on a line in the complex plane turned against the real line by the phase of $\sqrt{\cc{\t}}$. And we find for }
\vert z_m\vert &=\vert s^{-1} x_m \vert < \left\vert \sqrt{\tfrac{2\cc{\t}}{\n (1-\vert t\vert^2)} } \right\vert \sqrt{2m+1}=\vert \F\vert \sqrt{\tfrac{m+\frac{1}{2}}{\n}}.
\end{align}
Therefore $z_1,\ldots,z_m\in [-\F_m,\F_m]$.
\end{proof}

\begin{remark}
Note that there are better bounds for the largest zeros of the Hermite polynomials.
\end{remark}

\begin{remark}
All zeros of the orthonormal polynomials $(\op{k})_{k=1}^{\n}$ lie in $[-\F_\n,\F_\n]$ and $\F_\n$ tends to $\F$ when $\n\rightarrow \infty$.
\end{remark}

\begin{remark}
If $\phi\in (-\pi,\pi]$ is the argument of the parameter $\t$ then the argument of the zeros $z_1,\ldots,z_k$ is $-\tfrac{\phi}{2}$.
\end{remark}

\begin{remark}
That the zeros of $\op{m}$ are simple and lie on a line, follows also directly from the recursion relation \eqref{eq_rec}. For $m\ge2$ it shows further that between two neighboring zeros $z_k$ and $z_{k+1}$ (where $1\le k\le m-1$) lies exactly one zero of $\op{m-1}$.
\end{remark}

\subsection{Ellipse}
\begin{defin}
For $\vert \t \vert<1$ we define
\begin{align}
\mas=\sqrt{\tfrac{1+\vert \t\vert}{1-\vert\t\vert}}, \qquad \mis=\sqrt{\tfrac{1-\vert \t\vert}{1+\vert\t\vert}}, \qquad \text{and} \qquad \F=2\sqrt{\tfrac{\cc{\t}}{1-\vert \t\vert^2}}. \nn
\end{align}
\end{defin}


We know from calculation of the equilibrium measure (see \cite{felder-elbau}) --- which can be done explicitly in the Gaussian case --- that in the limit $\n\rightarrow\infty$ the eigenvalues will fill an ellipse with constant density, at least in a weak sense of convergence. The major and minor semi-axes of the ellipse are $\mas=\sqrt{\tfrac{1+\vert \t\vert}{1-\vert\t\vert}}$ and $\mis=\sqrt{\tfrac{1-\vert \t\vert}{1+\vert\t\vert}}$ respectively. So the foci $\pm \F$ lie at $\pm 2\sqrt{\tfrac{\cc{\t}}{1-\vert \t\vert^2}}$. Note that this $\F$ corresponds to the $\F$ we have already introduced in Proposition \ref{prop_zeros}, which therefore claims that the zeros of the orthonormal polynomials lie on the interval between the foci. 

In the following we concentrate on the case $\t>0$ so that the foci lie on the real axis. When $\t<0$ and the foci lie on the imaginary axis, the computations in the following sections could be done analogously.   
\begin{defin}
For $0<t<1$ we will call the inside of the ellipse
\begin{align}
\Et&=\left\{z\in \C \left\vert \tfrac{\left(\Re\, z\right)^2}{\mas^2} +\tfrac{\left(\Im\, z \right)^2}{\mis^2}<1 \right. \right\},
\intertext{and the ellipse itself}
\pEt&=\left\{z\in \C \left\vert \tfrac{\left(\Re\, z\right)^2}{\mas^2} +\tfrac{\left(\Im\, z\right)^2}{\mis^2}=1 \right. \right\}.
\end{align}
\end{defin}



\section{Identities for the Kernel of General Gaussian potentials}\label{sec_ident}
\sectionmark{Identities for the Kernel}

For simplicity we will now always assume that $\t$ is a real parameter. Then it is obvious from the recursion relation \eqref{eq_rec} or from the explicit formula \eqref{eq_op_herm} that the orthogonal polynomials have all real coefficients. 

\subsection{Identities for Unnormalized Kernel}
For simplicity of the calculations in the following sections we define the new quantity $\KM$ instead of the kernel $\Kn$ from Definition \ref{def_kernel}. 
\begin{defin}
Let $\n\in\N$. Then
\begin{equation}
\KM(w,z)=\frac{1}{\n}\sum_{m=0}^{\n-1} \op{m}(w)\op{m}(z), \qquad \forall w,z\in \C.
\end{equation}
\end{defin}
In contrast to $\Kn$, the new quantity $\KM$ comes without the complex conjugation of $\op{k}(w)$ in the sum defining the kernel. 

\begin{lemma}\label{lemma_km}
Let $\n\in\N$ and $w,z\in \C$. Then $\KM$ fulfills the identities
\begin{align}
\MoveEqLeft\frac{\partial \KM}{\partial w} (w,z)\pm\frac{\partial \KM}{\partial z}(w,z)\mp\n (1\mp\t)(w\pm z) \KM(w,z) \nn \\
&=\mp\sqrt{\frac{1\mp\t}{1\pm\t}}\bigl(\op{\n}(w) \op{\n-1}(z)\pm\op{\n-1}(w) \op{\n}(z)\bigr). \label{eq_id_km1}
\end{align}
\end{lemma}

For real $\t$, analogous identities in terms of $\Kn$ are given as $\KM$ is related to $\Kn$ by
\begin{equation}
\KM(\cc{w},z)=\frac{\Kn(w,z)}{\n},
\end{equation}
since $\cc{\op{m}(z)}=\op{m}(\cc{z})$ if $\t\in\R$.

\subsection{Identities for Normalized Kernel}
For the normalized kernel $\nK$ too, we define a new quantity $\HM$. 
\begin{defin}\label{def_hm}
Let $\n\in\N$. Then we define the ``pre-kernel'' $\HM$ as following.
\begin{equation}
\HM(w,z)=\e^{\n\left(-w z+\frac{\t w^2}{2}+\frac{\t z^2}{2}\right)} \KM =\frac{1}{\n}\e^{\n\left(-w z+\frac{\t w^2}{2}+\frac{\t z^2}{2}\right)}\sum_{m=0}^{\n-1} \op{m}(w)\op{m}(z), \qquad \forall w,z\in\C.
\end{equation}
\end{defin}
As $\t$ is real, $\HM$ is related to $\nK$ by
\begin{equation}\label{eq_HM_tKn}
\begin{split}
\HM(\cc{w},z)&=\frac{\nK(w,z)}{\n} e^{\n\left(\frac{\vert w \vert^2+\vert z \vert^2}{2}-\cc{w}z+\frac{\t\left(\cc{w}^2-w^2+z^2-\cc{z}^2\right)}{4}\right)}\\
&= \frac{\nK(w,z)}{\n}e^{\frac{\n}{2}\left(\vert w-z\vert^2 +w\cc{z}-\cc{w}z\right)+\frac{\n\t}{4}\left((z+w)(z-w)+(\cc{w}+\cc{z})(\cc{w}-\cc{z})\right)}. 
\end{split}
\end{equation}

\begin{prop}\label{prop_id}
Let $\n\in\N$ and $w,z\in\C$. Then $\HM$ fulfills the identities
\begin{align}
\frac{\partial \HM}{\partial w} (w,z)+\frac{\partial \HM}{\partial z} (w,z) &=-\sqrt{\frac{1-\t}{1+\t}}\bigl(\op{\n}(w) \op{\n-1}(z)+\op{\n-1}(w) \op{\n}(z)\bigr)\e^{\n\left(-w z+\frac{\t w^2}{2}+\frac{\t z^2}{2}\right)}, \label{eq_id1}
\intertext{and}
\frac{\partial \HM}{\partial w} (w,z)-\frac{\partial \HM}{\partial z} (w,z)&=\sqrt{\frac{1+\t}{1-\t}}\bigl(\op{\n}(w) \op{\n-1}(z)-\op{\n-1}(w) \op{\n}(z)\bigr)\e^{\n\left(-w z+\frac{\t w^2}{2}+\frac{\t z^2}{2}\right)}. \label{eq_id2}
\end{align}
\end{prop}

\begin{remark}
The identities in terms of $\nK$ can be found by the relation \eqref{eq_HM_tKn} between $\HM$ and $\nK$ for real $\t$. We abstain from stating them explicitly as it is simpler to do the following calculation for $\HM$ instead of $\nK$. Only at the end, in section \ref{subsec_asym_kernel}, having found the final result for $\HM$, we will turn back to $\nK$.
\end{remark}

\begin{remark}
$\HM(w,z)$ is a holomorphic function in both variables.
\end{remark}

\begin{remark}
Note that for real $\t$ the density is given by $\rho_\n(z)=\HM(\cc{z},z)$.
\end{remark}

\subsection{Similarity to Christoffel-Darboux Formula}
Here we want to state some facts about Hermite random matrix models and their orthonormal polynomials on the real line. Then we want to compare the Christoffel-Darboux formula to the identities we have found.

The polynomials $q_k$, $k\in\N_0$, of order $k$ with leading coefficient one which are orthogonal regarding the inner product
\begin{equation}\label{eq_inner_product_herm}
\ip{q_k}{q_l}_{\mathrm{Herm}}=\int_{-\infty}^\infty q_k(x) q_l(x) w(x)\dd x,
\end{equation}
we will call the orthogonal polynomial associated to the weight function $w$ which has to be so that the integral in \eqref{eq_inner_product_herm} converges. As in our case the weight function is typically of the form $w(x)=\e^{-n V(x)}$ for a real-valued potential $V$ (see for example \cite{bi}). Let
\begin{align}
h_k&=\ip{q_k}{q_k}_{\mathrm{Herm}}, \qquad k\in\N_0.
\shortintertext{Then we can define the orthonormal polynomial $p_k$ associated to the weight function $w$ by}
p_k(x)&=\frac{q_k(x)}{\sqrt{h_k}}.
\end{align}

Then there exist some coefficients $A_k, B_k, C_k\in \R$, $k\in\N$, so that the orthonormal polynomials fulfill the recurrence relation (see \cite{sze})
\begin{equation}\label{eq_rec_hermite}
p_k(x)=(A_k x+ B_k)p_{k-1}(x)-C_k p_{k-2}(x), \qquad \forall k\in \N, x\in\R,
\end{equation}
where we have set $p_{-1}=0$. $A_k$ and $C_k$ are given by
\begin{align}
A_k&=\sqrt{\frac{h_{k-1}}{h_k}}>0,\qquad \text{and}\qquad C_k=\frac{h_{k-1}}{\sqrt{h_k h_{k-2}}}, \qquad k\in \N, \nn
\end{align}
and $B_k=0$ for all $k\in\N$ if $w$ is an even function.

The Christoffel-Darboux formula follows from the recurrence relation (see \cite{sze})
\begin{equation}
\sum_{k=0}^{\n-1} p_k(x) p_k(y)=\sqrt{\frac{h_\n}{h_{\n-1}}} \frac{p_\n(x)p_{\n-1}(y)-p_{\n-1}(x)p_\n(y)}{x-y}, \qquad \forall \n\in\N,x,y\in\R,
\end{equation}
where the left-hand side, $\sum_{k=0}^{\n-1} p_k(x) p_k(y)$, is the unnormalized reproducing kernel for the Hermitian case analogous to Definition \ref{def_kernel}. With the help of this formula one can find the asymptotics of the kernel by the asymptotics of the orthonormal polynomials (see \cite{bi}).

Our identities \eqref{eq_id_km1}, \eqref{eq_id1} and \eqref{eq_id2} resemble the Christoffel-Darboux formula in the sense that on the right-hand side only orthonormal polynomials of order $\n$ and $\n-1$ are involved and the left-hand side has a relation to the reproducing kernel. While the Christoffel-Darboux formula in the Hermitian case gives directly the reproducing kernel, our identities only give a relation for the derivative of the reproducing kernel. 


\subsection{Limit $\t\rightarrow 0$}
As the orthonormal polynomials converge 
to the monomials of Proposition \ref{prop_gauss_onp} in the limit $\t\rightarrow 0$, we see that also identities \eqref{eq_id1} and \eqref{eq_id2} converge to the corresponding identities \eqref{eq_id_0} of the pure Gaussian case $\t=0$ as expected. 

\subsection{Proof of Identities}
\begin{lemma}\label{lemma_op_rel}
For all $k\in\N$ the orthonormal polynomials $\op{k}$ fulfill the relations
\begin{align}
\MoveEqLeft[3]\mp\sqrt{\frac{k}{\n}} \sqrt{\frac{1\mp\t}{1\pm\t}}\bigl(\op{k}(w) \op{k-1}(z)\pm\op{k-1}(w) \op{k}(z)\bigr)\nn\\
= {} &\mp\sqrt{\frac{k-1}{\n}} \sqrt{\frac{1\mp\t}{1\pm\t}}\bigl(\op{k-1}(w) \op{k-2}(z)\pm \op{k-2}(w) \op{k-1}(z)\bigr)\nn\\
& \mp(1\mp\t)(w\pm z) \op{k-1}(w)\op{k-1}(z)+\frac{1}{\n} \bigl( \dop{k-1}(w)\op{k-1}(z) \pm \op{k-1}(w)\dop{k-1}(z)\bigr), \label{eq_lemma_op_rel1}
\end{align}
where we have set $\op{-1}=0.$
\end{lemma}
\begin{proof}
Using the recursion relation \eqref{eq_rec} on the left-hand side of \eqref{eq_lemma_op_rel1}, we find
\begin{align}
\MoveEqLeft[3]\mp\sqrt{\tfrac{k}{\n}} \sqrt{\tfrac{1\mp\t}{1\pm\t}}\bigl(\op{k}(w) \op{k-1}(z)\pm\op{k-1}(w) \op{k}(z)\bigr)\nn \\
= {} &\mp\sqrt{\tfrac{k}{\n}} \sqrt{\tfrac{1\mp\t}{1\pm\t}}\left( \tfrac{w\op{k-1}(w)-\t \rn{k-1} \op{k-2}(w)}{\rn{k}}\op{k-1}(z) \pm \op{k-1}(w) \tfrac{z\op{k-1}(z)-\t \rn{k-1} \op{k-2}(z)}{\rn{k}} \right) \nn\\
= {} &\mp(1\mp\t)(w\pm z) \op{k-1}(w)\op{k-1}(z) \nn \\ 
& +\underbrace{(\t\mp\t^2)}_{\mathclap{\pm(1-\t^2)\mp(1\mp\t)}} \rn{k-1}\bigl(\op{k-1}(w) \op{k-2}(z)\pm\op{k-2}(w) \op{k-1}(z)\bigr)\nn\\
= {} &\mp(1\mp\t)(w\pm z) \op{k-1}(w)\op{k-1}(z)+\tfrac{1}{\n} \bigl( \dop{k-1}(w)\op{k-1}(z)\pm\op{k-1}(w)\dop{k-1}(z)\bigr)\nn\\
& \mp\sqrt{\tfrac{k-1}{\n}} \sqrt{\tfrac{1\mp\t}{1\pm\t}}\bigl(\op{k-1}(w) \op{k-2}(z)\pm\op{k-2}(w) \op{k-1}(z)\bigr),
\end{align}
where we have used the relation from Proposition \ref{prop_d_op} to replace $\op{k-2}$ by $\dop{k-1}$ in the term $$\pm(1-\t^2)\rn{k-1}\bigl(\op{k-1}(w) \op{k-2}(z)\pm\op{k-2}(w) \op{k-1}(z)\bigr)$$
of the third line.
\end{proof}

\begin{proof}[Proof of Lemma \ref{lemma_km}]
Using Lemma \ref{lemma_op_rel} recursively for $k=\n,\n-1,\ldots,1$ on the right-hand side of \eqref{eq_id_km1} we get
\begin{align}
\MoveEqLeft[3]\mp\sqrt{\frac{1\mp\t}{1\pm\t}}\bigl(\op{\n}(w) \op{\n-1}(z)\pm\op{\n-1}(w) \op{\n}(z)\bigr)\nn\\
= {} & \mp(1\mp \t)(w\pm z)\sum_{k=1}^\n \op{k-1}(w)\op{k-1}(z)\nn\\
&+\frac{1}{\n}\sum_{k=1}^\n \left(\dop{k-1}(w)\op{k-1}(z)\pm\op{k-1}(w)\dop{k-1}(z)\right) \nn\\
= {} & \mp\n (1\mp\t)(w\pm z) \KM(w,z)+\frac{\partial \KM}{\partial w} (w,z)\pm\frac{\partial \KM}{\partial z}(w,z).
\end{align}
\end{proof}

\begin{remark}
The proof has some similarity with the proof of the Christoffel-Darboux formula. The same way we use recursively Lemma \ref{lemma_km}, the Christoffel-Darboux formula can be proved by recursively using the recurrence relation \eqref{eq_rec_hermite}.
\end{remark}

\begin{proof}[Proof of Proposition \ref{prop_id}]
\begin{align}
\MoveEqLeft\frac{\partial \HM}{\partial w}(w,z)\pm\frac{\partial \HM}{\partial z}(w,z)\nn\\
& =\left( \frac{\partial \KM}{\partial w}(w,z)\pm\frac{\partial \KM}{\partial z}(w,z) \mp\n(1\mp\t)(w\pm z) \KM(w,z) \right) \e^{\n\left(-w z+\frac{\t w^2}{2}+\frac{\t z^2}{2}\right)}\nn\\
& =\mp\sqrt{\frac{1\mp\t}{1\pm\t}}\bigl(\op{\n}(w) \op{\n-1}(z)\pm\op{\n-1}(w) \op{\n}(z)\bigr)\e^{\n\left(-w z+\frac{\t w^2}{2}+\frac{\t z^2}{2}\right)},
\end{align}
where we have used Lemma \ref{lemma_km} in the last step.
\end{proof}

\subsection{Symmetries of Identities}\label{subsec_symm_id}
\begin{defin}
In the following we will use for the first quadrant the notation 
\begin{align}
\Cpp&=\{z\in \C\vert \Re\, z\ge 0 \land \Im\, z\ge 0\},\\
\intertext{and for the right-hand side of the complex plane without the negative imaginary axis}
\Cp&=\{z\in \C\vert \Re\, z> 0 \lor (\Re\, z=0 \land \Im\, z\ge 0\}.
\end{align}
\end{defin}


From the right-hand side of identities \eqref{eq_id1} and \eqref{eq_id2} we see that they change the sign under the map $(w,z)\mapsto (-w,-z)$ as exactly one of $\op{\n}$ and $\op{\n-1}$ is odd and that they get complex conjugated under $(w,z)\mapsto (\cc{w},\cc{z})$. 

Our main goal will be to find $\HM$ and $\tKn$, from which we can deduce the results we are looking for. Our interest is limited to the case when $w$ and $z$ are in a neighborhood of $z_0\in \C$. Because of the symmetries it is sufficient to consider only $z_0\in \Cpp$. 
But we will see that most of the time it won't be a big deal to consider all $z_0\in \C$. 

\section{Asymptotics of Identities}\label{ch_asymp}

In this chapter we are going to analyze the asymptotics as $\n\rightarrow \infty$ of identities \eqref{eq_id1} and \eqref{eq_id2} evaluated at $w=\cc{z}$. As before $t \in\R$ and for simplicity, we further assume that $t>0$. If not stated otherwise, we will assume this in the following chapters too. 

If we introduce, additionally to Definition \ref{def_rho}, the density 
\begin{align}
\rhon(x,y)&=\rhon(x+ i y)=\frac{\nK(x+i y,x+i y)}{\n}=\HM(x-i y,x+i y)
\intertext{as a function of two real variables $x$ and $y$,  then}
\frac{\partial \rhon}{\partial x}(x,y)&=\frac{\partial \HM}{\partial w}(x-i y,x+i y)+\frac{\partial \HM}{\partial z}(x-i y,x+i y) \label{eq_drho_x}
\shortintertext{corresponds to the left-hand side of the identity \eqref{eq_id1} evaluated at $w=x-iy$ and $z=x+iy$ and}
\frac{\partial \rhon}{\partial y}(x,y)&=-i \frac{\partial \HM}{\partial w}(x-i y,x+i y)+ i \frac{\partial \HM}{\partial z}(x-i y,x+i y) \label{eq_drho_y}
\end{align}
to the left-hand side of identity \eqref{eq_id2} respectively, multiplied by $-i$. Therefore the asymptotics of the identities will reveal the asymptotics of the derivatives of $\rhon(x,y)$.

To approximate the orthonormal polynomials we will use the Plancherel-Rotach asymptotics for the Hermite polynomials. If we are on the interval $(-\sqrt{2},\sqrt{2})$ a different asymptotics is valid than outside of the interval. We will review the Plancherel-Rotach asymptotics in these two cases in Section \ref{sec_pr}. 

First we will define and discuss some functions which help to shorten the asymptotic expressions and their computations.

\subsection{Some Useful Functions}
\subsubsection{Definitions of $\U{\F}$, $\T{\F}$ and $\W{\F}$}

\begin{defin}\label{def_U} For $\F=2\sqrt{\tfrac{\t}{1-\t^2}}\in (0,\infty)$ and $z\in\C\setminus[-\F,\F]$ we define
\begin{equation}
\U{\F}(z)=\left\{ \begin{array}{ll}\frac{z-\sqrt{z^2-\F^2}}{\F} &\text{if } z\in \Cp,\\ \frac{z+\sqrt{z^2-\F^2}}{\F} &\text{if } z\notin \Cp. \end{array}\right. 
\end{equation}
\end{defin}

\begin{remark}We use the standard convention for the root, i.e.\ $\sqrt{\e^{i\phi}}=\e^{i\frac{\phi}{2}}$ for $\phi\in (-\pi,\pi]$. 
\end{remark}

\begin{remark}
Note that $\U{\F}(-z)=-\U{\F}(z)$ and $\U{\F}(\cc{z})=\cc{\U{\F}(z)}$ for all $z\in\C\setminus[-\F,\F]$. 
\end{remark}

In Section \ref{subsection_pr_asym_outside} we will see that this function appears in the Plancherel-Rotach asymptotics which is valid outside of the interval $[-\F,\F]$. It will play a central role in the dominant term of the identities (see Section \ref{sec_f_g}).


\begin{defin}\label{def_T} For $\F\in (0,\infty)$ and $z\in\C\setminus[-\F,\F]$ we define
\begin{equation}
\T{\F}(z)=\left\{ \begin{array}{ll}\sqrt{z^2-\F^2} &\text{if } z\in \Cp,\\ -\sqrt{z^2-\F^2} &\text{if } z\notin \Cp. \end{array}\right. 
\end{equation}
\end{defin}

\begin{remark}
Note that $\T{\F}(-z)=-\T{\F}(z)$ and $\T{\F}(\cc{z})=\cc{\T{\F}(z)}$ for all $z\in\C\setminus[-\F,\F]$. 
\end{remark}

Additionally we will use $\U{\F}(x,y)=\U{\F}(x+ i y)$ and $\T{\F}(x,y)=\T{\F}(x+ i y)$ as functions of two real variables $x$ and $y$. We will see in Section \ref{subsec_holomorph} that $\U{\F}$ and $\T{\F}$ are holomorphic on $\C\setminus[\F,\F]$. The derivatives of $\T{\F}(x,y)=\pm \sqrt{x^2-y^2+2i xy -F^2}$ are
\begin{align}
\tfrac{\partial \T{\F}}{\partial x} (x,y)&=\tfrac{x+i y}{\pm \sqrt{x^2-y^2+2i xy -F^2}} =  \tfrac{x+i y}{\T{\F}(x,y)} , &  \tfrac{\partial \T{\F}}{\partial y} (x,y)&=\tfrac{i x- y}{\pm \sqrt{x^2-y^2+2i xy -F^2}} = i\tfrac{x+i y}{\T{\F}(x,y)}, \nn \\
\intertext{and of $\U{\F}(x,y)=\tfrac{x+i y - \T{\F}(x,y)}{\F}$}
\tfrac{\partial \U{\F}}{\partial x} (x,y)&=\tfrac{1}{\F}-\tfrac{x+i y}{\F \T{\F}(x,y)}=-\tfrac{\U{\F}(x,y)}{\T{\F}(x,y)} , &  \tfrac{\partial \U{\F}}{\partial y} (x,y)&=\tfrac{i}{\F}-\tfrac{i(x+i y)}{\F \T{\F}(x,y)}=\tfrac{\U{\F}(x,y)}{i \T{\F}(x,y)}. \nn
\end{align}

\begin{defin}\label{def_W} For $\F\in [0,\infty)$ and $z\in\C\setminus[-\F,\F]$ we define
\begin{equation}
\W{\F}(z)=\left( \frac{z+\F}{z-\F}\right)^{1/4}+\left( \frac{z-\F}{z+\F}\right)^{1/4} .
\end{equation}
\end{defin}

\begin{remark}Note that $\W{\F}(-z)=\W{\F}(z)$ and $\W{\F}(\cc{z})=\cc{\W{\F}(z)}$ for all $z\in\C\setminus[-\F,\F]$. The latter is true because $z^{1/4}$ has this symmetry under complex conjugation unless $z\in(-\infty,0)$. But $\tfrac{z\pm\F}{z\mp\F}\notin (-\infty,0]$ for all $z\in\C\setminus[-\F,\F]$ as we will see in the proof of Lemma \ref{lemma_holomorh_U_W}. 
\end{remark}

\subsubsection{Elliptic Coordinates}\label{sub_elliptic}
The simplest way to describe $\U{\F}$, $\T{\F}$ and $\W{\F}$ is to use elliptic coordinates on confocal ellipses with minor semi-axis $b>0$ and major semi-axis $a=\sqrt{b^2+\F^2}>b$. 

Then $\pm\F$ with $\F\in (0,\infty)$ are the common foci of all of these ellipses. For any $z\in \C\setminus[-\F,\F]$ we can find such an ellipse so that $z$ lies on it. For a chosen $z=x+i y$ the semi-axis $b$ is given by
\begin{align}
b^2=\frac{x^2+y^2-\F^2+\sqrt{4y^2 \F^2+(\F^2-x^2-y^2)^2}}{2}>0, \label{eq_coord_b}
\end{align}
since then
\begin{align}
\frac{x^2}{a^2}+\frac{y^2}{b^2}&=\frac{4x^2 b^2+4y^2 a^2}{4a^2 b^2} \nn \\
&=\tfrac{2x^2\left(x^2+y^2-\F^2+\sqrt{4y^2 \F^2+(\F^2-x^2-y^2)^2} \right)+2y^2\left(x^2+y^2+\F^2+\sqrt{4y^2 \F^2+(\F^2-x^2-y^2)^2} \right)}{\left(x^2+y^2+\sqrt{4y^2 \F^2+(\F^2-x^2-y^2)^2}\right)^2-\F^4} \nn \\
&=\tfrac{2(x^2+y^2)\left(x^2+y^2+\sqrt{4y^2 \F^2+(\F^2-x^2-y^2)^2}\right)+2(y^2-x^2)\F^2}{x^4+y^4+4y^2 \F^2+(\F^2-x^2-y^2)^2+2x^2 y^2+(2x^2+2y^2)\sqrt{4y^2 \F^2+(\F^2-x^2-y^2)^2}-\F^4} = 1 \label{eq_z_on_ellipse}
\end{align}
is fulfilled. Obviously for any given $z\in\C\setminus[-\F,\F]$ we find exactly one ellipse containing $z$. Thus we have found unique coordinates $(b,\phi)\in (0,\infty)\times(-\pi,\pi]$ so that $z=a \cos\phi+i b \sin\phi$.

\begin{prop}\label{prop_U}
Let $\F\in(0,\infty)$, $b>0$, $a=\sqrt{b^2+\F^2}>b$ and $z=a \cos\phi+i b \sin\phi $ with $\phi\in(\pi,\pi]$. Then 
\begin{align}
\U{\F}(z)&=r_F(b) (\cos \phi-i \sin \phi),
\intertext{where}
r_\F(b)&=\frac{\sqrt{b^2+\F^2}-b}{\F}=\frac{a-b}{\F}>0. 
\end{align}
Further
\begin{align}
\Re\, \U{\F}(z)&= r_F(b) \cos\phi, \quad \Im\, \U{\F}(z)= -r_F(b) \sin \phi, \quad\text{and}\quad \vert \U{\F}(z)\vert=r_F(b)>0.
\end{align}
\end{prop}
\begin{proof}
Using the parametrization $z=a \cos\phi+i b \sin\phi $, $\F^2=\F^2 (\sin^2\phi+\cos^2\phi)$ and $a^2=b^2+\F^2$ we get
\begin{align}
z^2-\F^2&=(a^2-\F^2)\cos^2\phi +2i ab \cos\phi\sin\phi-(b^2+\F^2)\sin^2\phi\nonumber\\
&=b^2 \cos^2\phi +2i ab \cos\phi \sin\phi-a^2 \sin^2\phi\nonumber\\
&=(b \cos\phi+i a \sin\phi)^2.
\end{align}
Inserting this in $\U{\F}$ we find
\begin{equation}
\U{\F}(z)=\frac{(a-b) \cos\phi+i (b-a) \sin\phi}{\F},
\end{equation}
where we have considered the change of sign in front of the square root in $\U{\F}$. The rest of the proposition is obvious from that.
\end{proof}

\begin{remark}
We see that $\Re\, \U{\F}(z)\ge 0$ if $\phi\in[-\tfrac{\pi}{2}, \tfrac{\pi}{2}]$ and therefore $-\tfrac{\pi}{2} \le \arg \U{\F}(z)\le \tfrac{\pi}{2}$. 
\end{remark}

\begin{cor}\label{cor_U1}
Let $\F\in(0,\infty)$, $b>0$ and $a=\sqrt{b^2+\F^2}>b$. Then $\U{F}$ maps the ellipse
\begin{align}
&\{z\in\C \vert z=a \cos\phi+i b \sin\phi, \phi\in (\pi,\pi] \}\\
\intertext{to the circle}
&\{z\in\C\vert z=r_F(b) (\cos \phi-i \sin \phi), \phi\in (\pi,\pi] \},
\end{align}
where the radius $r_\F(b)$ is given in Proposition \ref{prop_U}.
\end{cor}
\begin{proof}
This follows directly from Proposition \ref{prop_U}.
\end{proof}
\begin{remark}
Note that the part of the ellipse in the first quadrant gets mapped to the part of the circle in the forth quadrant and vice versa. The positive imaginary axis gets mapped to $(-i,0)$.
\end{remark}

\begin{remark}
If $b=\mis$ it follows that $a=\mas$. Therefore the ellipse $\pEt$ gets mapped by $\U{\F}$ to a circle with radius $r_\F(\mis)=\sqrt{\t}$.
\end{remark}


\begin{cor}\label{cor_U2}
Let $\F\in(0,\infty$). Then $\U{F}$ maps $\C\setminus [-F,F]$ bijectively to $B_1(0)\setminus\{0\}$
and infinity gets mapped to zero. 
\end{cor}
\begin{proof}
Every $z\in \C\setminus [-\F,\F]$ lies on an ellipse with minor semi-axis $b\in(0,\infty)$ and major semi-axis $a=\sqrt{b^2-\F^2}$. 
From Proposition \ref{prop_U} we see that $\U{\F}$ maps every point $z=a \cos\phi+i b\sin \phi$ of it bijectively to a point on the circle $\U{\F}(z)=r_\F(b)(\cos \phi+i\sin(-\phi))$ with $r_\F(b)=\frac{\sqrt{b^2+\F^2}-b}{\F}$ and $\phi\in (-\pi,\pi]$. We see that $r_\F(b)$ is continuous and that
\begin{align}
\lim_{b\rightarrow 0} r_\F(b)=1, \qquad\text{and}\qquad \lim_{b\rightarrow \infty} r_\F(b)=0.
\end{align}
Because
\begin{align}
\frac{\partial r_\F(b)}{\partial b}=\frac{b-\sqrt{b^2+\F^2}}{\F \sqrt{b^2+\F^2}}<0, \qquad \forall b\in(0,\infty),
\end{align}
$r_\F(b)$ is strictly monotonic decreasing from one to zero. Therefore $\U{\F}$ obviously is injective. We further see that it maps $\C\setminus [-\F,\F]$ to $B_1(0)\setminus\{0\}$. On the other hand, if $w \in B_1(0)\setminus\{0\}$, there exists exactly one $b$ with $r_\F(b)=\vert w\vert$. And then again, we find a point on the ellipse with semi-axis $b$ that gets mapped by $\U{\F}$ to $w$, thus $\U{\F}$ is surjective.
\end{proof}

Since $\U{\F}$ is bijective, we know that there exists an inverse function.
\begin{prop}\label{prop_U_inv}
Let $\F\in(0,\infty)$. Then the inverse function of $\U{\F}$ is given by
\begin{align}
\U{\F}^{-1}(u)=\frac{\F}{2}\left(u+\frac{1}{u}\right), \qquad \forall u\in B_1(0)\setminus\{0\}.
\end{align}
\end{prop}
\begin{proof}
We will show that
\begin{align}
\tfrac{\F}{2}\left(\U{\F}(z)+\tfrac{1}{\U{\F}(z)}\right)=z, \qquad \forall z\in \C\setminus[-\F,\F].
\end{align}
Expanding the fraction $\tfrac{1}{\U{\F}(z)}$ with $z\pm\sqrt{z^2-\F^2}$ we get
\begin{align*}
\F\left(\U{\F}(z)+\tfrac{1}{\U{\F}(z)}\right)=z\mp \sqrt{z^2-\F^2}+\tfrac{\F^2}{z\mp\sqrt{z^2-\F^2}}=z\mp \sqrt{z^2-\F^2}+\tfrac{\F^2(z\pm\sqrt{z^2-\F^2})}{z^2-z^2+\F^2}=2z,
\end{align*}
where the upper sign has been for $z\in\Cp$ and the lower sign for $z\notin\Cp$.
\end{proof}

\begin{lemma}\label{lemma_T}
Let $z\in \C\setminus [-\F,\F]$ be parametrized by $z=a \cos\phi+i b \sin\phi $ with $b>0$, $a=\sqrt{b^2+\F^2}>b$ and $\phi\in(-\pi,\pi]$. Then 
\begin{align}
\T{\F}(z)&= b\cos \phi+ i a \sin \phi,
\intertext{and}
\vert \T{\F}(z)\vert^2&=b^2+\F^2\sin^2\phi >0.
\end{align}
\end{lemma}
\begin{proof}
Using the parametrization $z=a \cos\phi+i b \sin\phi $ and $a^2=b^2+\F^2$, we get
\begin{align}
z^2-\F^2&=(a^2-\F^2)\cos^2\phi +2i ab \cos\phi\sin\phi-(b^2+\F^2)\sin^2\phi\nonumber\\
&=b^2 \cos^2\phi +2i ab \cos\phi \sin\phi-a^2 \sin^2\phi\nonumber\\
&=(b \cos\phi+i a \sin\phi)^2.
\end{align}
Inserting this in $\T{\F}$, we find
\begin{equation}
\T{\F}(z)=b \cos\phi+i a \sin\phi,
\end{equation}
where we have considered the change of sign in front of the root in $\T{\F}$. The rest of the proposition is obvious from that.
\end{proof}

\begin{lemma}\label{lemma_W} 
Let $z\in \C\setminus [-\F,\F]$ be parametrized by $z=a \cos\phi+i b \sin\phi $ with $b>0$, $a=\sqrt{b^2+\F^2}>b$ and $\phi\in(-\pi,\pi]$. Then 
\begin{align}
\W{\F}(z)&=\frac{\sqrt{2(a+b)} \e^{i\frac{\phi}{2}} }{\sqrt{b \cos \phi+ i a\sin \phi}}. 
\end{align}
Further
\begin{align}
\vert \W{\F}(z)\vert^2&= \frac{2(a+b)}{\sqrt{b^2+\F^2 \sin^2\phi}}>0, \qquad \text{and}\qquad \arg \W{\F}(z)\in (-\tfrac{\pi}{4},\tfrac{\pi}{4}]. \nn
\end{align}
\end{lemma}
\begin{proof}
Using the parametrization 
\begin{align}
z&=a \cos\phi+i b \sin\phi= \left(\tfrac{a+b}{2}+\tfrac{a-b}{2}\right)\cos \phi+i \left(\tfrac{a+b}{2}-\tfrac{a-b}{2}\right) \sin\phi\nn\\
&=\tfrac{a+b}{2} \e^{i\phi}+ \tfrac{a-b}{2} \e^{-i\phi}=\left( \sqrt{\tfrac{a+b}{2}} \e^{i\frac{\phi}{2}}\pm \sqrt{\tfrac{a-b}{2}} \e^{-i\frac{\phi}{2}}\right)^2\mp\sqrt{a^2-b^2},
\end{align}
and $a^2=b^2+\F^2$, we get
\begin{align}
z\pm \F &=\left( \sqrt{\tfrac{a+b}{2}} \e^{i\frac{\phi}{2}}\pm \sqrt{\tfrac{a-b}{2}} \e^{-i\frac{\phi}{2}}\right)^2.
\end{align}
Inserting this in $\W{\F}$, we find
\begin{align}
\W{\F}(z)&= \left( \tfrac{ \sqrt{\tfrac{a+b}{2}} \e^{i\frac{\phi}{2}}+ \sqrt{\tfrac{a-b}{2}} \e^{-i\frac{\phi}{2}} }{ \sqrt{\tfrac{a+b}{2}} \e^{i\frac{\phi}{2}}- \sqrt{\tfrac{a-b}{2}} \e^{-i\frac{\phi}{2}}} \right)^{1/2}+\left( \tfrac{ \sqrt{\tfrac{a+b}{2}} \e^{i\frac{\phi}{2}}- \sqrt{\tfrac{a-b}{2}} \e^{-i\frac{\phi}{2}} }{ \sqrt{\tfrac{a+b}{2}} \e^{i\frac{\phi}{2}}+ \sqrt{\tfrac{a-b}{2}} \e^{-i\frac{\phi}{2}}} \right)^{1/2}\nn \\
&=\frac{\sqrt{2(a+b)} \e^{i\frac{\phi}{2}} }{ \sqrt{ \tfrac{a+b}{2}\e^{i\phi} -\tfrac{a-b}{2} \e^{-i\phi} } } .
\intertext{From this we can also see that}
\vert \W{\F}(z) \vert^2&=\frac{\vert \sqrt{2(a+b)} \e^{i\frac{\phi}{2}} \vert^2}{ \vert \sqrt{b\cos\phi+i a\sin \phi  } \vert^2 } = \frac{2(a+b)}{\sqrt{b^2+\F^2 \sin^2\phi}}>0.
\end{align}
The last claim is obvious as the argument of the root in the standard convention fulfills $-\pi/4 < \arg \left(z^{1/4}\right) \le \pi/4$ for all $z\in \C\setminus \{0\}$.
\end{proof}

\subsubsection{Domain of Holomorphy of $\U{\F}$, $\T{\F}$ and $\W{\F}$} \label{subsec_holomorph}
\begin{lemma}\label{lemma_holomorh_U_W}
Let $\F \in (0,\infty)$. Then $\U{\F}$, $\T{\F}$ and $\W{\F}$ are holomorphic on $\C\setminus [-F,F]$.
\end{lemma}
\begin{proof}
We can see $\W{\F}$ as consisting of compositions of the functions 
\begin{align}
&\frac{z\pm \F}{z \mp \F}, \qquad \text{and}\qquad z^{1/4}. \nn
\end{align}
The first ones are holomorphic if $z\neq \pm \F$ and the second one if $z\notin (-\infty,0]$. But 
\begin{align}
\frac{z\pm \F}{z \mp \F}= \frac{(z\pm \F)(\cc{z}\mp \F ) }{(z \mp \F)(\cc{z}\mp \F )}=\frac{\vert z\vert^2\pm 2i\F\Im\, z-\F^2}{\vert z\vert^2 \pm 2F\Re\, z +\F^2}
\end{align}
is only real if $\Im\, z=0$. For $z\in\R$ we see that $\tfrac{z\pm \F}{z \mp \F}\in (-\infty,0)$ is equivalent to $z\in (-\F,\F)$, and $\tfrac{z\pm \F}{z \mp \F}=0$ only if $z=\mp \F$. Therefore $\W{\F}$ is holomorphic if $z\notin [-\F,\F]$.

\begin{align}
\T{\F}(z)=\pm \sqrt{z^2-\F^2}=\pm \sqrt{z^2}\sqrt{1-\tfrac{\F^2}{z^2}}= z \sqrt{1-\tfrac{\F^2}{z^2}}
\end{align}
is holomorphic if either $1-\tfrac{\F^2}{z^2}$ is not real or if $1-\tfrac{\F^2}{z^2}$ is bigger than zero. The first one is fulfilled if $z\notin \R\cup i\R$. If $1-\tfrac{\F^2}{z^2}\in \R$ the second condition is fulfilled if either $z^2<0$, i.e.\ $z\in i\R$, or $\F^2<z^2$, i.e.\ $z\in \R\setminus [-\F,\F]$. Therefore $\T{\F}$ is holomorphic on $\C$ everywhere except on $[-\F,\F]$.


Obviously $\U{\F}$ is holomorphic where $\T{\F}$ is holomorphic.
\end{proof}

\subsubsection{Limit $\t\rightarrow 0$ of $\U{\F}$, $\F^{-1} \U{\F}$ and $\W{\F}$}
\begin{lemma}\label{lemma_limit_U_W} 
Let $z\in \C\setminus\{0\}$. Then
\begin{align}
\lim_{\F\rightarrow 0} \U{F}(z)&=0, &\lim_{\F\rightarrow 0} \tfrac{\U{F}(z)}{\F}&=\tfrac{1}{2 z}, &\lim_{\F\rightarrow 0}\W{F}(z)&=\W{0}(z)=2. \nn
\end{align}
\end{lemma}
\begin{remark}
Note that the limits do not hold uniformly in the neighborhood of $z=0$.
\end{remark}

\begin{proof}
Choose $z\in \C\setminus\{0\}$. If $z\in \R$ then $\U{F}$ and $\W{\F}$ are only defined if $\vert \Re\, z\vert>\F$. For the limit $\F\rightarrow 0$ we can consider only such $\F$ so that $\U{F}$ and $\W{\F}$ are defined at $z$. Then it is obvious that
\begin{align}
\lim_{\F\rightarrow 0}\W{F}(z)&= \left( \tfrac{z}{z}\right)^{1/4}+\left( \tfrac{z}{z}\right)^{1/4}=2=\W{0}(z).
\end{align}
For small $\F$ we can expand the square root in $\U{F}$ 
\begin{align}
\tfrac{\sqrt{z^2-\F^2}}{\F}&= \tfrac{\sqrt{z^2}\left(1-\frac{\F^2}{2z^2}+\lO\left(\frac{\F^4}{z^4} \right)\right)}{\F}.
\intertext{Thus we get for $\U{\F}$, taking the change of sign in front of the root into account,}
\lim_{\F\rightarrow 0} \U{\F}(z)&=\lim_{\F\rightarrow 0}\tfrac{\frac{\F^2}{2z}+\lO\left(\frac{\F^4}{z^3}\right)}{\F}=0.\\
\intertext{The same way we find}
\lim_{\F\rightarrow 0} \tfrac{\U{\F}(z)}{\F}&=\lim_{\F\rightarrow 0}\tfrac{1}{2z}+\lO\bigl(\tfrac{\F^2}{z^3}\bigr)=\tfrac{1}{2z}.
\end{align}
\end{proof}

\subsection{Plancherel-Rotach Asymptotics}\label{sec_pr}


\begin{defin}\label{def_X_delta}
For all $\delta>0$ and $0<\F<\infty$ we define the regions
\begin{align}
\Adelta&=\{z \in \C \left\vert \vert \Re\, z\vert \ge \F+\delta  \lor \vert \Im\, z\vert \ge \delta\right. \},\\
\Bdelta&=\{z \in \C \left\vert -\F+\delta \le \Re\, z \le \F-\delta \land - \delta < \Im\, z < \delta\right. \},\\
\Cdelta&=\{z \in \C \left\vert \vert F-\vert \Re\, z\vert \vert < \delta  \land - \delta < \Im\, z < \delta \right.\}.
\end{align}
\end{defin}

\begin{remark}
For all $\delta>0$ and $0<\F<\infty$ we see that $\Adelta\cup\Bdelta\cup\Cdelta$ covers $\C$.
\end{remark}

\begin{remark}
When we set $\F=1$, there is a similar definition of regions $A_\delta$, $B_\delta$ and $C_\delta$ in \cite{dkmvz2001}. The shape of ours is a bit different, but this is nonrelevant as long as $\delta$ can be arbitrary small.
\end{remark}


\begin{defin}
We define the monic $\sqrt{m}$-rescaled Hermite polynomial of order $m$ as
\begin{equation}\label{eq_q}
Q_m(z)=2^{-m} m^{-m/2} H_m\left(\sqrt{m} z\right).
\end{equation}
\end{defin}
The rescaling of the argument is as we need it for our orthonormal polynomials \eqref{eq_op_herm}, which we can now write as
\begin{align}
\op{m}(z)&=Q_m\left(\tfrac{\sqrt{2}z}{\F} \sqrt{\tfrac{\n}{m}} \right) \tfrac{(2mt)^{m/2} \sqrt{\n}(1-\t^2)^{1/4}}{\sqrt{\pi m!}}, \qquad \forall m\in\N_0, \label{eq_op_m}\\
\intertext{and in particular}
\op{\n}(z)&=Q_\n\left(\tfrac{\sqrt{2}z}{\F} \right) \tfrac{(2t)^{\n/2} \n^{\frac{\n+1}{2}}(1-\t^2)^{1/4}}{\sqrt{\pi \n!}},\label{eq_op_n}\\
\intertext{and}
\op{\n-1}(z)&=Q_\n\left(\tfrac{\sqrt{2}z}{\F} \sqrt{\tfrac{\n}{\n-1}} \right) \tfrac{(2t)^{\frac{\n-1}{2}} \n^{\frac{\n+1}{2}} \left(\frac{\n-1}{\n}\right)^{\frac{\n-1}{2}} (1-\t^2)^{1/4}}{\sqrt{\pi \n!}}. \label{eq_op_n-1}
\end{align}

\subsubsection{\sloppy Plancherel-Rotach Asymptotics Outside of the Interval $[-\sqrt{2},\sqrt{2}]$}\label{subsection_pr_asym_outside} \fussy 
The zeros of $Q_m$ lie on the interval $[-\sqrt{2},\sqrt{2}]$ (see \cite{sze}). For $z\in \C\setminus [-\sqrt{2},\sqrt{2}]$ we have the following asymptotics for $Q_m(z)$ as $m\rightarrow \infty$. 
\begin{align}\label{eq_pi}
\pi_m(z)&=\tfrac{1}{2} \left(\left(\tfrac{z-\sqrt{2}}{z+\sqrt{2}}\right)^{1/4}+\left(\tfrac{z+\sqrt{2}}{z-\sqrt{2}}\right)^{1/4}\right) \frac{\e^{\frac{m}{4} \left(z\mp\sqrt{z^2-2}\right)^2}}{\left(z\mp\sqrt{z^2-2}\right)^{m}}\nn\\ 
&=\tfrac{1}{2}\W{\!\sqrt{2}}(z) \e^{\frac{m}{2}\left(\U{\!\sqrt{2}}(z)\right)^2} \left(\sqrt{2} \U{\!\sqrt{2}}(z) \right)^{-m},
\end{align}
with
\begin{align}
\frac{Q_m(z)}{\pi_m(z)}&=1+\lO\left(m^{-1}\right).\label{eq_pi_o}
\end{align}
\begin{remark}
The upper sign in the first line in \eqref{eq_pi} is valid for $z\in\Cp$ and the lower sign for $z\notin\Cp$. This comes from the fact that the Hermite polynomial have the symmetry $H_m(-z)=(-1)^m H_m(z)$.
\end{remark}

The asymptotics \eqref{eq_pi} can be found in \cite[p.206, (7.93)]{deift}, which was done with the help of a Riemann-Hilbert problem. It was first proved for real $z$ by M.~Plancherel and W.~Rotach in \cite{plancherel-rotach}, using the method of steepest descent. Note that \cite{plancherel-rotach} uses a different convention for the Hermite polynomials, there $H_m(z)$ is used for $2^{m/2} H_m(\sqrt{2}z)$ in terms of our notation. 

The $\lO$-term in \eqref{eq_pi_o} holds uniformly for $z$ in any closed domain that is a subset of $\C\setminus [-\sqrt{2},\sqrt{2}]$, i.e.\ for all $z\in \Xdelta{A}{\delta}{\sqrt{2}}$ for any fixed $\delta>0$ (see \cite[p.201, Theorem 8.22.9, (8.22.13)]{sze} and for a more general class of orthogonal polynomials \cite[p.51, Theorem 1.4(i)]{dkmvz2001}). See also \cite[Section 2.3, (2.16)]{wong} and \cite[p.10, Theorem 2.18]{dkmvz1999}.

\subsubsection{Plancherel-Rotach Asymptotics in a Neighborhood of the Interval $(-\sqrt{2}+\delta,\sqrt{2}-\delta)$}\label{subsection_pr_asym_inside}
On the interval $(-\sqrt{2},\sqrt{2})$ the asymptotics of $Q_m(z)$ is a bit more complicated as it has to take the oscillations between the zeros into account. There the Plancherel-Rotach asymptotics is
\begin{align}
\pi^r_m(z)= \Bigg( & \left( \tfrac{\sqrt{2}-z}{\sqrt{2}+z} \right)^{1/4} \cos\left(
   \tfrac{m}{2} \left(z \sqrt{2-z^2} +2 \arcsin \bigl(\tfrac{z}{\sqrt{2}}\bigr)-\pi\right)-\tfrac{\pi}{4} \right)\nn\\
& +\left(\tfrac{\sqrt{2}+z}{\sqrt{2}-z}\right)^{1/4} \cos \left(\tfrac{m}{2} \left(z \sqrt{2-z^2} +2\arcsin\bigl(\tfrac{z}{\sqrt{2}}\bigr)-\pi \right)+\tfrac{\pi }{4}\right)\Bigg) \e^{\frac{m}{2} \left(z^2-1-\log 2\right)}, \label{eq_pir}
\end{align}
with
\begin{align}
Q_m(z)=& \pi^r_m(z)+ \e^{\frac{m}{2} \left(z^2-1-\log 2\right)} \lO\left(m^{-1}\right). \label{eq_pir_o}
\end{align}
The asymptotics \eqref{eq_pir} can be found from \cite[p.233, (7.187)]{deift} considering the case $V(x)=x^2$. A proof that restricts to the Hermite polynomial can be found in \cite{plancherel-rotach}.  

There exists a $\delta_0>0$ small enough so that the $\lO$-term in \eqref{eq_pir_o} holds uniformly for $z\in \Xdelta{B}{\delta}{\sqrt{2}}$ if $\delta\le \delta_0$ (see \cite[p.51, Theorem 1.4(ii)\&(v)]{dkmvz2001} and \cite[p.201, Theorem 8.22.9, (8.22.12)]{sze}). See also \cite{wong} and \cite{dkmvz1999}.

\begin{remark}
In \eqref{eq_pi_o} we have a relative error term as it must hold uniformly at $\vert z\vert \rightarrow \infty$, where $\vert Q_m(z)\vert$ and $\vert \pi_m(z)\vert$ get big, while a region near the zeros of $Q_m(z)$ and $\pi_m(z)$, i.e.\ a neighborhood of $[-\F,\F]$, is excluded. In $\eqref{eq_pir_o}$ on the other hand we must have an absolute error term for the part describing the oscillations around zero.
\end{remark}

\begin{remark}
Note that at the points $z=\pm \sqrt{2}$ neither \eqref{eq_pi} nor \eqref{eq_pir} holds. But there exists a further approximation which is valid there or even in the neighborhood $\Xdelta{C}{\delta}{\sqrt{2}}$ if $\delta<\delta_0$, see \cite[Section 7.7]{deift}, \cite[Chapitre II]{plancherel-rotach}, \cite[Theorem 8.22.9, (8.22.14)]{sze}, \cite[Theorem 1.4(iii)\&(iv)]{dkmvz2001}, \cite[Section 5, (5.16)]{wong} and \cite[Theorem 2.18]{dkmvz1999}. But we are not going to use this approximation there as we will find a way without using it.
\end{remark}

\subsection{Asymptotics of Identities Outside of the Interval $[-\F,\F]$} \label{sec_asymp_ident_outside}
The parameter $\F$ will be present very frequently in the following calculations, so we will often use the notion $\Fz=\tfrac{z}{\F}$ to shorten the computations. 

\begin{remark}
If $\t=\sqrt{5}-2$ then we have $F=1$ and therefore $z=\Fz$.
\end{remark}

\begin{prop}\label{prop_asym} Let $\delta>0$, $0<\t<1$ and $z\in \Adelta$. 
Then as $\n\rightarrow\infty$
\begin{align}
\left.\frac{\partial \HM}{\partial w}\right\vert_{w=\cc{z}}+\left.\frac{\partial \HM}{\partial z}\right\vert_{w=\cc{z}} &=\e^{\n \f(z)+\frac{1}{2}\log \n + \log (\gp(z))-\frac{1}{2}\log(2 \pi^3)+o(1)}\nn \\ 
&=\sqrt{\tfrac{\n}{2 \pi^3}} \gp(z) \e^{\n \f(z)}(1+o(1)), \label{eq_lemma_asym1} \\
\shortintertext{and}
i \left.\frac{\partial \HM}{\partial w}\right\vert_{w=\cc{z}}-i \left.\frac{\partial \HM}{\partial z}\right\vert_{w=\cc{z}} &= \e^{\n \f(z)+\frac{1}{2}\log \n + \log (\gm(z))-\frac{1}{2}\log(2 \pi^3)+o(1)}\nn \\
&=\sqrt{\tfrac{\n}{2 \pi^3}} \gm(z) \e^{\n \f(z)}(1+o(1)) \label{eq_lemma_asym2},
\end{align}
where
\begin{align}
\f(z)
&= \tfrac{4\t}{1-\t^2}\left(\t\Re{(\Fz^2)}-\vert \Fz\vert^2\right) + \Re{\left(\left(\Fz\mp\sqrt{\Fz^2-1} \right)^2\right)}\nn\\
& \phantom{=}\quad -2 \log\left\vert\Fz\mp\sqrt{\Fz^2-1} \right\vert+\log\t+1, \nn\\
&=\t \Re(z^2) -\vert z\vert^2+\Re\left( \left(\U{\F}(z)\right)^2 \right)-2\log\left\vert\U{\F}(z)\right\vert+\log t+1 \label{eq_f},\\
\gp(z)
&=-\Re\left(\Fz\mp\sqrt{\Fz^2-1}\right) \left\vert\sqrt{\Fz+1}+\sqrt{\Fz-1}\right\vert^2 \left\vert \Fz^2-1 \right\vert^{-1/2} \tfrac{1-\t}{2\sqrt{\t}} \nn \\
&=-\tfrac{1-\t}{2\sqrt{\t}}\Re\left(\U{\F}(z)\right)\left\vert\W{\F}(z)\right\vert^2, \label{eq_gp}\\
\shortintertext{and}
\gm(z)
&=-\Im\left(\Fz\mp\sqrt{\Fz^2-1}\right) \left\vert\sqrt{\Fz+1}+\sqrt{\Fz-1}\right\vert^2\left\vert \Fz^2-1 \right\vert^{-1/2} \tfrac{1+\t}{2\sqrt{\t}}\nn \\
&= -\tfrac{1+\t}{2\sqrt{\t}}\Im\left(\U{\F}(z)\right)\left\vert\W{\F}(z)\right\vert^2\nn \\
&=\gp(z) \tfrac{1+\t}{1-\t} \frac{\Im\,\U{\F}(z)}{\Re\,\U{\F}(z)}. \label{eq_gm}
\end{align}
\end{prop}
\begin{proof}
When $z\in \Adelta$, we can use the Plancherel-Rotach asymptotics \eqref{eq_pi} as $\n\rightarrow \infty.$ Using \eqref{eq_q} and \eqref{eq_pi} in \eqref{eq_op_herm}, we get for the orthonormal polynomial of order $\n$
\begin{align}
\op{\n}(z)&=\tfrac{Q_\n\left(\sqrt{2}\Fz\right)}{\pi_\n\left(\sqrt{2}\Fz\right)}\tfrac{\n^{\frac{\n+1}{2}}}{2\sqrt{\pi}\sqrt{\n !}}\e^{\frac{\n}{2}\left(\Fz\mp\sqrt{\Fz^2-1}\right)^2} \left( \Fz\mp\sqrt{\Fz^2-1}\right)^{-\n} \nn \\
& \phantom{=}\quad\cdot \t^{\frac{\n}{2}} \left(1-\t^2\right)^{1/4} \left(\Fz^2-1\right)^{-1/4} \left(\sqrt{\Fz+1}+\sqrt{\Fz-1}\right) \nn \\
&=\tfrac{Q_\n\left(\frac{\sqrt{2}z}{\F}\right)}{\pi_\n\left(\frac{\sqrt{2}z}{\F} \right)}\tfrac{\n^{\frac{\n+1}{2}}}{2\sqrt{\pi}\sqrt{\n !}}\e^{\frac{\n}{2}\left(\U{\F}(z)\right)^2} \left( \U{\F}(z) \right)^{-\n} \t^{\frac{\n}{2}} \left(1-\t^2\right)^{1/4} \W{\F}(z),\label{eq_op1}\\
\intertext{and for the orthonormal polynomial of order $\n-1$}
\op{\n-1}(z)&=\tfrac{Q_{\n-1}\left(\sqrt{\frac{\n}{\n-1}}\sqrt{2}\Fz\right)}{\pi_{\n-1}\left(\sqrt{\frac{\n}{\n-1}}\sqrt{2}\Fz\right)}\tfrac{\n^{\tfrac{\n+1}{2}}}{2\sqrt{\pi}\sqrt{\n !}}\left(\tfrac{\n-1}{\n}\right)^{\n-1}\e^{\frac{\n}{2}\left(\Fz \mp \sqrt{\Fz^2-1+\frac{1}{\n}}\right)^2} \nn \\
&\phantom{=}\quad\cdot \left( \Fz \mp \sqrt{\Fz^2-1+\tfrac{1}{\n}}\right)^{-(\n-1)} \t^{\frac{\n-1}{2}} \left(1-\t^2\right)^{1/4} \nn \\
&\phantom{=}\quad\cdot\left(\Fz^2-1+\tfrac{1}{\n}\right)^{-1/4} \left(\sqrt{\Fz+\sqrt{1-\tfrac{1}{\n}}}+\sqrt{\Fz-\sqrt{1-\tfrac{1}{\n}}}\right) \nn \\
&=\tfrac{Q_{\n-1}\left(\sqrt{\frac{\n}{\n-1}}\frac{\sqrt{2} z}{\F}\right)}{\pi_{\n-1}\left(\sqrt{\frac{\n}{\n-1}}\frac{\sqrt{2} z}{\F}\right)}\tfrac{\n^{\tfrac{\n+1}{2}}}{2\sqrt{\pi}\sqrt{\n !}}\left(\tfrac{\n-1}{\n}\right)^{\n-1}  \e^{\frac{\n}{2}\Bigl( \sqrt{1-\frac{1}{\n}} \U{\F}\left(  z \cramped[\scriptscriptstyle]{\left(1-\frac{1}{\n}\right)^{-1/2}} \right) \Bigr)^2}  \nn \\
&\phantom{=}\quad\cdot \left( \sqrt{1-\tfrac{1}{\n}}  \U{\F}\left(  z \left(1-\tfrac{1}{\n}\right)^{-1/2} \right) \right)^{-(\n-1)} \t^{\frac{\n-1}{2}} \left(1-\t^2\right)^{1/4} \W{\F}\left(  z \left(1-\tfrac{1}{\n}\right)^{-1/2} \right).\label{eq_op2}
\end{align}

\eqref{eq_op1} and \eqref{eq_op2} are still exact, even for finite $\n$. Now we are going to use the following expansions up to order $\lO(\n^{-1})$ as $\n\rightarrow \infty$.
\begin{align}
\MoveEqLeft \frac{\n}{2}\left( \sqrt{1-\tfrac{1}{\n}}  \U{\F}\left(  z \left(1-\tfrac{1}{\n}\right)^{-1/2} \right) \right)^2=\frac{\n}{2}\left(\Fz\mp\sqrt{\Fz^2-1}\sqrt{1 +\tfrac{1}{\n(\Fz^2-1)}}\right)^2\nn\\
&=\frac{\n}{2}\left(\Fz\mp\sqrt{\Fz^2-1}\left(1+\tfrac{1}{2\n(\Fz^2-1)}+\lO(\n^{-2})\right)\right)^2\nn\\
&=\frac{\n}{2}\left(\left(\Fz\mp\sqrt{\Fz^2-1}\right)^2-\tfrac{\Fz\mp\sqrt{\Fz^2-1}}{\pm\n\sqrt{\Fz^2-1}}+\lO(\n^{-2})\right)\nn\\
&=\frac{\n}{2}\left(\Fz\mp\sqrt{\Fz^2-1}\right)^2 -\tfrac{\Fz\mp\sqrt{\Fz^2-1}}{\pm2\sqrt{\Fz^2-1}}+\lO(\n^{-1})\nn\\
&=\frac{\n}{2}\left(\U{\F}(z)\right)^2-\tfrac{\F \U{\F}(z)}{2\T{\F}(z)}+\lO(\n^{-1}),\label{eq_expan1}\\
\MoveEqLeft \left( \sqrt{1-\tfrac{1}{\n}}  \U{\F}\left(  z \left(1-\tfrac{1}{\n}\right)^{-1/2} \right) \right)^{-(\n-1)}=\left(\Fz\mp\sqrt{\Fz^2-1+\tfrac{1}{\n}}\right)^{-(\n-1)}\nn\\
&=\left(\Fz\mp\sqrt{\Fz^2-1}-\tfrac{1}{\pm2 \n\sqrt{\Fz^2-1}}+\lO(\n^{-2})\right)^{-(\n-1)}\nn\\
&=\left(\Fz\mp\sqrt{\Fz^2-1}\right)^{-(\n-1)}\left(1-\tfrac{1}{\pm 2 \n \sqrt{\Fz^2-1}\left( \Fz\mp\sqrt{\Fz^2-1}\right)}+\lO(\n^{-2})\right)^{-(\n-1)}\nn\\
&=\left(\Fz\mp\sqrt{\Fz^2-1}\right)^{-(\n-1)}\left(\e^{\frac{1}{\pm 2 \sqrt{\Fz^2-1}\left( \Fz\mp\sqrt{\Fz^2-1}\right)}}+\lO(\n^{-1})\right)\nn\\
&=\left(\U{\F}(z)\right)^{-(\n-1)}\left(\e^{\frac{\F}{2 \T{\F}(z) \U{\F}(z)}}+\lO(\n^{-1})\right),\label{eq_expan2}\\
\MoveEqLeft \left(\Fz^2-1+\tfrac{1}{\n}\right)^{-1/4}=\left(\Fz^2-1\right)^{-1/4}\left(1+\lO(\n^{-1})\right),\label{eq_expan3}\\
\MoveEqLeft \sqrt{\Fz+\sqrt{1-\tfrac{1}{\n}}}+\sqrt{\Fz-\sqrt{1-\tfrac{1}{\n}}}=\sqrt{\Fz+1+\lO(\n^{-1})}+\sqrt{\Fz-1+\lO(\n^{-1})}\nn\\
&=\left(\sqrt{\Fz+1}+\sqrt{\Fz-1}\right)\left(1+\lO(\n^{-1})\right).\label{eq_expan4}\\
\intertext{Combining \eqref{eq_expan3} and \eqref{eq_expan4} we get}
\MoveEqLeft \W{\F}\left(  z \left(1-\tfrac{1}{\n}\right)^{-1/2} \right)=\W{\F}(z)\left(1+\lO(\n^{-1})\right).\label{eq_expan4b}
\end{align}
For our purpose, we only need from the Plancherel-Rotach asymptotics \eqref{eq_pi_o} that
\begin{equation}\label{eq_pr_asymp}
\tfrac{Q_\n\left(\sqrt{2}\Fz\right)}{\pi_\n\left(\sqrt{2}\Fz\right)}=1+o(1), \qquad\text{and}\qquad  \tfrac{Q_{\n-1}\left(\sqrt{\tfrac{\n}{\n-1}}\sqrt{2}\Fz\right)}{\pi_{\n-1}\left(\sqrt{\tfrac{\n}{\n-1}}\sqrt{2}\Fz\right)}=1+o(1).
\end{equation} 
According to Section \ref{subsection_pr_asym_outside} this is valid for all $\tilde{\delta}>0$ when $\sqrt{2}\tfrac{z}{\F}$ and $\sqrt{\tfrac{\n}{\n-1}}\sqrt{2}\tfrac{z}{\F}$ lie in $\Xdelta{A}{\tilde{\delta}}{\sqrt{2}}$. When we set $\tilde{\delta}=\sqrt{2}\tfrac{\delta}{\F}$, this is equivalent to $z\in \Adelta\cap \Xdelta{A}{\delta}{\sqrt{1-\frac{1}{\n}}\F}\subset\Adelta$ as in the condition of the proposition. 

Using \eqref{eq_expan1}, \eqref{eq_expan2}, \eqref{eq_expan4b} and \eqref{eq_pr_asymp} in \eqref{eq_op1} and \eqref{eq_op2}, we get
\begin{align}
\MoveEqLeft[3] \op{\n}(\cc{z})\op{\n-1}(z)\pm\op{\n-1}(\cc{z})\op{\n}(z)\nn\\
= {} & \frac{\n^{\n+1}}{4\pi \n !} \t^{\n-\frac{1}{2}} \left(1-\t^2\right)^{1/2} \left(\tfrac{\n-1}{\n}\right)^{\n-1}\e^{\frac{\n}{2}\left(\left( \U{\F}(\cc{z}) \right)^2+ \left( \U{\F}(z) \right)^2\right)} \bigl( \U{\F}(\cc{z}) \U{\F}(z) \bigr)^{-\n} \nn\\
& \cdot  \W{\F}(\cc{z}) \W{\F}(z) \left(1 + o(1)\right)\nn\\
& \begin{aligned}[t] \cdot \Bigg\lgroup &\e^{-\frac{\F \U{\F}(z)}{2\T{\F}(z)}+\lO(\n^{-1})} \U{\F}(z) \left(\e^{\frac{\F}{2 \T{\F}(z) \U{\F}(z)}}+\lO(\n^{-1})\right) \nn\\
& \pm \e^{-\frac{\F \U{\F}(\cc{z})}{2\T{\F}(\cc{z})}+\lO(\n^{-1})} \U{\F}(\cc{z}) \left(\e^{\frac{\F}{2 \T{\F}(\cc{z}) \U{\F}(\cc{z})}}+\lO(\n^{-1})\right) \Bigg\rgroup \end{aligned} \nn\\
%
= {} & \frac{\n^{\n+1}}{4\pi \n !} \t^{\n-\frac{1}{2}} \left(1-\t^2\right)^{1/2} \left(\tfrac{\n-1}{\n}\right)^{\n-1} \e^{\n\Re\left(\left( \U{\F}(z) \right)^2 \right)} \left\vert \U{\F}(z) \right\vert^{-2\n} \left\vert \W{\F}(z) \right\vert^2 \left(1+o(1)\right) \nn\\
& \cdot \left\lgroup \U{\F}(z) \e^{-\frac{\F \U{\F}(z) }{2 \T{\F}(z)}+\tfrac{\F}{2 \T{\F}(z) \U{\F}(z)}} \pm \U{\F}(\cc{z}) \e^{-\frac{\F \U{\F}(\cc{z}) }{2 \T{\F}(\cc{z})}+\tfrac{\F}{2 \T{\F}(\cc{z}) \U{\F}(\cc{z})}} +\lO(\n^{-1})\right\rgroup \nn \\
%
= {} & \frac{\n^{\n+1}}{2\pi \n !} \t^{\n-\frac{1}{2}} \left(1-\t^2\right)^{1/2} \left(\tfrac{\n-1}{\n}\right)^{\n-1} \e^{\n\Re\left(\left( \U{F}(z) \right)^2 \right)} \left\vert \U{F}(z) \right\vert^{-2\n} \left\vert \W{F}(z) \right\vert^2  \nn \\
& \cdot \e^1 \left(1+o(1) \right)\left\{\begin{array}{l}\Re\, \U{\F}(z),\\ i \Im\, \U{\F}(z), \end{array}\right. \label{eq_asymp1}
\end{align}
where in the last step we have used that
\begin{equation}\label{eq_asymp1b}
\e^{-\frac{\F \U{\F}(z) }{2 \T{\F}(z)}+\frac{\F}{2 \T{\F}(z) \U{\F}(z)}}= \e^{-\frac{\Fz\mp\sqrt{\Fz^2-1}}{\pm 2\sqrt{\Fz^2-1}} + \frac{1}{\pm 2 \sqrt{\Fz^2-1}\left( \Fz\mp \sqrt{\Fz^2-1}\right)}}
=\e^1.
\end{equation}

When $z$ lies on the imaginary axis (or on the real axis) $\Re\, \U{\F}(z)$ (or $\Im\, \U{\F}(z)$ respectively) is zero and can't be factored out of the $(1+o(1))$ term. But \eqref{eq_asymp1} is still true as in this case the left-hand side is obviously identically zero for all $\n\in\N$.

Using the Stirling approximation
\begin{align}
\n !&= \sqrt{2 \pi \n}\n^\n \e^{-\n} \left(1+\lO(\n^{-1})\right), \label{eq_stirling}
\intertext{and the expansion}
\left(\frac{\n-1}{\n}\right)^{\n-1}&=\e^{-1}\left(1+\lO(\n^{-1})\right), \label{eq_expan5}
\end{align}
we get from \eqref{eq_asymp1}
\begin{align}
\MoveEqLeft[3] \op{\n}(\cc{z})\op{\n-1}(z)\pm\op{\n-1}(\cc{z})\op{\n}(z)\nn \\
= {} & \tfrac{\sqrt{\n} \e^\n }{(2\pi)^{3/2}} \t^\n  \sqrt{\tfrac{1-\t^2}{\t}} \e^{\n\Re\left(\left(\U{\F}(z) \right)^2 \right)} \left\vert \U{\F}(z) \right\vert^{-2\n} \left\vert \W{\F}(z) \right\vert^2 (1 +o(1)) \left\{\begin{array}{l} \Re\, \U{\F}(z),\\ i \Im\, \U{\F}(z). \end{array}\right. \label{eq_asymp2}
\end{align}

Using \eqref{eq_id1} and \eqref{eq_id2} at $w=\cc{z}$ and then \eqref{eq_asymp2}, we finally get
\begin{align}
\left.\frac{\partial \HM}{\partial w}\right\vert_{w=\cc{z}}+\left.\frac{\partial \HM}{\partial z}\right\vert_{w=\cc{z}} &=-\sqrt{\tfrac{1-\t}{1+\t}} \bigl(\op{\n}(\cc{z}) \op{\n-1}(z)+\op{\n-1}(\cc{z}) \op{\n}(z)\bigr) \e^{\n\left(\t \Re(z^2)-\vert z\vert^2\right)} \nn\\
&=-\tfrac{\sqrt{\n} }{(2 \pi)^{3/2}} \tfrac{1-\t}{\sqrt{\t}} \e^{\n \left(  1+\log(\t)+\Re\left(\left( \U{\F}(z) \right)^2 \right) -2 \log \vert \U{\F}(z) \vert +\t \Re{(z^2)} -\vert z\vert^2 \right)  } \nn \\
&\phantom{=} \quad \cdot \vert \W{\F}(z)\vert^2 \Re{\,\U{F}(z)}(1+o(1)),  \label{eq_asymp3}
\intertext{and}
i \left.\frac{\partial \HM}{\partial w}\right\vert_{w=\cc{z}}-i\left.\frac{\partial \HM}{\partial z}\right\vert_{w=\cc{z}} &=\sqrt{\tfrac{1+\t}{1-\t}} i \bigl( \op{\n}(\cc{z}) \op{\n-1}(z)- \op{\n-1}(\cc{z}) \op{\n}(z) \bigr)\e^{\n\left(\t \Re(z^2)-\vert z\vert^2\right) } \nn\\
&=-\tfrac{\sqrt{\n} }{(2 \pi)^{3/2}} \tfrac{1+\t}{\sqrt{\t}} \e^{\n \left(  1+\log(\t)+\Re\left(\left( \U{\F}(z) \right)^2 \right) -2 \log \vert \U{\F}(z) \vert +\t \Re{(z^2)} -\vert z\vert^2 \right)  } \nn \\
&\phantom{=} \quad \cdot \vert \W{\F}(z)\vert^2 \Im{\,\U{F}(z)}(1+o(1)), \label{eq_asymp4}
\end{align}
which proves the proposition.
\end{proof}


\subsection{Asymptotics of Identities on the Interval $(-\F,\F)$} \label{sec_asymp_ident_inside}
Here we only look at real $z\in (-F,F)$, which is equivalent to $\Fz\in (-1,1)$. 

\begin{prop}\label{prop_asym_B} Let $0<\delta<1$, $0<\t<1$ and $z\in (-\F(1-\delta),\F(1-\delta))$. Then
\begin{align}
\left.\frac{\partial \HM}{\partial w}\right\vert_{w=z}+\left.\frac{\partial \HM}{\partial z}\right\vert_{w=z} &=o(1), \qquad (\n\rightarrow\infty),\\
\intertext{and}
i \left.\frac{\partial \HM}{\partial w}\right\vert_{w=z}-i \left.\frac{\partial \HM}{\partial z}\right\vert_{w=z} &=0,\qquad \forall \n\in\N.
\end{align}
\end{prop}
\begin{proof}
When $z\in\R$ it is obvious for all $\n\in \N$ that \eqref{eq_id2} evaluated at $w=z$ is identically zero whereas \eqref{eq_id1} simplifies to 
\begin{equation}\label{eq_id1_real}
\left.\frac{\partial \HM}{\partial w}\right\vert_{w=z}+\left.\frac{\partial \HM}{\partial z}\right\vert_{w=z}=-2 \sqrt{\frac{1-\t}{1+\t}} \op{\n}(z)\op{\n-1}(z)\e^{-\n(1-\t)z^2}.
\end{equation}
Let $\n_0>\frac{2}{\delta} \gg 1$. Then it follows for all $\n>\n_0$ and $z\in (-\F(1-\delta),\F(1-\delta))$ that 
\begin{equation}
\sqrt{\frac{\n}{\n-1}}\sqrt{2}\Fz=\sqrt{2} \frac{z}{\F}\left(1+\frac{1}{2\n}\right)+\lO(\n^{-2})\in \left(-\sqrt{2}(1-\tfrac{\delta}{2}),\sqrt{2}(1-\tfrac{\delta}{2})\right).
\end{equation}
Therefore \eqref{eq_pir_o} is valid for $Q_{\n-1}\left(\sqrt{\tfrac{\n}{\n-1}}\sqrt{2}\Fz\right)$ and $ Q_\n\left(\sqrt{2}\Fz\right)$. On the real line we can estimate the cosine in \eqref{eq_pir} by one. Then the Plancherel-Rotach asymptotics in \eqref{eq_q} together with \eqref{eq_op_herm} gives an estimation for the orthonormal polynomial of order $\n$,
\begin{align}
\vert\op{\n}(z)\vert&\le \left( \tfrac{ \sqrt{1+\Fz} + \sqrt{1-\Fz} }{(1-\Fz^2)^{1/4}} + o(1) \right) \e^{\n \left(\Fz^2-\frac{1+\log 2}{2}\right)} \tfrac{(2\t\n)^{\frac{\n}{2}} }{\sqrt{\n!} } \tfrac{\sqrt{\n} (1-\t^2)^{1/4}}{\sqrt{\pi}},\\
\intertext{and for the orthonormal polynomial of order $\n-1$,}
\vert\op{\n-1}(z)\vert&\le \left( \tfrac{ \sqrt{\sqrt{1-\frac{1}{\n}}+\Fz} + \sqrt{\sqrt{1-\frac{1}{\n}}-\Fz} }{(1-\frac{1}{\n}-\Fz^2)^{1/4}} + o(1) \right)\e^{\n \left(\Fz^2-\frac{(\n-1)(1+\log 2)}{2\n}\right)}\tfrac{(2\t(\n-1))^{\frac{\n-1}{2}} }{\sqrt{\n!} } \tfrac{\n (1-\t^2)^{1/4}}{\sqrt{\pi}}.
\end{align}
From this we find for the asymptotics of identity \eqref{eq_id1_real} 
\begin{align}\label{eq_id1_real1}
\left.\left\vert\tfrac{\partial \HM}{\partial w}+\tfrac{\partial \HM}{\partial z}\right\vert \right\vert_{w=z}\le& 2 \sqrt{\tfrac{1-\t}{1+\t}} \left( \tfrac{ \sqrt{1+\Fz} + \sqrt{1-\Fz} }{(1-\Fz^2)^{1/4}}\tfrac{ \sqrt{\sqrt{1-\frac{1}{\n}}+\Fz} + \sqrt{\sqrt{1-\frac{1}{\n}}-\Fz} }{(1-\frac{1}{\n}-\Fz^2)^{1/4}} + o(1) \right) \nonumber\\
& \quad \cdot \e^{\n \left(2\Fz^2-1-\log 2\right)+\frac{1+\log 2}{2}} \tfrac{\n^{\n+1} \left(\frac{\n-1}{\n}\right)^{\frac{\n-1}{2}} (2\t)^{\n-\frac{1}{2}} \sqrt{1-\t^2} }{\n! \pi} \e^{-\n\F^2(1-\t)\Fz^2}\nonumber\\
=&\tfrac{\sqrt{\n}}{\sqrt{2\pi^3}} (g_{\n+}^r(z)+o(1)) \e^{\n \f^r(z)},\\
\intertext{where}
f^r(z)=&\frac{2 \Fz^2(1-\t)}{1+\t}+\log \t,\label{eq_fr}\\
\intertext{and}
g_{\n+}^r(z)=&\frac{2(1-\t)}{\sqrt{\t}} \left( \frac{ \sqrt{1+\Fz} + \sqrt{1-\Fz} }{(1-\Fz^2)^{1/4}}\tfrac{ \sqrt{\sqrt{1-\frac{1}{\n}}+\Fz} + \sqrt{\sqrt{1-\frac{1}{\n}}-\Fz} }{(1-\frac{1}{\n}-\Fz^2)^{1/4}} \right).
\end{align}
In the last step in \eqref{eq_id1_real} we have used the Stirling approximation \eqref{eq_stirling}. As $\sqrt{1-\tfrac{1}{\n}}=1-\tfrac{1}{2\n}+\lO(\n^{-2})>1-\delta$, it is obvious that $g_{\n+}^r$ is continuous on $(-\F(1-\delta),\F(1-\delta))$ and it converges to
\begin{equation}\label{eq_gpr}
\gp^r(z)=\lim_{\n\rightarrow\infty} g_{\n+}^r(z)=\tfrac{4(1-\t)\left(1+\sqrt{1-\Fz^2}\right) }{\sqrt{\t} \sqrt{1-\Fz^2}},
\end{equation}
which is bounded on $(-\F(1-\delta),\F(1-\delta))$. From \eqref{eq_fr} we see that \begin{equation}
\f^r(z)<f^r(\F)=\frac{2(1-\t)}{1+\t}+\log \t, \qquad \forall z\in (-\F(1-\delta),\F(1-\delta)), 0<\t<1.
\end{equation}
$\f^r(\F)$ goes to minus infinity when $\t$ goes to zero and it reaches zero when $\t$ goes to one. Because
\begin{equation}
\frac{\dd (\f^r(\F))}{\dd \t}=\frac{(1-\t)^2}{\t(1+\t)^2}>0, \qquad \forall 0<\t<1, \label{eq_fr_estim}
\end{equation}
$\f^r(\F)$ as a function of $\t$ is strictly monotonic increasing from minus infinity to zero for $\t\in (0,1)$. Therefore $\f^r(z)$ is strictly negative for all $z\in (-\F(1-\delta),\F(1-\delta))$ and all $\t\in (0,1)$. Because of this, $g_{\n+}^r \e^{\n \f^r(z)}$ must converge to zero when $\n\rightarrow \infty$, which proves the proposition. 
\end{proof}

\begin{remark}
We have found the asymptotics of the identities \eqref{eq_id1} and \eqref{eq_id2} when $z$ is outside of $[\F,\F]$ and when $z$ is inside of $(-\F,\F)$. But the points $\pm\F$ are not included as neither \eqref{eq_pi_o} nor \eqref{eq_pir_o} holds there. We are going to care for those later in Chapter \ref{sec_verification}.
\end{remark}


\subsection{Properties of $\protect\f$ and $\protect\gpm$}\label{sec_f_g}


Now we are going to discuss how the functions $\f$ and $\gpm$ look like. In Proposition \ref{prop_f_zero} we will show that the global maximum of $\f$ lies on the ellipse $\pEt$, where $\f$ takes the value zero. In Proposition \ref{prop_g} we will find a relation between $\gpm$ and the second derivatives of $\f$.

\subsubsection{Symmetries and Basic Properties of $\f$ and $\gpm$}
According to \eqref{eq_f}, \eqref{eq_gp} and \eqref{eq_gm} we see that $\f$ and $\gpm$ are continuous real-valued functions on $\C\setminus [-\F,\F]$. The identities and its asymptotics must share the same symmetries (see Section \ref{subsec_symm_id} for the symmetries of identities \eqref{eq_id1} and \eqref{eq_id2}). So $\f$, which is present in the exponent, is an even function. The sign change under the map $z\mapsto -z$ comes from $\gpm$, which are odd functions. Note that there is an additional change of sign in $\gm$ under the map $z\mapsto \cc{z}$, which comes from the multiplication of \eqref{eq_id2} with $i$ on the left-hand side of \eqref{eq_lemma_asym2}.

When we introduce $\f(x,y)=\f(x+i y)$ as a function of two real variables, we can write
\begin{align}
\f(x,y)= & \tfrac{4\t}{1-\t^2} \left(t \Fx^2-\t \Fy^2 -\Fx^2-\Fy^2\right)\nn\\
& + \tfrac{1}{2} \left(\Fx+i \Fy \mp\sqrt{\Fx^2-\Fy^2+2 i \Fx \Fy-1}\right)^2+\tfrac{1}{2} \left(\Fx-i \Fy \mp\sqrt{\Fx^2-\Fy^2-2 i \Fx \Fy-1}\right)^2\nn\\
& -\log \left(\bigl(\Fx+i \Fy \mp\sqrt{\Fx^2-\Fy^2+2 i \Fx \Fy-1}\bigr) \bigl(\Fx-i \Fy \mp\sqrt{\Fx^2-\Fy^2-2 i \Fx \Fy-1}\bigr)\right)\nn\\
& +\log \t +1 \nn\\
= & t x^2-t y^2 -x^2 -y^2 +\tfrac{1}{2}\left(\U{\F}(x,y)\right)^2 +\tfrac{1}{2}\left(\U{\F}(x,-y)\right)^2 \nn\\
& -\log\left(\U{\F}(x,y) \U{\F}(x,-y) \right) +\log\t+1 .\label{eq_f_xy}
\end{align}
Again we have used the notation $\Fx=\tfrac{x}{\F}$ and $\Fy=\tfrac{y}{\F}$.

The same way we use the functions $\gpm(x,y)=\gpm(x+i y)$. It is easy to see that $\f(x,y)$ and $\gpm(x,y)$ are smooth functions on $\{(x,y)\in\R^2\vert x\notin [-\F,\F] \lor y \neq 0\}$.

\subsubsection{$\f=0$ on the Ellipse $\pEt$}
\begin{lemma}\label{lemma_f_zero}
On the ellipse $\pEt$ the value of $\f$ is zero.
\end{lemma}



\begin{proof}
We can parametrize the ellipse $\pEt$ as
\begin{align}
z&=\mas \cos \phi + i \mis \sin \phi = \sqrt{\frac{1+\t}{1-\t}} \cos \phi + i \sqrt{\frac{1-\t}{1+\t}} \sin \phi, \qquad \phi\in (-\pi,\pi]\\
\shortintertext{which is equivalent to}
\Fz&=\frac{1+\t}{2\sqrt{\t}} \cos \phi + i \frac{1-\t}{2\sqrt{\t}} \sin \phi.
\end{align}

We are going to evaluate the terms of $\f$ on the ellipse. The important part happens in
\begin{align}
\tfrac{\T{\F}(z)}{\F}&=\pm \sqrt{\Fz^2-1} =\tfrac{ \pm \sqrt{(1+\t)^2\cos^2 \phi+2i(1-\t^2)\cos \phi\sin \phi-(1-\t)^2\sin^2 \phi-4\t} }{2\sqrt{\t}} \nn \\
&=\tfrac{ \pm \sqrt{(1-2\t+\t^2)\cos^2 \phi+2i(1-\t^2)\cos \phi\sin \phi-(1+2\t+\t^2)\sin^2 \phi}}{2\sqrt{\t}} \nn \\
&=\pm \tfrac{1}{2\sqrt{\t}}\sqrt{\left((1-\t)\cos \phi+i(1+\t)\sin \phi \right)^2}\nn \\
&=\tfrac{1}{2\sqrt{\t}}\left((1-\t)\cos \phi+i(1+\t)\sin \phi \right), \label{eq_f_zero1}
\intertext{when we use $1=\sin^2 \phi+\cos^2 \phi$ on the second line. This way we have got rid of the $-1$ term under the square root and it swaps the factors $(1+\t)^2$ and $(1-\t)^2$ in front of $\cos^2 \phi$ and $\sin^2 \phi$. So we get further}
\U{\F}(z) &=\Fz\mp\sqrt{\Fz^2-1}=\sqrt{\t}\cos \phi-i \sqrt{\t}\sin \phi. \label{eq_f_zero2}
\end{align}
Using \eqref{eq_f_zero2} we find for the terms of $\f$
\begin{align}
\Re{\left(\left(\U{\F}(z)\right)^2\right)} &= \Re \left( \t \cos^2\phi -2i\t \cos \phi\sin \phi-\t \sin^2\phi\right) =\t \left( \cos^2\phi - \sin^2\phi\right), \label{eq_f_zero3}
\intertext{and}
\log\left(\left\vert \U{\F}(z) \right\vert^2\right)&=\log \left(\t\left(\cos \phi-i \sin \phi\right) \left(\cos \phi+i \sin \phi\right)\right) = \log\left(\t(\cos^2\phi+\sin^2\phi)\right)=\log \t. \label{eq_f_zero4}
\intertext{The remaining term of $\f$ is}
\tfrac{4\t(\t \Re(\Fz^2)-\vert \Fz\vert^2)}{1-\t^2} &=\tfrac{\t \left((1+\t)^2\cos^2\phi-(1-\t)^2 \sin^2 \phi\right)-\left((1+\t)^2 \cos^2\phi+(1-\t)^2\sin^2\phi \right)}{1-\t^2}\nn \\
&=\tfrac{\t(1+\t^2)(\cos^2\phi-\sin^2\phi)+2\t^2-(1+\t^2)-2\t(\cos^2\phi-\sin^2 \phi) }{1-\t^2}\nn \\
&=-\t(\cos^2\phi-\sin^2\phi)-1. \label{eq_f_zero5}
\end{align}
Putting \eqref{eq_f_zero3}, \eqref{eq_f_zero4} and \eqref{eq_f_zero5} in $\f$ from \eqref{eq_f} we get 
\begin{equation}
\f(z)=-\t(\cos^2\phi-\sin^2 \phi)-1+ \t \left( \cos^2\phi - \sin^2\phi\right)- \log \t+ \log \t+1=0.
\end{equation}
\end{proof}

\subsubsection{Maximum of $f$}
%
%
%
\paragraph*{On the Real Line}
On $\R\setminus [-\F,\F]$ $\f$ simplifies to
\begin{align}
\f(x,0)&=-\frac{4\t \Fx^2}{1+\t} + \left(\U{\F}(x,0)\right)^2-2\log\left\vert\U{\F}(x,0)\right\vert +\log \t +1.
\intertext{We have}
\frac{\partial\f}{\partial x}(x,0)&=\frac{1}{\F}\frac{\partial\f}{\partial \Fx}(\Fx,0) =\frac{1}{\F}\left(-\tfrac{8\t\Fx}{1+\t}+2\left(\Fx\mp\sqrt{\Fx^2-1}\right)\left(1\mp\tfrac{\Fx}{\sqrt{\Fx^2-1}}\right)-\tfrac{2\bigl(1\mp\tfrac{\Fx}{\sqrt{\Fx^2-1}} \bigr)}{\Fx\mp\sqrt{\Fx^2-1}} \right)\nn \\
&=\tfrac{4}{\F}\left(\Fx \tfrac{1-\t}{1+\t}\mp\sqrt{\Fx^2-1} \right),
\intertext{and}
\frac{\partial^2\f}{\partial x^2}(x,0)&=\frac{1}{\F^2}\frac{\partial^2\f}{\partial \Fx^2}(\Fx,0)=\tfrac{4}{\F^2}\left(\tfrac{1-\t}{1+\t}\mp\tfrac{\Fx}{\sqrt{\Fx^2-1}} \right).
\end{align}
As
\begin{equation}
\frac{\partial \f}{\partial x}(x,0)=0 \quad \Longleftrightarrow \quad \Fx \frac{1-\t}{1+\t}\mp\sqrt{\Fx^2-1}=0 \quad \Longleftrightarrow \quad \Fx=\pm \frac{1+\t}{2\sqrt{\t}}, 
\end{equation}
and
\begin{equation}
\frac{1}{\F^2}\frac{\partial^2\f}{\partial \Fx^2}\bigl( \pm \tfrac{1+\t}{2\sqrt{\t}},0\bigr)=-4<0,
\end{equation}
we find that the maximum of $\f$ is at $\Fx=\pm\tfrac{1+\t}{2\sqrt{\t}}$ which is equivalent to $x=\pm\sqrt{\tfrac{1+t}{1-t}}=\pm\mas$. As expected these are the points where the ellipse $\pEt$ cuts the real line and $\f$ takes the value zero. A plot of $\f$ on the real line can be found in Figure \ref{fig_f_real}.

\begin{figure}[!tp]
\begin{center}
\includegraphics[height=6cm,angle=0,scale=0.9]{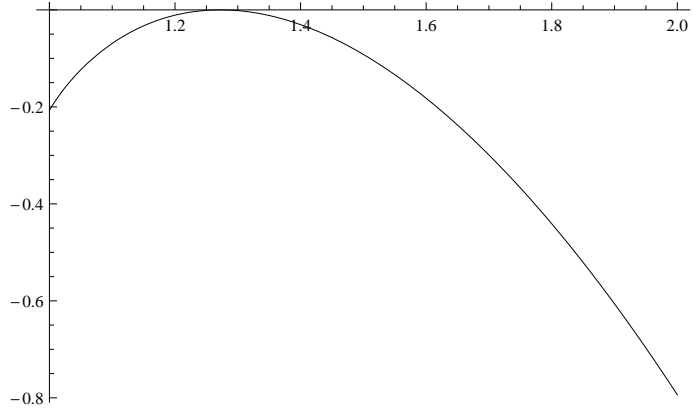}
\end{center}
\caption{\label{fig_f_real} Graph of $\f$ on the real line in the case $\t=\sqrt{5}-2$.}
\end{figure} 

\paragraph*{Variant 1}\label{subsubsec_f}
First we calculate the derivatives of \eqref{eq_f_xy}.
\begin{align}
\frac{\partial \f}{\partial x}(x,y)
&=2x(\t-1)+\U{\F}(x,y) \tfrac{\partial \U{\F}(x,y)}{\partial x} +\U{\F}(x,-y) \tfrac{\partial \U{\F}(x,-y)}{\partial x} \nn \\
& \quad -\tfrac{1}{\U{\F}(x,y)} \tfrac{\partial \U{\F}(x,y)}{\partial x} -\tfrac{1}{\U{\F}(x,-y)} \tfrac{\partial \U{\F}(x,-y)}{\partial x} \nn\\
&=2x(\t-1)-\tfrac{\left(\U{\F}(x,y)\right)^2}{\T{\F}(x,y)}-\tfrac{\left(\U{\F}(x,-y)\right)^2}{\T{\F}(x,-y)} +\tfrac{1}{\T{\F}(x,y)} +\tfrac{1}{\T{\F}(x,-y)} \nn\\
&=2x(\t-1)+\tfrac{2 \U{\F}(x,y)}{\F}+\tfrac{2 \U{\F}(x,-y)}{\F}=2x(\t-1)+\tfrac{4}{\F}\Re{\, \U{\F}(x,y)} \nn\\
&=\tfrac{4}{\F}\left(\Fx \tfrac{1-\t}{1+\t}\mp \Re \sqrt{\Fx^2-\Fy^2+2i\Fx \Fy -1} \right), \label{eq_df_x}
\intertext{and analogously}
\frac{\partial \f}{\partial y}(x,y)
&=-2y(\t+1)-\tfrac{4}{\F}\Im{\, \U{\F}(x,y)} \nn\\
&=\tfrac{4}{\F}\left(-\Fy \tfrac{1+\t}{1-\t}\pm \Im \sqrt{\Fx^2-\Fy^2+2i\Fx \Fy -1} \right). \label{eq_df_y}
\end{align}

\begin{lemma}\label{lemma_df_zero}
On the ellipse $\pEt$ we have $\frac{\partial\f}{\partial x}(x,y)=0$ and $\frac{\partial\f}{\partial y}(x,y)=0$.
\end{lemma}
\begin{proof}
We parametrize the ellipse $\pEt$ as in the proof of Lemma \ref{lemma_f_zero},
\begin{equation}
\Fz=\Fx+i\Fy, \qquad \Fx=\frac{1+\t}{2\sqrt{\t}} \cos \phi, \qquad \Fy=\frac{1-\t}{2\sqrt{\t}} \sin \phi, \qquad \phi\in (-\pi,\pi].
\end{equation}
Putting this in \eqref{eq_df_x} and \eqref{eq_df_y}, we find similarly as in \eqref{eq_f_zero1}
\begin{align}
\frac{\partial \f}{\partial x}(x,y)&=\tfrac{2(1-\t)\cos \phi \mp2 \Re \sqrt{(1+\t)^2\cos^2 \phi-(1-\t)^2\sin^2\phi+2i(1-\t^2)\cos \phi\sin \phi -4\t}}{\F\sqrt{\t}}\nn\\
&=\tfrac{2(1-\t)\cos \phi \mp 2\Re \sqrt{(1-\t)^2\cos^2\phi -(1+\t)^2\sin^2\phi+2i(1-\t^2)\cos \phi\sin \phi}}{\F\sqrt{\t}}\nn \\
&=\tfrac{2}{\F\sqrt{\t}}\left((1-\t)\cos \phi-\Re\left((1-\t)\cos \phi+i(1+\t)\sin \phi \right) \right)=0, \label{eq_df_zero1} \\
\intertext{and}
\frac{\partial \f}{\partial y}(x,y)&=\tfrac{-2(1+\t)\sin \phi \pm 2 \Im \sqrt{(1+\t)^2\cos^2 \phi-(1-\t)^2\sin^2\phi+2i(1-\t^2)\cos \phi\sin \phi -4\t}}{\F\sqrt{\t}} \nn\\
&=\tfrac{2}{\F\sqrt{\t}}\left(-(1+\t)\sin \phi+\Im\left((1-\t)\cos \phi+i(1+\t)\sin \phi \right) \right)=0. \label{eq_df_zero2}
\end{align} 
\end{proof}

\begin{lemma}\label{lemma_df_values}
On the ellipse 
\begin{equation*}
\{z\in \C \vert z = a \cos \phi+ i b \sin \phi, \phi\in (-\pi,\pi] \},
\end{equation*}
with the minor semi-axis $b\in (0,\infty)$ and major semi-axis $a=\sqrt{b^2+\F^2}>b$ the derivatives of $\f$ take the values 
\begin{equation}
\frac{\partial \f}{\partial x}(a \cos \phi,b \sin \phi)=\tfrac{\cos \phi}{\t} \left(a(1-\t)^2 -b(1-\t^2) \right),
\end{equation}
and
\begin{equation}
\frac{\partial \f}{\partial y}(a \cos \phi,b \sin \phi)=\tfrac{\sin \phi}{\t} \left(a(1-\t^2) -b(1+\t)^2 \right).
\end{equation}
\end{lemma}
\begin{proof}
We start to evaluate the derivatives of $\f$ on a general ellipse with independent semi-axes $a>0$ and $b>0$. From \eqref{eq_df_x} and \eqref{eq_df_y} we get
\begin{align}
\frac{\partial \f}{\partial x}(a \cos \phi,b \sin \phi)&=\tfrac{4a\cos \phi \frac{1-\t}{1+\t}\mp 4\Re\sqrt{a^2\cos^2 \phi-b^2\sin^2 \phi+2iab \cos \phi\sin \phi-\F^2}}{\F^2} \nn \\
&=\tfrac{4a\cos \phi \frac{1-\t}{1+\t}\mp 4\Re\sqrt{(a^2-F^2)\cos^2 \phi-(b^2+F^2)\sin^2 \phi+2iab \cos \phi\sin \phi}}{\F^2}, \label{eq_f_ellipse1}\\
\intertext{and analogously}
\frac{\partial \f}{\partial y}(a \cos \phi,b \sin \phi)&=\tfrac{-4b\sin \phi \frac{1+\t}{1-\t}\pm4\Im\sqrt{(a^2-F^2)\cos^2 \phi-(b^2+F^2)\sin^2 \phi+2iab \cos \phi\sin \phi}}{\F^2}. \label{eq_f_ellipse2}
\end{align}
Now, similarly as in \eqref{eq_f_zero1}, we want to use the $-\F^2$ to swap the factors in front of $\cos^2 \phi$ and $\sin^2 \phi$. This gives us the following dependence between the semi-axes $a$ and $b$
\begin{equation}\label{eq_f_ellipse_cond}
a^2=b^2+\F^2,
\end{equation}
as it is required by the lemma. Using this condition in \eqref{eq_f_ellipse1} and \eqref{eq_f_ellipse2} we get
\begin{align}
\frac{\partial \f}{\partial x}(a \cos \phi,b \sin \phi)&=\tfrac{4a\cos \phi \tfrac{1-\t}{1+\t}\mp4\Re\sqrt{(b\cos \phi+i a \sin \phi)^2}}{\F^2}\nn \\
&=\tfrac{4}{\F^2}\left(\tfrac{1-\t}{1+\t} a \cos \phi -b\cos \phi \right)=\tfrac{\cos \phi}{\t} \left(a(1-\t)^2 -b(1-\t^2) \right), \label{eq_f_ellipse3}
\intertext{and}
\frac{\partial \f}{\partial y}(a \cos \phi,b \sin \phi)&=\tfrac{4}{\F^2}\left(-\tfrac{1+\t}{1-\t} b \sin \phi +a\sin \phi \right) =\tfrac{\sin \phi}{\t} \left(a(1-\t^2) -b(1+\t)^2 \right). \label{eq_f_ellipse4}
\end{align}
\end{proof}

\begin{lemma}\label{lemma_df}
For $z=x+i y$ we have on the inside of the ellipse $\pEt$
\begin{align}
\sign(x) \frac{\partial \f}{\partial x}(x,y)&>0, \quad \forall z\in \Et \setminus (i\R\cup [-\F,\F]), \label{eq_lemma_df1a} \\
\shortintertext{and}
\sign(y) \frac{\partial \f}{\partial y}(x,y)&>0, \quad \forall z\in \Et \setminus \R, \label{eq_lemma_df1b} \\
\intertext{and on the outside of the ellipse $\pEt$}
\sign(x)\frac{\partial \f}{\partial x}(x,y)&<0, \quad \forall z\in\C\setminus(\Et\cup \pEt \cup i\R), \label{eq_lemma_df2a}\\
\shortintertext{and}
\sign(y)\frac{\partial \f}{\partial y}(x,y)&<0, \quad \forall z\in\C\setminus(\Et\cup \pEt \cup \R). \label{eq_lemma_df2b}
\end{align}
\end{lemma}
\begin{proof}
We choose $z=x+i y\in\C$. Using the elliptic coordinates of Section \ref{sub_elliptic} we define the minor semi-axis as
\begin{equation}\label{eq_semi_axis_b}
b^2=\frac{x^2+y^2-\F^2+\sqrt{4y^2 \F^2+(\F^2-x^2-y^2)^2}}{2}>0,
\end{equation} 
and the major semi-axis as $a^2=b^2+\F^2$. This ellipse fulfills the condition of Lemma \ref{lemma_df_values} and from \eqref{eq_z_on_ellipse} we know that $z$ lies on the ellipse. 
If $b<\mis=\sqrt{\frac{1-\t}{1+\t}}$, which is equivalent to
\begin{equation}
b^2<\frac{1-\t}{1+\t}=\frac{1+\t}{1-\t}-\frac{4\t}{1-\t^2}=\mas^2-\F^2\qquad \Longleftrightarrow \qquad a^2<\mas^2,
\end{equation} 
then the major semi-axis $a$ is also smaller than $\mas$. Then the ellipse with semi-axes $a$ and $b$, including the point $z$, lies completely inside of the ellipse $\pEt$. On the other hand, if $b>\mis$ also $a>\mas$ and the ellipse, including the point $z$, lies completely outside of $\pEt$.

By Lemma \ref{lemma_df_values} we find $\frac{\partial \f}{\partial x}$ at $z=x+iy=a \cos \phi+ib\sin \phi$,
\begin{equation} 
\frac{\partial \f}{\partial x}(x,y)=\tfrac{\cos \phi}{\t} \left(a(1-\t)^2 -b(1-\t^2) \right).
\end{equation}
If $\phi\neq \pm \tfrac{\pi}{2}$, obviously $\sign(x) \cos\phi >0$. Therefore $\sign(x)\frac{\partial \f}{\partial x}(x,y)>0$ if
\begin{equation}\label{eq_lemma_df_equiv}
\begin{split}
&a(1-\t)^2 -b(1-\t^2)>0\qquad \Longleftrightarrow\qquad a(1-\t) > b(1+\t)\\
&\qquad\Longleftrightarrow\qquad  \left(b^2+\frac{4\t}{1-\t^2}\right)(1-\t)^2=a^2(1-\t)^2 >b^2(1+\t)^2\\
&\qquad\Longleftrightarrow\qquad b^2<\frac{1-\t}{1+\t}=\mis^2, 
\end{split}
\end{equation}
which is equivalent to that $z$ lies inside of $\pEt$. On the other hand, $\sign(x) \frac{\partial \f}{\partial x}(x,y)<0$ if $z$ lies outside of $\pEt$.

By Lemma \ref{lemma_df_values} we find analogously
\begin{equation} 
\frac{\partial \f}{\partial y}(x,y)= \tfrac{\sin \phi}{\t} \left(a(1-\t^2) -b(1+\t)^2 \right).
\end{equation}
Here $\sin \phi$ compensates the change of sign of $\sign(y)$ in \eqref{eq_lemma_df1b} and \eqref{eq_lemma_df2b}. 
  
\begin{equation}
a(1-\t^2) -b(1+\t)^2 >0 \qquad \Longleftrightarrow\qquad a(1-\t)>b(1+\t),
\end{equation}
which is equivalent to \eqref{eq_lemma_df_equiv}. Therefore the discussion of the sign of $\frac{\partial \f}{\partial y}$ follows as above.
\end{proof}

\begin{prop}\label{prop_f_zero}
$\f(z)=0$ if $z\in \pEt$ and $\f(z)<0$ if $z\notin (\pEt\cup [-\F,\F])$.
\end{prop}
\begin{proof}
We choose a point $z_0\in \Et\setminus [-\F,\F]$. For simplicity we assume that $z_0\in \Cp$. By Lemma \ref{lemma_f_zero} we have that $\f=0$ on the ellipse $\pEt$. Lemma \ref{lemma_df} tells us that $\frac{\partial \f}{\partial x}$ is strictly positive on the right-hand side of $\Et$. So if we move from $z_0$ on a path parallel to the real axis to the right, the value of $\f$ must rise until $\f$ gets zero when we hit $\pEt$. Therefore $\f(z_0)$ must have been negative. Moving from there further to the right, the value of $\f$ must decrease as the lemma tells us that $\frac{\partial \f}{\partial x}$ is strictly negative outside of the ellipse. So it is obvious that for every point $z_1$ outside of $\pEt$ with $-\mis<\Im\, z_1<\mis$ we have that $\f(z_1)<0$. We are left to show that $\f$ is negative at a point $z_1$ outside of $\pEt$ when $\vert \Im\, z_1\vert\ge \mis$. Each such point can be reached by a path parallel to the imaginary axis starting either at a point of $\pEt$ (if $\Re\, z_1\le \mas$) or at a point of the real axis (if $\Re\, z_1> \mas$). In the first case we start with $f=0$, in the second case we have already shown that at the starting point on the real axis $f<0$. If $z_1$ lies in the first quadrant we have to move upwards and, by Lemma \ref{lemma_df}, $\frac{\partial \f}{\partial y}<0$ along this path. In any case, $f$ decreases until we reach $z_1$ where $f$ must be negative. Analogously we find the same result if $z_1$ lies in the forth quadrant where we have to move downwards and $\frac{\partial \f}{\partial y}>0$ along the path.
\end{proof}

A plot of $\f$ can be found in Figure \ref{fig_f_3d}. 

\begin{figure}[!th]
\begin{center}
\includegraphics[height=10cm,angle=0,scale=0.9]{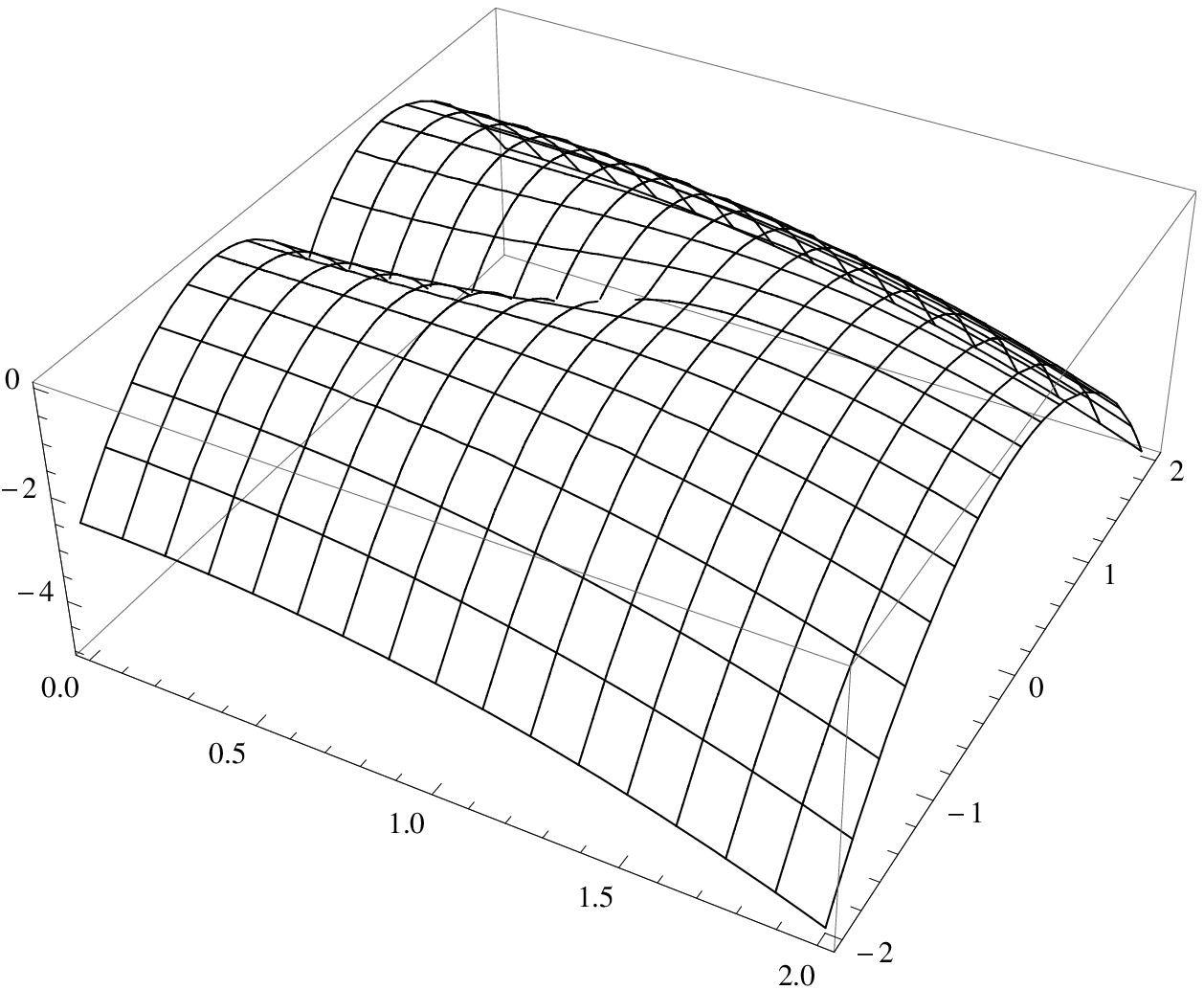}
\end{center}
\caption{\label{fig_f_3d} Graph of $\f$ in the case $\t=\sqrt{5}-2$ on $\{z\in\C\vert 0<\Re{\, z}<2, -2<\Im{\, z}<2 \}$.}
\end{figure} 


\paragraph*{Variant 2}

\begin{defin}
Let $0<\t<1$. For $u\in B_1(0)\setminus\{0\}$ we define the real-valued function $\f_U$ as
\begin{align}
\f_U(u)=\tfrac{\t^2}{1-\t^2} \Re\left(\left(u+u^{-1}\right)^2\right)-\tfrac{\t}{1-\t^2}\left\vert u+u^{-1} \right\vert^2 +\Re\left(u^2\right)-2 \log\vert u \vert+\log \t+1.
\end{align}
\end{defin}

With the help of $\f_U$ and Proposition \ref{prop_U_inv} we see that we can express $\f$ by the following composition
\begin{align}
\f(z)=\f_U\left(\U{\F}(z)\right), \qquad \forall z\in \C\setminus [-\F,\F]. \label{eq_f_fU}
\end{align}

\begin{prop} \label{prop_fU_max}
Let $\t\in (0,1)$. Then on $B_1(0)\setminus\{0\}$ $\f_U$ reaches its global maximum on the circle
\begin{align}
\{ z\in \C \vert \vert z \vert= \sqrt{\t} \}.
\end{align}
\end{prop}
\begin{proof}
We express $\f_U$ in polar coordinates $u=r\e^{i\phi}$
\begin{align}
\f_U(r,\phi)=& \tfrac{\t^2}{1-\t^2} \left(2 + \cos\left( 2\phi \right)\left(r^2+r^{-2}\right)\right) - \tfrac{\t}{1-\t^2} \left( 2\cos\left( 2\phi \right) +r^2+r^{-2} \right) + r^2 \cos\left( 2\phi \right) \nn \\
 & -2\log r +\log \t+1, \qquad \forall 0<r<1, \phi\in(-\pi,\pi]. \label{eq_fU_polar}
\end{align}
We have to find the critical points where $\left(\tfrac{\partial \f_U}{\partial r},\tfrac{\partial \f_U}{\partial \phi} \right)=(0,0)$.
For the partial derivatives we find
\begin{align}
\frac{\partial \f_U}{\partial r}(r,\phi)&=\frac{2 \left(\t - r^2 \right) \left(1 + r^2 \t - \left(\t + r^2\right) \cos(2 \phi) \right)}{r^3 \left(1 - \t^2\right)}, \\
\intertext{and}
\frac{\partial \f_U}{\partial \phi}(r,\phi)&=-\frac{2 \left(\t - r^2\right)^2 \sin( 2 \phi)}{ r^2 \left(1 - \t^2 \right) }.
\end{align}
So we see that when $r=\sqrt{\t}$, both partial derivatives are identically zero for all $\phi\in (-\pi,\pi]$ and that for $\vert r\vert < 1$ there are no further critical points. Thus all critical points lie on the circle around the origin with radius $\sqrt{\t}$. Since the derivative in tangential direction is identically zero for $r=\sqrt{\t}$, it is clear that $\f_U$ is constant on this circle. For the second derivative regarding $r$ evaluated at $\sqrt{\t}$ we get
\begin{align}
\frac{\partial^2 \f_U}{\partial r^2}(\sqrt{\t},\phi)&=-\frac{4 \left(1 + \t^2 - 2 \t \cos(2 \phi) \right)}{t \left(1 - t^2\right)}<0, \qquad \forall \phi\in (-\pi,\pi],
\end{align}
therefore the critical points are maxima.
\end{proof}

Additionally, with a straightforward computation, we get from \eqref{eq_fU_polar} that $\f_U$ is zero on the circle $\{ z\in \C \vert \vert z \vert= \sqrt{\t} \}$. We also know that $\U{\F}$ maps $\pEt$ to this circle and that $\f$ is given by \eqref{eq_f_fU}. Together with this, Proposition \ref{prop_fU_max} is equivalent to Proposition \ref{prop_f_zero}.

\subsubsection{Relation Between $g_\pm$ and Second Derivatives of $\f$ on the Ellipse}
From \eqref{eq_df_x} and \eqref{eq_df_y} we calculate the second derivatives of $\f$
\begin{align}
\frac{\partial^2 \f}{\partial x^2}(x,y)&=\frac{4}{\F^2}\left(\frac{1-\t}{1+\t} \mp \Re \frac{\Fx+i \Fy}{\sqrt{\Fx^2-\Fy^2+2i\Fx \Fy -1}} \right)= 2(t-1)-\frac{4}{\F} \Re{\left(\frac{\U{\F}(x,y)}{\T{\F}(x,y)} \right)} ,\label{eq_df_xx}
\intertext{and}
\frac{\partial^2 \f}{\partial y^2}(x,y)&=\frac{4}{\F^2}\left(-\frac{1+\t}{1-\t} \pm \Re \frac{\Fx+i \Fy}{\sqrt{\Fx^2-\Fy^2+2i\Fx \Fy -1}} \right) = -2(t+1)+\frac{4}{\F} \Re{\left(\frac{\U{\F}(x,y)}{\T{\F}(x,y)} \right)}. \label{eq_df_yy}
\end{align}


\begin{prop}\label{prop_g}
On the ellipse $\pEt$ parametrized by 
\begin{equation}
z=x+i y, \qquad x=\mas \cos \phi, \qquad y=\mis \sin \phi, \qquad \phi\in (-\pi,\pi],
\end{equation}
we have
\begin{align}
-\sign(x) \gp(x,y)&=\sqrt{-\frac{\partial^2\f}{\partial x^2}(x,y)}=\frac{2(1-\t)\vert \cos \phi\vert }{\sqrt{1+\t^2-2\t\cos(2\phi)}},
\intertext{and}
\sign(y) \gm(x,y)&=\sqrt{-\frac{\partial^2\f}{\partial y^2}(x,y)}=\frac{2(1+\t)\vert \sin \phi\vert }{\sqrt{1+\t^2-2\t\cos(2\phi)}}.
\end{align}
\end{prop}

\begin{proof}
We proceed similarly as in the proof of Lemma \ref{lemma_df_zero}. 
Putting the parametrization in \eqref{eq_df_xx} and \eqref{eq_df_yy}, we find as in \eqref{eq_df_zero1}
\begin{align}
\frac{\partial^2 \f}{\partial x^2}(x,y)&=\frac{4}{\F^2}\left( \frac{1-\t}{1+\t}- \Re \frac{(1+\t)\cos \phi+i(1-\t)\sin \phi}{(1-\t)\cos \phi+i(1+\t)\sin \phi} \right)  \nn\\
&=\frac{4}{\F^2}\left( \frac{1-\t}{1+\t}-\frac{1-\t^2}{1+\t^2-2\t \cos (2\phi)} \right) =-\frac{4 (1-\t)^2 \cos^2 \phi}{1+\t^2-2\t \cos(2\phi)}, \label{eq_lemma_g1}
\intertext{and analogously }
\frac{\partial^2 \f}{\partial y^2}(x,y)&=\frac{4}{\F^2}\left(-\frac{1+\t}{1-\t}+\frac{1-\t^2}{1+\t^2-2\t \cos (2\phi)} \right) =-\frac{4 (1+\t)^2 \sin^2 \phi}{1+\t^2-2\t \cos(2\phi)}. \label{eq_lemma_g2}
\end{align}

Now we evaluate the terms of $g_\pm$ from Proposition \ref{prop_asym} on the ellipse $\pEt$. From \eqref{eq_f_zero2} we find
\begin{align}
\Re \, \U{\F}(z)&=\sqrt{t}\cos \phi, \label{eq_lemma_g3}
\shortintertext{and}
\Im \, \U{\F}(z)&=-\sqrt{t}\sin \phi. \label{eq_lemma_g4}
\end{align}
From \eqref{eq_f_zero1} we see that
\begin{align}
\left\vert \Fz^2-1 \right\vert^{-1/2}&=\frac{1}{\vert \sqrt{\Fz^2-1}\vert}=\frac{2\sqrt{\t}}{\vert (1-\t) \cos \phi+i (1+\t) \sin \phi \vert}\nn\\
&=\frac{2\sqrt{\t}}{\sqrt{1+\t^2 + 2\t (\sin^2 \phi-\cos^2 \phi)}}=\frac{2\sqrt{\t}}{\sqrt{1+\t^2 - 2\t \cos(2\phi)}}.\label{eq_lemma_g5}
\intertext{We are left to evaluate}
\sqrt{\Fz\pm 1}&=\frac{1}{(4\t)^{1/4}}\sqrt{(1+\t)\cos \phi+i(1-\t)\sin \phi\pm 2\sqrt{t}}\nn \\
&=\frac{1}{(4\t)^{1/4}}\sqrt{\e^{i\phi}+\t \e^{-i\phi} \pm 2\sqrt{t}}=\frac{1}{(4\t)^{1/4}}\sqrt{(\e^{i\phi/2}\pm \sqrt{\t} \e^{-i\phi/2})^2 }\nn \\
&=\frac{1}{(4\t)^{1/4}}(\e^{i\phi/2}\pm \sqrt{\t} \e^{-i\phi/2}), \label{eq_lemma_g6}
\intertext{and get for the remaining term}
\left\vert\sqrt{\Fz+1}+\sqrt{\Fz-1}\right\vert^2&=\frac{1}{2\sqrt{\t}}\left\vert \e^{i\phi/2}+ \sqrt{\t} \e^{-i\phi/2}+\e^{i\phi/2}- \sqrt{\t} \e^{-i\phi/2} \right\vert^2\nn\\
&=\frac{1}{2\sqrt{\t}}\left\vert 2\e^{i\phi/2} \right\vert^2=\frac{2}{\sqrt{\t}}.\label{eq_lemma_g7}
\end{align}

Putting \eqref{eq_lemma_g3}, \eqref{eq_lemma_g4}, \eqref{eq_lemma_g5} and \eqref{eq_lemma_g7} in $\gp$ and $\gm$ from \eqref{eq_gp} and \eqref{eq_gm} respectively, we finally find
\begin{align}
\gp(\mas \cos \phi+i \mis \sin \phi)&=-\sqrt{t}\cos \phi \frac{2}{\sqrt{\t}} \frac{2\sqrt{\t}}{\sqrt{1+\t^2 - 2\t \cos(2\phi)}} \frac{1-\t}{2\sqrt{\t}}\nn \\
&=-\frac{2(1-\t)\cos \phi}{\sqrt{1+\t^2 - 2\t \cos(2\phi)}},
\intertext{and}
\gm(\mas \cos \phi+i \mis \sin \phi)&=\sqrt{t}\sin \phi \frac{2}{\sqrt{\t}} \frac{2\sqrt{\t}}{\sqrt{1+\t^2 - 2\t \cos(2\phi)}} \frac{1+\t}{2\sqrt{\t}} \nn\\
&= \frac{2(1+\t)\sin \phi}{\sqrt{1+\t^2 - 2\t \cos(2\phi)}}.
\end{align}
\end{proof}

A plot of $\gp$ and $\gm$ can be found in Figure \ref{fig_gpm}. 
%
%
\begin{figure}[!tbph]
\begin{center}
\includegraphics[height=8cm,angle=0,scale=0.9]{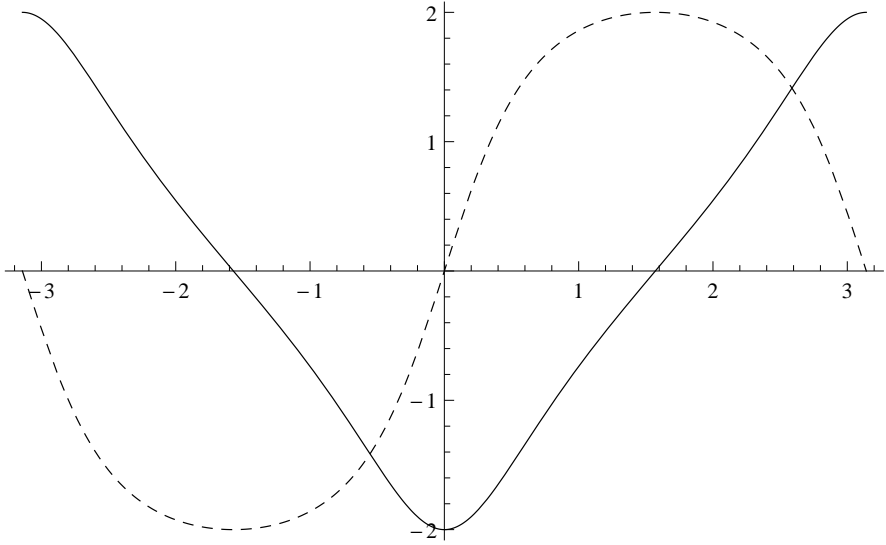}
\end{center}
\caption{\label{fig_gpm} Graph of $\gp$ (solid line) and $\gm$ (dashed line) in the case $\t=\sqrt{5}-2$ on the ellipse $\pEt$ parametrized by $\phi\in (-\pi,\pi]$ .}
\end{figure} 


We will see that in all of the following calculation we will nearly only need from $\f$ and $g_\pm$ the properties of $\f$ from Proposition \ref{prop_f_zero} and the relations and values of $g_\pm$ from Proposition \ref{prop_g}.

\subsubsection{Limit $\t\rightarrow 0$}
\begin{lemma}\label{lemma_lim_f_g}
Let $z\in \C\setminus \{0\}$. Then 
\begin{align}
\lim_{\t\rightarrow 0} \f(z)&=-\vert z\vert^2+ 2\log \vert z\vert +1, \\
\lim _{\t\rightarrow 0} \gp(z)&= -2\Re\left( \tfrac{1}{z} \right) \qquad \text{and} \qquad \lim _{\t\rightarrow 0} \gm(z)= -2\Im\left( \tfrac{1}{z} \right). \nn
\end{align}
\end{lemma}
\begin{proof}
For $z\neq 0$ according to Lemma \ref{lemma_limit_U_W}
\begin{align}
\lim_{\t\rightarrow 0} \U{\F}(z)&=\lim_{\F \rightarrow 0} \U{\F}(z)=0, \qquad \lim_{\t \rightarrow 0} \tfrac{ \U{\F}(z) }{\sqrt{\t}}=\lim_{\F\rightarrow 0} \tfrac{2 \U{\F}(z) }{\F}=\tfrac{1}{z}, \nn \\
\intertext{and}
\lim_{\t\rightarrow 0} \W{\F}(z)&=\lim_{\F \rightarrow 0} \W{\F}(z)=2.
\intertext{Therefore we find}
\lim_{\t\rightarrow 0} \f(z)&=\lim_{\t\rightarrow 0} \left(\t \Re(z^2)-\vert z\vert^2+\Re\left((\U{\F}(z))^2 \right)-2 \log \tfrac{\vert \U{\F}(z) \vert}{\sqrt{\t}} +1\right)\nn \\
&=-\vert z\vert^2+ 2\log \vert z\vert +1, \\
\lim_{\t\rightarrow 0} \gp(z)&=\lim_{\t\rightarrow 0} \left( -\Re\bigl(\tfrac{\U{\F}(z)}{\F} \bigr) \left\vert \W{\F}(z) \right\vert^2 \sqrt{\tfrac{1-\t}{1+\t}} \right)=-2\Re\left( \tfrac{1}{z} \right),
\intertext{and}
\lim_{\t\rightarrow 0} \gm(z)&=\lim_{\t\rightarrow 0} \left( -\Im\bigl(\tfrac{\U{\F}(z)}{\F} \bigr) \left\vert \W{\F}(z) \right\vert^2 \sqrt{\tfrac{1+\t}{1-\t}} \right)=-2\Im\left( \tfrac{1}{z} \right).
\end{align}
\end{proof}
\begin{remark}
Comparing the result of Lemma \ref{lemma_lim_f_g} with $\f_0$ and $g_{0,\pm}$ in \eqref{eq_f_g_0}, we see that in the limit $\t\rightarrow 0$ the functions $\f$ and $\gpm$ converge to the pure Gaussian case $t=0$.
\end{remark}

\subsubsection{Asymptotics of $\f$ for Large $\vert z\vert$}
\begin{prop}
For $z\gg \F$ the asymptotics of $\f$ is
\begin{align}
\f(z)=-\vert z \vert^2 +\t \Re(z^2) +2 \log \vert z\vert + \log (1-\t^2)+1+\lO \left(\tfrac{1}{z^2}\right).
\end{align}
\end{prop}
\begin{proof}
By expanding the square root in $\U{\F}$, we find for $z\gg \F$
\begin{align}
\U{\F}(z)=\tfrac{\F}{2z}+ \lO\left( \tfrac{1}{z^3} \right).
\end{align}
So we get for $\f$
\begin{align}
\f(z)&=-\vert z \vert^2 +\t \Re(z^2) +\Re\left((\U{\F}(z))^2\right) -2 \log \left(\tfrac{\F}{2\vert z\vert}\left(1+\lO\left(\tfrac{1}{z^2}\right) \right)\right)+\log \t+1 \nn \\
&=-\vert z \vert^2 +\t \Re(z^2) +2 \log \vert z\vert + \log\left(\tfrac{4\t}{\F^2}\right)+1+\lO \left(\tfrac{1}{z^2}\right).
\end{align}
\end{proof}

\section{Density for the Gaussian Potential}\label{sec_density}
Here we are going to calculate the density and in the next chapter the reproducing kernel. First we are just getting the result and we won't worry if our approximations are good enough and if exchange of limit and integration is allowed. We will verify in Chapter \ref{sec_verification} that the result is correct. 

\subsection{Density at Zero}
\begin{prop}\label{prop_density_zero}
Let $0<\delta<1$. Then for all $\t\in \C$ with $\vert \t\vert<1$ the density of eigenvalues at zero is a constant with the value
\begin{equation}
\rho(0)=\lim_{\n\rightarrow\infty} \rhon(0) = \frac{1}{\pi},
\end{equation}
and the limit converges uniformly for all $\vert \t\vert <1-\delta$.
\end{prop}
\begin{proof}
The density at zero can be calculated directly using that (see \cite[p. 432]{hille})
\begin{align}
H_m(0)&=\left\{\begin{array}{ll} 0 & \text{if $m$ is odd,}\\ \frac{(-1)^{m/2} m!}{(\frac{1}{2}m)!} & \text{if $m$ is even.} \end{array}\right.
\intertext{So we find for the orthonormal polynomials \eqref{eq_op_herm}}
\left\vert \op{m}(0)\right\vert^2&=\left\{\begin{array}{ll} 0 & \text{if $m$ is odd,}\\ \sqrt{1-\vert \t\vert^2}  \frac{ \n m! \vert \t\vert^m}{\pi \left(\left(\frac{1}{2}m\right)!\right)^2 2^m} & \text{if $m$ is even.} \end{array}\right. \label{eq_op_zero}
\end{align}

The function $(1-4x)^{-1/2}$ is analytic on $B_{\!\frac{1}{4}}(0)$. Using that
\begin{equation}
\begin{split}
\left. \frac{\dd^l}{\dd y^l} (1-4y)^{-1/2}\right\vert_{y=0}& =4^l \frac{1}{2}\cdot\frac{3}{2}\cdots\frac{2l-1}{2} (1-4y)^{-1/2-l}\bigg\vert_{y=0}\\
&=2^l (2l-1)!! = \frac{(2l)!}{l!},
\end{split}
\end{equation}
we get its Taylor series
\begin{equation}\label{eq_density_zero_taylor}
\frac{1}{\sqrt{1-4x}}=\sum_{l=0}^{\infty} \left. \frac{\dd^l}{\dd y^l} \frac{1}{\sqrt{1-4y}}\right\vert_{y=0} \frac{x^l}{l!}=\sum_{l=0}^{\infty} \frac{(2l)!}{(l!)^2}x^l,
\end{equation}
which converges uniformly for all $\vert x\vert \in [0,\tfrac{1-\delta}{4}]$.

Setting $m=2l$ in \eqref{eq_op_zero}, we find for the density of eigenvalues (Definition \ref{def_rho}) at zero
\begin{equation}
\begin{split}
\rho(0)&=\lim_{\n\rightarrow\infty}\rho_\n(0)=\lim_{\n\rightarrow\infty}\frac{1}{\n}\e^{-\n V(0)}\sum_{m=0}^{\n-1}\left\vert\op{m}(0)\right\vert^2 \\
&=\lim_{\n\rightarrow\infty} \frac{1}{\n} \sum_{l=0}^{\lfloor\frac{\n-1}{2}\rfloor}\sqrt{1-\vert \t\vert^2}  \frac{ \n (2l)! \vert \t\vert^{2l}}{\pi (l!)^2 2^{2l}} \\
&=\frac{1}{\pi}\sqrt{1-\vert\t\vert^2} \sum_{l=0}^\infty \frac{(2l)!}{(l!)^2} \left(\frac{\vert\t\vert^2}{4}\right)^{l}=\frac{1}{\pi}.
\end{split}
\end{equation}
In the last step we have used \eqref{eq_density_zero_taylor} with $x=\tfrac{\vert \t\vert^2}{4}$. Because $\vert x\vert \le\frac{(1-\delta)^2}{4}<\frac{1}{4}$, the sum converges uniformly for all $\vert \t\vert\le 1-\delta$.
\end{proof}
\begin{remark}
Contrary to other sections, in this proof we nowhere have used that $\t$ must be real.
\end{remark}

\subsection{Density at fixed $\protect\chapz$}\label{sec_density_fixed}

We can find the density $\rho_\n$ at a point $z_0\in \C$ by integration,
\begin{align}
\rho_\n(z_0)=\rho_\n(0)&+\int_0^{\Re\, z_0} \frac{\partial \rho_\n}{\partial x}(x,0) \dd x+\int_0^{\Im\, z_0} \frac{\partial \rho_\n}{\partial y}(\Re\, z_0,y) \dd y.
\intertext{According to \eqref{eq_drho_x} and \eqref{eq_drho_y} $\frac{\partial \rho_\n}{\partial x}$ and $\frac{\partial \rho_\n}{\partial y}$ can be expressed with the help of identities \eqref{eq_id1} and \eqref{eq_id2}, for which we know the asymptotics from Proposition \ref{prop_asym}. Therefore we can approximate the density as $\n \rightarrow \infty$ by}
\rho_\n(z_0)=\rho_\n(0)&+\int_0^{\Re\, z_0}  \sqrt{\tfrac{\n}{2\pi^3}} \gp(x,0) \e^{\n \f(x,0)}(1+o(1)) \dd x\nn\\
&+\int_0^{\Im\, z_0} (-1) \sqrt{\tfrac{\n}{2\pi^3}} \gm(\Re\, z_0,y) \e^{\n \f(\Re\, z_0,y)}(1+o(1)) \dd y. \label{eq_density_z0}
\end{align}
We will show that we only get a contribution from the integral if the integration path crosses the ellipse $\pEt$. Therefore the density will be constant inside of the ellipse, with the value $\tfrac{1}{\pi}$ according to Proposition \ref{prop_density_zero}. Across the ellipse the density has a jump and will be constant again outside. It is obvious (for example from Definition \ref{def_rho} together with Corollary \ref{cor_norm}) that the density is normalized to have integral $1$. Therefore the constant density on the non-compact outside of the ellipse has to be zero. We will show this by an explicit calculation which additionally will reveal that the density on $\pEt$ itself is $\tfrac{1}{2\pi}$.



\begin{prop}\label{prop_density}
The density of eigenvalues for the potential $V(z)=\vert z\vert^2-\Re(\t z^2)$  is given by
\begin{equation}
\rho(z)=\lim_{\n\rightarrow\infty} \rho_\n(z)=\left\{\begin{array}{lll} \tfrac{1}{\pi} & \text{if $z$ lies inside of $\pEt$,}\\ \tfrac{1}{2 \pi} & \text{if $z\in \pEt$,}\\ 0 & \text{if $z$ lies outside of $\pEt$.}\end{array}\right. 
\end{equation}
\end{prop}
\begin{proof}
Proposition \ref{prop_density_zero} has shown that the density at zero is $\frac{1}{\pi}$. First we will prove that the density is constant inside and outside of $\pEt$. And then we will calculate the jump across $\pEt$. 

We can find the density by \eqref{eq_density_z0}. We will show that parts of these integrals will contribute nothing as long as we don't cross the ellipse $\pEt$. We look at two point $z_1$ and $z_2$ which lie either on a line parallel to the real axis or on a line parallel to the imaginary axis, i.e.\ $\Im\, z_1=\Im\, z_2$ or $\Re\, z_1=\Re\, z_2$ respectively. We further assume that $z_1$ and $z_2$ lie both inside or outside of $\pEt$, i.e.\ the line segment $[z_1,z_2]$ does not cross the ellipse. Then Proposition \ref{prop_f_zero} tells us that $\f$ is strictly negative on $[z_1,z_2]$, i.e.\ there exists a $\epsilon>0$ so that $\f(z)<-\epsilon$ for all $z\in [z_1,z_2]$. $g_\pm$ are continuous functions and therefore are bounded on a compact set, i.e.\ there exists a constant $C<\infty$ so that $\vert \gp(z)\vert<C$ and $\vert\gm(z)\vert<C$ for all $z\in [z_1,z_2]$. Then we can estimate
\begin{align}
\vert\rho(z_2)-\rho(z_1)\vert&=\lim_{\n\rightarrow\infty} \left\vert \int_{z_1}^{z_2} \sqrt{\frac{\n}{2\pi^3}} g_\pm(\xi) \e^{\n \f(\xi)}\dd \xi \right \vert \le\lim_{\n\rightarrow\infty} \sqrt{\frac{\n}{2\pi^3}} C \e^{-\n \epsilon} \left\vert \int_{z_1}^{z_2} \dd\xi \right \vert=0, 
\end{align}
where $\xi=(x,\Im\, z_1)$ (if we integrate along a path parallel to the real axis) or $\xi=(\Re\, z_1,y)$ (if we integrate along a path parallel to the imaginary axis) and $\dd \xi=\dd x$ or $ \dd \xi=\dd y$ respectively. We have showed that the density is constant inside and outside of $\pEt$.

Now we consider a point $z_1=x_1+i y_1$ inside but near the ellipse $\pEt$ and a point $z_2=x_2+i y_2$ near $z_1$ ($z_2$ is free to lie inside, outside or on $\pEt$). Further we assume that either $y_1=y_2$ or $x_1=x_2$. As $z_1$ lies inside the ellipse we know that $\rho(z_1)=\frac{1}{\pi}$. Again we will consider a line segment starting at $z_1$ either parallel to the real axis or parallel to the imaginary axis. The point where the line crosses $\pEt$ will be called $z_E=x_E+i y_E$. As $g_\pm$ are continuous they won't change much in the near of $z_E$. We assume that $z_1$ and $z_2$ lie so near to $z_E$ that the change of $g_\pm$ along the line segment $[z_1,z_2]$ can be neglected, i.e.\ $g_\pm(z)\approx g_\pm(z_E)$ for all $z\in [z_1,z_2]$.
Now we make a Taylor expansion up to order two of $\f$ at $z_E$ as a function of either $x$ or $y$ respectively. So we get
\begin{align}
\f(x,y_E)&=\f(x_E,y_E)+\frac{\partial \f}{\partial x}(x_E,y_E)(x-x_E)+\frac{\partial^2 \f}{\partial x^2}(x_E,y_E)\tfrac{(x-x_E)^2}{2}+\mathcal{O}\left((x-x_E)^3\right),
\intertext{or}
\f(x_E,y)&=\f(x_E,y_E)+\frac{\partial \f}{\partial y}(x_E,y_E)(y-y_E)+\frac{\partial^2 \f}{\partial y^2}(x_E,y_E)\tfrac{(y-y_E)^2}{2}+\mathcal{O}\left((y-y_E)^3\right).
\end{align}
By Proposition \ref{prop_f_zero} and Lemma \ref{lemma_df_zero} we know that on the ellipse $\pEt$ we have that $\f=0$ as well as the first derivatives $\frac{\partial \f}{\partial x}=0$ and $\frac{\partial \f}{\partial y}=0$. Again we assume that $z_1$ and $z_2$ lie so near to $z_E$ that we can neglect the corrections $\mathcal{O}\left((x-x_E)^3\right)$ or $\mathcal{O}\left((y-y_E)^3\right)$, i.e.\ we have $f(x,y_E)\approx \frac{\partial^2 \f}{\partial x^2}(x_E,y_E)\tfrac{(x-x_E)^2}{2}$ or $f(x_E,y)\approx \frac{\partial^2 \f}{\partial y^2}(x_E,y_E)\tfrac{(y-y_E)^2}{2}$ on $[z_1,z_2]$. Then we find for the density at $z_2$
\begin{equation}\label{eq_proof_density}
\begin{split}
\rho(z_2)&=\rho(z_1)+\lim_{\n\rightarrow\infty} \int_{x_1}^{x_2} \sqrt{\frac{\n}{2\pi^3}} \gp(x,y_E) \e^{\n \f(x,y_E)}\dd x \\
&\approx \frac{1}{\pi}+\lim_{\n\rightarrow\infty}\sqrt{\frac{\n}{2\pi^3}} \gp(x_E,y_E) \int_{x_1}^{x_2} \e^{-\eta_+^2 (x-x_E)^2}\dd x\\ 
&= \frac{1}{\pi}+\lim_{\n\rightarrow\infty}\sqrt{\frac{\n}{2\pi^3}} \frac{\gp(x_E,y_E)}{\eta_+} \int_{\eta_+ (x_1-x_E)}^{\eta_+ (x_2-x_E)} \e^{-x^2} \dd x\\ 
&= \frac{1}{\pi}-\sqrt{\frac{1}{\pi^3}} \sign(x_E) \int_{-\infty \sign(x_E)}^{\infty (x_2-x_E)} \e^{-x^2} \dd x\\ 
&= \frac{1}{\pi}-\left\{ \begin{array}{lll} \pi^{-3/2}\int_{-\infty}^{-\infty} \e^{-x^2} \dd x \\ \pi^{-3/2}\int_{-\infty}^{0}\e^{-x^2} \dd x\\ \pi^{-3/2}\int_{-\infty}^{\infty}\e^{-x^2} \dd x \end{array}\right. = \left\{ \begin{array}{lll} \frac{1}{\pi} &\text{if $\vert x_2\vert<\vert x_E\vert$,}\\ \frac{1}{2\pi} &\text{if $x_2=x_E$,}\\ 0 &\text{if $\vert x_2\vert>\vert x_E\vert$.} \end{array}\right.
\end{split}
\end{equation}
In the step from line two to line three we have substituted $\eta_+ (x-x_E)$ by $x$ where $\eta_+=\sqrt{-\cramped{\frac{\partial^2 \f}{\partial x^2}}(x_E,y_E) \cramped{\frac{\n}{2}}>0}$ and $\eta_+\rightarrow \infty$ when $\n\rightarrow \infty$. Next we have used that by Proposition~\ref{prop_g} $\sqrt{-\cramped{\frac{\partial^2 \f}{\partial x^2}}(x_E,y_E)}=-\sign(x_E) \gp(x_E,y_E)$.

When we integrate on a path parallel to the imaginary axis we get analogously
\begin{equation}
\begin{split}
\rho(z_2)&=\rho(z_1)+\lim_{\n\rightarrow\infty} \int_{y_1}^{y_2} (-1) \sqrt{\frac{\n}{2\pi^3}} \gm(x_E,y) \e^{\n \f(x_E,y)}\dd y \\
&\approx \frac{1}{\pi}-\lim_{\n\rightarrow\infty}\sqrt{\frac{\n}{2\pi^3}} \gm(x_E,y_E) \int_{y_1}^{y_2} \e^{-\eta_-^2 (y-y_E)^2}\dd y\\
&= \frac{1}{\pi}-\lim_{\n\rightarrow\infty}\sqrt{\frac{\n}{2\pi^3}} \frac{\gm(x_E,y_E)}{\eta_-} \int_{\eta_- (y_1-y_E)}^{\eta_- (y_2-y_E)} \e^{-y^2} \dd y\\ 
&= \frac{1}{\pi}-\sqrt{\frac{1}{\pi^3}} \sign(y_E) \int_{-\infty\sign(y_E)}^{\infty (y_2-y_E)} \e^{-y^2} \dd y\\ 
&= \left\{ \begin{array}{lll} \frac{1}{\pi} &\text{if $\vert y_2\vert<\vert y_E\vert$,}\\ \frac{1}{2\pi} &\text{if $y_2=y_E$,}\\ 0 &\text{if $\vert y_2\vert>\vert y_E\vert$,} \end{array}\right. 
\end{split}
\end{equation}
where $\eta_-=\sqrt{-\cramped{\frac{\partial^2 \f}{\partial y^2}}(x_E,y_E) \cramped{\frac{\n}{2}}}>0$. Again from Proposition \ref{prop_g}, we have used that $\sqrt{-\cramped{\frac{\partial^2 \f}{\partial y^2}}(x_E,y_E)}=\sign(y_E)\gm(x_E,y_E)$. 
We have proven that the jump across $\pEt$ is $\frac{1}{\pi}$ and that the density on the ellipse itself is half of the density inside.
\end{proof}
\begin{remark}
The area of the ellipse is $\pi \mas \mis=\pi$. So the density is normalized to 1 as it is required.
\end{remark}

\subsection{Density on the Ellipse in a Scaled Limit}
Now we want to study how the density drops to zero across the ellipse. For this we have to scale the coordinates appropriate with the numbers of eigenvalues $\n$. In the Hermitian case, 
the suitable scaling was to hold the distance between eigenvalues constant. Therefore the coordinates have been scaled with $\n$. As in our case the eigenvalues fill a domain in the complex plain, we have to scale the coordinates with $\sqrt{\n}$ for that the mean distance between eigenvalues stays constant. In the following we are going to calculate the density at $z_0+\frac{a}{\sqrt{\n}}$ when $z_0\in\pEt$ and $a\in \C$. We will find that the density drops to zero as the complementary error function 
\begin{equation}
\erfc(z)=1-\erf(z)=\frac{2}{\sqrt{\pi}}\int_z^\infty \e^{-w^2} \dd w.
\end{equation}

\begin{theorem}\label{th_density_scaled_limit}
The density of eigenvalues at a point that scales to $z_0$ with $\n^{-1/2}$ is given by 
\begin{align}
\lim_{\n\rightarrow\infty}\rho_\n\left(z_0+\frac{a}{\sqrt{\n}}\right)&=\frac{1}{\pi}\left\{\begin{array}{lll}1 & \text{if $z_0$ lies inside of $\pEt$,}\\ \frac{1}{2} \erfc\left(\zeta(a) \right) & \text{if $z_0\in \pEt$,}\\ 0 & \text{if $z_0$ lies outside of $\pEt$,}\end{array}\right.
\intertext{with}
\zeta(a)&=\frac{\sqrt{2}\bigl((1-\t)\cos\phi \Re\, a+(1+\t)\sin\phi\Im\, a \bigr)}{\sqrt{(1+\t)^2-4\t \cos^2\phi}}, \label{eq_zeta}
\end{align}
where $\phi$ is the parameter so that $z_0=\mas \cos \phi+i \mis \sin \phi$.
\end{theorem}
\begin{proof}
For simplicity we assume that $z_0\in \Cp$. If $z_0$ lies inside (or outside) of $\pEt$, for all $\n$ large enough $z_0+\tfrac{a}{\sqrt{\n}}$ will lie inside (or outside) of the ellipse too and the density there will be $\tfrac{1}{\pi}$ (or zero respectively). That the density does not change is obvious from \eqref{eq_proof_density} as nothing gets changed since the integration boundaries still go from $-\infty$ to $\infty$.

So we are left to prove the more interesting case if $z_0$ lies on the ellipse. Let $z_0=x_0+i y_0\in\pEt$ with $z_0=\mas \cos\phi+i\mis \sin\phi$. As in Section \ref{sec_density_fixed} we integrate along a line parallel to the real axis. We choose a point $z_1=x_1+i y_1$ 
 near enough $z_0$ with $y_1=y_0$ so that $z_1+\tfrac{i \Im\, a}{\sqrt{n}}$ lies inside the ellipse for all sufficient large $\n$. This is possible if $z_0$ is not on the imaginary axis, in which case we alternatively have to integrate along a line parallel to the imaginary axis.

The path of integration will be the line $[z_1+\tfrac{i \Im{\, a}}{\sqrt{\n}},z_0+\tfrac{a}{\sqrt{\n}}]$. The path will slightly get moved parallel to the real axis when we increase $\n$. Again we will call $z_E=x_E+i y_E$ to be the point of intersection between the ellipse $\pEt$ and the line parallel to the real axis through $z_0+\tfrac{a}{\sqrt{\n}}$. So $z_E$ depends on $\n$ with $y_E=y_0+\tfrac{\Im{\, a}}{\sqrt{\n}}$. We assume that $\n$ is large enough and $z_1$ near enough $z_0$ so that the error of $\gp(z)\approx \gp(z_E)$ and $\f(x,y_E)\approx \tfrac{\partial^2\f}{\partial x^2}(x_E,y_E)\frac{(x-x_E)^2}{2}$ on $[z_1+\tfrac{i \Im{\, a}}{\sqrt{\n}},z_0+\tfrac{a}{\sqrt{\n}}]$ can be neglected. Then we will get the density as in \eqref{eq_proof_density}
\begin{equation}\label{eq_th_scaled_density1}
\begin{split}
\rho_\n\left(z_0+\tfrac{a}{\sqrt{\n}}\right)&\approx\rho_\n\left(z_1+\tfrac{\Im{\, a}}{\sqrt{\n}}\right)+\sqrt{\tfrac{\n}{2\pi^3}} \int_{x_1}^{x_0+\frac{\Re\, a}{\sqrt{\n}}} \gp(x,y_E) \e^{\n \f(x,y_E)}\dd x \\
&\approx \rho_\n\left(z_1+\tfrac{\Im{\, a}}{\sqrt{\n}}\right)+\sqrt{\tfrac{\n}{2\pi^3}} \gp(x_E,y_E) \int_{x_1}^{x_0+\frac{\Re\, a}{\sqrt{\n}}} \e^{-\eta_+^2 (x-x_E)^2}\dd x\\ 
&=\rho_\n\left(z_1+\tfrac{\Im{\, a}}{\sqrt{\n}}\right)-\pi^{-3/2} \int_{\eta_+ (x_1-x_E)}^{\eta_+ (x_0+\frac{\Re\, a}{\sqrt{\n}}-x_E)} \e^{-x^2} \dd x,
\end{split}
\end{equation}
where $\eta_+=\sqrt{-\frac{\partial^2 \f}{\partial x^2}(x_E,y_E) \frac{\n}{2}}=-\gp(x_E,y_E)\sqrt{\tfrac{\n}{2}}>0$ as in the proof of Proposition \ref{prop_density}. The lower boundary of integration, $\eta_+ (x_1-x_E)$, still goes to $-\infty$ when $\n\rightarrow\infty$, but the upper boundary needs a more careful consideration. For $x_E$ we find
\begin{align}
x_E&=\frac{\mas}{\mis}\sqrt{\mis^2-y_E^2}=\frac{\mas}{\mis}\sqrt{\mis^2-y_0^2-\frac{2 y_0 \Im\, a}{\sqrt{\n}}-\frac{(\Im\, a)^2}{\n}}\nn\\
&=\frac{\mas}{\mis}\sqrt{\mis^2-y_0^2}\sqrt{1-\frac{2 y_0 \Im\, a}{\sqrt{\n}(\mis^2-y_0^2)}-\frac{(\Im\, a)^2}{\n (\mis^2-y_0^2)}}\nn\\
&=x_0 \left(1-\frac{y_0 \Im\, a}{\sqrt{\n}(\mis^2-y_0^2)} \right)+\mathcal{O}\left(\n^{-1}\right).
\intertext{So the upper boundary of integration gets}
\eta_+\left(x_0+\frac{\Re\, a}{\sqrt{\n}}-x_E\right)&=\eta_+\left(\frac{\Re\, a}{\sqrt{\n}}+\frac{x_0 y_0 \Im\, a}{\sqrt{\n}(\mis^2-y_0^2)}+\mathcal{O}\left(\n^{-1}\right) \right)\nn\\
&=\sqrt{-\frac{\partial^2 \f}{\partial x^2}(x_E,y_E) \frac{1}{2}}\left(\Re\, a+\frac{x_0 y_0 \Im\, a}{\mis^2-y_0^2}\right)+\mathcal{O}\left(\n^{-1/2}\right).
\end{align}
As $\frac{\partial^2 \f}{\partial x^2}$ is continuous and $z_E$ goes to $z_0$ when $\n\rightarrow\infty$, we will find
\begin{equation}
\begin{split}
\zeta(a)&=\lim_{\n\rightarrow\infty} \eta_+\left(x_0+\frac{\Re\, a}{\sqrt{\n}}-x_E\right)=\sqrt{-\frac{\partial^2 \f}{\partial x^2}(x_0,y_0) \frac{1}{2}}\left(\Re\, a+\frac{x_0 y_0 \Im\, a}{\mis^2-y_0^2}\right)\\
&=\frac{\sqrt{2}(1-\t)\cos\phi}{\sqrt{1+\t^2-2\t \cos(2\phi)}} \left(\Re\, a+\frac{\mas \cos\phi \mis \sin\phi \Im\, a}{\mis^2-\mis^2\sin^2\phi}\right)\\
&=\frac{\sqrt{2}(1-\t)\cos\phi\Re\, a}{\sqrt{(1+\t)^2-4\t \cos^2\phi}}+\frac{\sqrt{2}(1-\t)\cos\phi}{\sqrt{(1+\t)^2-4\t \cos^2\phi}}\frac{\mas \sin\phi \Im\, a}{\mis \cos\phi}\\
&=\frac{\sqrt{2}\bigl((1-\t)\cos\phi\Re\, a+(1+\t)\sin\phi\Im\, a\bigr)}{\sqrt{(1+\t)^2-4\t \cos^2\phi}},
\end{split}
\end{equation}
where we have used the value of $\sqrt{-\frac{\partial^2 \f}{\partial x^2}(x_0,y_0)}$ from Proposition \ref{prop_g}.

We already have shown that $\rho_\n\bigl(z_1+\tfrac{\Im{\, a}}{\sqrt{\n}}\bigr)$ goes to $\rho(z_1)=\frac{1}{\pi}$ when $\n\rightarrow\infty$. So we finally find for \eqref{eq_th_scaled_density1}, in the limit $\n\rightarrow \infty$,
\begin{equation}
\begin{split}
\lim_{\n\rightarrow\infty} \rho_\n\left(z_0+\tfrac{a}{\sqrt{\n}}\right)&=\frac{1}{\pi}-\pi^{-3/2} \int_{-\infty}^{\zeta(a)} \e^{-x^2} \dd x =\pi^{-3/2}\left(\int_{-\infty}^{\infty} \e^{-x^2} \dd x- \int_{-\infty}^{\zeta(a)} \e^{-x^2} \dd x\right)\\
&=\pi^{-3/2}\int_{\zeta(a)}^{\infty} \e^{-x^2} \dd x=\frac{1}{2\pi}\erfc(\zeta(a)).
\end{split}
\end{equation}
\end{proof}
\begin{remark}
The same result we can get with an analogous calculation by integrating along a line parallel to the imaginary axis.
\end{remark}

From \eqref{eq_zeta} it is obvious that the density depends on the position on the ellipse and on the parameter $\t$. This is not astonishing as the system of coordinates for $a$ is parallel to the real axis. Now we want to look what happens when we turn the system of coordinates for $a$ so that the real axis is normal to the tangent at the ellipse in $z_0$ and points outward. 

\begin{cor}\label{cor_density_scaled_limit}
Let $z_0=\mas \cos \phi+i\mis \sin\phi$ be a point on the ellipse $\pEt$ with $\phi\in (-\pi,\pi]$  and $\psi$ the angle between the real axis and the outward normal to the tangent at $\pEt$ in $z_0$, i.e.\ 
\begin{align}
\psi&=\left\{\begin{array}{ll} \arctan\left(\frac{1+\t}{1-\t}\tan \phi \right) & \text{if } \phi\in (-\tfrac{\pi}{2},\tfrac{\pi}{2}), \\ \arctan\left(\frac{1+\t}{1-\t}\tan \phi \right)+\sign(\phi)\pi & \text{if } \phi\notin [-\tfrac{\pi}{2},\tfrac{\pi}{2}],  \\  \phi & \text{if }\phi=\pm \tfrac{\pi}{2}. \end{array}\right.
\end{align}
Then the density of eigenvalues is given by
\begin{equation}
\lim_{\n\rightarrow\infty}\rho_\n\left(z_0+ \frac{a \e^{i\psi}}{\sqrt{\n}}\right)=\frac{1}{2\pi}\erfc(\sqrt{2}\Re\, a).
\end{equation}
\end{cor}
\begin{proof}
For simplicity we assume that $\phi\in (-\tfrac{\pi}{2},\tfrac{\pi}{2})$ --- the proof for the other cases works analogously. The slope of the outward normal to the tangent at $\pEt$ in $z_0=x_0+i y_0$ is
\begin{equation}
\begin{split}
\tan \psi&=-\left(\frac{\partial}{\partial x_0}\Bigl(\pm\tfrac{\mis}{\mas}\sqrt{\mas^2-x_0^2}\Bigr) \right)^{-1}=\pm \frac{\mas}{\mis}\frac{\sqrt{\mas^2-x_0^2}}{x_0}\\
&=\pm \frac{\mas}{\mis} \frac{\sqrt{\mas^2 \sin^2\phi}}{\mas \cos \phi}=\frac{1+\t}{1-\t}\tan\phi.
\end{split}
\end{equation}
Substituting $a$ by $a \e^{i\psi}$ in Theorem \ref{th_density_scaled_limit} and using that
\begin{align}
\Re \left(a \e^{i\psi} \right)&=\cos \psi \Re\, a-\sin \psi \Im\, a,\qquad \Im\left(a \e^{i\psi} \right)=\sin \psi \Re\, a+\cos \psi \Im\, a, \nn \\
\cos \psi&=\frac{1}{\sqrt{1+\tan^2\psi}}=\frac{1}{\sqrt{1+\left(\frac{1+\t}{1-\t}\tan\phi \right)^2}}=\frac{(1-\t)\cos\phi}{\sqrt{(1+\t)^2 -4\t \cos^2 \phi}}, \label{eq_cos_psi}
\shortintertext{and}
\sin \psi&=\frac{\tan\psi}{\sqrt{1+\tan^2\psi}}=\frac{\frac{1+\t}{1-\t}\tan\phi}{\sqrt{1+\left(\frac{1+\t}{1-\t}\tan\phi \right)^2}}=\frac{(1+\t)\sin\phi}{\sqrt{(1+\t)^2 -4\t \cos^2 \phi}}, \label{eq_sin_psi}
\end{align}
we get 
\begin{align}
\zeta(a \e^{i\psi})=&\frac{\sqrt{2}\bigl((1-\t)\cos\phi \Re\! \left(a\e^{i\psi}\right) +(1+\t)\sin\phi\Im\! \left(a\e^{i\psi}\right) \bigr)}{\sqrt{(1+\t)^2-4\t \cos^2\phi}}\nn\\
=&\frac{\sqrt{2}\bigl( (1-\t)^2\cos^2\phi \Re\, a- (1-\t^2)\sin\phi\cos\phi \Im\, a\bigr)}{(1+\t)^2-4\t \cos^2\phi}\nn\\
& +\frac{\sqrt{2}\left((1+\t)^2\sin^2\phi \Re\, a +(1-\t^2)\sin\phi\cos\phi \Im\, a \right)}{(1+\t)^2-4\t \cos^2\phi} = \sqrt{2}\Re\, a.
\end{align}
Together with the result of Theorem \ref{th_density_scaled_limit} on $\pEt$, this proves the corollary.
\end{proof}
So we got a density that is independent of the position $z_0$ on the ellipse, of $\t$ and of $\Im\, a$. A plot of it can be found in Figure \ref{fig_erfc}.
\begin{figure}[!tbp]
\begin{center}
\includegraphics[height=8cm,angle=0,scale=0.9]{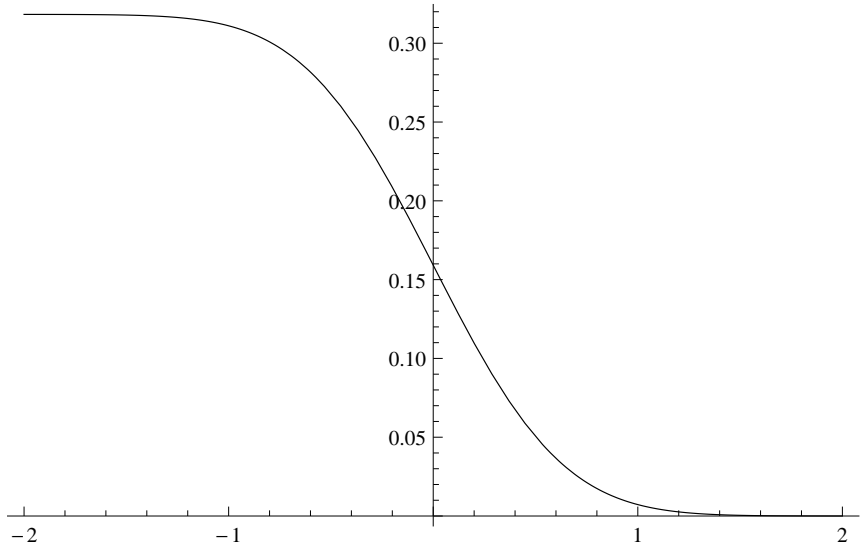}
\end{center}
\caption{\label{fig_erfc} Graph of the density $\lim_{\n\rightarrow\infty} \rhon(\mas+\tfrac{a}{\sqrt{n}})$ for $a\in[-2,2]$.}
\end{figure} 

\begin{remark}
This result is consistent with the result of Proposition \ref{prop_density} as $$\lim_{a\rightarrow-\infty}\erfc(a)=2,\qquad\lim_{a\rightarrow\infty}\erfc(a)=0\qquad \text{and}\qquad \erfc(0)=1.$$ 
\end{remark}


\section{Asymptotics of the Kernel and Correlation Functions for Gaussian Potentials}\label{ch_kernel_corr} 
\sectionmark{Asymptotics of the Kernel}
So far we only have calculated the density, i.e.\ we have evaluated $\HM(w,z)$ for $w=\cc{z}$. Now we are interested in $\HM(w,z)$ if $w\neq z$. But for our purpose it is sufficient to calculate $\HM(w,z)$ when the difference between $\cc{w}$ and $z$ is small, i.e.\ if $\cc{w}$ and $z$ converge to a common point $z_0\in \C$ as $\n\rightarrow \infty$. 
As when we have been interested in the density in scaled coordinates, we will use a scaling such that the mean distance between eigenvalues stays fixed. 
Therefore, in the complex plane, we have to use the scaling $w=\cc{z}_0+\frac{\cc{a}}{\sqrt{\n}}$, $z=z_0+\frac{b}{\sqrt{\n}}$ with $a,b\in\C$.

We can get $\HM(w,z)$ by integration
\begin{equation}\label{eq_hm_wz}
\begin{split}
\HM\left(\cc{z}_0+\frac{\cc{a}}{\sqrt{\n}},z_0+\frac{b}{\sqrt{\n}}\right) = & \HM(\cc{z}_0,z_0)+\int_0^{\frac{\cc{a}}{\sqrt{\n}}} \frac{\partial \HM}{\partial w}\left(\cc{z}_0+u,z_0\right)\dd u\\
&+\int_0^{\frac{b}{\sqrt{\n}}} \frac{\partial \HM}{\partial z}\left(\cc{z}_0+\frac{\cc{a}}{\sqrt{\n}},z_0+v\right)\dd v,
\end{split}
\end{equation}
where $\HM(\cc{z}_0,z_0)=\rho_\n(z_0)$ is known in the limit $\n\rightarrow \infty$ from Proposition \ref{prop_density}. Since $\tfrac{\partial \HM}{\partial w}\left(\cc{z}_0+u,z_0\right)$ is holomorphic in $u$ and $\frac{\partial \HM}{\partial z}\left(w,z_0+v\right)$ in $v$ the integrals are independent of of the integration paths and for simplicity we will choose them to be the straight lines $[0,\frac{\cc{a}}{\sqrt{\n}}]$ and $[0,\frac{b}{\sqrt{\n}}]$ in the complex plane.

\subsection{Identities for $\frac{\partial \protect\HM}{\partial \protect\chapw}$ and $\frac{\partial \protect\HM}{\partial \protect\chapz}$}
\begin{prop}\label{prop_ident_wz}
$H_\n(w,z)$ fulfills the following identities
\begin{align}
\frac{\partial \HM}{\partial w}(w,z)&=\bigl( \t \op{\n}(w) \op{\n-1}(z)-\op{\n-1}(w) \op{\n}(z)\bigr)\frac{\e^{\n\left(-w z+\frac{\t w^2}{2}+\frac{\t z^2}{2}\right)}}{\sqrt{1-\t^2}}, \label{eq_id3}
\intertext{and}
\frac{\partial \HM}{\partial z}(w,z)&=\bigl( -\op{\n}(w) \op{\n-1}(z)+\t \op{\n-1}(w) \op{\n}(z)\bigr)\frac{\e^{\n\left(-w z+\frac{\t w^2}{2}+\frac{\t z^2}{2}\right)}}{\sqrt{1-\t^2}}. \label{eq_id4}
\end{align}
\end{prop}

\begin{proof}
These identities follow directly by addition or subtraction of the identities \eqref{eq_id1} and \eqref{eq_id2}.
\begin{align}
\frac{\partial \HM}{\partial w}(w,z)&\begin{aligned}[t] = \tfrac{1}{2}\Big( &-\sqrt{\tfrac{1-\t}{1+\t}}\bigl(\op{\n}(w) \op{\n-1}(z)+\op{\n-1}(w) \op{\n}(z)\bigr)\nn \\
& + \sqrt{\tfrac{1+\t}{1-\t}}\bigl(\op{\n}(w) \op{\n-1}(z)-\op{\n-1}(w) \op{\n}(z)\bigr) \Big)\e^{\n\left(-w z+\frac{\t w^2}{2}+\frac{\t z^2}{2}\right)}\end{aligned} \nn \\
&=\bigl( \t \op{\n}(w) \op{\n-1}(z)-\op{\n-1}(w) \op{\n}(z)\bigr)\frac{\e^{\n \left(-w z+\frac{\t w^2}{2}+\frac{\t z^2}{2}\right)}}{\sqrt{1-\t^2}},
\shortintertext{and}
\frac{\partial \HM}{\partial z}(w,z)&\begin{aligned}[t]=\tfrac{1}{2}\Big( &-\sqrt{\tfrac{1-\t}{1+\t}}\bigl(\op{\n}(w) \op{\n-1}(z)+\op{\n-1}(w) \op{\n}(z)\bigr)\nn \\
& - \sqrt{\tfrac{1+\t}{1-\t}}\bigl(\op{\n}(w) \op{\n-1}(z)-\op{\n-1}(w) \op{\n}(z)\bigr) \Big)\e^{\n\left(-w z+\frac{\t w^2}{2}+\frac{\t z^2}{2}\right)} \end{aligned} \nn \\
&=\bigl( -\op{\n}(w) \op{\n-1}(z)+\t \op{\n-1}(w) \op{\n}(z)\bigr)\frac{\e^{\n\bigl(-w z+\frac{\t w^2}{2}+\frac{\t z^2}{2}\bigr)}}{\sqrt{1-\t^2}}.
\end{align}
\end{proof}

As in Chapter \ref{ch_asymp} we are going to find the asymptotics of these identities. If we look at the integrals in \eqref{eq_hm_wz} we see that we need this identities evaluated at $w=\cc{z}_0+u$ and $z=z_0$ or $w=\cc{z}_0+\frac{\cc{a}}{\sqrt{\n}}$ and $z=z_0+v$ respectively, instead of $w=\cc{z}$ as in \eqref{eq_lemma_asym1} and \eqref{eq_lemma_asym2}. From the boundaries of integration we see that $u=\mathcal{O}(\n^{-1/2})$ and $v=\mathcal{O}(\n^{-1/2})$. When we have calculated the asymptotics of the identities \eqref{eq_lemma_asym1} and \eqref{eq_lemma_asym2} we have expanded its terms up to order $\mathcal{O}(\n^{-1})$. Now we also have to consider the integration parameters $u$ and $v$ in this expansion up to order two to get the same order of error. This we are going to do in Proposition \ref{prop_asym_wz}.

\subsection{Asymptotics of $\frac{\partial \protect\HM}{\partial \protect\chapw}$ and $\frac{\partial \protect\HM}{\partial \protect\chapz}$}
\subsubsection{Outside of the Interval $[-\F,\F]$}
As in previous sections we are using the notation $\Fz_0=\frac{z_0}{\F}$, $\Fu=\frac{u}{\F}$ and $\Fv=\frac{v}{\F}$.

\begin{prop}\label{prop_asym_wz}
Let $z_0\in \C\setminus[-\F,\F]$, $A>0$ and $u,v\in B_{\!\frac{A}{\sqrt{\n}}}(0)$. As $\n\rightarrow\infty$ we have
\begin{align}
\left.\frac{\partial \HM}{\partial w}\right\vert_{\substack{\phantom{}_{w=\cc{z}_0+u} \\ \phantom{}_{z=z_0+v} \\ \phantom{}}}&=\sqrt{\tfrac{\n}{2 \pi^3}} \gw(z_0) \e^{\n \left( \f(z_0) +\hu(\cc{z}_0)u+\huu(\cc{z}_0)u^2 +\hu(z_0)v+\huu(z_0)v^2+\huv uv \right)}(1+o(1)), \label{eq_lemma_asym_wz1}
\intertext{and}
\left.\frac{\partial \HM}{\partial z}\right\vert_{\substack{\phantom{}_{w=\cc{z}_0+u} \\ \phantom{}_{z=z_0+v} \\ \phantom{}}}&=\sqrt{\tfrac{\n}{2 \pi^3}} \gz(z_0) \e^{\n \left( \f(z_0) +\hu(\cc{z}_0)u+\huu(\cc{z}_0)u^2 +\hu(z_0)v+\huu(z_0)v^2+\huv uv \right)}(1+o(1)), \label{eq_lemma_asym_wz2}
\intertext{where}
\f(z_0)&=\t \Re{\left(z_0^2\right)}-\left\vert z_0\right\vert^2+\Re\left(\left(\U{\F}(z_0)\right)^2 \right)-2\log\left\vert \U{\F}(z_0) \right\vert +\log\t+1 \label{eq_f_wz},\\
\gw(z_0)&=\tfrac{1}{4} \left( \t^{1/2}\Bigl(\Fz_0\mp\sqrt{\Fz_0^2-1}\Bigr)-\t^{-1/2}\Bigl(\cc{\Fz}_0\mp\sqrt{\cc{\Fz}_0^2-1}\Bigr) \right)  \tfrac{\left\vert\sqrt{\Fz_0+1}+\sqrt{\Fz_0-1}\right\vert^2 }{ \sqrt{\left\vert \Fz_0^2-1 \right\vert}}\nn\\
&=\tfrac{1}{4}\vert \W{\F}(z_0) \vert^2 \left(\t^{1/2} \U{\F}(z_0)-\t^{-1/2} \U{\F}(\cc{z}_0)\right) =\frac{\gp(z_0)-i \gm(z_0)}{2}, \label{eq_gw}\\
\gz(z_0)&=\gw(\cc{z}_0)=\cc{\gw(z_0)}=\frac{\gp(z_0)+i \gm(z_0)}{2},\label{eq_gz} \\
\hu(z_0)&=\tfrac{1-\left(\Fz_0\mp\sqrt{\Fz_0^2-1}\right)^2}{\pm \F \sqrt{\Fz_0^2-1}}+F(\t \Fz_0-\cc{\Fz}_0)=\frac{ 2 \U{\F}(z_0) }{\F}+\t z_0-\cc{z}_0 \label{eq_hu},  \\ 
\huu(z_0)&=\tfrac{\left(\Fz_0\mp\sqrt{\Fz_0^2-1}\right)^2-1}{2\F^2 \left(\Fz_0^2-1\right)}+\frac{\t}{2}= \frac{\left(\U{\F}(z_0) \right)^2-1 }{2 \left(\T{\F}(z_0) \right)^2} =\frac{\t}{2}-\frac{\U{\F}(z_0)}{\F \T{\F}(z_0)}, \label{eq_huu} \\ 
\huv&=-1 , \label{eq_huv}
\end{align}
and $\gpm$ are the same functions as in Proposition \ref{prop_asym}. 
\end{prop}
\begin{proof}
We look at the proof of Proposition \ref{prop_asym}, how the terms of \eqref{eq_op1} and \eqref{eq_op2} depending on $z$ or $w$ have been going into $\f$. Only these terms will be important when we have to calculate the corrections in $u$ and $v$. In the remaining terms the corrections in $u$ and $v$ will only give errors of order $\mathcal{O}(\n^{-1/2})$ and can be neglected. Therefore we need the following additional expansions when we set $\Fz=\Fz_0+\Fv$
\begin{align}
\MoveEqLeft[3] \tfrac{\n}{2}\bigl( \U{\F}(z_0+v) \bigr)^2=\tfrac{\n}{2}\left(\Fz_0+\Fv\mp\sqrt{(\Fz_0+\Fv)^2-1}\right)^2 =\tfrac{\n}{2}\Big(\Fz_0+\Fv\mp\sqrt{\Fz_0^2-1}\sqrt{1+\tfrac{2\Fz_0\Fv+\Fv^2}{\Fz_0^2-1}}\Big)^2 \nn\\
= {} & \tfrac{\n}{2}\left(\Fz_0+\Fv\mp\sqrt{\Fz_0^2-1}\left(1+\tfrac{2\Fz_0\Fv+\Fv^2}{2(\Fz_0^2-1)}-\tfrac{\Fz_0^2\Fv^2}{2(\Fz_0^2-1)^2}+\lO(\n^{-3/2}) \right) \right)^2 \nn\\
= {} & \tfrac{\n}{2}\Big(\Fz_0\mp\sqrt{\Fz_0^2-1}\Big)^2-\tfrac{\n \Fv \left(\Fz_0\mp\sqrt{\Fz_0^2-1}\right)^2}{\pm\sqrt{\Fz_0^2-1}}+ \tfrac{\n \Fv^2 \left(\Fz_0\mp\sqrt{\Fz_0^2-1}\right)}{\pm 2 \left(\Fz_0^2-1 \right)^{3/2}} + \tfrac{\n \Fv^2 \left(\Fz_0\mp\sqrt{\Fz_0^2-1}\right)^2}{2 \left(\Fz_0^2-1 \right)}+\mathcal{O}(\n^{-1/2}) \nn\\
= {} & \tfrac{\n}{2}\left( \U{\F}(z_0) \right)^2 -\tfrac{\n v \left( \U{\F}(z_0) \right)^2}{\T{\F}(z_0) } + \tfrac{\n v^2 \F \U{\F}(z_0)}{ 2 \left( \T{\F}(z_0) \right)^3} + \tfrac{\n v^2 \left( \U{\F}(z_0) \right)^2}{2 \left(\T{\F}(z_0) \right)^2}+\mathcal{O}(\n^{-1/2}), \label{eq_expan6}
\intertext{and}
\MoveEqLeft[3] \bigl( \U{\F}(z_0+v) \bigr)^{-\n}=\e^{-\n \log\left(\Fz_0+\Fv \mp \sqrt{(\Fz_0+\Fv)^2-1}\right)}=\e^{-\n \log\left(\Fz_0+\Fv\mp\sqrt{\Fz_0^2-1}\sqrt{1+\frac{2\Fz_0\Fv+\Fv^2}{\Fz_0^2-1}}\right)}\nn\\
= {} & \e^{-\n \log\left(\Fz_0+\Fv\mp\sqrt{\Fz_0^2-1}\left(1+\frac{2\Fz_0\Fv+\Fv^2}{2(\Fz_0^2-1)}-\frac{\Fz_0^2\Fv^2}{2(\Fz_0^2-1)^2}+\lO(\n^{-3/2}) \right) \right)}\nn\\
= {} & \Big(\Fz_0\mp\sqrt{\Fz_0^2-1}\Big)^{-\n}\e^{-\n \log\left(1-\frac{\Fv}{\pm\sqrt{\Fz_0^2-1}}+\frac{\Fv^2}{\pm 2(\Fz_0^2-1)^{3/2}\left(\Fz_0\mp\sqrt{\Fz_0^2-1}\right)} +\lO(\n^{-3/2}) \right)}\nn\\
= {} & \Big(\Fz_0\mp\sqrt{\Fz_0^2-1}\Big)^{-\n}\e^{-\n \left(-\frac{\Fv}{\pm\sqrt{\Fz_0^2-1}}+\frac{\Fv^2}{\pm 2\left(\Fz_0^2-1\right)^{3/2}\left(\Fz_0\mp\sqrt{\Fz_0^2-1}\right)}-\frac{\Fv^2}{2 \left(\Fz_0^2-1\right)} +\lO(\n^{-3/2}) \right)}\nn\\
= {} & \Big(\Fz_0\mp\sqrt{\Fz_0^2-1}\Big)^{-\n}\e^{\frac{\n \Fv}{\pm\sqrt{\Fz_0^2-1}}- \frac{\n \Fv^2 \Fz_0}{\pm2 \left(\Fz_0^2-1\right)^{3/2}} } \left(1+\lO(\n^{-1/2})\right) \nn\\
= {} & \left(\U{\F}(z_0) \right)^{-\n}\e^{\frac{\n v}{\T{\F}(z_0)}- \frac{\n v^2 z_0}{2 \left(\T{\F}(z_0)\right)^3} } \left(1+\lO(\n^{-1/2})\right). \label{eq_expan7}
\intertext{So \eqref{eq_expan1} becomes }
\MoveEqLeft[3]\tfrac{\n}{2} \Big( \sqrt{1-\tfrac{1}{\n}} \U{\F}\left( (z_0+v) (1-\tfrac{1}{\n})^{-1/2} \right) \Big)^2 = \tfrac{\n}{2}\Big(\Fz_0\mp\sqrt{\Fz_0^2-1}\Big)^2-\tfrac{\n \Fv \left(\Fz_0\mp\sqrt{\Fz_0^2-1}\right)^2}{\pm\sqrt{\Fz_0^2-1}} \nn\\
& + \tfrac{\n \Fv^2 \left(\Fz_0\mp\sqrt{\Fz_0^2-1}\right)}{\pm 2 \left(\Fz_0^2-1 \right)^{3/2}}+\tfrac{\n \Fv^2 \left(\Fz_0\mp\sqrt{\Fz_0^2-1}\right)^2}{2 \left(\Fz_0^2-1 \right)}-\tfrac{\Fz_0\mp\sqrt{\Fz_0^2-1}}{\pm 2\sqrt{\Fz_0^2-1}}+\lO(\n^{-1/2})\nn \\
= {} & \tfrac{\n}{2}\left( \U{\F}(z_0) \right)^2-\tfrac{\n v \left( \U{\F}(z_0)\right)^2}{\T{\F}(z_0)} + \tfrac{\n v^2 \F \U{\F}(z_0) }{ 2 \left(\T{\F}(z_0) \right)^3}+\tfrac{\n v^2 \left( \U{\F}(z_0) \right)^2}{2 \left( \T{\F}(z_0) \right)^2}-\tfrac{\F\U{\F}(z_0)}{ 2 \T{\F}(z_0)}+\lO(\n^{-1/2}), \label{eq_expan8}
\intertext{and \eqref{eq_expan2}}
\MoveEqLeft[3] \Big( \sqrt{1-\tfrac{1}{\n}} \U{\F}\left( (z_0+v) (1-\tfrac{1}{\n})^{-1/2} \right) \Big)^{-(\n-1)} \nn\\
= {} & \Big(\Fz_0\mp\sqrt{\Fz_0^2-1}\Big)^{-(\n-1)}\e^{\frac{\n \Fv}{\pm\sqrt{\Fz_0^2-1}}- \frac{\n \Fv^2 \Fz_0}{\pm 2 \left(\Fz_0^2-1\right)^{3/2}} + \frac{1}{\pm 2 \sqrt{\Fz_0^2-1}\left( \Fz_0\mp\sqrt{\Fz_0^2-1}\right)}} \left(1+\lO(\n^{-1/2})\right)\nn\\
= {} & \left( \U{\F}(z_0) \right)^{-(\n-1)}\e^{\frac{\n v}{\T{\F}(z_0)}- \frac{\n v^2 z_0}{ 2 \left(\T{\F}(z_0)\right)^3} +\frac{\F}{ 2 \T{\F}(z_0) \U{\F}(z_0) }} \left(1+\lO(\n^{-1/2})\right) . \label{eq_expan9}
\end{align}

The expansion \eqref{eq_expan4b} gets no further corrections, it stays valid if we replace on the left-hand side $\Fz$ by $\Fz_0+\Fv$ and on the right-hand side $\Fz$ by $\Fz_0$ and $\lO(\n^{-1})$ by $\lO(\n^{-1/2})$.

Setting $z=z_0+v$ in \eqref{eq_op1} and using \eqref{eq_pr_asymp} (which is valid at $z_0+v$ too), the expansions \eqref{eq_expan6}, \eqref{eq_expan7} and also
\begin{equation}\label{eq_asym_wz1}
\e^{\tfrac{\n v^2 \F \U{\F}(z_0) - \n v^2 z_0  }{ 2 \left( \T{\F}(z_0) \right)^3}}= \e^{\tfrac{\n \Fv^2 \left(\Fz_0\mp\sqrt{\Fz_0^2-1}\right)}{\pm 2 \left(\Fz_0^2-1 \right)^{3/2}}}\e^{-\frac{\n \Fv^2 \Fz_0}{\pm 2 \left(\Fz_0^2-1\right)^{3/2}}}=\e^{-\frac{\n \Fv^2}{2 \left(\Fz_0^2-1\right)}}=\e^{-\frac{\n v^2}{2 \left(\T{\F}(z_0)\right)^2}},
\end{equation}
we get
\begin{align}
\op{\n}(z_0+v)&=\frac{\n^{\frac{\n+1}{2}}}{2\sqrt{\pi}\sqrt{\n !}}\t^{\frac{\n}{2}}\left(1-\t^2\right)^{1/4} \e^{\frac{\n}{2}\left( \U{\F}(z_0) \right)^2} \left( \U{\F}(z_0) \right)^{-\n} \e^{\n v \tfrac{1- \left( \U{\F}(z_0) \right)^2}{\T{\F}(z_0) } } \e^{\n v^2 \frac{\left(\U{\F}(z_0)\right)^2-1}{2 \left( \T{\F}(z_0) \right)^2}}\nn \\
&\phantom{=}\quad \cdot \W{\F}(z_0) (1+o(1)). \label{eq_op_wz1}
\intertext{Analogously we are setting $w=z_0+v$ in \eqref{eq_op2} and using \eqref{eq_pr_asymp}, the expansions \eqref{eq_expan4b}, \eqref{eq_expan8}, \eqref{eq_expan9} and also \eqref{eq_asymp1b} and \eqref{eq_asym_wz1} to get}
\op{\n-1}(z_0+v)&=\frac{\n^{\frac{\n+1}{2}}}{2\sqrt{\pi}\sqrt{\n !}}\t^{\frac{\n}{2}}\left(1-\t^2\right)^{1/4} \t^{-1/2} \left(\tfrac{\n-1}{\n}\right)^{\n-1} \e^1  \e^{\frac{\n}{2}\left( \U{\F}(z_0) \right)^2} \left( \U{\F}(z_0) \right)^{-\n} \U{\F}(z_0) \nn \\
&\phantom{=}\quad \cdot \e^{\n v \tfrac{1- \left( \U{\F}(z_0) \right)^2}{\T{\F}(z_0)} } \e^{\n v^2 \frac{\left(\U{\F}(z_0)\right)^2-1}{2 \left(\T{\F}(z_0)\right)^2}} \W{\F}(z_0)(1+o(1)). \label{eq_op_wz2}
\end{align}

All factors that appear in both of \eqref{eq_op_wz1} and \eqref{eq_op_wz2} and all factors that don't depend on $z_0$ can get excluded from 
\begin{align}
\MoveEqLeft[3]\pm\t^{1/2}\left(\t^{\pm 1/2}\op{\n}(\cc{z}_0+u) \op{\n-1}(z_0+v)-\t^{\mp 1/2}\op{\n-1}(\cc{z}_0+u) \op{\n}(z_0+v) \right)\nn\\
= {} & \frac{\n^{\n+1}}{4\pi\n !}\t^{\n}\left(1-\t^2\right)^{1/2} \e^{\n\Re\left(\left(\U{\F}(z_0)\right)^2\right)} \left\vert \U{\F}(z_0) \right\vert^{-2\n} \left\vert\W{\F}(z_0)\right\vert^2 \nn\\
&\cdot \e^{\n u \tfrac{1- \left(\U{\F}(\cc{z}_0)\right)^2}{\T{\F}(\cc{z}_0)}}\e^{\n v \tfrac{1- \left(\U{\F}(z_0)\right)^2}{\T{\F}(z_0)} } \e^{\n u^2 \frac{\left(\U{\F}(\cc{z}_0)\right)^2-1}{2 \left(\T{\F}(\cc{z}_0)\right)^2}} \e^{ \n v^2 \frac{\left(\U{\F}(z_0)\right)^2-1}{2 \left(\T{\F}(z_0)\right)^2}} \nn\\
&\cdot  \left\lgroup \pm \t^{\pm 1/2} \U{\F}(z_0) \mp \t^{\mp 1/2} \U{\F}(\cc{z}_0) \right\rgroup(1+o(1)). \label{eq_asym_wz2}
\end{align}
The only factors that we could not exclude are those in the last row. We further have used \eqref{eq_expan5} which has canceled the factor $\e^1$ in \eqref{eq_op_wz2}.

We are left to evaluate the last factor which is part of the identities \eqref{eq_id3} and \eqref{eq_id4} at $w=\cc{z}_0+u$ and $z_0+v$
\begin{equation}\label{eq_asym_wz3}
\e^{\n(-w z+\frac{\t w^2}{2}+\frac{\t z^2}{2})}=\e^{\n\left(-\vert z_0\vert^2+\t\Re(z_0^2)+u(\t \cc{z}_0-z_0)+v(\t z_0-\cc{z}_0)+\frac{u^2\t}{2}+ \frac{v^2\t}{2}-uv\right)}.
\end{equation}

Using the Stirling approximation \eqref{eq_stirling} in \eqref{eq_asym_wz2} we finally find with the help of identities \eqref{eq_id3} and \eqref{eq_id4} (Proposition \ref{prop_ident_wz}) evaluated at $w=\cc{z}_0+u$ and $z_0+v$
\begin{align}
\left.\frac{\partial \HM}{\partial w}\right\vert_{\substack{\phantom{}_{w=\cc{z}_0+u} \\ \phantom{}_{z=z_0+v} \\ \phantom{}}}&= \e^{\n\left(  \t\Re(z_0^2)-\vert z_0\vert^2 +1+\log\t +\Re\left(\left(\U{\F}(z_0)\right)^2\right) -2 \log \left\vert\U{\F}(z_0)\right\vert \right)} \nn \\
&\phantom{=}\quad\cdot \e^{\n u  \left(   \t \cc{z}_0-z_0 + \tfrac{1- \left(\U{\F}(\cc{z}_0)\right)^2}{\T{\F}(\cc{z}_0)} \right) } \e^{\n v \left(  \t z_0-\cc{z}_0 +  \tfrac{1- \left(\U{\F}(z_0)\right)^2}{\T{\F}(z_0)} \right)}\nn \\ 
&\phantom{=}\quad\cdot \e^{ \n u^2 \left(  \frac{\t}{2}+ \frac{\left(\U{\F}(\cc{z}_0)\right)^2-1}{2 \left(\T{\F}(\cc{z}_0)\right)^2}  \right)}   \e^{\n v^2 \left(  \frac{\t}{2} + \frac{\left(\U{\F}(z_0)\right)^2-1}{2 \left(\T{\F}(z_0)\right)^2}  \right)} \e^{-\n u v} \nn \\
&\phantom{=}\quad\cdot \tfrac{\sqrt{\n}}{\sqrt{2\pi^3}}\tfrac{1}{4} \left\vert \W{\F}(z_0) \right\vert^2 \left(\t^{1/2} \U{\F}(z_0) - \t^{-1/2} \U{\F}(\cc{z}_0)\right) (1+o(1)), \label{eq_asym_wz4}
\intertext{and}
\left.\frac{\partial \HM}{\partial z}\right\vert_{\substack{\phantom{}_{w=\cc{z}_0+u} \\ \phantom{}_{z=z_0+v} \\ \phantom{}}}&= \e^{\n\left(  \t\Re(z_0^2)-\vert z_0\vert^2 +1+\log\t +\Re\left(\left(\U{\F}(z_0)\right)^2\right) -2 \log \left\vert\U{\F}(z_0)\right\vert \right)} \nn \\
&\phantom{=}\quad\cdot \e^{\n u  \left(   \t \cc{z}_0-z_0 + \tfrac{1- \left(\U{\F}(\cc{z}_0)\right)^2}{\T{\F}(\cc{z}_0)} \right) } \e^{\n v \left(  \t z_0-\cc{z}_0 +  \tfrac{1- \left(\U{\F}(z_0)\right)^2}{\T{\F}(z_0)} \right)}\nn \\ 
&\phantom{=}\quad\cdot \e^{ \n u^2 \left(  \frac{\t}{2}+ \frac{\left(\U{\F}(\cc{z}_0)\right)^2-1}{2 \left(\T{\F}(\cc{z}_0)\right)^2}  \right)}   \e^{\n v^2 \left(  \frac{\t}{2} + \frac{\left(\U{\F}(z_0)\right)^2-1}{2 \left(\T{\F}(z_0)\right)^2}  \right)} \e^{-\n u v} \nn \\
&\phantom{=}\quad\cdot \tfrac{\sqrt{\n}}{\sqrt{2\pi^3}}\tfrac{1}{4} \left\vert \W{\F}(z_0) \right\vert^2 \left(-\t^{-1/2} \U{\F}(z_0) + \t^{1/2} \U{\F}(\cc{z}_0)\right) (1+o(1)). \label{eq_lemma_wz5}
\end{align}
This proves the proposition.
\end{proof}

\begin{remark}
We note that in \eqref{eq_lemma_asym_wz1} and \eqref{eq_lemma_asym_wz2} $\n \left(\hu(\cc{z}_0)u+\hu(z_0)v\right)$ is of order $\mathcal{O}(\n^{1/2})$ and $\n \left(\huu(\cc{z}_0)u^2+\huu(z_0)v^2+\huv uv \right)$ is of order $\mathcal{O}(1)$.
\end{remark}

\begin{remark}
The term of highest order in the exponent which is of order $\lO(\n)$ has not changed, i.e.\ $\f$ here is the same as in Proposition \ref{prop_asym}. Neither the linear combination of the identities nor the corrections $u$ and $v$ of order $\mathcal{O}(\n^{-1/2})$ have any influence on this term.
\end{remark}

\begin{remark}
The linear combinations of identities \eqref{eq_id1} and \eqref{eq_id2} are described by $\gw$ and $\gz$ alone. This is obvious when we set $u=v=0$ in Proposition \ref{prop_asym_wz}. Comparing this with Proposition \ref{prop_asym} we see that
\begin{equation}
\left.\frac{\partial \HM}{\partial w}\right\vert_{\cc{w}=z=z_0}=\frac{\sqrt{\n}}{\sqrt{2\pi^3}} \e^{\n \f(z_0)} \underbrace{\left(\frac{\gp(z_0)-i \gm(z_0)}{2} \right)}_{\gw(z_0)} \left(1+o(1) \right),
\end{equation}
and
\begin{equation}
\left.\frac{\partial \HM}{\partial z}\right\vert_{\cc{w}=z=z_0}=\frac{\sqrt{\n}}{\sqrt{2\pi^3}} \e^{\n \f(z_0)} \underbrace{\left(\frac{\gp(z_0)+i \gm(z_0)}{2} \right)}_{\gz(z_0)} \left(1+o(1) \right).
\end{equation}
\end{remark}

\subsubsection{On the Interval $[-\F,\F]$}
\begin{prop}\label{prop_H_B}
Let $z_0\in (-F,F)$, $A>0$ and $u,v\in B_{\!\frac{A}{\sqrt{\n}}}(0)$. As $\n\rightarrow\infty$ we have
\begin{align}
\left\vert  \left.\frac{\partial \HM}{\partial w}\right\vert_{\substack{\phantom{}_{w=\cc{z}_0+u} \\ \phantom{}_{z=z_0+v} \\ \phantom{}}}   \right\vert& = \lO(\n^{-1/2}), \qquad \text{and} \quad \left\vert  \left.\frac{\partial \HM}{\partial z}\right\vert_{\substack{\phantom{}_{w=\cc{z}_0+u} \\ \phantom{}_{z=z_0+v} \\ \phantom{}}}   \right\vert&= \lO(\n^{-1/2}). \nn
\end{align}
\end{prop}
We postpone the proof of this proposition to Chapter \ref{sec_verification}, where we will prove a more precise version of it in Proposition \ref{prop_estim_B_C}.

\subsection{Properties of $\protect\hu$, $\protect\huu$, $\protect\gw$ and $\protect\gz$}
We can see in \eqref{eq_gw}
-\eqref{eq_huu} that the functions $\gw$, $\gz$, $\hu$ and $\huu$ are complex-valued and continuous on $$\{z\in \C\vert z\notin [-\F,\F]\}.$$

\begin{lemma}\label{lemma_hu_zero}
$\hu(z_0)=\hu(\cc{z}_0)=0$ for all $z_0\in\pEt$.
\end{lemma}
\begin{proof}
We set $z_0= \mas \cos \phi+ i \mis \sin \phi$ with $\phi\in(-\pi,\pi]$. From \eqref{eq_f_zero1} and \eqref{eq_f_zero2} we find
\begin{align}
\hu(z_0)&=\cc{\hu(\cc{z}_0)}=\frac{ 2 \U{\F}(z_0) }{\F}+\t z_0-\cc{z}_0\nn \\
&=\frac{1-(\sqrt{\t}\cos \phi-i \sqrt{\t}\sin \phi)^2}{\frac{1}{\sqrt{1-t^2}}\left( (1-\t)\cos \phi+i(1+\t)\sin \phi \right)} + \tfrac{1}{\sqrt{1-t^2}} \left( \t^2 \cos\phi-i\t^2 \sin\phi-\cos\phi+i\sin\phi \right)\nn\\
&=\sqrt{1-t^2}\left(\frac{1-\t \e^{-2i\phi}}{\e^{i\phi}-\t\e^{-i\phi} }+\frac{(t^2-1)\e^{-i\phi}}{1-\t^2}\right)=0. 
\end{align}
\end{proof}

\begin{lemma}\label{lemma_gwz} 
For $z_0=\mas \cos \phi+i\mis \sin \phi$ with $\phi\in(-\pi,\pi]$ we have the following properties for $\huu$, $\gw$ and $\gz$.
\begin{align}
\vert 2 \huu(z_0)\vert&=\vert\gw(z_0)\vert=\vert \gz(z_0)\vert=1, \\
\pm\sqrt{-2\huu(\cc{z}_0)}&=-\gw(z_0)=\frac{\e^{i\phi}-\t \e^{-i\phi}}{\sqrt{1+\t^2-\t\e^{2i\phi}-\t\e^{-2i\phi}}},
\intertext{where the upper sign is valid if $z_0\in \Cp$ and the lower sign else, and}
\pm\sqrt{-2\huu(z_0)}&=-\gz(z_0)=\frac{\e^{-i\phi}-\t \e^{i\phi}}{\sqrt{1+\t^2-\t\e^{2i\phi}-\t\e^{-2i\phi}}},
\end{align}
where the upper sign is valid if $\cc{z}_0\in \Cp$ and the lower sign else.
\end{lemma}
\begin{proof}
The simplest way to evaluate $\gw$ and $\gz$ on $\pEt$ is to use Proposition \ref{prop_g} for $\gpm$. We get
\begin{align}
\gw(z_0)&=\frac{\gp(z_0)-i \gm(z_0)}{2}=\frac{-2(\t-1)\cos\phi-2i(\t+1)\sin\phi}{2\sqrt{1+\t^2-2\t\cos(2\phi)}}\nn\\
&=-\frac{\e^{i\phi}-\t \e^{-i\phi}}{\sqrt{1+\t^2-\t\e^{2i\phi}-\t\e^{-2i\phi}}}, \label{eq_lemma_gwz1}
\shortintertext{and}
\gz(z_0)&=\frac{\gp(z_0)+i \gm(z_0)}{2}=-\frac{\e^{-i\phi}-\t \e^{i\phi}}{\sqrt{1+\t^2-\t\e^{2i\phi}-\t\e^{-2i\phi}}}.\label{eq_lemma_gwz2}
\intertext{From this we can easily verify}
\vert\gw(z_0)\vert^2&=\vert\gz(z_0)\vert^2=\frac{\left(\e^{i\phi}-\t \e^{- i\phi}\right)\left( \e^{- i\phi}-\t \e^{ i\phi}\right) }{1+\t^2-\t\e^{2i\phi}-\t\e^{-2i\phi}} =1. \label{eq_lemma_gwz_one}
\intertext{For $\huu$ from \eqref{eq_huu} we can use \eqref{eq_f_zero2} and \eqref{eq_f_zero1}. So we find}
-2\huu(z_0)&=\frac{1-(\U{\F}(z_0))^2}{\left( \T{\F}(z_0)\right)^2}-\t=\frac{\left(1-\t(\cos\phi-i\sin\phi)^2\right)(1-\t^2)}{\left((1-\t)\cos \phi+i(1+\t) \sin\phi \right)^2}-\t\nn\\
&=\frac{1+\t^2-\t\e^{-2i\phi}-\t \e^{2i\phi}}{\e^{2i\phi}-2\t+\t^2\e^{-2i\phi}}=\frac{\left(1-\t\e^{2i\phi}\right)\left(1-\t\e^{-2i\phi}\right) }{\left(\e^{2i\phi}-t \right)\left(1-\t\e^{-2i\phi}\right)}\nn\\
&=\frac{1-\t\e^{2i\phi}}{\e^{2i\phi}-t}=\frac{\e^{-i\phi}-\t\e^{i\phi}}{\e^{i\phi}-t\e^{-i\phi}}=\frac{\left(\e^{-i\phi}-\t\e^{i\phi}\right)^2}{1+t^2-\t\e^{2i\phi}-t\e^{-2i\phi}}.
\intertext{As $\Re(\e^{- i\phi}-\t\e^{i\phi})=(1-\t)\cos \phi$, we see which sign we have to use when we are taking the square root.}
\pm \sqrt{-2\huu(\cc{z}_0)}&=\frac{\e^{i\phi}-\t\e^{-i\phi}}{\sqrt{1+t^2-\t\e^{2i\phi}-t\e^{-2i\phi}}}=-\gw(z_0).
\intertext{ where we have used the plus sign if $z_0\in \Cp$. Analogously we get}
\pm \sqrt{-2\huu(z_0)}&=\frac{\e^{-i\phi}-\t\e^{i\phi}}{\sqrt{1+t^2-\t\e^{2i\phi}-t\e^{-2i\phi}}}=-\gz(z_0),.
\end{align}
where we have used the plus sign if $\cc{z}_0\in \Cp$.
\end{proof}

We have seen that the absolute value of $\gw$ on $\pEt$ is one. The phase of $-\gw$ on $\pEt$ can be found in Figure \ref{fig_gw_phase}. $\gz$ is complex conjugated to $\gw$.

\begin{figure}[!tbp]
\begin{center}
\includegraphics[height=8cm,angle=0,scale=0.9]{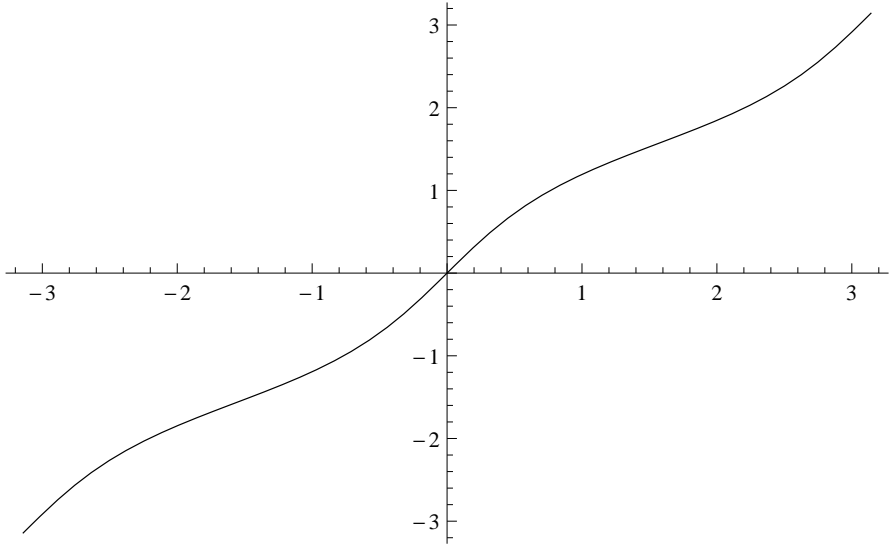}
\end{center}
\caption{\label{fig_gw_phase} Plot of the phase of $-\gw$ in the case $\t=\sqrt{5}-2$ on the ellipse $\pEt$ parametrized by $\phi\in (-\pi,\pi]$. }
\end{figure} 

\subsection{Asymptotics of $\protect\HM$ in the Bulk and Outside of the Ellipse $\protect\pEt$}
Here we consider the case when $z_0\notin \pEt$, i.e.\ $z_0$ is in the domain where the density is constant. We will assume that $\n$ is large enough that $\cc{B}_{\!\frac{\vert a\vert}{\sqrt{\n}}}(z_0)$ and $\cc{B}_{\!\frac{\vert b\vert}{\sqrt{\n}}}(z_0)$ have no intersection with $\pEt$. This is always possible as $\C\setminus \pEt$ is open.

For $a,b\in\C$ we can find $\HM(w,z)$ at $w=\cc{z}_0+\tfrac{\cc{a}}{\sqrt{\n}}$ and $z=z_0+\tfrac{b}{\sqrt{\n}}$ by \eqref{eq_hm_wz}, where we know the asymptotics as $\n\rightarrow \infty$ of 
$$\frac{\partial \HM}{\partial w}\left(\cc{z}_0+u,z_0 \right)\qquad \text{and}\qquad \frac{\partial \HM}{\partial z}\left(\cc{z}_0+\tfrac{\cc{a}}{\sqrt{\n}},z_0+v\right)$$ 
by Proposition \ref{prop_asym_wz} when $u\in [0,\tfrac{\cc{a}}{\sqrt{\n}}]$ and $v\in [0,\tfrac{b}{\sqrt{\n}}]$.

\begin{prop}\label{prop_hm_wz}
Let $z_0\notin \pEt$ and $a,b\in\C$. Then
\begin{equation}
\lim_{\n\rightarrow\infty}\HM\left(\cc{z}_0+\frac{\cc{a}}{\sqrt{\n}},z_0+\frac{b}{\sqrt{\n}}\right)=\lim_{\n\rightarrow\infty}\HM(\cc{z}_0,z_0)=\rho(z_0).
\end{equation}
\end{prop}
\begin{proof}
We only have to show that both integrals in \eqref{eq_hm_wz} are zero. According to Proposition \ref{prop_asym_wz} the asymptotics of their integrands is
\begin{align}
\frac{\partial \HM}{\partial w}\left(\cc{z}_0+u,z_0 \right)&=\sqrt{\tfrac{\n}{2 \pi^3}} \gw(z_0) \e^{\n \left( \f(z_0) +\hu(\cc{z}_0)u+\huu(\cc{z}_0)u^2\right)},
\intertext{and}
\frac{\partial \HM}{\partial z}\left(\cc{z}_0+\tfrac{\cc{a}}{\sqrt{\n}},z_0+v\right)&=\sqrt{\tfrac{\n}{2 \pi^3}} \gz(z_0) \e^{\n \left( \f(z_0) +\hu(z_0)v+\huu(z_0)v^2 \right)  +  {\sqrt{\n}\cc{a} \left(\hu(\cc{z}_0)+\huv v\right)  + \cc{a}^2 \huu(\cc{z}_0) } }.
\intertext{Substituting $\sqrt{\n} u$ by $u$ and $\sqrt{\n} v$ by $v$ 
\eqref{eq_hm_wz} becomes}
\HM\left(\cc{z}_0+\tfrac{\cc{a}}{\sqrt{\n}},z_0+\tfrac{b}{\sqrt{\n}}\right)&=\HM(\cc{z}_0,z_0)+ \tfrac{\gw(z_0)}{\sqrt{2 \pi^3}} \int_0^{\cc{a}} \e^{\n \Phi_\n^w(z_0,u)} \dd u + \tfrac{\gz(z_0)}{\sqrt{2 \pi^3}} \int_0^{b} \e^{\n \Phi_\n^z(z_0,v)} \dd v,
\intertext{where}
\Phi_\n^w(z_0,u)&= \f(z_0) + \frac{u \hu(\cc{z}_0)}{\sqrt{\n}}+\frac{u^2 \huu(\cc{z}_0)}{\n},
\shortintertext{and}
\Phi_\n^z(z_0,v)&=\f(z_0) + \frac{v\hu(z_0)+ \cc{a} \hu(\cc{z}_0)}{\sqrt{\n}}+\frac{v^2\huu(z_0)+v\cc{a}\huv +\cc{a}^2\huu(\cc{z}_0)}{\n}. 
\end{align}
By Proposition \ref{prop_f_zero} we know that $f(z_0)$ is strictly negative. The functions $\gw$, $\gz$, $\hu$ and $\huu$ are continuous and can't have a singularity at $z_0$ or $\cc{z}_0$. $u$ and $v$ are bounded by the boundary of integration. So for $\n$ large enough $\Phi_\n^w$ and $\Phi_\n^z$ are strictly negative, uniformly for $u\in [0,\cc{a}]$ or $v\in [0,b]$ respectively. Therefore the integrals become zero as $\n\rightarrow\infty$. 
\end{proof}

\subsection{Asymptotics of $\protect\HM$ on the Ellipse $\protect\pEt$}

\paragraph*{Variant 1}
We start to evaluate $\HM(w,z)$ at $w=\cc{z}_0+\tfrac{\cc{a}}{\sqrt{\n}}$ and $z=z_0$.
\begin{lemma}\label{lemma_hm_w_ellipse}
Let $z_0\in\pEt$ and $a\in\C$. Then 
\begin{equation}
\lim_{\n\rightarrow\infty}\HM\left(\cc{z}_0+\tfrac{\cc{a}}{\sqrt{\n}},z_0\right)=\frac{1}{2\pi}\erfc\left(\zeta(\cc{a},\phi) \right), 
\end{equation}
with
\begin{equation}\label{eq_zeta2}
\zeta(\cc{a},\phi)=\frac{\cc{a}}{\sqrt{2}}\frac{(1-\t)\cos\phi +(1+\t)i\sin\phi }{\sqrt{(1+\t)^2-4\t \cos^2\phi}},
\end{equation}
where $\phi$ is the parameter so that $z_0=\mas \cos \phi+i \mis \sin \phi$.
\end{lemma}
\begin{proof}
As before we get $\HM(w,z)$ at $w=\cc{z}_0+\tfrac{\cc{a}}{\sqrt{\n}}$ and $z=z_0$ by integration
\begin{equation}
\HM\left(\cc{z}_0+\tfrac{\cc{a}}{\sqrt{\n}},z_0\right)=\HM(\cc{z}_0,z_0)+\int_0^{\frac{\cc{a}}{\sqrt{\n}}} \frac{\partial \HM}{\partial w}\left(\cc{z}_0+u,z_0\right)\dd u,
\end{equation}
where we have the asymptotics of $\frac{\partial \HM}{\partial w}\left(\cc{z}_0+u,z_0\right)$ from Proposition \ref{prop_asym_wz} and we know that $\HM(\cc{z}_0,z_0)=\tfrac{1}{2\pi}$ in the limit $\n\rightarrow\infty$ from Proposition \ref{prop_density}. Now as $z_0\in\pEt$, we know by Proposition \ref{prop_f_zero} that $f(z_0)$ in \eqref{eq_lemma_asym_wz1} is not negative but zero and by Lemma \ref{lemma_hu_zero} that also $\hu(\cc{z}_0)=0$. Therefore as $\n\rightarrow \infty$
\begin{equation}\label{eq_integral_erfc1}
\begin{split}
\HM\left(\cc{z}_0+\tfrac{\cc{a}}{\sqrt{\n}},z_0\right)&\approx \frac{1}{2\pi}+\sqrt{\tfrac{\n}{2 \pi^3}} \gw(z_0) \int_0^{\frac{\cc{a}}{\sqrt{\n}}} \e^{\n \huu(\cc{z}_0)u^2} \dd u\\
&= \frac{1}{2\pi}+\sqrt{\tfrac{\n}{2 \pi^3}} \gw(z_0) \int_0^{\frac{\cc{a}}{\sqrt{\n}}} \e^{-\eta_w^2 u^2} \dd u\\
&=\frac{1}{2\pi} - \pi^{-3/2} \int_0^{\eta_w \frac{\cc{a}}{\sqrt{\n}}} \e^{- u^2} \dd u\\
&=\pi^{-3/2}\left(\int_0^\infty \e^{-u^2}\dd u  - \int_0^{\zeta(\cc{a},\phi)} \e^{- u^2} \dd u \right)\\
&=\pi^{-3/2} \int_{\zeta(\cc{a},\phi)}^\infty \e^{- u^2} \dd u=\frac{1}{2\pi }\erfc(\zeta(\cc{a},\phi)),
\end{split}
\end{equation}
where we have substituted $\eta_w u$ by $u$ with $\eta_w=\pm \sqrt{-\huu(\cc{z}_0)\n}=-\gw(z_0)\sqrt{\tfrac{\n}{2}}$ according to Lemma \ref{lemma_gwz}. For $\zeta(\cc{a},\phi)=\eta_w \frac{\cc{a}}{\sqrt{\n}}$ we find again by Lemma \ref{lemma_gwz}
\begin{equation}
\zeta(\cc{a},\phi)=\frac{\cc{a} (-\gw(z_0))}{\sqrt{2}}=\frac{\cc{a}}{\sqrt{2}} \frac{\e^{i\phi}-\t \e^{-i\phi}}{\sqrt{1+\t^2-\t\e^{2i\phi}-\t\e^{-2i\phi}}},
\end{equation}
which is equivalent to the definition of $\zeta$ in the lemma.
\end{proof}

We finally want to find $\HM(w,z)$ at $w=\cc{z}_0+\tfrac{\cc{a}}{\sqrt{\n}}$ and $z=z_0+\tfrac{b}{\sqrt{\n}}$.
\begin{prop}\label{prop_hm_wz_ellipse}
Let $z_0\in\pEt$ and $a,b\in\C$. Then 
\begin{equation}
\lim_{\n\rightarrow\infty}\HM\left(\cc{z}_0+\tfrac{\cc{a}}{\sqrt{\n}},z_0+\tfrac{b}{\sqrt{\n}}\right)=\frac{1}{2\pi}\erfc\left(\zeta(\cc{a},\phi)+ \zeta(b,-\phi) \right), 
\end{equation}
with
\begin{equation}\label{eq_zeta3}
\zeta(\cc{a},\phi)=\frac{\cc{a}}{\sqrt{2}}\frac{(1-\t)\cos\phi +(1+\t)i\sin\phi }{\sqrt{(1+\t)^2-4\t \cos^2\phi}},
\end{equation}
where $\phi$ is the parameter so that $z_0=\mas \cos \phi+i \mis \sin \phi$.
\end{prop}
\begin{proof}
First we set $\cc{w}_1=\cc{z}_0+\tfrac{\cc{a}}{\sqrt{\n}}$ and $z_1=z_0+\tfrac{\cc{b}}{\sqrt{\n}}$. We will proceed similarly as in the proof of Lemma \ref{lemma_hm_w_ellipse}. We get $\HM(\cc{w}_1,z_1)$ by integration
\begin{equation}
\HM\left(\cc{w}_1,z_1\right)=\HM\left(\cc{w}_1,z_0\right)+\int_0^{\frac{b}{\sqrt{\n}}} \frac{\partial \HM}{\partial z}\left(\cc{w}_1,z_0+v\right)\dd v,
\end{equation}
where we know the asymptotics of $\frac{\partial \HM}{\partial z}\left(\cc{w}_1,z_0+v\right)$ from Proposition \ref{prop_asym_wz} and that in the limit $\n\rightarrow\infty$, $\HM(\cc{w}_1,z_0)=\tfrac{\erfc(\zeta(\cc{a},\phi))}{2\pi}$ according to Lemma \ref{lemma_hm_w_ellipse}. Note that the definition for $\zeta$ in this proposition here is the same as in Lemma \ref{lemma_hm_w_ellipse}. $f(z_0)=0$ by Proposition \ref{prop_f_zero} and also $\hu(\cc{z}_0)=\hu(z_0)=0$ by Lemma \ref{lemma_hu_zero}, which simplifies \eqref{eq_lemma_asym_wz2}. Therefore as $\n\rightarrow \infty$
\begin{equation}\label{eq_lemma_hm1}
\begin{split}
\HM\left(\cc{w}_1,z_1\right) &\approx \tfrac{\erfc(\zeta(\cc{a},\phi))}{2\pi}+\sqrt{\tfrac{\n}{2 \pi^3}} \gz(z_0) \e^{\huu(\cc{z}_0)\cc{a}^2} \int_0^{\frac{b}{\sqrt{\n}}} \e^{\n \left(\huu(z_0)v^2-\frac{\cc{a}}{\sqrt{\n}}v\right)} \dd v\\
&=  \tfrac{\erfc(\zeta(\cc{a},\phi))}{2\pi}+\sqrt{\tfrac{\n}{2 \pi^3}} \gz(z_0) \e^{\huu(\cc{z}_0)\cc{a}^2+\eta_z^2\frac{\xi^2}{\n}}  \int_0^{\frac{b}{\sqrt{\n}}} \e^{-\eta_z^2 \left( v+ \frac{\xi}{\sqrt{\n}}\right)^2 } \dd v\\
&=  \tfrac{\erfc(\zeta(\cc{a},\phi))}{2\pi}+\sqrt{\tfrac{\n}{2 \pi^3}} \tfrac{\gz(z_0)}{\eta_z} \e^{\frac{\cc{a}^2}{2}\left( 2\huu(\cc{z}_0)-\frac{1}{2\huu(z_0)} \right)} \int_{\frac{\eta_z}{\sqrt{\n}} \xi}^{\frac{\eta_z}{\sqrt{\n}} (b+\xi )} \e^{-v^2 } \dd v,
\end{split}
\end{equation}
where $\xi=-\frac{\cc{a}}{2\huu(z_0)}$ and where we have substituted $\eta_z (v+\frac{\xi}{\sqrt{\n}}) $ by $v$ and have used that $\eta_z=\pm\sqrt{-\huu(z_0)\n}=-\gz(z_0)\sqrt{\tfrac{\n}{2}}$ according to Lemma \ref{lemma_gwz}. We further know that $2\huu(\cc{z}_0)-\frac{1}{2\huu(z_0)}=0$ as $\vert 2\huu(z_0)\vert=1$. For the upper boundary of integration we get 
\begin{align}
\tfrac{\eta_z}{\sqrt{\n}}(b+\xi)&
=-\gz(z_0) \tfrac{b}{\sqrt{2}}-\gw(z_0)\tfrac{\cc{a}}{\sqrt{2}}=\zeta(b,-\phi)+\zeta(\cc{a},\phi).
\end{align}
Finally \eqref{eq_lemma_hm1} becomes
\begin{equation}\label{eq_integral_erfc2}
\begin{split}
\HM\left(\cc{w}_1,z_1\right) &\approx  \pi^{-3/2} \int_{\zeta(\cc{a},\phi)}^\infty \e^{- v^2} \dd v -\pi^{-3/2} \int_{\zeta(\cc{a},\phi)}^{\zeta(\cc{a},\phi)+\zeta(b,-\phi)} \e^{-v^2 } \dd v\\
&=\pi^{-3/2} \int_{\zeta(\cc{a},\phi)+\zeta(b,-\phi)}^\infty \e^{-v^2 } \dd v=\erfc(\zeta(\cc{a},\phi)+\zeta(b,-\phi)).
\end{split}
\end{equation}
\end{proof}

\begin{cor}\label{cor_hm_wz_ellipse}
Let $a,b\in\C$ and $z_0=\mas \cos \phi+i\mis \sin\phi$ be a point on the ellipse $\pEt$ with $\phi\in (-\pi,\pi]$  and $\psi$ the angle between the real axis and the outward normal to the tangent at $\pEt$ in $z_0$, i.e.\ $\psi$ is defined as in Corollary \ref{cor_density_scaled_limit}. Then
\begin{equation}
\lim_{\n\rightarrow\infty} \HM\left(\cc{z}_0+\tfrac{\cc{a}\e^{-i\psi}}{\sqrt{\n}},z_0+\tfrac{b\e^{i\psi}}{\sqrt{\n}}\right)=\frac{1}{2\pi}\erfc\left(\frac{\cc{a}+b}{\sqrt{2}} \right). 
\end{equation}
\end{cor}
\begin{proof}
Using \eqref{eq_cos_psi} and \eqref{eq_sin_psi} from the proof of Corollary \ref{cor_density_scaled_limit}, we get
\begin{equation}\label{eq_psi_e}
\e^{\pm i \psi}=\frac{(1-\t)\cos\phi\pm(1+\t) i \sin\phi}{\sqrt{(1+\t)^2 -4\t \cos^2 \phi}}.
\end{equation}
Substituting in Proposition \ref{prop_hm_wz_ellipse} $a$ by $a \e^{i \psi}$ and $b$ by $b \e^{i \psi}$, we find 
\begin{equation}
\lim_{\n\rightarrow\infty} \HM\left(\cc{z}_0+\tfrac{\cc{a}\e^{-i\psi}}{\sqrt{\n}},z_0+\tfrac{b\e^{i\psi}}{\sqrt{\n}}\right)=\frac{1}{2\pi}\erfc\left(\zeta(\cc{a} \e^{-i\psi},\phi )+ \zeta(b\e^{i\psi},-\phi) \right), 
\end{equation}
Using \eqref{eq_psi_e} we get 
\begin{align}
\zeta(\cc{a} \e^{-i\psi},\phi)&=\frac{\cc{a}}{\sqrt{2}} \frac{\e^{-i\psi}\left((1-\t)\cos\phi +(1+\t)i\sin\phi\right)}{\sqrt{(1+\t)^2-4\t \cos^2\phi}}\nn \\
&=\frac{\cc{a}}{\sqrt{2}} \frac{(1-\t)^2\cos^2\phi +(1+\t)^2 \sin^2\phi}{(1+\t)^2-4\t \cos^2\phi}=\frac{\cc{a}}{\sqrt{2}}, 
\intertext{and analogously}
\zeta(b\e^{i\psi},-\phi)&=\frac{b}{\sqrt{2}}.
\end{align}
\end{proof}

\paragraph*{Variant 2}
We see that Corollary \ref{cor_density_scaled_limit} is consistent with Corollary \ref{cor_hm_wz_ellipse} by setting $b=\cc{a}$. Another way to proof the latter one, would have been to use that
\begin{equation}
\begin{split}
\HM\left(\cc{z}_0+\tfrac{\cc{a}\e^{-i\psi}}{\sqrt{\n}},z_0+\tfrac{b\e^{i\psi}}{\sqrt{\n}}\right)=&\HM\left(\cc{z}_0+\tfrac{\cc{b}\e^{-i\psi}}{\sqrt{n}},z_0+\tfrac{b\e^{i\psi}}{\sqrt{n}} \right)\\
&+\int_{\frac{\cc{b}\e^{-i\psi}}{\sqrt{\n}}}^{\frac{\cc{a}\e^{-i\psi}}{\sqrt{\n}}} \frac{\partial \HM}{\partial w}\left(\cc{z}_0+u,z_0+\tfrac{b\e^{i\psi}}{\sqrt{\n}} \right)\dd u. 
\end{split}
\end{equation}
This way we would have been left to only one integration.

\subsection{Asymptotics of the Kernel $\tKn$}\label{subsec_asym_kernel}
Now we want to summarize our calculations in the following theorem for the kernel $\tKn$.
\begin{theorem}\label{th_km}
Let $a,b,z_0\in\C$. Then the asymptotics as $\n\rightarrow\infty$ of the kernel $\tKn$ for the potential $V(z)=\vert z\vert^2-\Re(t z^2)$ is given by
\begin{equation}
\begin{split}
&\frac{\tKn}{\n}\left(z_0+\tfrac{a\e^{i\psi}}{\sqrt{\n}},z_0+\tfrac{b\e^{i\psi}}{\sqrt{\n}}\right)\\
&\qquad=\e^{i\nphase(z_0,a\e^{i\psi},b\e^{i\psi})}\e^{-\frac{\vert a-b\vert^2}{2}} \left\{\begin{array}{lll}\frac{1}{\pi} & \text{if $z_0$ lies inside of $\pEt$,}\\ \frac{1}{2\pi}\erfc\left(\frac{\cc{a}+b}{\sqrt{2}}\right) & \text{if $z_0\in \pEt$,}\\ 0 & \text{if $z_0$ lies outside of $\pEt$,}\end{array}\right. 
\end{split}
\end{equation}
with
\begin{align}
\nphase(z_0,a\e^{i\psi},b\e^{i\psi})&=\sqrt{\n}\Im\left(\e^{-i\psi}(\cc{a}-\cc{b})(z_0-\t \cc{z}_0) \right)+\Im\left( \cc{a}b+\tfrac{\t}{2}\e^{2i\psi}(a^2-b^2) \right),
\end{align}
where $\psi\in(-\pi,\pi]$ can be arbitrary if $z_0\notin\pEt$ and is defined as in Corollary \ref{cor_density_scaled_limit} if $z_0=\mas \cos \phi+i\mis \sin\phi\in\pEt$.
\end{theorem}
\begin{remark}
We have only introduced the rotation $e^{i\psi}$ in the coordinate system of $a$ and $b$ for $z_0\in\pEt$ so that the real axis is normal to the tangent at the ellipse in $z_0$. If $z_0\notin\pEt$, this rotation does neither help nor harm and $\psi$ can be set to zero. It will only give an additional rotation in the phase which is oscillating with $\n$. 
\end{remark}

\begin{proof}
Comparing Definition \ref{def_norm_kernel} with Definition \ref{def_hm}, we see that for $\t\in\R$ $\tKn$ is related to $\HM$ as
\begin{align}
\frac{\tKn}{\n}(w,z)&=\HM(\cc{w},z) \e^{\n \left(\cc{w}z-\frac{\vert z\vert^2+\vert w\vert^2 }{2}+\frac{\t}{2}\left( \Re(z^2)+\Re(w^2)-\cc{w}^2-z^2 \right) \right)}.
\intertext{Setting $w=z_0+\tfrac{a}{\sqrt{\n}}$ and $z=z_0+\tfrac{b}{\sqrt{\n}}$, we get}
\frac{\frac{\tKn}{\n}\left(z_0+\tfrac{a}{\sqrt{\n}},z_0+\tfrac{b}{\sqrt{\n}}\right)}{\HM \left(\cc{z}_0+\tfrac{\cc{a}}{\sqrt{\n}},z_0+\tfrac{b}{\sqrt{\n}}\right)}&=\e^{\frac{\t}{2}\left(\sqrt{\n}2i\Im(a z_0+\cc{b}\cc{z}_0)+i\Im(a^2+b^2)  \right) } \e^{\sqrt{\n}i \Im(\cc{a}z_0+b\cc{z}_0)+i\Im(\cc{a}b)-\frac{\vert a-b\vert^2}{2}} \nn \\
&=\e^{\sqrt{\n}i\Im\left((\cc{a}-\cc{b})(z_0-\t \cc{z}_0) \right)+i\Im\left(\cc{a}b+\frac{\t}{2}(a^2-b^2) \right) } \e^{-\frac{\vert a-b\vert^2}{2}}.
\end{align}
Substituting $a$ by $a\e^{i\psi}$ and $b$ by $a\e^{i\psi}$ and using for 
$$\HM\left(\cc{z}_0+\tfrac{\cc{a}\e^{-i\psi}}{\sqrt{\n}},z_0+\tfrac{b\e^{i\psi}}{\sqrt{\n}}\right)$$
Corollary \ref{cor_hm_wz_ellipse} if $z_0\in\pEt$ or Proposition \ref{prop_hm_wz} together with Proposition \ref{prop_density} if $z_0\notin\pEt$, proves the theorem.
\end{proof}

\subsection{Universality of $\vert \tKn\vert$}
We have seen that if $z_0$ lies either inside, outside or on the ellipse $\pEt$, then in each of these domains the asymptotics as $\n\rightarrow\infty$ of $\HM$ will be independent of $z_0$ and $\t$ and depends only on $a$ and $b$ if $z_0\in\pEt$ (in a coordinate system where the real axis is normal to the tangent at the ellipse in $z_0$). In contrast to this, $\tKn$ has the phase $\nphase$ that not only depends on $z_0$ and $\t$ (additionally to $a$ and $b$) but is also highly oscillating with $\n$ so that $\tKn$ has not a well defined limit when $\n\rightarrow\infty.$

However, we can formulate the following universality theorem for the absolute value of $\tKn$.

\begin{theorem}[Universality]\label{th_univ_kernel}
Let $a,b\in\C$. Then for every potential $V(z)=\vert z\vert^2-\Re(t z^2)$ the absolute value of the kernel is
\begin{align}
\lim_{\n\rightarrow\infty}\left\vert\frac{\tKn}{\n}\left(z_0+\tfrac{a}{\sqrt{\n}},z_0+\tfrac{b}{\sqrt{\n}}\right)\right\vert&=\frac{1}{\pi}\e^{-\frac{\vert a-b\vert^2}{2}}, \label{eq_univ_kernel1}
\intertext{if $z_0$ lies inside of the ellipse $\pEt$ and}
\lim_{\n\rightarrow\infty}\left\vert\frac{\tKn}{\n}\left(z_0+\tfrac{a\e^{i\psi}}{\sqrt{\n}},z_0+\tfrac{b\e^{i\psi}}{\sqrt{\n}}\right)\right\vert&=\frac{1}{2\pi} \e^{-\frac{\vert a-b\vert^2}{2}} \left\vert\erfc\left(\frac{\cc{a}+b}{\sqrt{2}}\right)\right\vert,  
\end{align}
if $z_0=\mas \cos \phi+i\mis \sin\phi\in\pEt$ where $\psi$ is defined as in Corollary \ref{cor_density_scaled_limit}.
\end{theorem}
\begin{proof}
The proof follows directly by taking the absolute value of $\tKn$ in Theorem \ref{th_km}.
\end{proof}

\begin{remark}
Comparing \eqref{eq_univ_kernel1} with the universality of the well-known sine-kernel for Hermitian matrix models, 
\begin{equation}
\lim_{\n \rightarrow \infty} \frac{K^{\mathrm{GUE}}_\n\left(x+\frac{a}{K_\n(x,x)}, x+\frac{b}{K^{\mathrm{GUE}}_\n(x,x)} \right) }{K_\n(x,x)}=\frac{\sin \pi(a-b)}{\pi(a-b)},
\end{equation}
where $x,a,b\in \R$ and $K^{\mathrm{GUE}}_\n$ is the analogous kernel for the Hermitian case, we see that the normal matrix model has not the same kernel though they share some similarities.
\end{remark}
 
\subsection{Universality of Correlation Functions}\label{sec_univ_corr}
Now as we know the kernel $\tKn$, we can find the correlation functions by Definition \ref{def_correlation}. 
\begin{theorem}[Universality of Correlation Functions]\label{th_corr}
Let $\Rn{m}$ be the $m$-point correlation function as in Definition \ref{def_correlation} and $a_1,\ldots,a_m\in\C$. Then
\begin{align}
\lim_{\n\rightarrow\infty} \frac{\Rn{m}}{\n^m} \left(z_0+\tfrac{a_1}{\sqrt{\n}},\ldots,z_0+\tfrac{a_m}{\sqrt{\n}} \right)&= \det \left( \tfrac{1}{\pi}\e^{-\frac{1}{2}\vert a_k-a_l \vert^2+i\Im \left(\cc{a}_k a_l\right) }  \right)_{k,l=1}^m,
\intertext{if $z_0$ lies inside of the ellipse $\pEt$,}
\lim_{\n\rightarrow\infty} \frac{\Rn{m}}{\n^m} \left(z_0+\tfrac{a_1}{\sqrt{\n}},\ldots,z_0+\tfrac{a_m}{\sqrt{\n}} \right)&=0,
\intertext{if $z_0$ lies outside of $\pEt$ and}
\lim_{\n\rightarrow\infty} \frac{\Rn{m}}{\n^m} \left(z_0+\tfrac{a_1\e^{i\psi}}{\sqrt{\n}},\ldots,z_0+\tfrac{a_m\e^{i\psi}}{\sqrt{\n}}\right)&=\det \left( \tfrac{\erfc\left(\frac{\cc{a}_k+a_l}{\sqrt{2}}\right)}{2\pi}\e^{-\frac{1}{2}\vert a_k-a_l \vert^2+ i\Im \left(\cc{a}_k a_l\right) }  \right)_{k,l=1}^m, 
\end{align}
if $z_0=\mas \cos \phi+i\mis \sin\phi\in\pEt$ where $\psi$ is defined as in Corollary \ref{cor_density_scaled_limit}.
\end{theorem}
\begin{proof}
From Definition \ref{def_correlation} we get
\begin{equation}
\begin{split}\label{eq_corr1}
\frac{\Rn{m}}{\n^m} \left(z_0+\tfrac{a_1}{\sqrt{\n}},\ldots,z_0+\tfrac{a_m}{\sqrt{\n}} \right)&=\det\left(\frac{\tKn}{\n}\left(z_0+\tfrac{a_k}{\sqrt{\n}},z_0+\tfrac{a_l}{\sqrt{\n}}\right)\right)_{k,l=1}^m\\
&=\sum_{\sigma\in S_m} \sign(\sigma) \corrdet{m}(\sigma),
\end{split}
\end{equation}
where
\begin{equation}\label{eq_corr2}
\corrdet{m}(\sigma)=\prod_{k=1}^m \frac{\tKn}{\n}\left(z_0+\tfrac{a_k}{\sqrt{\n}},z_0+\tfrac{a_{\sigma(k)}}{\sqrt{\n}}\right), \qquad \underline{a}=(a_1,\ldots,a_m).
\end{equation}
From Theorem \ref{th_univ_kernel} (or Theorem \ref{th_km} if $z_0$ lies outside of $\pEt$) we see that $\vert\corrdet{m}\vert$ is well-defined and finite when $\n\rightarrow\infty$. So we can move the limit in front of $\Rn{m}$ in the theorem into the sum over the permutation $\sigma$, provided the limit $\n\rightarrow\infty$ of $\corrdet{m}$ is well-defined, i.e.\ has no longer an oscillating phase. 

If $z_0$ lies outside of $\pEt$, it is obvious from Theorem \ref{th_km} that then \eqref{eq_corr1} goes to zero when $\n\rightarrow\infty$.

Next we want to analyze $\corrdet{m}$ in the case if $z_0$ lies inside of $\pEt$. We will see that not only the phase oscillating with $\sqrt{\n}$ will disappear, but also that the remaining part of the phase will not depend on $\t$ anymore. Using Theorem \ref{th_km} for $\tfrac{\tKn}{\n}\left(z_0+\tfrac{a_k}{\sqrt{\n}},z_0+\tfrac{a_{\sigma(k)}}{\sqrt{\n}}\right)$ we find
\begin{align}
\MoveEqLeft\lim_{\n\rightarrow\infty}\corrdet{m}(\sigma) \nn \\
& = \lim_{\n\rightarrow\infty}  \tfrac{1}{\pi^m} \e^{ i\sqrt{\n} \Im \sum_{k=1}^m (\cc{a}_k-\cc{a}_{\sigma(k)})(z_0-\t \cc{z}_0)+i\Im \sum_{k=1}^m \cc{a}_k a_{\sigma(k)} + i \Im \sum_{k=1}^m  \frac{\t}{2}(a_k^2-a_{\sigma(k)}^2) -\sum_{k=1}^m \frac{\vert a_k-a_{\sigma(k)}\vert^2}{2}}. \label{eq_corr3}
\end{align}
Using that a summation of all $\cc{a}_k$ is invariant under permutations, we see that the sum in the oscillating phase cancels
\begin{align}
\sum_{k=1}^m (\cc{a}_k-\cc{a}_{\sigma(k)})(z_0-\t \cc{z}_0)&= (z_0-\t \cc{z}_0)\left( \sum_{k=1}^m \cc{a}_k - \sum_{k=1}^m  \cc{a}_{\sigma(k)}\right)=0, \label{eq_corr4}
\intertext{as the third sum}
\sum_{k=1}^m  \tfrac{\t}{2}(a_k^2-a_{\sigma(k)}^2)&=\frac{\t}{2}\left(\sum_{k=1}^m  a_k^2-\sum_{k=1}^m  a_{\sigma(k)}^2 \right)=0. \label{eq_corr5}
\end{align}
So the only term left in the phase is $\Im \sum_{k=1}^m \cc{a}_k a_{\sigma(k)}$, which is independent of $\n$. Therefore the limit in \eqref{eq_corr3} is well-defined and becomes
\begin{equation}\label{eq_corr6}
\lim_{\n\rightarrow\infty}\corrdet{m}(\sigma)=\tfrac{1}{\pi^m}\e^{\sum_{k=1}^m \left( i \Im(\cc{a}_k a_{\sigma(k)})-\frac{1}{2}\left\vert a_k-a_{\sigma(k)}\right\vert^2 \right)}.
\end{equation}
Putting this in \eqref{eq_corr1}, taking the limit $\n\rightarrow\infty$ and rewriting the result as a determinant, we find
\begin{equation}\label{eq_corr7}
\lim_{\n\rightarrow\infty}\frac{\Rn{m}}{\n^m} \left(z_0+\tfrac{a_1}{\sqrt{\n}},\ldots,z_0+\tfrac{a_m}{\sqrt{\n}} \right)= \det \left( \tfrac{1}{\pi}\e^{-\frac{1}{2}\vert a_k-a_l \vert^2+ i\Im \left(\cc{a}_k a_l\right) }  \right)_{k,l=1}^m.
\end{equation}

We are left to the case if $z_0\in\pEt$. Using again Theorem \ref{th_km} and noting that the right-hand side of \eqref{eq_corr6} is invariant under the simultaneous substitution of $a_k$ by $a_k\e^{i\psi}$ and $a_{\sigma(k)}$ by $a_{\sigma(k)}\e^{i\psi}$, we get the same way as before
\begin{equation}\label{eq_corr8}
\lim_{\n\rightarrow\infty}\acorrdet{m}{\ve{a}\e^{i\psi}}(\sigma)=\tfrac{1}{(2\pi)^m}\prod_{k=1}^m\e^{i \Im(\cc{a}_k a_{\sigma(k)})- \frac{1}{2}\left\vert a_k-a_{\sigma(k)}\right\vert^2} \erfc\left(\tfrac{\cc{a}_k+a_{\sigma(k)}}{\sqrt{2}}\right).
\end{equation}
We can write the result for the correlation function again as a determinant, 
\begin{equation}\label{eq_corr9}
\lim_{\n\rightarrow\infty}\frac{\Rn{m}}{\n^m} \left(z_0+\tfrac{a_1\e^{i\psi}}{\sqrt{\n}},\ldots,z_0+\tfrac{a_m\e^{i\psi}}{\sqrt{\n}} \right)= \det \left( \tfrac{\erfc\left(\frac{\cc{a}_k+a_l}{\sqrt{2}}\right)}{2\pi}\e^{-\frac{1}{2}\vert a_k-a_l \vert^2+ i\Im \left(\cc{a}_k a_l\right) }  \right)_{k,l=1}^m.
\end{equation}
\end{proof}

\begin{remark}
In each of its domains (inside, outside or on the ellipse $\pEt$) the correlation function is independent of $z_0$. Further it doesn't depend on $\t$ at all. So we have found a law of universality for Gaussian random normal matrices. 
\end{remark} 


\begin{remark}
The imaginary part of the sum $\sum_{k=1}^m \cc{a}_k a_{\sigma(k)}$ in \eqref{eq_corr6} can only be different from zero if $\sigma$ contains a cycle with length larger than two. So for $m<3$ the correlation function simplifies to
\begin{equation}
\lim_{\n\rightarrow\infty}\frac{\Rn{m}}{\n^m} \left(z_0+\tfrac{a_1}{\sqrt{\n}},\ldots,z_0+\tfrac{a_m}{\sqrt{\n}} \right)= \det \left( \tfrac{1}{\pi}\e^{-\frac{1}{2}\vert a_k-a_l \vert^2}\right)_{k,l=1}^m,
\end{equation}
if $z_0$ lies inside of the ellipse. If $\sigma$ has a cycle of length of at least three then $\sum_{k=1}^m \cc{a}_k a_{\sigma(k)}$ in generally does not vanish and gives a complex phase for the term \eqref{eq_corr6}. Then we find for $\sigma^{-1}$ that $\lim_{\n\rightarrow\infty}\corrdet{m}(\sigma^{-1})=\lim_{\n\rightarrow\infty}\cc{\corrdet{m}(\sigma)}$. Further $\sign(\sigma^{-1})=\sign(\sigma)$. So we can combine these two terms to 
\begin{equation} 
\lim_{\n\rightarrow\infty}\left(\corrdet{m}(\sigma)+\corrdet{m}(\sigma^{-1})\right) = \frac{1}{\pi^m} \cos\left(\sum_{k=1}^m \Im\left(\cc{a}_k a_{\sigma(k)}\right) \right) \e^{-\sum_{k=1}^m \frac{1}{2}\left\vert a_k-a_{\sigma(k)}\right\vert^2 }, 
\end{equation} 
if $z_0$ lies inside the ellipse. 
From this it is obvious that the correlation function is real as it must be. The term $\sum_{k=1}^m \cc{a}_k a_{\sigma(k)}$ contains information how the $a_k$ are arranged. For example if $a_k=x_k \e^{i\phi}+c$ with $x_k\in \R$ for all $k=1,\ldots,m$ and $\phi\in[0,2\pi)$, $c\in\C$ we find that 
\begin{equation}
\Im \sum_{k=1}^m \cc{a}_k a_{\sigma(k)}= \Im \sum_{k=1}^m \left( x_k x_{\sigma(k)}+x_k \e^{-i \phi} c+x_{\sigma(k)}\e^{i\phi} \cc{c}+\cc{c}c\right)=\Im \left(x \e^{-i \phi} c+x \e^{i\phi} \cc{c}\right)=0,
\end{equation}
where $x=\sum_{k=1}^m x_k=\sum_{k=1}^m x_{\sigma(k)}.$ So if all $a_k$ lie on a line in the complex plane, this term vanishes.
\end{remark}

\begin{remark}
Because $\Rn{m}\left(z_0+\tfrac{a_1}{\sqrt{\n}},\ldots,z_0+\tfrac{a_m}{\sqrt{\n}} \right)$ on $\pEt$ does neither depend on the position $z_0\in\pEt$ nor on the parameter of the ellipse, $\t$, we see that the correlation function can't  ``feel'' the curvature of the ellipse at the point $z_0$.
\end{remark}


\subsection{Comparison to the Pure Gaussian Case}
In the limit $\t\rightarrow 0$ we find
\begin{align}
\lim_{\t\rightarrow 0} \gw(z_0)&=\tfrac{1}{\cc{z}_0}, &\lim_{\t\rightarrow 0} \gz(z_0)&=\tfrac{1}{z_0}, \nn \\
\intertext{and}
\lim_{\t\rightarrow 0} \hu(z_0)&=\tfrac{1}{z_0}-\cc{z}_0, & \lim_{\t\rightarrow 0} \huu(z_0)&=-\tfrac{1}{2 z_0^2}. \nn
\end{align}
Thus it is obvious that the asymptotics of identities \eqref{eq_id3} and \eqref{eq_id4}, which we have computed in Proposition \ref{prop_asym_wz}, converge to \eqref{eq_asym_id_kernelw_0} and \eqref{eq_asym_id_kernelz_0} of the pure Gaussian case $t=0$.

In Theorem \ref{th_km} we have found that the reproducing kernel $\nK$ is independent of $\t$ except for a nonrelevant phase. Taking the limit $\t\rightarrow 0$ of this phase and comparing the kernel with the result for the pure Gaussian case in Proposition \ref{prop_kernel_0}, we see that they are identical.

Further we have found that the correlation functions are independent of $\t$ and look exactly as for the pure Gaussian potential (see Proposition \ref{prop_corr_0}). 

As $\t$ goes to zero, $\pEt$ becomes the unit circle and therefore we can get all results for the special case $\t=0$ by the limit. This is not very surprising as the circle is not a degenerated case in view of the random matrix model. Thus we can remove the condition $\t>0$ in our results for the kernel and correlation functions and include the case $\t=0$.


\section{Verification of the Approximations}\label{sec_verification}
\subsection{Error Term in the Asymptotics of Identities}
In this chapter we are going to validate the approximation we have made in the previous chapters. We will give estimations of all error terms.

\subsubsection{Plancherel-Rotach Asymptotics on $\Adelta$} 
We start to review Proposition \ref{prop_asym}.
\begin{prop}\label{prop_asymp_R}
Let $\delta>0$, $0<\delta_0<1$ and $\n_0\in\N$. Then there exists functions 
\begin{align}
\begin{array}{rcl}
\xi_k: \N\setminus\{1,\ldots,\n_0\} & \longrightarrow & \left[0,\tfrac{1}{\n_0+1}\right]\\
\n & \longmapsto & \xi_k(\n)\le\frac{1}{\n}
\end{array}, \qquad k=1,2,3,
\end{align}
so that for all $0<\t<1-\delta_0$, $z\in \Adelta$ and $\n>\n_0$, 
\begin{align}
\left.\left(\tfrac{\partial \HM}{\partial w}+\tfrac{\partial \HM}{\partial z}\right)\right\vert_{w=\cc{z}} &=\sqrt{\tfrac{\n}{2 \pi^3}}\frac{\gp(z) \e^{\n \f(z)}}{ \Re{\, \U{\F}(z)} } \Re\left( \U{\F}(z) \bigl(1+\Rtz\bigr)\right), \label{eq_rho_dx_R}\\
\intertext{and}
i \left.\left(\tfrac{\partial \HM}{\partial w}-\tfrac{\partial \HM}{\partial z}\right)\right\vert_{w=\cc{z}} &=\sqrt{\tfrac{\n}{2 \pi^3}} \frac{ \gm(z) \e^{\n \f(z)}}{\, \Im{\U{\F}(z)} } \Im\left( \U{\F}(z) \bigl(1+\Rtz\bigr)\right), \label{eq_rho_dy_R}
\intertext{where $\f$ and $\gpm$ are as in Proposition \ref{prop_asym} and}
1+\Rtz&=\e^{\Rzero+\Rone+\Rtwo+\Rthree},
\intertext{with} 
\Rzero&=\log \left( \tfrac{ Q_\n\left( \sqrt{2}\frac{\cc{z}}{\F} \right)}{\pi_\n \left( \sqrt{2}\frac{\cc{z}}{\F} \right) } \tfrac{ Q_{\n-1}\left( \sqrt{\frac{\n}{\n-1}} \sqrt{2}\frac{z}{\F} \right)}{\pi_{\n-1} \left(\sqrt{\frac{\n}{\n-1}} \sqrt{2}\frac{z}{\F} \right) } \right) ,\\ 
\Rone&=-\tfrac{1}{12} \xi_1(\n) ,\\  
\Rtwo&=-\tfrac{\Fz\mp\sqrt{\Fz^2-1+\xi_2(\n)}}{\pm4 \n (1-\xi_2(\n))\sqrt{\Fz^2-1+\xi_2(\n)}} = -\tfrac{\F\U{\F}\left(\frac{z}{\sqrt{1-\xi_2(\n) } }  \right)  }{4\n\left(1-\xi_2(\n) \right) \T{\F}\left(\frac{z}{\sqrt{1-\xi_2(\n) } }  \right) },\\
\Rthree&= -\tfrac{\Fz(\Fz\mp\sqrt{\Fz^2-1+\xi_3(\n) }) }{4\n (1-\xi_3(\n)) (\Fz^2-1+\xi_3(\n))} =-\tfrac{z\F\U{\F}\left(\frac{z}{\sqrt{1-\xi_3(\n) } } \right)  }{4\n\left(1-\xi_3(\n) \right)^{3/2} \left(\T{\F}\left(\frac{z}{\sqrt{1-\xi_3(\n) } }  \right)\right)^2  }  .
\end{align}
\end{prop}
\begin{remark}
Note that $$\frac{\gp(z)}{\Re{\,\U{\F}(z) }  } \qquad \text{and} \qquad \frac{\gm(z)}{\Im{\,\U{\F(z)}}  }$$ are well-defined even where $\Re{\,\U{\F}(z)}=0$ or $\Im{\,\U{\F}(z)}=0$ since $\gpm(z)$ contain the factor $\Re{\,\U{\F}(z)}$ or $\Im{\,\U{\F}(z)}$ respectively.
\end{remark}

\begin{remark}
Like in Proposition \ref{prop_asym}, we can write \eqref{eq_rho_dx_R} with a relative error 
\begin{align}
\left.\left(\tfrac{\partial \HM}{\partial w}+\tfrac{\partial \HM}{\partial z}\right)\right\vert_{w=\cc{z}} &=\sqrt{\tfrac{\n}{2 \pi^3}} \gp(z) \e^{\n \f(z)}\left(1+ \tfrac{ \Re\left( \U{\F}(z) \Rtz\right)}{\Re{\,\U{\F}(z)}}\right), \\
\intertext{if $\Re{\,\U{\F}(z)} \neq 0$ and analogously for \eqref{eq_rho_dy_R} }
i\left.\left(\tfrac{\partial \HM}{\partial w}+\tfrac{\partial \HM}{\partial z}\right)\right\vert_{w=\cc{z}} &=\sqrt{\tfrac{\n}{2 \pi^3}} \gm(z) \e^{\n \f(z)}\left(1+ \tfrac{ \Im\left(\,\U{\F}(z)\Rtz\right)}{\Im{\, \U{\F}(z)}}\right),
\end{align} 
if $\Im{\,\U{\F}(z)} \neq 0$.
\end{remark}

\begin{proof}
We are going back to the proof of Proposition \ref{prop_asym} and we will reexamine all the approximations and expansions we did in that proof. 
We start with the error terms coming from the Plancherel-Rotach asymptotics \eqref{eq_pi_o}. We define 
\begin{equation}\label{eq_proof_R0}
\Rzero=\log \left( \tfrac{ Q_\n\left( \sqrt{2}\frac{\cc{z}}{\F} \right)}{\pi_\n \left( \sqrt{2}\frac{\cc{z}}{\F} \right) } \tfrac{ Q_{\n-1}\left( \sqrt{\frac{\n}{\n-1}} \sqrt{2}\frac{z}{\F} \right)}{\pi_{\n-1} \left(\sqrt{\frac{\n}{\n-1}} \sqrt{2}\frac{z}{\F} \right) } \right).
\end{equation}
The logarithm is well defined because the zeros of $Q_\n$ and $Q_{\n-1}$ are on the interval $[-\sqrt{2},\sqrt{2}]$ for all $\n\in\N$. Note that under complex conjugation it has the symmetry $\Rzeroz{\cc{z}}=\cc{\Rzero}$ because $Q_\n$, $Q_{\n-1}$, $\pi_\n$ and $\pi_{\n-1}$ have the same symmetry. 

We have used the Plancherel-Rotach asymptotics in \eqref{eq_asymp1} when we applied \eqref{eq_pr_asymp}. Now we will use $\e^{\Rzero}$ instead of the $1+o(1)$ terms in \eqref{eq_pr_asymp}. 

The Plancherel-Rotach asymptotics can be applied for all $\tilde{\delta}>0$ when both $\sqrt{2}\tfrac{z}{\F}$ and $\sqrt{\tfrac{\n}{\n-1}} \sqrt{2}\tfrac{z}{\F}$ lie in $\Xdelta{A}{\tilde{\delta}}{\sqrt{2}}$. This is fulfilled by the condition $z\in \Adelta$ of the proposition when we set $\tilde{\delta}=\sqrt{2}\tfrac{\delta}{\F_0}$, with $\F_0=\sqrt{\tfrac{4(1-\delta_0)}{2\delta_0-\delta_0^2}}$ the largest value of $\F$, since
\begin{align}
z\in \Adelta \iff \sqrt{2}\tfrac{z}{\F}\in \Xdelta{A}{\frac{\sqrt{2}\delta}{\F}}{\sqrt{2}}\subset \Xdelta{A}{\frac{\sqrt{2}\delta}{\F_0}}{\sqrt{2}}, 
\end{align}
and also $\sqrt{\tfrac{\n}{\n-1}} z\in \Adelta$. This is enough for our purpose and we don't discuss this error term here any further. More on it can be found in the references we mentioned in Section \ref{subsection_pr_asym_outside}.

Next we have a look at \eqref{eq_stirling}. For the Stirling approximation we know that 
\begin{align}
\frac{\sqrt{2\pi \n} \n^\n \e^{-\n} }{\n!}&=\e^{-\frac{1}{12} \xi_1(\n)}=\e^{\Rone}, \label{eq_proof_R1}
\intertext{where $\xi_1(\n)\in (\tfrac{12}{12\n+1},\tfrac{1}{\n})\subset (0,\tfrac{1}{\n})$. We have defined the error term}
\Rone&=-\tfrac{1}{12} \xi_1(\n) < 0. 
\end{align}

At last we care for the error terms in the expansions \eqref{eq_expan1}--\eqref{eq_expan4b} and \eqref{eq_expan5}. We know from Lemma \ref{lemma_R2_R3_new} that for all $z\in \C \setminus [-\F,\F]$
\begin{align}
\e^{\n\ftwo\left(\tfrac{1}{\n}\right)}&=\left(\tfrac{\n-1}{\n}\right)^{\n-1} \e^{\frac{\n}{2} \Bigl( \sqrt{1-\frac{1}{\n}} \U{\F}\left(  z \cramped[\scriptscriptstyle]{\left(1-\frac{1}{\n}\right)^{-1/2}} \right) \Bigr)^2 }  \left( \sqrt{1-\tfrac{1}{\n}}  \U{\F}\left(  z \left(1-\tfrac{1}{\n}\right)^{-1/2} \right) \right)^{-(\n-1)} \nn\\
&=\e^{\frac{\n}{2}\left(\U{\F}(z)\right)^2} \left(\U{\F}(z) \right)^{-(\n-1)} \e^{\Rtwo}, \label{eq_proof_R2}
\shortintertext{and}
\e^{\fthree{}\left(\tfrac{1}{\n} \right) } 
&=\W{\F}\left(  z \left(1-\tfrac{1}{\n}\right)^{-1/2} \right) =\W{\F}(z) \e^{\Rthree}, \label{eq_proof_R3} 
\end{align}
where
\begin{align}
\Rtwo 
&= -\tfrac{\F\U{\F}\left(\frac{z}{\sqrt{1-\xi_2(\n) } }  \right)  }{4\n\left(1-\xi_2(\n) \right) \T{\F}\left(\frac{z}{\sqrt{1-\xi_2(\n) } }  \right) },\\
\shortintertext{and}
\Rthree 
&=-\tfrac{z\F\U{\F}\left(\frac{z}{\sqrt{1-\xi_3(\n) } } \right)  }{4\n\left(1-\xi_3(\n) \right)^{3/2} \left(\T{\F}\left(\frac{z}{\sqrt{1-\xi_3(\n) } }  \right)\right)^2  } .
\end{align}
with $\xi_2(\n),\xi_3(\n)\in \left[0,\tfrac{1}{\n}\right].$

We see that $\e^{\n\ftwo \left(\frac{1}{\n}\right)}$ represents the terms for which we have used expansions \eqref{eq_expan1}, \eqref{eq_expan2} and \eqref{eq_expan5} and $\e^{\fthree{}\left(\frac{1}{\n}\right)}$ those of \eqref{eq_expan4b}. 

Note that like $\Rzero$ (and $\Rone$ trivially) all the error terms have the symmetry 
\begin{equation}
\Rtwoz{\cc{z}}=\cc{\Rtwo} \qquad \text{and}\qquad \Rthreez{\cc{z}}=\cc{\Rthree}.  \nn
\end{equation}

With \eqref{eq_proof_R2} and \eqref{eq_proof_R3} we can write $\op{\n-1}$ in \eqref{eq_op2} as
\begin{align} 
\op{\n-1}(z)&=\tfrac{Q_{\n-1}\left(\sqrt{\frac{\n}{\n-1}}\sqrt{2}\Fz\right)}{\pi_{\n-1}\left(\sqrt{\frac{\n}{\n-1}}\sqrt{2}\Fz\right)} \tfrac{\n^{\frac{\n+1}{2}} \t^{\frac{\n-1}{2}} \left(1-\t^2\right)^{1/4} \e^{\n \ftwo\left(\frac{1}{\n}\right)+\fthree{}\left(\frac{1}{\n}\right) }  }{2\sqrt{\pi}\sqrt{\n !}}\nn\\
&=\tfrac{Q_{\n-1}\left(\sqrt{\frac{\n}{\n-1}}\sqrt{2}\Fz\right)}{\pi_{\n-1}\left(\sqrt{\frac{\n}{\n-1}}\sqrt{2}\Fz\right)} \tfrac{\n^{\frac{\n+1}{2}}}{2\sqrt{\pi}\sqrt{\n !}}  \e^{\frac{\n}{2}\left(\U{\F}(z)\right)^2} \left(\U{\F}(z) \right)^{-(\n-1)}\nn\\
& \phantom{=} \quad \cdot \t^{\frac{\n-1}{2}} \left(1-\t^2\right)^{1/4} \W{\F}(z) \e^{\Rtwo+\Rthree}.\label{eq_op2_error}
\end{align}

Now we can put \eqref{eq_op1} and \eqref{eq_op2_error} in the right-hand side of identities \eqref{eq_id1} and \eqref{eq_id2} evaluated at $w=\cc{z}$ and use \eqref{eq_proof_R0} and \eqref{eq_proof_R1}. So we get as in \eqref{eq_asymp3} and \eqref{eq_asymp4}, but this time with an explicit error term, 
\begin{align}
\left.\left(\frac{\partial \HM}{\partial w}\pm\frac{\partial \HM}{\partial z}\right)\right\vert_{w=\cc{z}}&= \frac{\sqrt{\n} \e^{\n \f(z)}}{\sqrt{2\pi^3}}\left\{ \begin{array}{c}\frac{\gp(z) \Re\left(\U{\F}(z)\left(1+\Rtz \right)\right)}{\Re{\, \U{\F}(z)}} \\ \frac{-i \gm(z) \Im\left(\U{\F}(z)\left(1+\Rtz\right)\right)}{\Im{\,\U{\F}(z)}} \end{array}\right. ,
\end{align}
where 
\begin{equation*}
1+\Rtz=\e^{\Rzero+\Rone+\Rtwo+\Rthree}.
\end{equation*}
\end{proof}

\begin{defin}
For $\delta>0$ we define
\begin{align}
\C_{x,\delta}&=\{z\in\C \vert \vert \Re\, z\vert > \delta\},\qquad \text{and}\qquad \C_{y,\delta}=\{z\in\C \vert \vert \Im\, z\vert > \delta\}.
\end{align}
\end{defin}

\begin{prop}\label{prop_estim_R}
Let $\epsilon>0$, $\delta>0$, $0<\delta_0<1$, $\delta_1\ge\delta$ and $R>0$. Then there exists $\n_0=\n_0(\epsilon,\delta,\delta_0,R)\in\N$ and $C=C(\delta,\delta_0)>0$    
 so that 

\begin{center}
\begin{minipage}{130mm}
\begin{enumerate}[\normalfont (i)]
\item $\displaystyle \sup_{\substack{0<\t<1-\delta_0 \\\n>\n_0 }} \sup_{\substack{z\in \Adelta \\ \xi_1,\xi_2,\xi_3\in[0,\frac{1}{\n}] } } \left\vert \Rtz \right\vert <\epsilon $, 
\item $\displaystyle \sup_{\substack{0<\t<1-\delta_0 \\\n>\n_0 }} \sup_{\substack{z\in \Adelta \cap B_R(0)\\ \vert \Re\, z\vert \ge\delta_1  \\ \xi_1,\xi_2,\xi_3\in[0,\frac{1}{\n}] }} \left\vert \tfrac{\Re\left( \U{\F}(z) \Rtz \right)}{\Re{\, \U{\F}(z)}} \right\vert <\epsilon$, 
\item $\displaystyle \sup_{\substack{0<\t<1-\delta_0 \\\n>\n_0 }} \sup_{ \substack{z\in B_R(0)\\ \vert \Im\, z\vert \ge\delta_1  \\ \xi_1,\xi_2,\xi_3\in[0,\frac{1}{\n}]} } \left\vert \tfrac{\Im\left(\U{\F}(z) \Rtz \right)}{\Im{\, \U{\F}(z)}} \right\vert <\epsilon,$
\item $\displaystyle \sup_{\substack{0<\t<1-\delta_0 \\\n>\n_0 }} \sup_{\substack{z\in \Adelta \\ \xi_1,\xi_2,\xi_3\in[0,\frac{1}{\n}] } } \left\vert \tfrac{\vert \W{\F}(z)\vert^2 }{\F} \sqrt{\tfrac{1\mp\t}{1\pm\t}} \Re\left(\U{\F}(z) \bigl(1+\Rtz\bigr)\right) \right\vert < C,$
\end{enumerate}
\end{minipage}
\end{center}
where $\Rtz$ is as in Proposition \ref{prop_asymp_R}. 

\end{prop}
\begin{proof}
Without loss of generality $\epsilon<1$. The largest value of $\F$ is $\F_0=\sqrt{\tfrac{4(1-\delta_0)}{2\delta_0-\delta_0^2}}$. 

We start with the first claim for 
\begin{equation}
\Rtz=\e^{\Rzero+\Rone+\Rtwo+\Rthree}-1.
\end{equation}
Let $\epsilon_1=\tfrac{\epsilon}{8}$. When $z\in \Adelta$ it follows that $\tfrac{\sqrt{2}\cc{z}}{\F}, \sqrt{\tfrac{\n}{\n-1}}\tfrac{\sqrt{2}z}{\F} \in \Xdelta{A}{\frac{\sqrt{2}\delta}{\F}}{\sqrt{2}}\subset \Xdelta{A}{\frac{\sqrt{2}\delta}{\F_0}}{\sqrt{2}}=\Xdelta{A}{\tilde{\delta}}{\sqrt{2}}$ where we have defined $\tilde{\delta}=\tfrac{\sqrt{2}\delta}{\F_0}>0$. From Section \ref{subsection_pr_asym_outside} we know that for any $\tilde{\delta}>0$ the Plancherel-Rotach asymptotics \eqref{eq_pi_o} holds, i.e.\ $\tfrac{Q_{\n}(z)}{\pi_{\n}(z)}$ converges to one as $\n\rightarrow\infty$ uniformly for all $z\in \Xdelta{A}{\tilde{\delta}}{\sqrt{2}}$. Therefore
\begin{align}
\lim_{\n \rightarrow\infty} \tfrac{ Q_\n\left( \sqrt{2}\frac{\cc{z}}{\F} \right)}{\pi_\n \left( \sqrt{2}\frac{\cc{z}}{\F} \right) }=1, \qquad \text{and} \qquad \lim_{\n \rightarrow\infty} \tfrac{ Q_{\n-1}\left( \sqrt{\frac{\n}{\n-1}} \sqrt{2}\frac{z}{\F} \right)}{\pi_{\n-1} \left(\sqrt{\frac{\n}{\n-1}} \sqrt{2}\frac{z}{\F} \right) }=1, \nn
\end{align} 
uniformly for all $0<\F<\F_0$ and $z\in \Adelta$. Since the logarithm is continuous at $1$, it follows that for all $\n$ big enough
\begin{align}\label{eq_prop_estim_R_1}
\vert \Rzero \vert =\left \vert \log \left( \tfrac{ Q_\n\left( \sqrt{2}\frac{\cc{z}}{\F} \right)}{\pi_\n \left( \sqrt{2}\frac{\cc{z}}{\F} \right) } \tfrac{ Q_{\n-1}\left( \sqrt{\frac{\n}{\n-1}} \sqrt{2}\frac{z}{\F} \right)}{\pi_{\n-1} \left(\sqrt{\frac{\n}{\n-1}} \sqrt{2}\frac{z}{\F} \right) } \right) \right \vert <\epsilon_1, \qquad \forall 0<\t<1-\delta_0, z\in \Adelta.
\end{align}

$ \vert \Rone \vert =\tfrac{1}{12} \xi_1(\n) < \epsilon_1$ if $\xi_1(\n)<12\epsilon_1$, which is fulfilled if $\n>\tfrac{1}{12\epsilon_1}$ since $\xi_1(\n) \in [0,\tfrac{1}{\n}]$.

When $\epsilon_1>0$ and $\xi_2,\xi_3\in[0,\tfrac{1}{2}]$, we know from Lemma \ref{lemma_estim_R2_R3} that for $\n$ big enough
\begin{align}\label{eq_prop_estim_R_3}
\sup_{0<\t<1-\delta_0} \sup_{\substack{z\in \Adelta \\ \xi_2(\n)\in [0,\frac{1}{2}] } } \vert \Rtwo \vert <\epsilon_1, \quad \text{and}\quad \sup_{0<\t<1-\delta_0} \sup_{\substack{z\in \Adelta \\ \xi_3(\n)\in [0,\frac{1}{2}] } } \vert \Rthree \vert <\epsilon_1.
\end{align}

So we can choose $\n_0$ big enough so that for all $\n \ge\n_0$
\begin{align}
\vert \Rzero+\Rone+\Rtwo+\Rthree \vert &<4\epsilon_1=\tfrac{\epsilon}{2},\\
\intertext{and therefore}
\left\vert \e^{\Rzero+\Rone+\Rtwo+\Rthree}-1\right\vert &<\epsilon,
\end{align}
uniformly for all $0<\t<1-\delta_0$, $z\in \Adelta$ and $\xi_1(\n),\xi_2(\n),\xi_3(\n)\in [0,\frac{1}{\n}]$. 

To proof the second and third claim we use the parametrization $z=a \cos\phi+i b \sin\phi$ with $a^2=b^2+F^2$. Since $\vert z\vert<R$ it follows that $a<R$ and $b<R$. Let $\epsilon_2=\tfrac{\epsilon \delta_1}{R}>0$. From Proposition \ref{prop_U} we know that $\vert \U{\F}(z)\vert= \tfrac{a-b}{\F}$ and from Lemma \ref{lemma_estim_U2} that $\vert \Re\, \U{\F}(z)\vert > \tfrac{\delta_1}{\F}\left(1-\tfrac{b}{a}\right)$ if $\vert\Re\, z\vert>\delta_1$ and $\vert \Im\, \U{\F}(z)\vert > \tfrac{\delta_1}{\F}(\tfrac{a}{b}-1)$ if $\vert\Im\, z\vert>\delta_1$. Using the first claim, there exists $\n_0$ so that $\left\vert \Rtz \right\vert <\epsilon_2$ for all $0<\t<1-\delta_0$, $z\in \Adelta$ and $\xi_1(\n),\xi_2(\n),\xi_3(\n)\in [0,\frac{1}{\n}]$. Thus we find that 
\begin{align}
\left\vert \tfrac{\Re\left(\U{\F}(z) \Rtz \right)}{\Re\, \U{\F}(z)} \right\vert & \le  \tfrac{\left\vert \U{\F}(z) \Rtz \right\vert }{\left\vert\Re\, \U{\F}(z)\right\vert}\le \tfrac{\epsilon_2 \frac{a-b}{\F} }{\frac{\delta_1}{\F}\left(1-\tfrac{b}{a}\right)}=\tfrac{\epsilon_2 a}{\delta_1}< \tfrac{\epsilon_2 R}{\delta_1} =\epsilon,
\intertext{and}
\left\vert \tfrac{\Im\left(\U{\F}(z) \Rtz \right)}{\Im\, \U{\F}(z) } \right\vert & \le  \tfrac{\left\vert \U{\F}(z) \Rtz \right\vert }{\left\vert\Im\, \U{\F}(z)\right\vert}\le \tfrac{\epsilon_2 \frac{a-b}{\F} }{ \tfrac{\delta_1}{\F}(\tfrac{a}{b}-1) }=\tfrac{\epsilon_2 b}{\delta_1}< \tfrac{\epsilon_2 R}{\delta_1} =\epsilon,
\end{align}
uniformly for all $0<\t<1-\delta_0$, $\xi_1(\n),\xi_2(\n),\xi_3(\n)\in [0,\frac{1}{\n}]$ and $z \in B_R(0) \cap \Adelta\cap \C_{x,\delta_1}$ or $z \in B_R(0) \cap \Adelta \cap \C_{y,\delta_1}$ respectively.

The last claim follows if we can find $\n_0\in\N$ and $C\in\R$ so that for all $\n>\n_0$
\begin{align}
\sup_{\substack{0<\t<1-\delta_0 }} \sup_{\substack{z\in \Adelta  } } \left( \vert \W{\F}(z)\vert^2  \sqrt{\tfrac{1\mp\t}{1\pm\t}} \tfrac{\vert \U{\F}(z) \vert}{\F} \sup_{ \xi_1,\xi_2,\xi_3\in[0,\frac{1}{\n}]}  \left\vert1+\Rtz\right\vert\right) < C.
\end{align}
This is fulfilled for $C=\tfrac{4}{\delta \delta_0}\left(1+\sqrt{1+\tfrac{\F_0^2}{\delta^2} } \right)$ since we know from the first claim that there exists $\n_0$ so that
\begin{align}
\sup_{\substack{0<\t<1-\delta_0 \\\n>\n_0 }} \sup_{\substack{z\in \Adelta \\ \xi_1,\xi_2,\xi_3\in[0,\frac{1}{\n}] } } \left\vert \Rtz \right\vert <1,
\end{align}
and from Lemma \ref{lemma_estim_W}, Lemma \ref{lemma_estim_U_T} and Lemma \ref{lemma_estim_b} that for all $0<\t<1-\delta_0$
\begin{align}
\sup_{\substack{z\in \Adelta  } }  \vert \W{\F}(z)\vert^2 < 2\left(1+\sqrt{1+\tfrac{\F_0^2}{\delta^2}} \right), \quad \text{and} \quad  \sup_{\substack{z\in \Adelta}} \tfrac{\vert \U{\F}(z) \vert}{\F}< \tfrac{1}{2\delta}. \nn
\end{align}
\end{proof}
\begin{remark}
From Proposition \ref{prop_asymp_R} together with Proposition \ref{prop_estim_R} we can see that the $o(1)$ term in Proposition \ref{prop_asym} holds uniformly for $z\in \Adelta$. 
\end{remark}

\subsubsection{Maximum Principle for Asymptotics on $\Cdelta$}
Up to now we have neglected the asymptotics in a neighborhood of $\F$ and $-\F$. In this section we are going to take care of this neighborhood. Instead of another Plancherel-Rotach asymptotics which holds around $\pm\F$, we will use the maximum principle for holomorphic functions to estimate the orthonormal polynomials. 

\begin{defin}
For $\delta>0$ we define the complement of $\Adelta$
\begin{align}
\Adelta^C&=\{z\in \C \vert z\notin \Adelta\},\\
\intertext{the boundary}
\partial\Adelta&=\{z\in \C \vert (\vert \Im\, z\vert =\delta \land \vert \Re\, z\vert \le \F)\lor (\vert \Re\, z\vert =\F+\delta \land \vert \Im\, z \vert \le \delta\},\\
\intertext{and for $0\le \lambda< \F $ the following subset of $\Adelta^C$ }
\Xdelta{C}{\delta}{\F,\lambda}&=\{z\in \Adelta^C \vert  \vert \Re\, z\vert \ge \F-\lambda\}.
\end{align}
\end{defin}

The maximum principle tells us that $\vert \op{m}(z) \vert$ has its maximum on the boundary, i.e.\ 
\begin{align}
\max_{z\in \Adelta^C\cup \partial \Adelta} \vert \op{m}(z) \vert = \max_{z\in \partial \Adelta} \vert \op{m}(z) \vert.
\end{align}
Thus we can use the Plancherel-Rotach asymptotics on $\Adelta$ to estimate the maximum of $\vert \op{\n}(z) \vert$ and $\vert \op{\n-1}(z) \vert$ on $\partial \Adelta$ and so we have an estimation on $\Adelta^C$. For $t<1$ we will even see that this is not only useful in a neighborhood of $\F$ and $-\F$ but also on a larger part of $\Adelta^C$, for small $\t$ even all on $\Adelta^C$.

Using \eqref{eq_proof_R1} we can see from \eqref{eq_op1} that for $\n>1$ we can estimate $\op{\n}$ as
\begin{align}
\vert \op{\n}(z)\vert &=\left \vert \tfrac{Q_\n\left (\sqrt{2}\frac{z}{\F} \right)}{\pi_\n\left(\sqrt{2}\frac{z}{\F} \right)} \right\vert \tfrac{\n^{\frac{\n+1}{2}}}{2\sqrt{\pi}\sqrt{\n !}}\e^{\frac{\n}{2}\Re\left(\left( \U{\F}(z) \right)^2\right)} \left\vert \U{\F}(z)\right\vert^{-\n} \t^{\frac{\n}{2}} \left(1-\t^2\right)^{1/4} \vert \W{\F}(z)\vert \nn \\
&= \left \vert \tfrac{Q_\n\left (\sqrt{2}\frac{z}{\F} \right)}{\pi_\n\left(\sqrt{2}\frac{z}{\F} \right)} \right\vert \e^{\frac{\Rone}{2}}\n^{1/4} \gmax(z) \e^{\frac{\n}{2} \fmax(z)}, \label{eq_op1_max}\\
\intertext{and from \eqref{eq_op2_error} that we have for $\op{\n-1}$}
\vert\op{\n-1}(z)\vert&=\left\vert \tfrac{ Q_{\n-1}\left(\sqrt{\frac{\n}{\n-1}}\sqrt{2}\frac{z}{\F} \right)}{\pi_{\n-1}\left(\sqrt{\frac{\n}{\n-1}}\sqrt{2}\frac{z}{\F}\right)} \right\vert \tfrac{\n^{\frac{\n+1}{2}}}{2\sqrt{\pi}\sqrt{\n !}}  \e^{\frac{\n}{2} \Re\left( \left(\U{\F}(z)\right)^2 \right) } \left\vert\U{\F}(z) \right\vert^{-(\n-1)} \t^{\frac{\n-1}{2}} \nn\\ 
& \qquad \cdot \left(1-\t^2\right)^{1/4} \vert\W{\F}(z)\vert \e^{\Re\left(\Rtwoz{z}+\Rthreez{z} \right)} \nn\\
&=\left\vert \tfrac{ Q_{\n-1}\left(\sqrt{\frac{\n}{\n-1}}\sqrt{2}\frac{z}{\F} \right)}{\pi_{\n-1}\left(\sqrt{\frac{\n}{\n-1}}\sqrt{2}\frac{z}{\F}\right)} \right\vert \e^{\frac{\Rone}{2}+\Re\left(\Rtwoz{z}+\Rthreez{z} \right)}\tfrac{\n^{1/4} \vert \U{\F}(z)\vert \gmax(z) }{\sqrt{\t}}  \e^{\frac{\n}{2}\fmax(z)}, \label{eq_op2_max}\\
\intertext{where $\Rone$, $\Rtwotz{}$, $\Rthreetz{}$ are as in Proposition \ref{prop_asymp_R} with suitable $\xi_1, \xi_2, \xi_3\in [0,\tfrac{1}{\n}]$, }
\fmax(z)&=1+\Re\left( \left(\U{\F}(z)\right)^2 \right)-2 \log \left\vert \U{\F}(z) \right\vert+\log \t,
\intertext{and}
\gmax(z)&=\tfrac{ (1-\t^2)^{1/4} \vert \W{\F}(z) \vert } {2 (2\pi^3)^{1/4} }.
\end{align}

\begin{lemma}\label{lemma_fmax}
Let $\delta>0$ and $\F\in (0,\infty)$. Then the maxima of 
\begin{align}
\fmax(z)&=1+\Re\left( \left(\U{\F}(z)\right)^2 \right)-2 \log \left\vert \U{\F}(z) \right\vert+\log \t
\end{align}
on $\partial \Adelta$ are at $\F+\delta\pm i \delta$ and $-\F-\delta\pm i \delta$.
\end{lemma}
\begin{proof}
We define $\fmax(x,y)=\fmax(x+iy)$ 
as a function of two real variables $x$ and $y$. The derivatives of $\fmax$ are
\begin{align}
\frac{\partial \fmax}{\partial x}(x,y)&=\frac{4}{\F}\Re\,\U{\F}(x,y), &\text{and}\quad \frac{\partial \fmax}{\partial y}(x,y)&=-\frac{4}{\F}\Im\,\U{\F}(x,y), \nn
\intertext{and the second derivatives}
\frac{\partial^2 \fmax}{\partial x^2}(x,y)&=\frac{4}{\F}\Re\left(\frac{\U{\F}(x,y)}{\T{\F}(x,y)}\right), &\text{and}\quad \frac{\partial \fmax}{\partial y}(x,y)&=-\frac{4}{\F}\Re\left(\frac{\U{\F}(x,y)}{\T{\F}(x,y)}\right). \nn
\end{align}
With the help of Proposition \ref{prop_U} we see that the critical points on the lines $\Im\, z=\pm \delta$ are at $z=\pm i\delta$ and that $\tfrac{\partial^2 \fmax}{\partial x^2}(0,\pm\delta)>0$. So the critical points are minima. Analogously the critical points on the lines $\Re\, z=\pm (F+\delta)$ are at $z=\pm (F+\delta)$ and these are also minima. Therefore the global maximum of $\fmax$ on $\partial\Adelta$ are at $\F+\delta(1\pm i)$ and $-\F-\delta(1\pm i)$. Because of symmetries all maxima are of same value. 
\end{proof}

\begin{prop}\label{prop_using_maximum_principle}
Let $\epsilon>0$, $0<\delta_0<\e^{-2}$ and $0<\delta_1<1-\delta_0$. Then there exists $\delta=\delta(\epsilon,\delta_0,\delta_1)>0$ so that for all $\n>1$
\begin{align}
\sup_{z\in \Adelta^C} \vert \op{\n}(z)\vert&< g_1(\n,\t,\zmax) \e^{\frac{\n}{2} \left(2+\log \t+\epsilon \right)},  &\forall \t &\in(\delta_1,1-\delta_0), \label{eq_estim_max_princ1} \\
\sup_{z\in \Adelta^C} \vert \op{\n-1}(z)\vert&< g_2(\n,\t,\zmax) \e^{\frac{\n}{2} \left(2+\log \t+\epsilon \right)},  &\forall \t &\in(\delta_1,1-\delta_0),  \label{eq_estim_max_princ2} \\
\sup_{z\in \Adelta^C} \vert \op{\n}(z)\vert&< g_1(\n,\t,\zmax) \e^{-\frac{\n \delta_0}{2}}, & \forall \t & \in(0,\e^{-2}-\delta_0), \label{eq_estim_max_princ3} \\
\intertext{and} 
\sup_{z\in \Adelta^C} \vert \op{\n-1}(z)\vert&< g_2(\n,\t,\zmax) \e^{-\frac{\n \delta_0}{2}}, & \forall \t & \in(0,\e^{-2}-\delta_0), \label{eq_estim_max_princ4}
\end{align}
where 
\begin{align}
\zmax&=\F+(1+i)\delta,\\
g_1(\n,\t,z)&= \left \vert \tfrac{Q_\n\left (\sqrt{2}\frac{z}{\F} \right)}{\pi_\n\left(\sqrt{2}\frac{z}{\F} \right)} \right\vert \n^{1/4} \gmax(z) \sup_{\xi_1 \in [0,\frac{1}{\n}] } \e^{\frac{\Rone}{2}} ,
\end{align}
\begin{align}
g_2(\n,\t,z)&= \left\vert \tfrac{ Q_{\n-1}\left(\sqrt{\frac{\n}{\n-1}}\sqrt{2}\frac{z}{\F} \right)}{\pi_{\n-1}\left(\sqrt{\frac{\n}{\n-1}}\sqrt{2}\frac{ z }{\F}\right)} \right\vert \tfrac{\n^{1/4} \vert \U{\F}( z )\vert \gmax( z ) }{\sqrt{\t}} \,\, \sup_{\mathclap{\xi_1,\xi_2,\xi_3 \in [0,\frac{1}{\n}] }} \,\, \e^{\frac{\Rone}{2}+\Re\left(\Rtwoz{ z }+\Rthreez{ z } \right)} ,\\
\intertext{with $\Rone$, $\Rtwotz{}$, $\Rthreetz{}$ as in Proposition \ref{prop_asymp_R} and}
\gmax(z)&=\tfrac{ (1-\t^2)^{1/4} \vert \W{\F}(z) \vert } {2 (2\pi^3)^{1/4} } .
\end{align}
\end{prop}
\begin{proof}
Without loss of generality we can assume that $\epsilon<\delta_0$. Since we have taken the supremum over $\xi_1,\xi_2,\xi_3 \in [0,\frac{1}{\n}]$ in the definition of $g_1$ and $g_2$ it is obvious from \eqref{eq_op1_max}, \eqref{eq_op2_max} and Lemma \ref{lemma_fmax} that for all $z\in \partial \Adelta$ and $\n\ge 2$ 
\begin{align}
\vert \op{\n}(z)\vert &\le g_1(\n,\t, z) \e^{\frac{\n}{2} \fmax(z)}\le g_1(\n,\t,\zmax) \e^{\frac{\n}{2} \fmax(\zmax)}, \\
\intertext{and} 
\vert \op{\n-1}(z)\vert&\le g_2(\n,\t,z) \e^{\frac{\n}{2} \fmax(z)}\le g_2(\n,\t,\zmax) \e^{\frac{\n}{2} \fmax(\zmax)},
\end{align}
where $\fmax$ is as in Lemma \ref{lemma_fmax}.
Because $\op{\n}$ and $\op{\n-1}$ are polynomials, we can use the maximum principle which tells us that 
\begin{align}
\sup_{z\in \Adelta^C} \vert \op{\n}(z)\vert&= \max_{z\in \partial \Adelta} \vert \op{\n}(z)\vert, \\
\intertext{and}
\sup_{z\in \Adelta^C} \vert \op{\n-1}(z)\vert&= \max_{z\in \partial \Adelta} \vert \op{\n-1}(z)\vert. 
\end{align}
Therefore we are only left to prove that there exists $\delta>0$ so that $\fmax(\zmax)<2+\log \t+\epsilon$ for all $\t \in [\delta_1,1-\delta_0]$ and $\fmax(\zmax)<-\delta_0$ for all $\t \in (0,\e^{-2}-\delta_0]$. 

We have defined $\U{\F}$ only on $\C\setminus [-\F,\F]$. But there is a continuous continuation of $\fmax$ on all of $\C$, which we will call $\fmaxc$. We see that $\fmaxc( \F)=2+\log \t$. Then there is a neighborhood of $\F$ so that $\vert \fmaxc(z)-\fmaxc(\F) \vert<\epsilon$ and therefore there exists $\delta=\delta(\t)>0$ so that $\fmaxc(z)<2+\log \t+\epsilon$ for all $z\in B_{2\delta}(\F)$. $\delta$ can still be depending on $\t$, but since $\fmaxc$ is also continuous in $\t$ it follows that $\delta$ is depending continuously on $\t$. Therefore
\begin{align}
\min_{\delta_1\le t \le 1-\delta_0} \delta(\t)
\end{align}
is strictly positive. With this minimum we have found a fixed $\delta$ so that $\fmax(\F+(1\pm i)\delta)=\fmax(-\F-(1\pm i)\delta)<2+\log \t+\epsilon$. 

From this follows further for $\delta_1<\t<\e^{-2}-\delta_0$ that $\fmax(\F+(1+i)\delta)<2+\log(\e^{-2}-\delta_0)+\epsilon\le -\delta_0 \e^2+\epsilon < -\delta_0$ if $\epsilon<\delta_0$. So the only thing left to prove is that in the limit $\t$ going to zero there is still $\delta>0$ so that $\fmax(\F+(1\pm i)\delta)<-\delta_0$. According to Lemma \ref{lemma_limit_U_W} 
\begin{align}
\lim_{\t\rightarrow 0} \U{\F}(\F+(1+i)\delta)&=\lim_{\F\rightarrow 0} \U{\F}(\F+(1+i)\delta)=0,\\
\intertext{and}
\lim_{\t\rightarrow 0} \left\vert\frac{\U{\F}(\F+(1+i)\delta)}{\sqrt{\t}}\right\vert&=\lim_{\F\rightarrow 0} \left\vert\frac{2\U{\F}((1+i)\delta)}{\F}\right\vert=\frac{1}{\vert 1+i\vert \delta}.
\intertext{Thus}
\lim_{\F\rightarrow 0} \fmax (\F+(1+i)\delta)&=1+\log (2\delta^2)<-1<-\delta_0\quad \text{if } \delta<\tfrac{1}{\sqrt{2}\e}.
\end{align}
\end{proof}

\begin{prop}\label{prop_estim_max_princ}
Let $0<\delta_0<\e^{-2}$. Then there exists $\n_0=\n_0(\delta_0)$, $\delta=\delta(\delta_0)>0$, $\epsilon=\epsilon(\delta_0)>0$, $C=C(\delta_0)>0$ and $\lambda=\lambda(\delta_0)>0$ so that for all $\n > \n_0$
\begin{align}
\sup_{\t \in(0,1-\delta_0)}\sup_{x+i y\in \Xdelta{C}{\delta}{\F,\lambda}} \left\vert \frac{\dd \rhon}{\dd x}(x,y) \right\vert&< C\sqrt{\n} \e^{-\n \epsilon},\\
\intertext{and}
\sup_{\t \in(0,1-\delta_0)}\sup_{x+i y\in \Xdelta{C}{\delta}{\F,\lambda}} \left\vert \frac{\dd \rhon}{\dd y}(x,y) \right\vert&< C\sqrt{\n} \e^{-\n \epsilon}. 
\end{align}
Further we can choose $\lambda=\F$ if $\t<\e^{-2}-\delta_0$.
\end{prop}
\begin{remark}The proposition gives an uniform estimation for $\left\vert \tfrac{\dd \rhon}{\dd x}(x,y) \right\vert$ and $\left\vert \tfrac{\dd \rhon}{\dd y}(x,y) \right\vert$ on a finite neighborhood of $-\F$ and $\F$ and even on all of $\Adelta^C$ if $\t<\e^{-2}$.
\end{remark}
\begin{proof}
Let $z=x+i y$, $\n>1$, $\epsilon= \tfrac{\delta_0^3}{12}$ and $\delta_1=\e^{-2}-\delta_0>0$. 

From identities \eqref{eq_id1} and \eqref{eq_id2} evaluated at $w=\cc{z}$ we see that for all $\delta>0$ and $ z\in \Adelta^C $
\begin{align}
\left\vert \tfrac{\dd \rhon}{\dd x}(z) \right\vert &\le 2\sqrt{\tfrac{1-\t}{1+\t}} \e^{\n\left(-\vert z\vert^2 +\t \Re(z^2) \right)} \sup_{z,w\in \Adelta^C} \left\vert \op{\n}(w)\op{\n-1}(z)\right\vert, \label{eq_estim_op1}
\intertext{and}
\left\vert \tfrac{\dd \rhon}{\dd y}(z) \right\vert &\le 2\sqrt{\tfrac{1+\t}{1-\t}} \e^{\n\left(-\vert z\vert^2 +\t \Re(z^2) \right)} \sup_{z,w\in \Adelta^C} \left\vert \op{\n}(w)\op{\n-1}(z)\right\vert. \label{eq_estim_op2}
\end{align}
Using Proposition \ref{prop_using_maximum_principle} we get $\delta>0$ so that for all $\t\in (\delta_1,1-\delta_0)$ and $z\in \Adelta^C$
\begin{align}
\left\vert \tfrac{\dd \rhon}{\dd x}(z) \right\vert &< 2\sqrt{\tfrac{1-\t}{1+\t}} g_1(\n,\t,\zmax) g_2(\n,\t,\zmax) \e^{\n\left(-\vert z\vert^2 +\t \Re(z^2) + 2 +\log \t+\epsilon\right)}, \label{eq_proof_estim_mp1} \\
\left\vert \tfrac{\dd \rhon}{\dd y}(z) \right\vert &< 2\sqrt{\tfrac{1+\t}{1-\t}} g_1(\n,\t,\zmax) g_2(\n,\t,\zmax) \e^{\n\left(-\vert z\vert^2 +\t \Re(z^2) + 2 +\log \t+\epsilon\right)}, \label{eq_proof_estim_mp2} \\
\intertext{and that for all $\t\in (0,\e^{-2}-\delta_0]$ and $z\in \Adelta^C$}
\left\vert \tfrac{\dd \rhon}{\dd x}(z) \right\vert &< 2\sqrt{\tfrac{1-\t}{1+\t}} g_1(\n,\t,\zmax) g_2(\n,\t,\zmax) \e^{\n\left(-\vert z\vert^2 +\t \Re(z^2) - \delta_0 \right)}, \label{eq_proof_estim_mp3} \\
\intertext{and}
\left\vert \tfrac{\dd \rhon}{\dd y}(z) \right\vert &< 2\sqrt{\tfrac{1+\t}{1-\t}} g_1(\n,\t,\zmax) g_2(\n,\t,\zmax) \e^{\n\left(-\vert z\vert^2 +\t \Re(z^2) - \delta_0 \right)}, \label{eq_proof_estim_mp4} 
\end{align}
where $\zmax$, $g_1$ and $g_2$ are as in Proposition \ref{prop_using_maximum_principle}. Since $\zmax\in\partial\Adelta\subset\Adelta$ we can find an estimation for
\begin{align}
2\sqrt{\tfrac{1\mp\t}{1\pm\t}} g_1(\n,\t,\zmax) g_2(\n,\t,\zmax)&=\sqrt{\tfrac{1\mp\t}{1\pm\t}} \tfrac{\sqrt{\n}\left\vert\U{\F}(\zmax)\right\vert \left\vert \W{\F}(\zmax)\right\vert^2  }{\sqrt{2\pi^3}\F} \left\vert R_{\ve{\xi},\n}{(\t,\zmax)}-1 \right\vert \nn \\
&\le \sqrt{\tfrac{\n}{ \pi^3 \delta_0}} \sup_{z\in\Adelta} \left( \tfrac{\vert\U{\F}(z)\vert}{\F} \vert \W{\F}(z)\vert^2 \left\vert\Rtz-1\right\vert\right) \nn \\
&< \tfrac{2\sqrt{\n} \left( 1+\sqrt{1+\frac{4}{\delta_0 \delta^2}} \right) }{\sqrt{\pi^3 \delta_0}\delta}=\sqrt{\n} C, \label{eq_proof_estim_mp4b}
\end{align}
for all $\n> \n_0$ and $\n_0\in\N$ big enough and with $R_{\ve{\xi},\n}$ as in Proposition \ref{prop_asymp_R}. In the last step we have used that according to Proposition \ref{prop_estim_R} there exists $\n_0\in \N$ so that 
\begin{align}
\sup_{\substack{0<\t<1-\delta_0 \\\n>\n_0 }} \sup_{\substack{z\in \Adelta \\ \xi_1,\xi_2,\xi_3\in[0,\frac{1}{\n}] } } \left\vert \Rtz \right\vert <1. 
\end{align}
We have further used that by Lemma \ref{lemma_estim_U_T}, \ref{lemma_estim_W} and \ref{lemma_estim_b} for all $\t\in (0,1-\delta_0)$
\begin{align}
\sup_{z\in \Adelta}\frac{\vert \U{\F}(z) \vert}{\F}<\frac{1}{2\delta}, \qquad \text{and}\qquad \sup_{z\in \Adelta} \vert \W{\F}(z) \vert^2< 2\left(1+\sqrt{1+\frac{\F_0^2}{\delta^2}} \right), \nn
\end{align}
where
\begin{align}
\F_0=\sqrt{\frac{4(1-\delta_0)}{1-(1-\delta_0)^2}}< \frac{2}{\sqrt{\delta_0}}.
\end{align}
So we have found a suitable constant $C$ and $\n_0\in\N$, only depending on $\delta_0$. 

Since $-\vert z\vert^2 +\t \Re(z^2) \le 0$ for all $z\in \C$ and $\t \in (0,1)$ it follows directly from \eqref{eq_proof_estim_mp3} and \eqref{eq_proof_estim_mp4} that for all $\n> \n_0$
\begin{align}
\sup_{\t \in(0,\e^{-2}-\delta_0]}\sup_{z\in \Adelta^C} \left\vert \frac{\dd \rhon}{\dd x}(x,y) \right\vert& \le C\sqrt{\n} \e^{-\n \delta_0} < C\sqrt{\n} \e^{-\n \epsilon} ,\\
\intertext{and}
\sup_{\t \in(0,\e^{-2}-\delta_0]}\sup_{z\in \Adelta^C} \left\vert \frac{\dd \rhon}{\dd y}(x,y) \right\vert&< C\sqrt{\n} \e^{-\n \epsilon}.
\end{align}
This proves the claim for $\t \in(0,\e^{-2}-\delta_0]$.

We are only left to consider $\t>\delta_1$ and have to show that there exists $\lambda>0$ so that $-\vert z\vert^2 +\t \Re(z^2) + 2 +\log \t+\epsilon<-\epsilon$ for all $z\in \Xdelta{C}{\delta}{\F,\lambda}$. Since $-\vert z\vert^2+\t\Re(z^2)$ reaches its maximum on $\Xdelta{C}{\delta}{\F,\lambda}$ at $z=\F-\lambda$, we find that for all $\t \in (\delta_1,1-\delta_0)$
\begin{align}
\sup_{z\in \Xdelta{C}{\delta}{\F,\lambda}} \left(-\vert z\vert^2 +\t \Re(z^2) + 2 +\log \t+\epsilon \right)=-(\F-\lambda)^2(1-\t) + 2 +\log \t+\epsilon, \label{eq_proof_estim_mp5}
\end{align}
and
\begin{align}
-(\F-\lambda)^2(1-\t) + 2 +\log \t+\epsilon<-\epsilon \iff (\F-\lambda)^2>\tfrac{2+2\epsilon+\log \t}{1-\t}, \label{eq_proof_estim_mp5b}
\end{align}
which is fulfilled trivially if $2+2\epsilon+\log \t\le 0$. Thus we can assume that $2+2\epsilon+\log \t > 0$. Then \eqref{eq_proof_estim_mp5b} is equivalent to 
\begin{align}
\lambda<\tfrac{ 2\sqrt{\t}-\sqrt{(1+\t)(2+2\epsilon+\log\t)} }{\sqrt{1-\t^2}}. \label{eq_proof_estim_mp6}
\end{align}
So we have to show that the infimum over $\t\in (\delta_1,1-\delta_0)$ of the right-hand side of \eqref{eq_proof_estim_mp6} is strictly positive, which is equivalent to 
\begin{align}
\inf_{\t\in (\delta_1,1-\delta_0)} \tilde{f}(\t)>0,
\end{align}
where $\tilde{f}(\t)=4\t-(1+\t)(2+2\epsilon+\log\t)$. But 
\begin{align}
\tfrac{\partial \tilde{f}(\t)}{\partial \t}=1-2\epsilon-\tfrac{1}{\t}-\log \t<0,
\end{align} 
since the maximum of $1-\tfrac{1}{\t}-\log \t$ at $\t=1$ is zero. Therefore 
\begin{align}
\inf_{\t\in (\delta_1,1-\delta_0)} \tilde{f}(\t)=\tilde{\f}(1-\delta_0)=-2\delta_0-2\epsilon-(2-\delta_0)\log(1-\delta_0).
\end{align}
Since 
\begin{align}
-(2-\delta_0)\log(1-\delta_0)=(2-\delta_0)\sum_{k=1}^\infty \tfrac{\delta_0^k}{k}=2\delta_0+\sum_{k=2}^\infty \tfrac{k-2}{(k-1)k}\delta_0^k> 2\delta_0+\tfrac{\delta_0^3}{6}+\tfrac{\delta_0^4}{6},
\end{align}
we see that 
\begin{align}
\inf_{\t\in (\delta_1,1-\delta_0)} \tilde{f}(\t) > -2\epsilon+\tfrac{\delta_0^3}{6}+\tfrac{\delta_0^4}{6}\ge \tfrac{\delta_0^4}{6}>0 \quad\text{if }\epsilon\le\tfrac{\delta_0^3}{12}.
\end{align}
Therefore we have found $\lambda>0$ so that \eqref{eq_proof_estim_mp5b} is fulfilled for all $\t\in (\delta_1,1-\delta_0)$. Then it follows from \eqref{eq_proof_estim_mp1} and \eqref{eq_proof_estim_mp2} that for all $\n>\n_0$
\begin{align}
\sup_{\t \in(\delta_1,1-\delta_0)}\sup_{z\in \Xdelta{C}{\delta}{\F,\lambda}} \left\vert \tfrac{\dd \rhon}{\dd x}(z) \right\vert &< C \sqrt{\n} \e^{-\n\epsilon}, \\
\intertext{and}
\sup_{\t \in(\delta_1,1-\delta_0)}\sup_{z\in \Xdelta{C}{\delta}{\F,\lambda}} \left\vert \tfrac{\dd \rhon}{\dd y}(z) \right\vert &< C \sqrt{\n} \e^{-\n\epsilon},
\end{align}
since $\Xdelta{C}{\delta}{\F,\lambda} \subset \Adelta^C$.
\end{proof}

\subsubsection{Plancherel-Rotach Asymptotics on $\Bdelta$}
So far we have found a uniform estimation for the derivative of $\rhon$ except on $\Bdelta\setminus\Xdelta{C}{\delta}{\F,\lambda}$ for $\delta>0$ and some $\lambda<\F$ --- and for small $\t$ even on all of $\C$. Thus in the following we will only concentrate on $z\in \Bdelta\setminus\Xdelta{C}{\delta}{\F,\lambda}$ and we can neglect small $\t$.

\begin{defin}
For $\delta>0$ and $0\le \lambda< \F $ we define
\begin{align}
\Xdelta{B}{\delta}{\F,\lambda}&=\{z\in \Bdelta \vert \vert \Re\, z\vert\le \F-\lambda \}.
\end{align}
\end{defin}

In the following we will review Proposition \ref{prop_asym_B} and give a more precise version.
\begin{prop}\label{prop_pr_B}
Let $\epsilon>0$, $\lambda>0$, $0<\delta_0<1$ and $0<\delta_1<1-\delta_0$. Then there exists $\delta=\delta(\epsilon,\delta_0,\delta_1,\lambda)>0$ and $\n_0=\n_0(\delta_0,\delta_1,\lambda)\in\N$ so that for all $\n>\n_0$, $\t\in[\delta_1,1-\delta_0]$ and $z\in \Xdelta{B}{\delta}{\F,\lambda}$
\begin{align}
\vert \op{\n}(z)\vert&< g_1^r(\n,\t) \e^{\n \left( \frac{\Re(z^2)}{\F^2} + \frac{\log \t+\epsilon}{2} \right) }, \label{eq_pr_B1} \\
\shortintertext{and}
\vert \op{\n-1}(z)\vert&<  g_2^r(\n,\t) \e^{\n \left( \frac{\Re(z^2)}{\F^2} + \frac{\log \t+\epsilon}{2} \right) }, \label{eq_pr_B2} 
\intertext{where}
g_1^r(\n,\t)&= \Bigl( 2\sup_{z\in \Xdelta{B}{\delta}{\F,\lambda} }\left\vert \tfrac{\F+z}{\F-z} \right\vert^{1/4} +1 \Bigr)  g_0^r(\n,\t) ,\\
\shortintertext{and}
g_2^r(\n,\t)&= \Bigl( 2\sup_{z\in \Xdelta{B}{\delta}{\F,\lambda} }\left\vert \tfrac{\F\sqrt{1-\frac{1}{\n}}+z}{\F\sqrt{1-\frac{1}{\n}}-z} \right\vert^{1/4} +1 \Bigr) g_0^r(\n,\t)\left(\tfrac{\n-1}{\n} \right)^{\frac{\n-1}{2}} \e^{\frac{1-\log 2}{2}},\\
\shortintertext{with}
g_0^r(\n,\t)&=\tfrac{\n^{1/4}(1-\t^2)^{1/4}}{\sqrt{\pi}(2\pi)^{1/4}} \sup_{\xi_1\in [0,\frac{1}{\n}]} \e^{\frac{\Rone}{2}},
\end{align}
and $\Rone$ as in Proposition \ref{prop_asymp_R}.
\end{prop}
\begin{proof}
Let $\F_0=\tfrac{2}{\sqrt{\delta_0}}$ and $\F_1=\sqrt{2\delta_1}$. Then $\F_1<\F<\F_0$ for all $\t \in [\delta_1,1-\delta_0]$. Without loss of generality we can assume that $\lambda<\F_0$.

To shorten the terms of the orthonormal polynomials we define 
\begin{align}
\gamma_1&=\tfrac{z\sqrt{\F^2-z^2}}{\F^2}+\arcsin\left(\tfrac{z}{\F}\right)-\pi, \\
\gamma_2&=\tfrac{z\sqrt{\F^2\left(1-\frac{1}{\n}\right)-z^2}}{\F^2}+\left(1-\tfrac{1}{\n}\right)\arcsin\left(\sqrt{\tfrac{\n}{\n-1}}\tfrac{z}{\F}\right)-\pi, \\
W^{\pm}_1&=\left(\tfrac{\F\pm z}{\F\mp z}\right)^{1/4}, \qquad W^{\pm}_2=\left(\tfrac{\F\sqrt{1-\frac{1}{\n}}\pm z}{\F\sqrt{1-\frac{1}{\n}}\mp z}\right)^{1/4}, \nn \\
R^r_{1}(\n,t,z)&=\tfrac{Q_\n\left(\sqrt{2}\frac{z}{\F} \right) - \pi^r_\n \left(\sqrt{2}\frac{z}{\F} \right)}{\e^{\n \left(\frac{z^2}{\F^2} -\frac{1+\log 2}{2}\right)}}, \\
\intertext{and}
R^r_{2}(\n,t,z)&=\tfrac{Q_{\n-1}\left(\sqrt{\frac{\n}{\n-1}}\sqrt{2}\frac{z}{\F} \right) - \pi^r_{\n-1} \left(\sqrt{\frac{\n}{\n-1}}\sqrt{2}\frac{z}{\F} \right)}{\e^{\n \left(\frac{z^2}{\F^2} -\frac{1+\log 2}{2}\right)+\frac{1+\log 2}{2}}}.
\end{align}

Using \eqref{eq_op_n}, \eqref{eq_op_n-1}, \eqref{eq_pir} and \eqref{eq_proof_R1} we can now write an estimation of the absolute value of the orthonormal polynomials of order $\n$ and $\n-1$ on $\Xdelta{B}{\delta}{\F,\lambda}$ as
\begin{align}
\vert \op{\n}(z)\vert &\le\left\vert W_1^+ \cos\left(\n \gamma_1+\tfrac{\pi}{4}\right)+ W_1^- \cos\left(\n \gamma_1-\tfrac{\pi}{4}\right) +R^r_{1}(\n,t,z)  \right\vert \nn \\
&\quad \cdot g_0^r(\n,t) \e^{\n \left( \frac{\Re(z^2)}{\F^2} + \frac{\log \t}{2} \right) }
, \label{eq_op1_r}
\intertext{and}
\vert\op{\n-1}(z)\vert&\le\left\vert W_2^+ \cos\left(\n \gamma_2+\tfrac{\pi}{4}\right)+ W_2^- \cos\left(\n \gamma_2-\tfrac{\pi}{4}\right) +R^r_{2}(\n,t,z)  \right\vert \nn\\
&\quad \cdot g_0^r(\n,t) \e^{\n \left( \frac{\Re(z^2)}{\F^2} + \frac{\log \t}{2} \right) } \e^{\frac{1}{2}-\frac{\log \t}{2}}  \left( \tfrac{\n-1}{\n} \right)^{\frac{\n-1}{2}}. \label{eq_op2_r}
\end{align}

From Section \ref{subsection_pr_asym_inside} we know that there exists $\delta_2>0$ so that the Plancherel-Rotach asymptotics \eqref{eq_pir_o} holds uniformly on $\Xdelta{B}{\delta_2}{\sqrt{2}}$. 
We can assume that $\delta_2< \tfrac{\sqrt{2}\lambda}{\F_0} \tfrac{\F_0-\lambda}{2\F_0-\lambda}$, which is strictly positive. If $\delta\le \tfrac{\F_1\delta_2}{2}$, it follows from the Plancherel-Rotach asymptotics that there exists $\n_0=\n_0(\delta_0,\delta_1,\lambda)>\tfrac{2\F_0}{\lambda}$ so that for all $z \in \Xdelta{B}{\delta}{\F,\lambda}$, $\t\in[\delta_1,1-\delta_0 ]$ and $\n>\n_0$
\begin{align}\label{eq_pr_estim}
R^r_{1}(\n,t,z) &<1, \qquad \text{and} \qquad R^r_{2}(\n,t,z) <1,
\end{align}
since then $\sqrt{2}\tfrac{z}{\F}$ and $ \sqrt{\tfrac{\n}{\n-1}}\sqrt{2}\tfrac{z}{\F} $ lie in $\Xdelta{B}{\delta_2}{\sqrt{2}}$. The $\n_0$ we so have found is the required $\n_0$ of the proposition and we assume that $\n>\n_0$ in the following. 

Let $\gamma$ be either $\gamma_1$ or $\gamma_2$. Then we can estimate
\begin{align}
\left\vert \cos\left(\n \gamma \pm \tfrac{\pi}{4}\right) \right\vert &\le \left\vert \cos\left(\n \Re \gamma \pm \tfrac{\pi}{4}\right) \cos\left(\n i \Im\gamma \right)  \right\vert + \left\vert \sin\left(\n \Re\gamma \pm \tfrac{\pi}{4}\right) \sin\left(\n i \Im\gamma \right)  \right\vert \nn \\
&\le \tfrac{1}{2}\left\vert e^{-\n \Im\gamma} + e^{\n \Im\gamma}\right\vert+\tfrac{1}{2}\left\vert e^{-\n \Im\gamma} - e^{\n \Im\gamma}\right\vert\le \e^{\n\vert\Im \gamma\vert}.
\end{align}
Since
\begin{align}
z\sqrt{\F^2-z^2}\qquad \text{and}\qquad  \arcsin\left(\tfrac{z}{\F}\right) \nn
\end{align}
are uniformly continuous on any compact subset of $\{z\in\C \vert \vert \Re{\, z}\vert < \F \}$, we see that for any $\delta>0$, $\gamma_1$ and $\gamma_2$ are uniformly continuous on $\Xdelta{B}{\delta}{\F,\lambda}$ for all $\t \in [\delta_1,1-\delta_0 ]$. Further $\gamma_1$ and $\gamma_2$ are uniformly continuous in $\t$ on $[\delta_1,1-\delta_0 ]$ and --- if we consider $\alpha=\tfrac{1}{\n}$ to be a real valued variable --- also in $\alpha$ on $[0,\tfrac{1}{\n_0}]$. Since $\gamma_1$ and $\gamma_2$ are real if $z\in \Xdelta{B}{\delta}{\F,\lambda}\cap \R$ it follows from uniform continuity that for all $\epsilon>0$ there exists $\delta=\delta(\epsilon,\delta_0,\delta_1,\lambda)>0$ so that
\begin{align}
\vert \Im \gamma_1 \vert < \tfrac{\epsilon}{2}, \qquad \text{and} \qquad \vert \Im \gamma_2 \vert < \tfrac{\epsilon}{2},
\end{align}
uniformly for all $z\in \Xdelta{B}{\delta}{\F,\lambda}$, $\t \in [\delta_1,1-\delta_0 ]$ and $\n>\n_0$. This will be the $\delta$ we have claimed in the proposition. So we can estimate 
\begin{align}
\max_{\gamma\in\{\gamma_1,\gamma_2 \} } \left\vert \cos\left(\n \gamma \pm \tfrac{\pi}{4}\right) \right\vert< \e^{\n \frac{\epsilon}{2}}.
\end{align}
Using this and \eqref{eq_pr_estim} for the orthonormal polynomials, we get from \eqref{eq_op1_r} and \eqref{eq_op2_r}
\begin{align}
\vert \op{\n}(z)\vert  & \le \left( 2\sup_{z\in \Xdelta{B}{\delta}{\F,\lambda} } \left\vert W_1^+\right\vert +1 \right) \e^{\n \frac{\epsilon}{2}}  g_0^r(\n,t) \e^{\n \left( \frac{\Re(z^2)}{\F^2} + \frac{\log \t}{2} \right) } , \\
\intertext{and}
\vert \op{\n-1}(z)\vert  & \le \left( 2\sup_{z\in \Xdelta{B}{\delta}{\F,\lambda} } \left\vert W_2^+\right\vert +1 \right) \e^{\n \frac{\epsilon}{2}}  g_0^r(\n,t)  \e^{\frac{1}{2}-\frac{\log \t}{2}}  \left( \tfrac{\n-1}{\n} \right)^{\frac{\n-1}{2}} \e^{\n \left( \frac{\Re(z^2)}{\F^2} + \frac{\log \t}{2} \right) } , 
\end{align}
for all $z\in \Xdelta{B}{\delta}{\F,\lambda}$, $\t \in [\delta_1,1-\delta_0 ]$ and $\n>\n_0$.
\end{proof}

\begin{prop}\label{prop_estim_pr_B}
Let $0<\delta_0<1$, $0<\delta_1<1-\delta_0$ and $\lambda>0$. Then there exists $\n_0=\n_0(\delta_0,\delta_1,\lambda)$, $\delta=\delta(\delta_0,\delta_1,\lambda)>0$, $\epsilon=\epsilon(\delta_0,\delta_1,\lambda)>0$, $C=C(\delta_0,\lambda)>0$ so that for all $\n > \n_0$
\begin{align}
\sup_{\t \in(\delta_1,1-\delta_0)}\sup_{x+i y\in \Xdelta{B}{\delta}{\F,\lambda}} \left\vert \frac{\dd \rhon}{\dd x}(x,y) \right\vert&< C\sqrt{\n} \e^{-\n \epsilon},\\
\intertext{and}
\sup_{\t \in(\delta_1,1-\delta_0)}\sup_{x+i y\in \Xdelta{B}{\delta}{\F,\lambda}} \left\vert \frac{\dd \rhon}{\dd y}(x,y) \right\vert&< C\sqrt{\n} \e^{-\n \epsilon}. 
\end{align}
\end{prop}
\begin{proof}
Let $z=x+ i y$, $\epsilon=-\tfrac{\delta_0}{2-\delta_0}-\log(1-\delta_0)>0$ and $\F_0=\tfrac{2}{\sqrt{\delta_0}}$ so that $\F_0>\F$ for all $\t \in [\delta_1,1-\delta_0]$.

As in \eqref{eq_estim_op1} and \eqref{eq_estim_op2} we get from identities \eqref{eq_id1} and \eqref{eq_id2} evaluated at $w=\cc{z}$ the following estimations for the derivatives of $\rhon$
\begin{align}
\left\vert \tfrac{\dd \rhon}{\dd x}(z) \right\vert &\le 2\sqrt{\tfrac{1-\t}{1+\t}} \e^{\n\left(-\vert z\vert^2 +\t \Re(z^2) \right)} \left\vert \op{\n}(z)\op{\n-1}(z)\right\vert,
\intertext{and}
\left\vert \tfrac{\dd \rhon}{\dd y}(z) \right\vert &\le 2\sqrt{\tfrac{1+\t}{1-\t}} \e^{\n\left(-\vert z\vert^2 +\t \Re(z^2) \right)} \left\vert \op{\n}(z)\op{\n-1}(z)\right\vert,
\end{align}
for all $z\in\C$. With the help of Proposition \ref{prop_pr_B}, we get $\n_0\in\N$ with $\n_0>\tfrac{2\F_0}{\lambda}$ and $\delta>0$ so that for all $\n>\n_0$, $\t\in (\delta_1,1-\delta_0)$ and $z\in \Xdelta{B}{\delta}{\F,\lambda}$
\begin{align}
\left\vert \tfrac{\dd \rhon}{\dd x}(z) \right\vert &< 2\sqrt{\tfrac{1-\t}{1+\t}} g_1^r(\n,\t) g_2^r(\n,\t) \e^{\n\left(-\vert z\vert^2 +\t \Re(z^2) + \frac{2\Re(z^2)}{\F^2} +\log \t+\epsilon\right)}, \label{eq_proof_estim_prB1} \\
\intertext{and}
\left\vert \tfrac{\dd \rhon}{\dd y}(z) \right\vert &< 2\sqrt{\tfrac{1+\t}{1-\t}} g_1^r(\n,\t) g_2^r(\n,\t) \e^{\n\left(-\vert z\vert^2 +\t \Re(z^2) + \frac{2\Re(z^2)}{\F^2} +\log \t+\epsilon\right)}, \label{eq_proof_estim_prB2} 
\end{align}
where $g_1^r$ and $g_2^r$ are as in Proposition \ref{prop_pr_B}. For the factors in front of the exponential term we find
\begin{align}
2\sqrt{\tfrac{1\mp\t}{1\pm\t}} g_1^r(\n,\t) g_2^r(\n,\t)=&(1\mp \t)\tfrac{\sqrt{2\n}  \e^{\frac{1-\log 2}{2}}}{\sqrt{\pi^3}} \left(\tfrac{\n-1}{\n} \right)^{\frac{\n-1}{2}}  \sup_{\xi_1\in [0,\frac{1}{\n}]} \e^{\Rone} \nn \\
&\quad \cdot \biggl( 2\sup_{z\in \Xdelta{B}{\delta}{\F,\lambda} }\left\vert \tfrac{\F+z}{\F-z} \right\vert^{1/4} +1 \biggr)  \biggl( 2\sup_{z\in \Xdelta{B}{\delta}{\F,\lambda} }\left\vert \tfrac{\F\sqrt{1-\frac{1}{\n}}+z}{\F\sqrt{1-\frac{1}{\n}}-z} \right\vert^{1/4} +1 \biggr), \label{eq_proof_estim_pr_B1}
\end{align}
where $\Rone$ is as in Proposition \ref{prop_asymp_R} with $\e^{\Rone}<1$. If $\n>\n_0>\tfrac{2\F_0}{\lambda}$, then for all $\t\in (\delta_1,1-\delta_0)$
\begin{align}
\inf_{z\in \Xdelta{B}{\delta}{\F,\lambda}} \left\vert \pm \F \sqrt{1-\tfrac{1}{\n}}- z \right\vert\ge \tfrac{\lambda}{2},
\end{align}
and so
\begin{align}
\sup_{z\in \Xdelta{B}{\delta}{\F,\lambda} }\left\vert \tfrac{\F\sqrt{1-\frac{1}{\n}}+z}{\F\sqrt{1-\frac{1}{\n}}-z} \right\vert &\le \tfrac{4\F_0}{\lambda}, \qquad \text{and}\qquad \sup_{z\in \Xdelta{B}{\delta}{\F,\lambda} }\left\vert \tfrac{\F+z}{\F-z} \right\vert \le \tfrac{2\F_0}{\lambda}. \nn
\end{align}
So we find an estimation for \eqref{eq_proof_estim_pr_B1} for all $\n>\n_0$ and $\t\in (\delta_1,1-\delta_0)$,
\begin{align}
2\sqrt{\tfrac{1\mp\t}{1\pm\t}} g_1^r(\n,\t) g_2^r(\n,\t)&\le \sqrt{\n} \tfrac{2\sqrt{2}  \e^{\frac{1-\log 2}{2}}}{\sqrt{\pi^3}} \left( 2\left(\tfrac{2\F_0}{\lambda} \right)^{1/4} +1 \right)  \left( 2  \left(\tfrac{2\F_0}{\lambda} \right)^{1/4} +1 \right)=\sqrt{\n} C, \label{eq_proof_estim_pr_B2}
\end{align}
with a finite $C\in\R$.

Now we look at the exponential term $\e^{\n\left(\f^r(z)+\epsilon\right)}$ with
\begin{align}
\f^r(z)&=-\vert z\vert^2 +\t \Re(z^2) + \frac{2\Re(z^2)}{\F^2} +\log \t.
\intertext{With}
\f^r(x+i y)&=\tfrac{(1-\t)^2 x^2}{2\t} - \tfrac{(1+\t)^2 y^2}{2\t}+\log \t
\shortintertext{we see that}
\sup_{z\in \Xdelta{B}{\delta}{\F,\lambda} }\f^r(z)&=\f^r(\F-\lambda)< \f^r(\F) =\tfrac{2(1-\t)}{1+\t}+\log \t.
\end{align}
$\f^r(\F)$ is strictly monotonic increasing in $\t$ for all $\t<1$ and reaches zero for $\t=1$, which we have already discussed in \eqref{eq_fr_estim}. Therefore $\f^r(\F)<0$ and we have defined $\epsilon>0$ so that $\f^r(\F)<-2\epsilon$ for all $\t\in [\delta_1,1-\delta_0]$. And so we finally find that 
\begin{align}
\sup_{\t \in [\delta_1,1-\delta_0]}\sup_{z\in \Xdelta{B}{\delta}{\F,\lambda}} \left\vert \tfrac{\dd \rhon}{\dd x}(z) \right\vert &< \sqrt{\n}C \e^{-\n\epsilon},\\
\shortintertext{and}
\sup_{\t \in [\delta_1,1-\delta_0]}\sup_{z\in \Xdelta{B}{\delta}{\F,\lambda}} \left\vert \tfrac{\dd \rhon}{\dd y}(z) \right\vert &< \sqrt{\n}C \e^{-\n\epsilon}.
\end{align}
\end{proof}

\subsubsection{Uniform Estimation of Derivatives of Density}

\begin{defin}
For $\delta<\mis$ we define the filled ellipses with major semi-axis $\mas\pm\delta$ and minor semi-axis $\mis\pm\delta$
\begin{align}
\Edelta{\pm,\delta}= \{z\in \C \vert \frac{(\Re\, z)^2}{(\mas\pm \delta )^2}+\frac{(\Im\, z)^2}{(\mis\pm \delta )^2}<1\}.  
\end{align}
Their boundaries we will call $\partial \Edelta{\pm,\delta}$ and the space between $\partial \Edelta{-,\delta}$ and $\partial \Edelta{+,\delta}$
\begin{align}
\Edelta{\delta}= \{z\in \C \vert \frac{(\Re\, z)^2}{(\mas+ \delta )^2}+\frac{(\Im\, z)^2}{(\mis+ \delta )^2}<1 \land \frac{(\Re\, z)^2}{(\mas- \delta )^2}+\frac{(\Im\, z)^2}{(\mis- \delta )^2}>1  \}.  
\end{align}
\end{defin}

\begin{defin}
For $\delta<\mis$ we define the regions
\begin{align}
\Omegadelta{x,+,\delta} &= \{z\in \C \vert \frac{(\Re\, z -\delta)^2}{\mas^2}+\frac{(\Im\, z)^2}{\mis^2}<1\} \cup \{z\in \C \vert \frac{(\Re\, z +\delta)^2}{\mas^2}+\frac{(\Im\, z)^2}{\mis^2}<1\}, \\
\Omegadelta{x,-,\delta} &= \{z\in \C \vert \frac{(\Re\, z -\delta)^2}{\mas^2}+\frac{(\Im\, z)^2}{\mis^2}<1\} \cap \{z\in \C \vert \frac{(\Re\, z +\delta)^2}{\mas^2}+\frac{(\Im\, z)^2}{\mis^2}<1\}, \\
\Omegadelta{y,+,\delta} &= \{z\in \C \vert \frac{(\Re\, z)^2}{\mas^2}+\frac{(\Im\, z-\delta)^2}{\mis^2}<1\} \cup \{z\in \C \vert \frac{(\Re\, z)^2}{\mas^2}+\frac{(\Im\, z +\delta)^2}{\mis^2}<1\}, \\  
\Omegadelta{y,-,\delta} &= \{z\in \C \vert \frac{(\Re\, z)^2}{\mas^2}+\frac{(\Im\, z-\delta)^2}{\mis^2}<1\} \cap \{z\in \C \vert \frac{(\Re\, z )^2}{\mas^2}+\frac{(\Im\, z+\delta)^2}{\mis^2}<1\}.
\end{align}
Their boundaries we will call $\partial \Omegadelta{x,\pm,\delta}$ and $\partial \Omegadelta{y,\pm,\delta}$ and the space between $\partial \Omegadelta{x,-,\delta}$ and $\partial \Omegadelta{x,+,\delta}$
\begin{align}
\Omegadelta{x,\delta}&=\Omegadelta{x,+,\delta}\setminus\left( \Omegadelta{x,-,\delta}\cup \partial \Omegadelta{x,-,\delta} \right),
\intertext{and analogous the space between $\partial \Omegadelta{y,-,\delta}$ and $\partial \Omegadelta{y,+,\delta}$}
\Omegadelta{y,\delta}&=\Omegadelta{y,+,\delta}\setminus\left( \Omegadelta{y,-,\delta}\cup \partial \Omegadelta{x,-,\delta} \right),
\end{align}
with their boundaries $\partial \Omegadelta{x,\delta}$ and $\partial \Omegadelta{y,\delta}$ respectively.
\end{defin}

The regions $\Omegadelta{x,\delta}$ and $\Omegadelta{y,\delta}$ are both a subset of $\Edelta{\delta}$. The minimal distance
\begin{align}
\inf_{0<\t<1-\delta_0} \dist(\pEt\cap \C_{x,\delta},\partial \Omegadelta{x,\delta}\cap \C_{x,\delta} )=\inf_{0<\t<1-\delta_0} \inf_{\substack{z_1\in \pEt\cap \C_{x,\delta} \\ z_2\in \partial\Omegadelta{x,\delta}\cap \C_{x,\delta}  } } \vert z_1-z_2\vert 
\end{align}
is strictly positive for all $0<\delta_0<1$. Analogously $\inf_{0<\t<1-\delta_0} \dist(\pEt\cap \C_{y,\delta},\partial \Omegadelta{y,\delta}\cap \C_{y,\delta} )>0$ for all $0<\delta_0<1$. We can see that $\partial\Omegadelta{x,-,\delta}$ (or $\partial\Omegadelta{x,+,\delta}$) on $\Cp$ is just the arc of the ellipse $\pEt$ moved to the left (or to the right respectively), parallel to the real axis by $\delta$. Analogously $\partial\Omegadelta{y,-,\delta}$ (or $\partial\Omegadelta{y,+,\delta}$) we get by moving $\pEt$ down (or up) parallel to the imaginary axis by $\delta$. So for every point $z_1\in \pEt\setminus\C_{x,2\delta}$ there is a point $z_2\in \Omegadelta{x,\pm,\delta}\setminus\C_{x,\delta}$ with $z_1=z_2\pm \delta$ and for every point $z_1\in \pEt\setminus\C_{y,2\delta}$ there is a point $z_2\in \Omegadelta{y,\pm,\delta}\setminus\C_{y,\delta}$ with $z_1=z_2\pm i \delta$. Our aim now is to find suitable estimates for $\tfrac{\partial \rhon}{\partial x}$ and $\tfrac{\partial \rhon}{\partial y}$ on $\C\setminus (\Omegadelta{x,\delta}\cup \Omegadelta{y,\delta})$, $\Omegadelta{x,\delta}\setminus\C_{x,\delta}$ and $\Omegadelta{y,\delta}\setminus\C_{y,\delta}$ so that we can find estimations for $\rhon$ by integration along a path parallel to the real or imaginary axis.

\begin{prop}\label{prop_estim_f}
Let $\delta>0$, $\delta_0>0$ and $\delta_1>0$. Then there exists $\epsilon=\epsilon(\delta,\delta_0)>0$ so that
\begin{align}
\sup_{\substack{0<\t<1-\delta_0}} \sup_{\substack{z\in \Xdelta{A}{\delta_1}{\F}\setminus \Edelta{\delta} }} \f(z) \le -\epsilon,
\end{align}
where $\f$ is as in Proposition \ref{prop_asym}.
\end{prop}
\begin{proof}
Without loss of generality we assume that $\delta>0$ and $\delta_1>0$ are small enough so that $\Edelta{\delta} \subset \Xdelta{A}{\delta_1}{\F}$.
From Lemma \ref{lemma_df} we see that the maximum of $\f$ on $\Xdelta{A}{\delta_1}{\F}\setminus \Edelta{\delta}$ must lie on $\partial \Edelta{-,\delta}\cup\partial \Edelta{+,\delta}$ and this maximum, which we define as $-\epsilon_\t$, is strictly negative according to Proposition \ref{prop_f_zero}. Since $\f$ is continuous in $z$ on $\Xdelta{A}{\delta_1}{\F}$ and in $\t$ on $(0,1-\delta_0]$ and $\partial \Edelta{-,\delta}\cup\partial \Edelta{+,\delta}$ transforms continuously with $\t$, we see that also $\epsilon_\t$ depends continuously on $\t$ and therefore has a strictly positive minimum on any compact subset of $(0,1-\delta_0]$. So we only have to check that this is still true in the limit $\t\rightarrow 0$. According to Lemma \ref{lemma_lim_f_g}
\begin{align}
\lim_{\t\rightarrow 0} \f(z)=-\vert z\vert^2+\log \vert z \vert^2+1,
\end{align}
which has its maximum at $\vert z\vert=1$, i.e.\ on the circle $\partial E_0$ and is obviously constant and negative on $\partial E_{0,-,\delta}$ and on $\partial E_{0,+,\delta}$. Thus we find that
\begin{align}
\epsilon= \inf_{(0,\t-\delta_0]} \epsilon_\t = -\sup_{(0,\t-\delta_0]} \sup_{\substack{z\in \Xdelta{A}{\delta_1}{\F}\setminus \Edelta{\delta} }} \f(z)>0.
\end{align}
\end{proof}

\begin{prop}\label{prop_unifrom_estimate_drho}
Let $\epsilon>0$, $0<\delta_0<1$ 
and $\delta>0$. Then there exists 
$\n_0=\n_0(\epsilon,\delta,\delta_0) \in \N$ so that for all $\n\ge \n_0$ and $0<\t<1-\delta_0$ 
\begin{align}
\sup_{x+i y\in \C \setminus\Edelta{\delta} } \left \vert \frac{\dd \rhon}{\dd x}(x,y)\right \vert< \epsilon, \qquad \text{and} \qquad \sup_{x+i y\in \C \setminus\Edelta{\delta} } \left \vert \frac{\dd \rhon}{\dd y} (x,y)\right\vert< \epsilon. \label{eq_unifrom_estimate_drho1}
\end{align}

Additionally there exist $C(\delta_0)>0$ and $R(\delta_0)>0$ so that for all $z=x+i y\notin B_R(0)$, $\n\ge \n_0$ and $0<\t<1-\delta_0$ follows that
\begin{align}
\left \vert \frac{\dd \rhon}{\dd x} (x,y)\right\vert< C \sqrt{\n} \e^{-\n \vert z\vert}, \qquad \text{and} \qquad \left \vert \frac{\dd \rhon}{\dd y}(x,y)\right\vert< C \sqrt{\n} \e^{-\n \vert z\vert}. \label{eq_unifrom_estimate_drho2}
\end{align}

Further let $\delta_1>0$ 
and $z_E=x_E+i y_E$. Then there exists $\delta_2<\delta_1$ so that for all $0<\delta<\delta_2$, $\n\ge \n_0$, $0<\t<1-\delta_0$
\begin{align} 
(1-\epsilon) \e^{(1+\epsilon) \n \frac{\partial^2 f}{\partial x^2}(x_E,y_E)\frac{(x-x_E)^2}{2}}&\le \frac{ \frac{\dd \rhon }{\dd x}(x,y_E) }{ \gp(z_E) \sqrt{\frac{\n}{2\pi^3}} }  \le  (1+\epsilon) \e^{(1-\epsilon) \n \frac{\partial^2 f}{\partial x^2}(x_E,y_E)\frac{(x-x_E)^2}{2}},\label{eq_unifrom_estimate_drho3} \\
\intertext{if $x \in (x_E-\delta,x_E+\delta) \cap \C_{x,\delta_1}$ and $z_E\in \pEt\cap \C_{x,\delta_1}$, and}
(1-\epsilon) \e^{(1+\epsilon)\n \frac{\partial^2 f}{\partial y^2}(x_E,y_E)\frac{(y-y_E)^2}{2}}&\le \frac{ \frac{\dd \rhon }{\dd y}(x_E,y) }{-\gm(z_E) \sqrt{\frac{\n}{2\pi^3}} }  \le  (1+\epsilon) \e^{(1-\epsilon) \n \frac{\partial^2 f}{\partial y^2}(x_E,y_E)\frac{(y-y_E)^2}{2}}, \label{eq_unifrom_estimate_drho4} 
\end{align}
if $y \in (y_E-\delta,y_E+\delta) \cap \C_{y,\delta_1}$ and $z_E\in \pEt\cap \C_{y,\delta_1}$.
\end{prop}
\begin{remark}
The proposition claims that the derivatives of $\rhon$ get arbitrarily small for $\n\rightarrow\infty$ on $\C$ except in a neighborhood of $\pEt$, where we have an estimation for $\frac{\dd \rhon }{\dd x}$ or $\frac{\dd \rhon }{\dd y}$. Further, for large $\vert z \vert$, the derivatives are going at least exponentially fast to zero.
\end{remark}
\begin{proof}
Without loss of generality $\delta_0<\e^{-2}$. Let $z=x+i y$. To prove the first claim we are going to combine Proposition \ref{prop_asymp_R}, \ref{prop_estim_R}, \ref{prop_estim_max_princ}, \ref{prop_estim_pr_B} and \ref{prop_estim_f} to find an estimation on $\C\setminus\Edelta{\delta}$. In this proof we are going to call the variables $\delta$, $\n_0$, $\epsilon$, $C$ and $\delta_1$ if they appear in those propositions $\delta_A$, $\n_{0_A}$, $\epsilon_A$, $C_A$ and $\delta_{1_A}$ in Proposition \ref{prop_asymp_R}, \ref{prop_estim_R} and \ref{prop_estim_f} --- $\delta_B$, $\n_{0_B}$, $\epsilon_B$, $C_B$ and $\delta_{1_B}$ in Proposition \ref{prop_estim_pr_B} --- and $\delta_C$, $\n_{0_C}$, $\epsilon_C$, and $C_C$ in Proposition \ref{prop_estim_max_princ} --- whereas $\delta$, $\n_0$, $\epsilon$, $C$ and $\delta_1$ 
always denote the variables form this proposition itself. We start with the region around $\pm \F$.
\begin{enumerate}
\item From Proposition \ref{prop_estim_max_princ} we get $\delta_C>0$, $\n_{0_C}\in\N$, $\epsilon_C>0$, $C_C\in\R$ and $\lambda>0$ so that for all $\n>\n_{0_C}$ and $\t \in(0,1-\delta_0)$
\begin{align}
\sup_{z\in \Xdelta{C}{\delta_C}{\F,\lambda}} \left\vert \frac{\dd \rhon}{\dd x}(z) \right\vert&< C_C \sqrt{\n} \e^{-\n \epsilon_C},\\
\shortintertext{and}
\sup_{z\in \Xdelta{C}{\delta_C}{\F,\lambda}} \left\vert \frac{\dd \rhon}{\dd y}(z) \right\vert&< C_C \sqrt{\n} \e^{-\n \epsilon_C}. 
\end{align}
For $\t\le \e^{-2}-\delta_0$ we can choose $\lambda=\F$ so that $\Xdelta{C}{\delta_C}{\F,\lambda}=\Xdelta{A}{\delta_C}{\F}^C$ and continue directly with paragraph 3, while setting $\delta_B=\delta_C$, $\n_{0_B}=\n_{0_C}$ and $C_B=C_C$.
\item Here we have now only to deal with $\t>\e^{-2}-\delta_0=\delta_{1_B}$. From Proposition \ref{prop_estim_pr_B} we get $\delta_B>0$, $\n_{0_B}\in\N$, $\epsilon_B>0$ and $C_B$ so that for all $\n>\n_{0_B}$ and $\t \in(\delta_{1_B},1-\delta_0)$
\begin{align}
\sup_{z\in \Xdelta{B}{\delta_B}{\F,\lambda}} \left\vert \frac{\dd \rhon}{\dd x}(z) \right\vert&< C_B \sqrt{\n} \e^{-\n \epsilon_B},\\
\shortintertext{and}
\sup_{z\in \Xdelta{B}{\delta_B}{\F,\lambda}} \left\vert \frac{\dd \rhon}{\dd y}(z) \right\vert&< C_B \sqrt{\n} \e^{-\n \epsilon_B}. 
\end{align}
\item Let $\delta_A=\min \{\delta_B,\delta_C \}$. Then $\Xdelta{A}{\delta_A}{\F} \cup \Xdelta{B}{\delta_B}{\F,\lambda} \cup \Xdelta{C}{\delta_C}{\F,\lambda}=\C$. From Proposition \ref{prop_asymp_R} we know that for all $\n>1$, $0<t<1-\delta_0$ and $z\in \Xdelta{A}{\delta_A}{\F}$
\begin{align}
\frac{\dd \rhon}{\dd x}(z) &= -\sqrt{\tfrac{\n}{2 \pi^3}}\tfrac{\vert \W{\F}(z)\vert^2 \e^{\n \f(z)}}{\F} \sqrt{\tfrac{1-\t}{1+\t}} \Re\left(\U{\F}(z) \bigl(1+\Rtz\bigr)\right), \\
\intertext{and}
\frac{\dd \rhon}{\dd y}(z) &=\sqrt{\tfrac{\n}{2 \pi^3}} \tfrac{\vert \W{\F}(z)\vert^2 \e^{\n \f(z)}}{\F} \sqrt{\tfrac{1+\t}{1-\t}} \Im\left( \U{\F}(z) \bigl(1+\Rtz\bigr)\right), 
\end{align}
with $\f$ as in Proposition \ref{prop_asym} and $R_{\ve{\xi},\n}$ as in Proposition \ref{prop_asymp_R}. With the help of Proposition \ref{prop_estim_R}, we get $\n_{0_A}\in\N$ and $C_A\in\R$, and from Proposition \ref{prop_estim_f} $\epsilon_A>0$ so that for all $\n>\n_{0_A}$ and $\t \in(0,1-\delta_0)$
\begin{align}
\sup_{z\in \Xdelta{A}{\delta_A}{\F} \setminus\Edelta{\delta} } \left\vert \frac{\dd \rhon}{\dd x}(z) \right\vert&< \tfrac{C_A}{\sqrt{2\pi^3}} \sqrt{\n} \e^{-\n \epsilon_A},\\
\shortintertext{and}
\sup_{z\in \Xdelta{A}{\delta_A}{\F}\setminus\Edelta{\delta} } \left\vert \frac{\dd \rhon}{\dd y}(z) \right\vert&< \tfrac{C_A}{\sqrt{2\pi^3}} \sqrt{\n} \e^{-\n \epsilon_A}. 
\end{align}

We define $C=\max \left\{\tfrac{C_A}{\sqrt{2\pi^3}},C_B,C_C\right\}$ and $\epsilon=\min \{ \epsilon_A,\epsilon_B,\epsilon_C \}>0$. Now the first claim follows since it is easy to see that there exists $\n_0>\max \left\{\n_{0_A},\n_{0_B},\n_{0_C} \right\}$, big enough so that for all $\n>\n_0$
\begin{align}
C \sqrt{\n} \e^{-\n \epsilon }<\epsilon.
\end{align}
Thus for all $\n>\n_0$ and $\t \in(0,1-\delta_0)$
\begin{align}
\sup_{z\in \C \setminus\Edelta{\delta} } \left\vert \frac{\dd \rhon}{\dd x}(z) \right\vert&< \epsilon, \qquad \text{and} \qquad \sup_{z\in \C \setminus\Edelta{\delta} } \left\vert \frac{\dd \rhon}{\dd y}(z) \right\vert< \epsilon. \nn 
\end{align}

\item In the second claim we choose $\delta_A>0$ small enough and assume without loss of generality that $\delta>0$ is small enough so that $\Edelta{\delta}\subset \Xdelta{A}{\delta_A}{\F}$. We only have to consider large $\vert z\vert$, thus $z\in \Adelta$. Again we use the constant $C_A$ and $\n_0\in\N$ from Proposition \ref{prop_estim_R} we have found in paragraph 3 and set $C=\tfrac{C_A}{\sqrt{2\pi^3}}$.
%
%
We are only left to show that there exists $R>0$ so that for all $\vert z \vert\ge R$ we have that $\f(z) \le - \vert z \vert$. For $\vert z \vert \gg \F_0$ we find that
\begin{align}
\U{\F}(z)= \tfrac{\F}{2z}+\tfrac{\F^3}{8z^3}+\lO\left(\tfrac{\F^5}{z^5} \right).
\end{align}
So we get an estimation if $\vert z \vert \gg \F_0$
\begin{align}
\f(z) &\le -\vert z\vert^2 (1-\t)+ \tfrac{F^2}{\vert z \vert^2}-2\log\left\vert \tfrac{\F}{2z}\right\vert+\log \t+1, \\
\intertext{and therefore}
 \tfrac{\f(z)}{\vert z\vert} &\le - \delta_0 \vert z\vert + \tfrac{F_0^2}{\vert z \vert^3}+\tfrac{\log(\vert z\vert^2)}{\vert z \vert}+\tfrac{1}{\vert z \vert} \rightarrow -\infty \quad (\vert z \vert \rightarrow\infty). \label{eq_f_large_z}
\end{align}

\item For the remaining claim we have to find an approximation for $\f$ and $\gpm$ in a neighborhood of $\pEt$. Already in the proof of Proposition \ref{prop_density} we have used the Taylor approximation of $\f$ as a function of $x$ at fixed $y=y_E$ or as a function of $y$ at fixed $x=x_E$. 

Without loss of generality $\epsilon<1$. Let $\delta_A>0$ be small enough so that for all $\t\in (0,1-\delta_0)$ and $\delta$ small enough $\Edelta{\delta}\subset \Xdelta{A}{\delta_A}{\F}$. Then for all $z \in \C_{x,\delta_1} \cap \Xdelta{A}{\delta_A}{\F} \supset \C_{x,\delta_1} \cap \Omegadelta{x,\delta}$ we see from Proposition \ref{prop_asymp_R} that
\begin{align}
\tfrac{\partial \rhon}{\partial x}(z) &=\sqrt{\tfrac{\n}{2 \pi^3}} \gp(z) \e^{\n \f(z)}\left(1+ \tfrac{ \Re\left( \U{\F}(z) \Rtz\right)}{\Re\left(\U{\F}(z)\right)}\right), \label{eq_proof_uniform_estim1} \\
\intertext{and for $z \in \C_{y,\delta_1} \cap \Xdelta{A}{\delta_A}{\F}$ }
\tfrac{\partial \rhon}{\partial y}(z) &=-\sqrt{\tfrac{\n}{2 \pi^3}} \gm(z) \e^{\n \f(z)}\left(1+ \tfrac{ \Im\left( \U{\F}(z) \Rtz\right)}{\Im\left(\U{\F}(z)\right)}\right). \label{eq_proof_uniform_estim2}
\end{align}

There exists $R>0$ so that $\Omegadelta{x,\delta}\cup \Omegadelta{y,\delta} \subset B_R(0)$ for all $\t\in (0,1-\delta_0)$. We set $\epsilon_A=\tfrac{\epsilon}{3}$. Then according to Proposition \ref{prop_estim_R} there exists $\n_{0_A}$ so that 
\begin{align}
\sup_{\substack{0<\t<1-\delta_0 \\\n>\n_{0_A} }} \sup_{\substack{z\in \Omegadelta{x,\delta}\cap \C_{x,\delta_1} \\ \xi_1,\xi_2,\xi_3\in[0,\frac{1}{\n}] }} \left\vert \tfrac{\Re\left(\U{\F}(z) \Rtz \right)}{\Re\, \U{\F}(z)} \right\vert &<\epsilon_A,\\
\shortintertext{and}
\sup_{\substack{0<\t<1-\delta_0 \\\n>\n_{0_A} }} \sup_{\substack{z\in \Omegadelta{y,\delta}\cap \C_{y,\delta_1} \\ \xi_1,\xi_2,\xi_3\in[0,\frac{1}{\n}] }} \left\vert \tfrac{\Im\left(\U{\F}(z) \Rtz \right)}{\Im\, \U{\F}(z) } \right\vert &<\epsilon_A.
\end{align}

From Proposition \ref{prop_g} it is obvious for all $\t\in (0,1-\delta_0)$ that both $\gp$ and $\tfrac{\partial^2 \f}{\partial x^2}$ on $\pEt$ only get zero where $\pEt$ crosses the imaginary line and analogously $\gm$ and $\tfrac{\partial^2 \f}{\partial y^2}$ where $\pEt$ intersects with the real line. Thus
\begin{align}
g_0&=\inf_{0<\t<1-\delta_0} \left(\min\left\{ \inf_{z_E\in \pEt \cap \C_{x,\delta_1} } \left\vert \gp(z_E) \right\vert,\inf_{z_E\in \pEt \cap \C_{y,\delta_1} } \left\vert \gm(z_E) \right\vert \right\}\right) \nn \\
&= \inf_{0<\t<1-\delta_0} \left( \min \biggl\{ \inf_{z_E \in \pEt \cap \C_{x,\delta_1} } \sqrt{\left\vert \tfrac{\partial^2 \f}{\partial x^2}(z_E) \right\vert}, \inf_{z_E \in \pEt \cap \C_{y,\delta_1} } \sqrt{\left\vert \tfrac{\partial^2 \f}{\partial y^2}(z_E) \right\vert} \biggr\} \right)
\end{align}
is strictly positive. Since $\gp$ and $\gm$ are uniformly continuous in $z$ on $\Xdelta{A}{\delta_A}{\F}$ 
and in $\t$ on any compact subset of $(0,1-\delta_0]$, it follows that there exists $\delta_2>0$ so that
\begin{align}
\vert \gp(x,y_E)-\gp(x_E,y_E)\vert < \tfrac{\epsilon}{3 g_0}, \qquad \forall z_E\in \pEt \cap \C_{x,\delta_1}, x\in [x_E-\delta_2,x_E+\delta_2], \\
\shortintertext{and}
\vert \gm(x_E,y)-\gm(x_E,y_E)\vert < \tfrac{\epsilon}{3 g_0}, \qquad \forall z_E\in \pEt \cap \C_{y,\delta_1}, y\in [y_E-\delta_2,y_E+\delta_2].
\end{align}
This is still true in the limit $\t\rightarrow 0$ since according to Lemma \ref{lemma_lim_f_g} 
\begin{align}
\lim_{\t\rightarrow 0} \gp(z)= -2\Re\left( \tfrac{1}{z} \right) \qquad \text{and} \qquad \lim_{\t\rightarrow 0} \gm(z)= -2\Im\left( \tfrac{1}{z} \right),
\end{align}
which are uniformly continuous for $z\in \Xdelta{A}{\delta_A}{0}$. 

We assume that $\delta_2<\tfrac{3}{4} \epsilon g_0^2 \delta_A^3$ and small enough that $\Edelta{\delta}\subset \Xdelta{A}{\delta_A}{\F}$. Then for all $\delta<\delta_2$
\begin{align}
(1-\tfrac{\epsilon}{3}) &< \tfrac{\gp(x,y_E)}{\gp(x_E,y_E)} <(1+\tfrac{\epsilon}{3}), 
\intertext{for all $z_E\in\pEt \cap \C_{x,\delta_1}$ and $x\in [x_E-\delta,x_E+\delta]$ and}
(1-\tfrac{\epsilon}{3}) &< \tfrac{\gm(x_E,y)}{\gm(x_E,y_E)} <(1+\tfrac{\epsilon}{3}), 
\end{align}
for all $z_E\in\pEt \cap \C_{y,\delta_1}$ and $y\in [y_E-\delta,y_E+\delta]$. Note that $x\in (x_E-\delta,x_E+\delta)$ and $z_E \in \pEt$ is equivalent to $(x,y_E)\in \Omegadelta{x,\delta}$ and analogously $y\in (y_E-\delta,y_E+\delta)$ and $z_E \in \pEt$ is equivalent to $(x_E,y)\in \Omegadelta{y,\delta}$.

According to Lemma \ref{lemma_T} and \ref{lemma_estim_b} $\vert \T{\F}(z) \vert^2> \delta_A^2$. With $\tfrac{\partial^3 \f}{\partial x^3}=4 \Re \left( (\T{\F}(z))^{-3} \right)$ and $\tfrac{\partial^3 \f}{\partial y^3}=4 \Im \left( (\T{\F}(z))^{-3} \right) $ we get
\begin{align}
\left\vert \tfrac{\partial^3 \f}{\partial x^3}(z) \right\vert \le \tfrac{4}{\vert \T{\F}(z)\vert^3} < \tfrac{4}{\delta_A^3},\qquad \text{and}\qquad \left\vert \tfrac{\partial^3 \f}{\partial y^3}(z) \right\vert < \tfrac{4}{\delta_A^3},
\end{align}
for all $\t\in (0,1-\delta_0)$ and $z\in \Xdelta{A}{\delta_A}{\F}$.

Since $\f$, $\tfrac{\partial \f}{\partial x}$ and $\tfrac{\partial \f}{\partial y}$ are zero on $\pEt$, the Taylor expansion at $x_E$ or $y_E$ up to order two, including a remainder in Lagrange form, is
\begin{align}
\f(x,y_E)&=\tfrac{\partial^2 \f}{\partial x^2}(x_E,y_E) \tfrac{(x-x_E)^2}{2}+\tfrac{\partial^3 \f}{\partial x^3}(\xi,y_E) \tfrac{(x-x_E)^3}{3!} \nn \\
&= \tfrac{\partial^2 \f}{\partial x^2}(x_E,y_E) \tfrac{(x-x_E)^2}{2} \left(1+ \left(\tfrac{\partial^2 \f}{\partial x^2}(x_E,y_E)\right)^{-1} \tfrac{\partial^3 \f}{\partial x^3}(\xi,y_E) \tfrac{x-x_E}{3} \right),
\intertext{with $\xi \in B_{\vert x-x_E \vert}(x_E)\cap \R$ and analogously}
\f(x_E,y)&= \tfrac{\partial^2 \f}{\partial y^2}(x_E,y_E) \tfrac{(y-y_E)^2}{2} \left(1+ \left(\tfrac{\partial^2 \f}{\partial y^2}(x_E,y_E)\right)^{-1} \tfrac{\partial^3 \f}{\partial y^3}(x_E,\xi) \tfrac{y-y_E}{3} \right),
\end{align}
with $\xi \in B_{\vert y-y_E \vert}(y_E)\cap \R$. For $\vert x-x_E\vert<\delta_2$ we get
\begin{align}
\left\vert \left(\tfrac{\partial^2 \f}{\partial x^2}(x_E,y_E)\right)^{-1} \tfrac{\partial^3 \f}{\partial x^3}(\xi,y_E) \tfrac{x-x_E}{3} \right\vert \le \tfrac{4 \delta_2}{3 g_0^2 \delta_A^3}<\epsilon,
\intertext{and for $\vert y-y_E\vert<\delta_2$}
\left\vert \left(\tfrac{\partial^2 \f}{\partial y^2}(x_E,y_E)\right)^{-1} \tfrac{\partial^3 \f}{\partial y^3}(x_E,\xi) \tfrac{y-y_E}{3} \right\vert \le \tfrac{4 \delta_2}{3 g_0^2 \delta_A^3}<\epsilon.
\end{align}
Noting that $\tfrac{\partial^2 \f}{\partial x^2}$ is negative, so for all $\t\in (0,1-\delta_0)$, $z_E\in\pEt \cap \C_{x,\delta_1}$ and $x\in [x_E-\delta,x_E+\delta]$ 
\begin{align}
\tfrac{\partial^2 \f}{\partial x^2}(x_E,y_E) \tfrac{(x-x_E)^2}{2} \left(1+\epsilon\right) < \f(x,y_E)< \tfrac{\partial^2 \f}{\partial x^2}(x_E,y_E) \tfrac{(x-x_E)^2}{2} \left(1-\epsilon\right),
\end{align}
and finally if additional $\n>\n_0=\n_{0_A}$, we get from \eqref{eq_proof_uniform_estim1}
\begin{align}
\left(1- \tfrac{\epsilon}{3} \right)^2  \e^{\n \frac{\partial^2 \f}{\partial x^2}(z_E) \frac{(x-x_E)^2}{2} \left(1+\epsilon\right) } &\le \tfrac{\frac{\partial \rhon}{\partial x}(x,y_E)}{ \sqrt{\frac{\n}{2 \pi^3}} \gp(z_E) } \le \left(1+ \tfrac{\epsilon}{3} \right)^2  \e^{\n \frac{\partial^2 \f}{\partial x^2}(z_E) \frac{(x-x_E)^2}{2} \left(1-\epsilon\right) }. \label{eq_pos1}
\end{align}
And analogously for all $\t\in (0,1-\delta_0)$, $z_E\in\pEt \cap \C_{y,\delta_1}$, $y\in [y_E-\delta,y_E+\delta]$ and $\n>\n_0$
\begin{align}
\left(1- \tfrac{\epsilon}{3} \right)^2  \e^{\n \frac{\partial^2 \f}{\partial y^2}(z_E) \frac{(y-y_E)^2}{2} \left(1+\epsilon\right) }& \le \tfrac{\frac{\partial \rhon}{\partial y}(x_E,y)}{-\sqrt{\frac{\n}{2 \pi^3}} \gm(z_E) } \le \left(1+ \tfrac{\epsilon}{3} \right)^2  \e^{\n \frac{\partial^2 \f}{\partial y^2}(z_E) \frac{(y-y_E)^2}{2} \left(1-\epsilon\right) }. \label{eq_pos2}
\end{align}
\end{enumerate}
\end{proof}

\begin{remark}
Note that in \eqref{eq_pos1} and \eqref{eq_pos2}
\begin{align}
\frac{\frac{\partial \rhon}{\partial x}(x,y_E)}{ \sqrt{\frac{\n}{2 \pi^3}} \gp(z_E) } \qquad \text{and}\qquad \frac{\frac{\partial \rhon}{\partial y}(x_E,y)}{-\sqrt{\frac{\n}{2 \pi^3}} \gm(z_E) } \nn
\end{align}
are positive. 
\end{remark}

\subsection{Verification of Estimation of Density}

Here we are going to verify the approximation we have done in Proposition \ref{prop_density}.
\begin{prop}\label{prop_veri_density}
Let $\delta>0$ and $0<\delta_0<1$. Then
\begin{align}
\lim_{\n\rightarrow\infty} \sup_{0<\t<1-\delta_0} \sup_{z\in \Edelta{-,\delta}} \left\vert \rhon(z)-\tfrac{1}{\pi}\right\vert &= 0, \\
\lim_{\n\rightarrow\infty} \sup_{0<\t<1-\delta_0} \sup_{z\in \C\setminus \Edelta{+,\delta}} \left\vert \rhon(z)\right\vert &= 0, \\
\lim_{\n\rightarrow\infty} \sup_{0<\t<1-\delta_0} \sup_{z\in \pEt} \left\vert \rhon(z)-\tfrac{1}{2\pi}\right\vert &= 0.
\end{align}
\end{prop}
\begin{proof}
We choose $\epsilon>0$. 
\begin{enumerate}
\item We first care about $z\in \Edelta{-,\delta}$. Proposition \ref{prop_unifrom_estimate_drho} gives us $\n_0\in\N$, independent of $z$ and $0<t<1-\delta_0$, and estimations \eqref{eq_unifrom_estimate_drho1}. For any two points $z_1,z_2\in \Edelta{-,\delta}$ there is a way from $z_1$ to $z_2$ inside of $\Edelta{-,\delta}$ consisting of paths parallel to either the real or the imaginary axis with total length smaller than $\mas+\mis<2\sqrt{\tfrac{2+\delta_0}{\delta_0}}$ for all $0<\t<1-\delta_0$. Therefore we get by integration of $\tfrac{\partial \rhon}{\partial x}$ and $\tfrac{\partial \rhon}{\partial y}$ that $\vert \rhon(z_2)-\rhon(z_1)\vert<2 \epsilon \sqrt{\tfrac{2+\delta_0}{\delta_0}}$ for all $\n>\n_0$, uniformly for $\t\in (0,1-\delta_0)$ and $z_1,z_2\in \Edelta{-,\delta}$. Taking the limit $\epsilon$ to zero it follows that the density is constant as $\n\rightarrow\infty$. We already have shown in Proposition \ref{prop_density_zero} that the density at zero converges to $\tfrac{1}{\pi}$, uniformly for $0<t<1-\delta_0$.

\item Next we consider $z\in \C\setminus \Edelta{+,\delta}$. From Proposition \ref{prop_unifrom_estimate_drho} we have $\n_0\in\N$ and estimations \eqref{eq_unifrom_estimate_drho1}, and additionally also $C>0$ and $R>0$, independent of $z$ and $0<t<1-\delta_0$, and estimations \eqref{eq_unifrom_estimate_drho2}, which hold uniformly for $z\notin B_R(0)$ and $\n>\n_0$. Again by integration along paths parallel to the real or imaginary axis, of total length smaller than $4R$, we find for $z_1,z_2\in (\C\setminus \Edelta{+,\delta})\cap B_R(0)$ that $\vert \rhon(z_2)-\rhon(z_1)\vert< 4\epsilon R$ for $\n>\n_0$. Therefore, as $\epsilon\rightarrow 0$, the density is constant, uniformly on $(\C\setminus \Edelta{+,\delta})\cap B_R(0)$ and $0<t<1-\delta_0$. 

For any $z_2\in \C \setminus \left(B_R(0) \cup \Edelta{+,\delta}\right)$ we can get the density by integration along paths parallel to the real or imaginary axis, of total length smaller than $2\vert z_2\vert$, starting at $z_1$ on the boundary of $B_R(0)$. So we get the estimation
\begin{align}
\vert \rhon(z_2)-\rhon(z_1)\vert < 2 C \sqrt{\n} \int_R^\infty \e^{-\n z} \dd z= \tfrac{2C}{\sqrt{\n}} \e^{-\n R}\rightarrow 0 \quad (\n\rightarrow\infty), 
\end{align}
uniformly for $z_2\in \C \setminus \left(B_R(0) \cup \Edelta{+,\delta}\right)$ and $0<t<1-\delta_0$. 

So we know that the density on the outside is constant too. We are left to show that the density is zero, i.e.\ that across $\Edelta{\delta}$ the density decreases by $\tfrac{1}{\pi}$. This we will do at the end of paragraph 3.

\item Now we consider the last claim, for $z\in\pEt$. In the following we will replace $z$ by $z_E$ and use $z$ as a free variable. Because of the symmetries of the density it is sufficient to restrict all of the following argumentation to $z_E\in\Cpp$. Choose $\delta_1$ small enough, for example $\delta_1<\tfrac{1}{8}$. Then we get from Proposition \ref{prop_unifrom_estimate_drho} $\n_0$, $\delta_2$ --- independent of $z_E$ and $0<t<1-\delta_0$ --- and estimations \eqref{eq_unifrom_estimate_drho3} and \eqref{eq_unifrom_estimate_drho4}. Assume that $\delta_2>0$ is small enough so that $\delta_2<\delta_1$ and $\Edelta{\delta_2} \subset \C_{x,2\delta_1}\cup\C_{y,2\delta_1}$. For any $z_E\in \pEt$ there is a path of length $\delta_2$ from $z_1$ to $z_E$ where $z_1\in \partial\Omegadelta{x,-,\delta_2}\cap \C_{x,\delta_1}$ and the path is parallel to the real axis or $z_1\in \partial\Omegadelta{y,-,\delta_2}\cap \C_{y,\delta_1}$ and the path is parallel to the imaginary axis. 

There is a finite separation of $\pEt$ from $\partial\Omegadelta{x,-,\delta_2}$ and $\partial\Omegadelta{y,-,\delta_2}$ for all $0<t<1-\delta_0$, i.e.\ 
\begin{align}
\tilde{\delta}=\inf_{0<t<1-\delta_0 }\inf_{\substack{z\in \partial\Omegadelta{x,-,\delta_2} \cup \partial\Omegadelta{y,-,\delta_2} \\ z_E\in \pEt} } \vert z_E-z \vert,
\end{align}
is strictly positive while the maximal distance of $z\in \partial\Edelta{-,\tilde{\delta}}$ from $\pEt$, i.e.\ 
\begin{align}
\sup_{0<t<1-\delta_0} \sup_{z\in \partial\Edelta{-,\tilde{\delta}} } \inf_{z_E\in \pEt} \vert z_E-z \vert
\end{align}
is not bigger than $\tilde{\delta}$ since
\begin{align}
\inf_{\phi_2 \in (-\pi,\pi]} & \left\vert \mas \cos \phi_2+i \mis \sin \phi_2 - (\mas-\tilde{\delta}) \cos \phi_1 - i(\mis -\tilde{\delta}) \sin\phi_1 \right\vert \nn \\
 \le & \left\vert \tilde{\delta} \cos \phi_1 + i \tilde{\delta} \sin\phi_1 \right\vert=\tilde{\delta}.
\end{align}
Therefore $\partial\Omegadelta{x,-,\delta_2} \cup \partial\Omegadelta{y,-,\delta_2}\subset \Edelta{-,\tilde{\delta}}$ and thus we know from what we have already proved in paragraph 1, with $\delta$ set to $\tilde{\delta}$, that the density at $z_1$ converges to $\tfrac{1}{\pi}$ uniformly for $z_1\in \partial\Omegadelta{x,-,\delta_2} \cup \partial\Omegadelta{y,-,\delta_2}$ and $0<t<1-\delta_0$.

On the path from $z_1=x_1+i y_1 \in \partial\Omegadelta{x,-,\delta_2}\cap \C_{x,\delta_1}$ to $z_E=x_E+i y_E$, i.e.\ for $z=x + i y$ with $x\in (x_1,x_E)=(x_E-\delta_2,x_E)$ and $y=y_E$, estimation \eqref{eq_unifrom_estimate_drho3} is valid for $\n>\n_0$. So we find by integration of the inequality \eqref{eq_unifrom_estimate_drho3}, like in Proposition \ref{prop_density}, with the substitution $\eta_+ (x-x_E)$ by $x$ where 
\begin{align}
\eta_+ &=\sqrt{-\frac{\partial^2 \f}{\partial x^2}(x_E,y_E) \frac{\n}{2}}>0
\end{align}
and with the help of Proposition \ref{prop_g} that
\begin{align}
\tfrac{(1-\epsilon)}{\pi^{3/2}}\int_{-\eta_+ (x_E-x_1)}^{0} \e^{-(1+\epsilon)x^2 } \dd x \le \rhon(z_1)-\rhon(z_E) \le \tfrac{(1+\epsilon)}{\pi^{3/2}} \int_{-\eta_+ (x_E-x_1) }^{0} \e^{-(1-\epsilon)x^2 } \dd x. \label{eq_veri_dens1}
\end{align}
For the integrals we get with $x_E-x_1=\delta_2$
\begin{align}
\tfrac{(1\mp\epsilon)}{\pi^{3/2}}\int_{-\eta_+ \delta_2}^{0} \e^{-(1\pm\epsilon)x^2 } \dd x= \tfrac{1\mp\epsilon}{\sqrt{1\pm\epsilon}} \tfrac{1}{2\pi} \erf\left(\delta_2 \sqrt{(1\pm\epsilon) \left(-\tfrac{\partial^2\f}{\partial x^2}(z_E) \right) \tfrac{\n}{2}}  \right). \label{eq_veri_dens2}
\end{align}
This converges to $\tfrac{1}{2\pi}$ as $\epsilon$ goes to zero and $\n$ to infinity, uniformly for $0<\t<1-\delta_0$ and $z_1 \in \partial\Omegadelta{x,-,\delta_2}\cap \C_{x,\delta_1}$ since $\delta_2$ is independent of $\t$ and $z_1$ and since $-\tfrac{\partial^2\f}{\partial x^2}(z_E)\ge g_0$ where
\begin{align}
g_0&=\inf_{0<\t<1-\delta_0} \left(\min\{ \inf_{z_E\in \pEt \cap \C_{x,\delta_1} } \left\vert \gp(z_E) \right\vert,\inf_{z_E\in \pEt \cap \C_{y,\delta_1} } \left\vert \gm(z_E) \right\vert \}\right)>0 .
\end{align}

Analogously on the path from $z_1 \in \partial\Omegadelta{y,-,\delta_2}\cap \C_{y,\delta_1}$ to $z_E$, i.e.\ for $z=x + i y$ with $x=x_E$ and $y\in (y_E-\delta_2,y_E)$, estimation \eqref{eq_unifrom_estimate_drho4} is valid for $\n>\n_0$. So we find by integration of the inequality that 
\begin{align}
\tfrac{(1-\epsilon)}{\pi^{3/2}}\int_{-\eta_- ( y_E-y_1 ) }^{0} \e^{-(1+\epsilon)y^2 } \dd y \le \rhon(z_1)-\rhon(z_E) \le \tfrac{(1+\epsilon)}{\pi^{3/2}} \int_{-\eta_- ( y_E-y_1 ) }^{0} \e^{-(1-\epsilon)y^2 } \dd y,
\end{align}
where
\begin{align}
\eta_- &=\sqrt{-\frac{\partial^2 \f}{\partial y^2}(x_E,y_E) \frac{\n}{2}}>0.
\end{align}
And
\begin{align}
\tfrac{(1\mp\epsilon)}{\pi^{3/2}}\int_{-\eta_- \delta_2}^{0} \e^{-(1\pm\epsilon)y^2 } \dd y= \tfrac{1\mp\epsilon}{\sqrt{1\pm\epsilon}} \tfrac{1}{2\pi} \erf\left(\delta_2 \sqrt{(1\pm\epsilon) \left(-\tfrac{\partial^2\f}{\partial y^2}(z_E) \right) \tfrac{\n}{2}}  \right)
\end{align}
converges to $\tfrac{1}{2 \pi}$ as $\epsilon$ goes to zero and $\n$ to infinity, uniformly for $0<\t<1-\delta_0$ and $z_1 \in \partial\Omegadelta{y,-,\delta_2}\cap \C_{y,\delta_1}$. So we have shown that on $\pEt$ the density $\rhon(z_E)$ converges to $\lim_{\n\rightarrow\infty} \rhon(z_1)-\tfrac{1}{2\pi}=\tfrac{1}{2\pi}$, uniformly for all $0<\t<1-\delta_0$ and $z_E\in\pEt$.

We still have to find a connection from $\Edelta{-,\delta}$ to $\C\setminus \Edelta{+,\delta}$ to prove that the constant density of the latter is really zero. It is sufficient to show this at one point and for simplicity we will do this along the real axis, i.e.\ across $z_E=\mas$. Without loss of generality $\delta\le\delta_2$. Since estimation \eqref{eq_unifrom_estimate_drho3} is valid on $[z_E-\delta,z_E+\delta]$ we find like in \eqref{eq_veri_dens1}
\begin{align}
\tfrac{(1-\epsilon)}{\pi^{3/2}}\int_{-\eta_+ \delta}^{\eta_+ \delta} \e^{-(1+\epsilon)x^2 } \dd x \le \rhon(z_E-\delta)-\rhon(z_E+\delta) \le \tfrac{(1+\epsilon)}{\pi^{3/2}} \int_{-\eta_+ \delta}^{\eta_+ \delta} \e^{-(1-\epsilon)x^2 } \dd x,
\end{align}
and
\begin{align}
\tfrac{(1\mp\epsilon)}{\pi^{3/2}} \int_{-\eta_+ \delta}^{\eta_+ \delta}   \e^{-(1\pm\epsilon)y^2 } \dd y= \tfrac{1\mp\epsilon}{\sqrt{1\pm\epsilon}} \tfrac{1}{\pi} \erf\left(\delta \sqrt{(1\pm\epsilon) \left(-\tfrac{\partial^2\f}{\partial y^2}(z_E) \right) \tfrac{\n}{2}}  \right),
\end{align}
which converges to $\tfrac{1}{\pi}$ as $\epsilon$ goes to zero and $\n$ to infinity, uniformly for $0<\t<1-\delta_0$. Therefore the density jumps from $\tfrac{1}{\pi}$ to zero across $\pEt$.
\end{enumerate}
\end{proof}
\begin{remark}
Since there is a discontinuity of the density across $\pEt$ it is obvious that the density can't converge uniformly in a neighborhood of $\pEt$. We can also see this for example in \eqref{eq_veri_dens2}. The error function does not converge uniformly for arbitrarily small $\delta_2$.
\end{remark}

So we have not only verified Proposition \ref{prop_density} but additionally shown where the convergence is uniform in $\t$ and $z$. It isn't necessary to review Theorem \ref{th_density_scaled_limit} as it will follow from Proposition \ref{prop_veri_kernel} setting $a=b$.

\subsection{Error Term of Kernel}

\subsubsection{Plancherel-Rotach Asymptotics on $\Adelta$}
Here we want to review the error term in Proposition \ref{prop_asym_wz}.
\begin{prop}\label{prop_asym_wz_R}
Let $\delta>0$, $0<\delta_0<1$ and $A\in (0,\infty)$. Then there exists $\n_0=\n_0(\delta,A) \in\N$ and functions 
\begin{align}
\begin{array}{rcl}
\xi_k: \N\setminus\{1,\ldots,\n_0-1\} & \longrightarrow & \left[0,\tfrac{1}{\n_0}\right]\\
\n & \longmapsto & \xi_k(\n)\le\frac{1}{\n}
\end{array}, \qquad k=1,2,3 \\
\intertext{and}
\begin{array}{rcl}
\xi_k: \cc{B}_{\!\frac{A}{\sqrt{\n_0}}}(0) & \longrightarrow & \cc{B}_{\!\frac{A}{\sqrt{\n_0}}}(0)\\
u & \longmapsto & \xi_k(u)\in B_{\vert u\vert}(0)
\end{array}, \qquad k=4,5,6,
\end{align}
so that for all $0<\t<1-\delta_0$, $z_0\in \Adelta$, $u,v\in \cc{B}_{\!\frac{A}{\sqrt{\n_0}}}(0)$ and $\n\ge\n_0$ 
\begin{align}
\left.\frac{\partial \HM}{\partial w}\right\vert_{\substack{\phantom{}_{w=\cc{z}_0+u} \\ \phantom{}_{z=z_0+v} \\ \phantom{}}} =&\sqrt{\tfrac{\n}{2 \pi^3}} \gw(z_0) \e^{\n \left( \f(z_0) +\hu(\cc{z}_0)u+\huu(\cc{z}_0)u^2 +\hu(z_0)v+\huu(z_0)v^2+\huv uv \right)} \nn \\ 
& \cdot \left(1+\Rxi{+}{(\t,z_0,u,v)}\right), \label{eq_lemma_asym_wz1_R} \\
\intertext{and}
\left.\frac{\partial \HM}{\partial z}\right\vert_{\substack{\phantom{}_{w=\cc{z}_0+u} \\ \phantom{}_{z=z_0+v} \\ \phantom{}}} =&\sqrt{\tfrac{\n}{2 \pi^3}} \gz(z_0) \e^{\n \left( \f(z_0) +\hu(\cc{z}_0)u+\huu(\cc{z}_0)u^2 +\hu(z_0)v+\huu(z_0)v^2+\huv uv \right)} \nn \\
& \cdot \left(1+\Rxi{-}{(\t,z_0,u,v)}\right), \label{eq_lemma_asym_wz2_R} 
\end{align}
where
\begin{align}
\f(z_0)&=\t \Re{\left(z_0^2\right)}-\left\vert z_0\right\vert^2+\Re\left(\left(\U{\F}(z_0)\right)^2 \right)-2\log\left\vert \U{\F}(z_0) \right\vert +\log\t+1 \label{eq_f_wz_R},\\
\gw(z_0)&=\tfrac{1}{4}\vert \W{\F}(z_0) \vert^2 \left(\t^{1/2} \U{\F}(z_0)-\t^{-1/2} \U{\F}(\cc{z}_0)\right), \label{eq_gw_R} \\
\gz(z_0)&=\gw(\cc{z}_0), \label{eq_gz_R} \\
\hu(z_0)&=\tfrac{ 2 \U{\F}(z_0) }{\F}+\t z_0-\cc{z}_0 , \label{eq_hu_R} \\ 
\huu(z_0)&=\tfrac{\t}{2}-\tfrac{\U{\F}(z_0)}{\F \T{\F}(z_0)}, \label{eq_huu_R} \\ 
\huv&=-1 \label{eq_huv_R}, \\
\intertext{and the error term is}
\Rxi{\pm}{(\t,z_0,u,v)} &= \tfrac{\t^{\pm 1/2} \U{\F}(z_0)\left(\e^{\Rsum{(\t,z_0,u,v)}}-1 \right) - \t^{\mp 1/2} \U{\F}(\cc{z}_0)\left(\e^{\Rsum{(\t,\cc{z}_0,v,u)}}-1 \right) }{\t^{\pm 1/2} \U{\F}(z_0)-\t^{\mp 1/2} \U{\F}(\cc{z}_0)}, \\
\intertext{with}
\Rsum{(\t,z_0,u,v)}=&\Rzeroz{z_0,u,v}+\Rone+\Rtwoz{z_0}+\Rthreez{z_0} \nn \\
& +\n \left(\Rfourz{\cc{z}_0,u}+\Rfour \right) \nn \\
& +\Rfivez{\cc{z}_0,u}+\Rfive+\Rsix, 
\end{align}
and
\begin{align}
\Rzeroz{z_0,u,v}=&\log\left(  \tfrac{Q_\n\left(\sqrt{2} \frac{\cc{z}_0+u}{\F} \right)}{\pi_\n\left(\sqrt{2} \frac{\cc{z}_0+u}{\F} \right)}   \tfrac{Q_{\n-1}\left(\sqrt{\frac{\n}{\n-1}}\sqrt{2}\frac{z_0+v}{\F} \right)}{\pi_{\n-1}\left(\sqrt{\frac{\n}{\n-1}}\sqrt{2}\frac{z_0+v}{\F}\right)}   \right),\\
\Rone=&-\tfrac{1}{12} \xi_1(\n) ,\\ 
\Rtwoz{z_0}=&-\tfrac{\F\U{\F}\left(\frac{z_0}{\sqrt{1-\xi_2(\n) } }  \right)  }{4\n\left(1-\xi_2(\n) \right) \T{\F}\left(\frac{z_0}{\sqrt{1-\xi_2(\n) } }  \right) },\\
\Rthreez{z_0}=&-\tfrac{z_0\F\U{\F}\left(\frac{z_0}{\sqrt{1-\xi_3(\n) } }  \right)  }{4\n\left(1-\xi_3(\n) \right)^{3/2} \left(\T{\F}\left(\frac{z_0}{\sqrt{1-\xi_3(\n) } }  \right)\right)^2  } ,\\
\Rfour=&\tfrac{v^3}{3 \left(\T{\F}\left(z_0+\xi_4(v) \right)\right)^3 }, \\
\Rfive=&-\tfrac{v \F \U{\F}(z_0+\xi_5(v))}{2 \left(\T{\F}(z_0+\xi_5(v))\right)^2 }, \\
\Rsix=& -\tfrac{v}{\T{\F}(z_0+\xi_6(v)) } +\tfrac{v \F^2}{4\n\left(1-\xi_2(\n) \right)^{3/2}\left(\T{\F}\left(\frac{z_0+\xi_6(v)}{\sqrt{1-\xi_2(\n) } } \right)\right)^3 } \nn \\
&+\tfrac{v \F^2(z_0+\xi_6(v))^2 - v ( 1-\xi_3(\n) )\F^4  \left(\U{\F}\left(\frac{z_0+\xi_6(v)}{\sqrt{1-\xi_3(\n) } }  \right) \right)^2  }{4\n\left(1-\xi_3(\n) \right)^{5/2} \left(\T{\F}\left(\frac{z_0+\xi_6(v)}{\sqrt{1-\xi_3(\n) } }  \right)\right)^5  }. 
\end{align}
\end{prop}
\begin{remark}
Note that $\f$, $\gw$, $\gz$, $\hu$, $\huu$ and $\huv$ are as in Proposition \ref{prop_asym_wz}. We see that not only $\f$ but also $\huu$ and trivially $\huv$ are symmetric under the map $z_0\mapsto -z_0$ whereas $\gw$, $\gz$ and $\hu$ are skew-symmetric.
\end{remark}

\begin{proof}
We choose $\n_0> \max\{\tfrac{4A^2}{\delta^2},1\}$ 
and assume in the following that $\n>\n_0$. Then $u,v \in \cc{B}_{\!\frac{\delta}{2}}(0)$ and $\cc{z}_0+u$ and $z_0+v$ lie in $\Xdelta{A}{\frac{\delta}{2}}{\F}$ for all $z_0\in \Adelta$. 

As in the proof of Proposition \ref{prop_asymp_R} in \eqref{eq_op2_error} we can use Lemma \ref{lemma_R2_R3_new} for all $\cc{z}_0+u,z_0+v\in \C \setminus [-\F,\F]$ to get for $\op{\n-1}$ from \eqref{eq_op2} 
\begin{align}
\op{\n-1}(z_0+v)=&\tfrac{Q_{\n-1}\left(\sqrt{\frac{\n}{\n-1}}\sqrt{2}\frac{z_0+v}{\F} \right)}{\pi_{\n-1}\left(\sqrt{\frac{\n}{\n-1}}\sqrt{2}\frac{z_0+v}{\F}\right)} \tfrac{\n^{\frac{\n+1}{2}}}{2\sqrt{\pi}\sqrt{\n !}}  \e^{\frac{\n \left(\U{\F}(z_0+v)\right)^2  }{2} } \left(\U{\F}(z_0+v) \right)^{-(\n-1)} \nn \\
&\cdot \t^{\frac{\n-1}{2}} \left(1-\t^2\right)^{1/4} \W{\F}(z_0+v) \e^{\Rtwoz{z_0+v}+\Rthreez{z_0+v}}.\label{eq_op2_wz_error}
\end{align}
and analogous expression for $\op{\n-1}(\cc{z}_0+u)$. The error terms $\Rtwotz{}$ and $\Rthreetz{}$ are
\begin{align}
\Rtwoz{z_0+v}&=-\tfrac{\F\U{\F}\left(\frac{z_0+v}{\sqrt{1-\xi_2(\n) } }  \right)  }{4\n\left(1-\xi_2(\n) \right) \T{\F}\left(\frac{z_0+v}{\sqrt{1-\xi_2(\n) } }  \right) },\\
\shortintertext{and}
\Rthreez{z_0+v}&=-\tfrac{(z_0+v)\F\U{\F}\left(\frac{z_0+v}{\sqrt{1-\xi_3(\n) } }  \right)  }{4\n\left(1-\xi_3(\n) \right)^{3/2} \left(\T{\F}\left(\frac{z_0+v}{\sqrt{1-\xi_3(\n) } }  \right)\right)^2  } ,
\end{align}
with $\xi_2(\n),\xi_3(\n)\in \left[0,\tfrac{1}{\n}\right]$. For $\op{\n}$ from \eqref{eq_op1} at $z=z_0+v$ we have
\begin{align}
\op{\n}(z_0+v)=&\tfrac{Q_\n\left(\sqrt{2} \frac{z_0+v}{\F} \right)}{\pi_\n\left(\sqrt{2} \frac{z_0+v}{\F} \right)}\tfrac{\n^{\frac{\n+1}{2}}}{2\sqrt{\pi}\sqrt{\n !}}\e^{\frac{\n \left(\U{\F}(z_0+v)\right)^2 }{2} } \left( \U{\F}(z_0+v) \right)^{-\n} \t^{\frac{\n}{2}} \left(1-\t^2\right)^{1/4} \W{\F}(z_0+v). \label{eq_op1_wz_error}
\end{align}

As $u,v\in \cc{B}_{\!\frac{\delta}{2}}(0)$ and $\xi_2,\xi_3\in \left[0,\tfrac{1}{\n}\right]\subset\left[0,\tfrac{1}{2}\right]$ we can use Lemma \ref{lemma_R4_R5_R6} which gives us
\begin{align} 
\e^{\n \ffour(v)}&=\e^{\frac{\n}{2}\left(\U{\F}(z_0+v)\right)^2} \left(\U{\F}(z_0+v) \right)^{-\n} \nn \\
&=\e^{\frac{\n}{2}\left(\U{\F}(z_0)\right)^2 } \left(\U{\F}(z_0)\right)^{-\n} \e^{\n \left( \frac{2 v \U{\F}(z_0)}{\F} -\frac{ v^2 \U{\F}(z_0)}{\F \T{\F}(z_0)} +  \Rfour \right)}, \label{eq_f4_error}  \\
\e^{\ffive(v)}&= \W{F}(z_0+v)=  \W{\F}(z_0)\e^{ \Rfive}, \label{eq_f5_error}
\shortintertext{and}
\e^{\fsix(v)}&= \U{\F}(z_0+v)\e^{\Rtwoz{z_0+v}}\e^{\Rthreez{z_0+v}} \nn \\ 
&=  \U{\F}(z_0)   \e^{\Rtwoz{z_0} + \Rthreez{z_0} + \Rsix}, \label{eq_f6_error}
\intertext{where}
\Rfour&=\tfrac{v^3}{3 \left(\T{\F}\left(z_0+\xi_4(v) \right)\right)^3 }, \\
\Rfive&=-\frac{v \F \U{\F}(z_0+\xi_5(v))}{2 \left(\T{\F}(z_0+\xi_5(v))\right)^2 }, \\
\Rsix&= -\tfrac{v}{\T{\F}(z_0+\xi_6(v)) } +\tfrac{v \F^2}{4\n\left(1-\xi_2(\n) \right)^{3/2}\left(\T{\F}\left(\frac{z_0+\xi_6(v)}{\sqrt{1-\xi_2(\n) } } \right)\right)^3 } \nn \\
&\phantom{=}  +\tfrac{v \F^2(z_0+\xi_6(v))^2 - v ( 1-\xi_3(\n) )\F^4  \left(\U{\F}\left(\frac{z_0+\xi_6(v)}{\sqrt{1-\xi_3(\n) } }  \right) \right)^2  }{4\n\left(1-\xi_3(\n) \right)^{5/2} \left(\T{\F}\left(\frac{z_0+\xi_6(v)}{\sqrt{1-\xi_3(\n) } }  \right)\right)^5  }, 
\end{align} 
with $\xi_4(v),\xi_5(v),\xi_6(v)\in \cc{B}_{\vert v\vert}(0)$. These error terms account for the $\lO(\n^{-1/2})$ in the expansions \eqref{eq_expan6}-\eqref{eq_expan9} of the proof of Proposition \ref{prop_asym_wz}. With the help of that, we find for \eqref{eq_op2_wz_error} and \eqref{eq_op1_wz_error}
\begin{align}
\op{\n-1}(z_0+v)&=\tfrac{Q_{\n-1}\left(\sqrt{\frac{\n}{\n-1}}\sqrt{2}\frac{z_0+v}{\F} \right)}{\pi_{\n-1}\left(\sqrt{\frac{\n}{\n-1}}\sqrt{2}\frac{z_0+v}{\F}\right)} \tfrac{\n^{\frac{\n+1}{2}} \t^{\frac{\n-1}{2}} \left(1-\t^2\right)^{1/4}  }{2\sqrt{\pi}\sqrt{\n !}}  \e^{\n \ffour(v)+\ffive(v)+ \fsix(v)},  \nn \\
&=\tfrac{Q_{\n-1}\left(\sqrt{\frac{\n}{\n-1}}\sqrt{2}\frac{z_0+v}{\F} \right)}{\pi_{\n-1}\left(\sqrt{\frac{\n}{\n-1}}\sqrt{2}\frac{z_0+v}{\F}\right)} \tfrac{\n^{\frac{\n+1}{2}} \t^{\frac{\n-1}{2}} \left(1-\t^2\right)^{1/4}  }{2\sqrt{\pi}\sqrt{\n !}}   \e^{\frac{\n}{2}\left(\U{\F}(z_0)\right)^2 } \left(\U{\F}(z_0)\right)^{-\n} \W{\F}(z_0)  \nn \\
& \quad \cdot \e^{\n \left( \frac{2 v \U{\F}(z_0)}{\F} -\frac{ v^2 \U{\F}(z_0)}{\F \T{\F}(z_0)} +  \Rfour \right)}  \e^{ \Rfive}  \nn \\
& \quad \cdot \U{\F}(z_0) \e^{\Rtwoz{z_0} + \Rthreez{z_0} + \Rsix},  \label{eq_op2_wz_error2} \\
\shortintertext{and}
\op{\n}(z_0+v)&=\tfrac{Q_\n\left(\sqrt{2} \frac{z_0+v}{\F} \right)}{\pi_\n\left(\sqrt{2} \frac{z_0+v}{\F} \right)}\tfrac{\n^{\frac{\n+1}{2}} \t^{\frac{\n}{2}} \left(1-\t^2\right)^{1/4}  }{2\sqrt{\pi}\sqrt{\n !}} \e^{\n\ffour(v)} \e^{\ffive(v)}\nn \\
&= \tfrac{Q_\n\left(\sqrt{2} \frac{z_0+v}{\F} \right)}{\pi_\n\left(\sqrt{2} \frac{z_0+v}{\F} \right)}\tfrac{\n^{\frac{\n+1}{2}} \t^{\frac{\n}{2}} \left(1-\t^2\right)^{1/4}  }{2\sqrt{\pi}\sqrt{\n !}} \e^{\frac{\n}{2}\left(\U{\F}(z_0)\right)^2 } \left(\U{\F}(z_0)\right)^{-\n}    \W{\F}(z_0) \nn \\
&\quad \cdot  \e^{\n \left( \frac{2 v \U{\F}(z_0)}{\F} -\frac{ v^2 \U{\F}(z_0)}{\F \T{\F}(z_0)} +  \Rfour \right)}   \e^{ \Rfive}  . \label{eq_op1_wz_error2}
\end{align}

The Plancherel-Rotach asymptotics \eqref{eq_pi_o} is valid because $\sqrt{2} \tfrac{z_0+v}{\F}$ and $\sqrt{\frac{\n}{\n-1}}\sqrt{2}\frac{z_0+v}{\F}$ lie in $\Xdelta{A}{\tilde{\delta}}{\F}$ when $z_0+v\in \Xdelta{A}{\frac{\delta}{2}}{\F}$ and $\tilde{\delta}=\tfrac{\delta}{\sqrt{2}\F_0}$ where $\F_0=\sqrt{\tfrac{4(1-\delta_0)}{2\delta_0-\delta_0^2}}$ is the largest value of $\F$. Further we know that $Q_\n$ and $Q_{\n-1}$ have no zeros there. Thus, similarly as in the proof of Proposition \ref{prop_asymp_R}, we can define the error term
\begin{align}\label{eq_pr_wz}
\Rzeroz{z_0,u,v}=\log\left(  \tfrac{Q_\n\left(\sqrt{2} \frac{\cc{z}_0+u}{\F} \right)}{\pi_\n\left(\sqrt{2} \frac{\cc{z}_0+u}{\F} \right)}   \tfrac{Q_{\n-1}\left(\sqrt{\frac{\n}{\n-1}}\sqrt{2}\frac{z_0+v}{\F} \right)}{\pi_{\n-1}\left(\sqrt{\frac{\n}{\n-1}}\sqrt{2}\frac{z_0+v}{\F}\right)}   \right).
\end{align}

Again we have an error term coming from the Stirling approximation
\begin{equation}\label{eq_proof_R1_wz}
\frac{\sqrt{2\pi \n} \n^\n \e^{-\n} }{\n!}=\e^{-\frac{1}{12} \xi_1(\n)}=\e^{\Rone},
\end{equation}
where $\xi_1(\n)\in (\tfrac{12}{12\n+1},\tfrac{1}{\n})\subset (0,\tfrac{1}{\n})$ and $\Rone=-\tfrac{1}{12} \xi_1(\n) < 0$.

Using \eqref{eq_op2_wz_error2}, \eqref{eq_op1_wz_error2}, \eqref{eq_pr_wz} and \eqref{eq_proof_R1_wz}, we get as in \eqref{eq_asym_wz2} of the proof of Proposition \ref{prop_asym_wz}
\begin{align}\label{eq_asym_wz2_R}
\pm\t^{1/2}&\left(\t^{\pm 1/2}\op{\n}(\cc{z}_0+u) \op{\n-1}(z_0+v)-\t^{\mp 1/2}\op{\n-1}(\cc{z}_0+u) \op{\n}(z_0+v) \right) \nn \\
&=\pm \tfrac{\sqrt{\n} \e^{\n} \t^\n \sqrt{1-t^2} }{4\pi\sqrt{2\pi}}\e^{\Rone} \e^{\n\left(\ffourz{\cc{z}_0}(u)+\ffour(v)\right)+\ffivez{\cc{z}_0}(u)+\ffive(v)}\nn \\
&\qquad\cdot \left( \t^{\pm 1/2} \e^{\Rzeroz{z_0,u,v} + \fsix(v)} - \t^{\mp 1/2} \e^{\Rzeroz{\cc{z}_0,v,u} + \fsixz{\cc{z}_0}(u)} \right)
\end{align}
Using this and \eqref{eq_asym_wz3}, we finally find for the identities \eqref{eq_id3} and \eqref{eq_id4} at $w=\cc{z}_0+u$ and $z=z_0+v$
\begin{align}
&\left. \begin{array}{l} \left.\frac{\partial \HM}{\partial w}\right\vert_{\substack{\phantom{}_{w=\cc{z}_0+u} \\ \phantom{}_{z=z_0+v} \\ \phantom{}}}  \\  \left.\frac{\partial \HM}{\partial z}\right\vert_{\substack{\phantom{}_{w=\cc{z}_0+u} \\ \phantom{}_{z=z_0+v} \\ \phantom{}}} \end{array}\right\} = \pm\t^{1/2} \frac{ \e^{\n\left( -(\cc{z}_0+u)(z_0+v)+\frac{\t(\cc{z}_0+u)^2+\t(z_0+v)^2 }{2}  \right)} }{\sqrt{1-\t^2}} \nn \\
& \qquad \cdot \left(\t^{\pm 1/2}\op{\n}(\cc{z}_0+u) \op{\n-1}(z_0+v)-\t^{\mp 1/2}\op{\n-1}(\cc{z}_0+u) \op{\n}(z_0+v) \right) \nn\\
& \quad = \pm \tfrac{\sqrt{\n}}{\sqrt{2\pi^3}} \e^{\n\left(1+\log \t +  \ffourz{\cc{z}_0}(u)+\ffour(v) -\vert z_0\vert^2+\t\Re(z_0^2)+u(\t \cc{z}_0-z_0)+v(\t z_0-\cc{z}_0)+\frac{u^2\t}{2}+ \frac{v^2\t}{2}-uv\right)} \nn \\
& \qquad \cdot \e^{\Rone +  \ffivez{\cc{z}_0}(u)+\ffive(v)} \nn \\ 
& \qquad \cdot \tfrac{1}{4} \left( \t^{\pm 1/2} \e^{\Rzeroz{z_0,u,v} + \fsix(v)} - \t^{\mp 1/2} \e^{\Rzeroz{\cc{z}_0,v,u} + \fsixz{\cc{z}_0}(u)} \right) \nn \\
& \quad = \pm \tfrac{\sqrt{\n}}{\sqrt{2\pi^3}} \e^{\n\left(1+\log \t + \Re\left( \left(\U{\F}(z_0)\right)^2 \right)-2\log \left\vert\U{\F}(z_0)\right\vert  -\vert z_0\vert^2+\t\Re(z_0^2)\right)} \nn \\
& \qquad \cdot \e^{\n\left(u \left(\t \cc{z}_0-z_0 + \frac{2 \U{\F}(\cc{z}_0)}{\F} \right) +v \left(\t z_0-\cc{z}_0 + \frac{2 \U{\F}(z_0)}{\F} \right)+u^2 \left(\frac{\t}{2}-\frac{ \U{\F}(\cc{z}_0)}{\F \T{\F}(\cc{z}_0)} \right)  + v^2\left( \frac{\t}{2} -\frac{\U{\F}(z_0)}{\F \T{\F}(z_0)} \right) - uv \right) }\nn \\
& \qquad \cdot \e^{\Rone +  \n \Rfourz{\cc{z}_0,u} + \n \Rfour + \Rfivez{\cc{z}_0,u}  +\Rfive  } \nn \\ 
& \qquad \cdot \tfrac{ \vert \W{\F}(z_0)\vert^2}{4}   \begin{aligned}[t] \Big\lgroup & \t^{\pm 1/2}  \U{\F}(z_0) \e^{\Rzeroz{z_0,u,v}+\Rtwoz{z_0} + \Rthreez{z_0} + \Rsix } \\ &- \t^{\mp 1/2}  \U{\F}(\cc{z}_0) \e^{\Rzeroz{\cc{z}_0,v,u}+\Rtwoz{\cc{z}_0} + \Rthreez{\cc{z}_0} + \Rsixz{\cc{z}_0,u}  }  \Big\rgroup \end{aligned} \nn \\
& \quad = \tfrac{\sqrt{\n}}{\sqrt{2\pi^3}} \e^{ \n\left(  \f(z_0)+u \hu(\cc{z}_0)+v \hu(z_0)+ u^2 \huu(\cc{z}_0)+v^2 \huu(z_0) +uv \huv \right) } \nn \\
& \qquad \cdot \left(1+\Rxi{\pm}{(\t,z_0,u,v)}\right) \left\{\begin{array}{l} \gw(z_0) \\ \gz(z_0), \end{array}\right. 
\end{align}
where we have defined the error term
\begin{align}
\Rxi{\pm}{(\t,z_0,u,v)} &= \tfrac{\t^{\pm 1/2} \U{\F}(z_0)\left(\e^{\Rsum{(\t,z_0,u,v)}}-1 \right) - \t^{\mp 1/2} \U{\F}(\cc{z}_0)\left(\e^{\Rsum{(\t,\cc{z}_0,v,u)}}-1 \right) }{\t^{\pm 1/2} \U{\F}(z_0)-\t^{\mp 1/2} \U{\F}(\cc{z}_0)},
\end{align}
with
\begin{align}
\Rsum{(\t,z_0,u,v)}=&\Rzeroz{z_0,u,v}+\Rone+\Rtwoz{z_0}+\Rthreez{z_0} \nn \\
& +\n \left(\Rfourz{\cc{z}_0,u}+\Rfour \right) \nn \\
& +\Rfivez{\cc{z}_0,u}+\Rfive+\Rsix.
\end{align}
\end{proof}


\begin{prop} \label{prop_Rxi}
Let $\epsilon>0$, $\delta>0$, $0<\delta_0<1$, $A>0$. Then there exists $\n_0=\n_0(\epsilon,\delta,\delta_0,A) \in\N$ so that
\begin{align}
\sup_{\substack{0<t<1-\delta_0 \\ \n>\n_0 }} \sup_{\substack{ z_0\in\Adelta \\ \xi_1,\xi_2,\xi_3\in [0,\frac{1}{\n}] \\ u,v,\xi_4,\xi_5,\xi_6\in \cc{B}_{\!\frac{A}{\sqrt{\n}}}(0) }} \left\vert \Rxi{\pm}{(\t,z_0,u,v)} \right\vert <\epsilon,
\end{align}
\end{prop}
where $\Rxi{\pm}$ is as in Proposition \ref{prop_asym_wz_R}.
\begin{proof}
Without loss of generality $\epsilon<1$. Let $\epsilon_1=\tfrac{\epsilon \delta_0}{36}$ and $\F_0=\sqrt{\tfrac{4(1-\delta_0)}{2\delta_0-\delta_0^2}}$ be the largest value of $\F$. 

We can now find an estimation for $\vert \Rsum \vert $, like in the proof of Proposition \ref{prop_estim_R} for $\vert \Rzerotz{}+\Rone{}+\Rtwotz{}+\Rthreetz{} \vert$. For $z_0\in \Adelta$ and $u,v\in \cc{B}_\frac{A}{\sqrt{\n}}(0)\subset \cc{B}_\frac{\delta}{2}(0)$ it follows that $\sqrt{2}\tfrac{z_0+u}{\F}$, $\sqrt{2}\tfrac{z_0+v}{\F}$, $\sqrt{\tfrac{2\n}{\n-1}}\tfrac{z_0+u}{\F}$, $\sqrt{\tfrac{2\n}{\n-1}}\tfrac{z_0+v}{\F}$ lie in $\Xdelta{A}{\tilde{\delta}}{\F}$ for all $\n>\tfrac{4A^2}{\delta^2}$ and $\tilde{\delta}=\tfrac{\sqrt{2}\delta}{\F_0}$. Therefore we know from Section \ref{subsection_pr_asym_outside} that the Plancherel-Rotach asymptotics holds uniformly for $z_0\in \Adelta$ and $u,v\in \cc{B}_\frac{A}{\sqrt{\n}}(0)$. Therefore there exists $\n_1\in\N$ so that, like in the proof of Proposition \ref{prop_estim_R}, 
\begin{align}
\left\vert \Rzeroz{z_0,u,v}\right\vert =\left \vert \log\left(  \tfrac{Q_\n\left(\sqrt{2} \frac{\cc{z}_0+u}{\F} \right)}{\pi_\n\left(\sqrt{2} \frac{\cc{z}_0+u}{\F} \right)}   \tfrac{Q_{\n-1}\left(\sqrt{\frac{\n}{\n-1}}\sqrt{2}\frac{z_0+v}{\F} \right)}{\pi_{\n-1}\left(\sqrt{\frac{\n}{\n-1}}\sqrt{2}\frac{z_0+v}{\F}\right)}   \right) \right\vert<\epsilon_1,
\end{align}
for all $\n>\n_1$, $0<\t<1-\delta_0$, $z_0\in\Adelta$ and $u,v\in \cc{B}_\frac{A}{\sqrt{\n}}(0)$.

Again $ \vert \Rone \vert =\tfrac{1}{12} \xi_1(\n) < \epsilon_1$ if $\n>\tfrac{1}{12\epsilon_1}$.

From Lemma \ref{lemma_estim_R2_R3} we know that there exists $\n_2\in\N$ so that for all $\n>\n_2$, $0<\t<1-\delta_0$, $z_0\in\Adelta$ and $\xi_2(\n),\xi_3(\n)\in [0,\tfrac{1}{2}]$
\begin{align}
\vert \Rtwoz{z_0} \vert &<\epsilon_1, & \vert \Rthreez{z_0} \vert &<\epsilon_1, \nn \\
\intertext{and additionally for all $v,\xi_4(v),\xi_5(v),\xi_6(v) \in \cc{B}_\frac{A}{\sqrt{\n}}(0)$, according to Lemma \ref{lemma_estim_R4_R6}, }
\vert \n \Rfour \vert &<\epsilon_1, & \vert \Rfive \vert &<\epsilon_1, & \vert \Rsix\vert < \epsilon_1. \nn
\end{align}

We can choose $\n_0> \max\{\n_1, \n_2, \tfrac{1}{12\epsilon_1}, \tfrac{4A^2}{\delta^2} \}$. Then it follows that $\vert \Rsum(t,z_0,u,v)\vert < 9\epsilon_1=\tfrac{\epsilon \delta_0}{4}$ for all $\n>\n_0$ and thus
\begin{align}
\left \vert \e^{\Rsum(t,z_0,u,v)}-1 \right\vert< \tfrac{\epsilon \delta_0}{2}, \qquad \text{and} \qquad \left \vert \e^{\Rsum(t,\cc{z}_0,v,u)}-1 \right\vert< \tfrac{\epsilon \delta_0}{2}, \label{eq_prop_Rxi1}
\end{align}
uniformly for $0<\t<1-\delta_0$, $z_0\in\Adelta$, $\xi_1(\n),\xi_2(\n),\xi_3(\n)\in [0,\tfrac{1}{\n}]$ and $v,u,\xi_4,\xi_5,\xi_6 \in \cc{B}_\frac{A}{\sqrt{\n}}(0)$.

Since $\U{\F}(\cc{z}_0)=\cc{\U{\F}(z_0)}$, we find for the absolute value squared of the denominator of $\Rxi{\pm}$
\begin{align}
&\left \vert \t^{\pm 1/2} \U{\F}(z_0) - \t^{\mp 1/2} \U{\F}(\cc{z}_0) \right\vert^2 \nn \\
&\qquad = \left \vert \Re(\U{\F}(z_0)) \left( \t^{\pm 1/2} - \t^{\mp 1/2} \right) +i\Im(\U{\F}(z_0))  \left( \t^{\pm 1/2} + \t^{\mp 1/2} \right)  \right\vert^2 \nn \\
&\qquad = (\Re\, \U{\F}(z_0))^2 \left( \t^{\pm 1/2} - \t^{\mp 1/2} \right)^2+ (\Im\, \U{\F}(z_0))^2 \left( \t^{\pm 1/2} + \t^{\mp 1/2} \right)^2.
\end{align}
This we can estimate as
\begin{align}
&(\Re\, \U{\F}(z_0))^2 \left( \t^{\pm 1/2} - \t^{\mp 1/2} \right)^2+ (\Im\, \U{\F}(z_0))^2 \left( \t^{\pm 1/2} + \t^{\mp 1/2} \right)^2 \nn \\
&\qquad > \left( (\Re\, \U{\F}(z_0))^2 + (\Im\, \U{\F}(z_0))^2 \right) \left( \t^{\pm 1/2} - \t^{\mp 1/2} \right)^2.
\end{align}
So we get for $z_0=a \cos\phi+i b\sin\phi$ with $a^2=b^2+F^2$ using Proposition \ref{prop_U}
\begin{align}
\left \vert \t^{\pm 1/2} \U{\F}(z_0) - \t^{\mp 1/2} \U{\F}(\cc{z}_0) \right\vert^2 > r_{\F}^2(b) \left( \t^{\pm 1/2} - \t^{\mp 1/2} \right)^2, \label{eq_prop_Rxi2}
\end{align}
with $r_{\F}(b)=\tfrac{a-b}{\F}$.

The numerator of $\Rxi{\pm}$ we can estimate as
\begin{align}
&\left\vert \t^{\pm 1/2} \U{\F}(z_0)\left(\e^{\Rsum{(\t,z_0,u,v)}}-1 \right) - \t^{\mp 1/2} \U{\F}(\cc{z}_0)\left(\e^{\Rsum{(\t,\cc{z}_0,v,u)}}-1 \right) \right\vert \nn \\
&\qquad < 2\t^{- 1/2} \left\vert \U{\F}(z_0)\left(\e^{\Rsum{(\t,z_0,u,v)}}-1 \right) \right\vert < \epsilon\delta_0 \t^{- 1/2} \left\vert \U{\F}(z_0) \right\vert , \label{eq_prop_Rxi3}
\end{align}
where we have used \eqref{eq_prop_Rxi1} in the last step. Using that
\begin{align}
\tfrac{\t^{-1}}{\left( \t^{\pm 1/2} - \t^{\mp 1/2} \right)^2}=\tfrac{1}{(1-\t)^2}<\tfrac{1}{\delta_0^2}, \qquad \forall \t\in(0,1-\delta_0), \label{eq_prop_Rxi4}
\end{align}
and that $\vert\U{\F}(z_0)\vert^2=r_{\F}^2(b)$ according to Proposition \ref{prop_U}, we find for $\Rxi{\pm}\!$ with \eqref{eq_prop_Rxi2} and \eqref{eq_prop_Rxi3} that finally
\begin{align}
\left\vert \Rxi{\pm}{(\t,z_0,u,v)} \right\vert <\epsilon,
\end{align}
uniformly for all $\t\in(0,1-\delta_0)$, $z_0\in\Adelta$, $\xi_1(\n),\xi_2(\n),\xi_3(\n)\in [0,\tfrac{1}{\n}]$ and $v,u,\xi_4,\xi_5,\xi_6 \in \cc{B}_\frac{A}{\sqrt{\n}}(0)$.
\end{proof}

\begin{prop}\label{prop_estim_A}
Let $\delta>0$, $0<\delta_0<1$, $\delta_1>0$ and $A\in (0,\infty)$. Then there exists $\n_0=\n_0(\delta_0,\delta_1,A)\in\N$, $\epsilon=\epsilon(\delta,\delta_0)>0$, $C=C(\delta_0,\delta_1)>0$ so that for all $\n > \n_0$
\begin{align}
\sup_{\t \in(0,1-\delta_0)}\sup_{\substack{z_0 \in \Xdelta{A}{\delta_1}{\F}\setminus \Edelta{\delta} \\ u,v \in \cc{B}_{\! \frac{A}{\sqrt{\n}} }(0) }} \left\vert  \left.\frac{\partial \HM}{\partial w}\right\vert_{\substack{\phantom{}_{w=\cc{z}_0+u} \\ \phantom{}_{z=z_0+v} \\ \phantom{}}}   \right\vert&< C\sqrt{\n} \e^{-\n \epsilon},\\
\shortintertext{and}
\sup_{\t \in(0,1-\delta_0)}\sup_{\substack{z_0 \in \Xdelta{A}{\delta_1}{\F}\setminus \Edelta{\delta} \\ u,v \in \cc{B}_{\! \frac{A}{\sqrt{\n}} }(0) }} \left\vert  \left.\frac{\partial \HM}{\partial z}\right\vert_{\substack{\phantom{}_{w=\cc{z}_0+u} \\ \phantom{}_{z=z_0+v} \\ \phantom{}}}   \right\vert&< C\sqrt{\n} \e^{-\n \epsilon}. 
\end{align}
\end{prop}
\begin{proof}
Let $\F_0=\sqrt{\tfrac{4(1-\delta_0)}{2\delta_0-\delta_0^2}}$ be the largest value of $\F$. 

We know there exists $\tilde{\n}_0\in \N$ so that according to Proposition \ref{prop_asym_wz_R}
\begin{align}
\left.\frac{\partial \HM}{\partial w}\right\vert_{\substack{\phantom{}_{w=\cc{z}_0+u} \\ \phantom{}_{z=z_0+v} \\ \phantom{}}} =&\sqrt{\tfrac{\n}{2 \pi^3}} \gw(z_0) \e^{\n \left( \f(z_0) +\hu(\cc{z}_0)u+\huu(\cc{z}_0)u^2 +\hu(z_0)v+\huu(z_0)v^2+\huv uv \right)} \nn \\ 
& \cdot \left(1+\Rxi{+}{(\t,z_0,u,v)}\right), \\
\left.\frac{\partial \HM}{\partial z}\right\vert_{\substack{\phantom{}_{w=\cc{z}_0+u} \\ \phantom{}_{z=z_0+v} \\ \phantom{}}} =&\sqrt{\tfrac{\n}{2 \pi^3}} \gz(z_0) \e^{\n \left( \f(z_0) +\hu(\cc{z}_0)u+\huu(\cc{z}_0)u^2 +\hu(z_0)v+\huu(z_0)v^2+\huv uv \right)} \nn \\
& \cdot \left(1+\Rxi{-}{(\t,z_0,u,v)}\right),
\end{align}
and according to Proposition \ref{prop_Rxi}
\begin{align}
\left\vert \Rxi{\pm}(t,z_0,u,v)\right\vert < 1,
\end{align}
for all $\n>\tilde{\n}_0$, $0<\t<1-\delta_0$, $z_0\in \Xdelta{A}{\delta_1}{\F}$, $\xi_1(\n),\xi_2(\n),\xi_3(\n)\in [0,\tfrac{1}{\n}]$ and $v,u,\xi_4,\xi_5,\xi_6 \in \cc{B}_\frac{A}{\sqrt{\n}}(0)$ with $\f$, $\gw$, $\gz$, $\hu$, $\huu$, $\huv$, $\Rxi{\pm}$ and $(\xi_k)_{k=1}^6$ as in Proposition \ref{prop_asym_wz_R}.  

From Lemma \ref{lemma_estim_W} and \ref{lemma_estim_b} we get that
\begin{align}
\vert \W{\F}(z_0)\vert^2 & < 2\left(1+\sqrt{1+\tfrac{\F_0^2}{\delta_1^2}}\right),
\end{align}
and from Lemma \ref{lemma_estim_U_T} and \ref{lemma_estim_b} that
\begin{align}
\t^{\pm 1/2} \vert \U{\F}(z_0)\vert & \le \t^{- 1/2} \vert \U{\F}(z_0)\vert = \tfrac{2 \vert \U{\F}(z_0)\vert }{\sqrt{1-\t^2} \F } < \tfrac{1}{\sqrt{\delta_0} \delta_1},
\end{align}
and therefore
\begin{align}
\vert \gw(z_0) \vert&= \vert \gz(z_0) \vert = \tfrac{1}{4}\vert \W{\F}(z_0) \vert^2 \left\vert \t^{1/2} \U{\F}(z_0)-\t^{-1/2} \U{\F}(\cc{z}_0) \right\vert \nn \\
& \le \tfrac{1}{2}\vert \W{\F}(z_0) \vert^2  \t^{-1/2} \left\vert \U{\F}(z_0) \right\vert < \left(1+\sqrt{1+\tfrac{\F_0^2}{\delta_1^2}}\right) \tfrac{1}{\sqrt{\delta_0} \delta_1},
\end{align}
for all $0<\t<1-\delta_0$ and $z_0\in \Xdelta{A}{\delta_1}{\F}$. Thus we find a finite constant $C>0$ so that 
\begin{align}
\left.\begin{array}{r} \left\vert\gw(z_0)\right\vert \\ \left\vert\gz(z_0)\right\vert \end{array} \right\} \sqrt{\tfrac{\n}{2 \pi^3}} \left\vert 1+\Rxi{+}{(\t,z_0,u,v)} \right\vert & < \sqrt{\tfrac{2 \n}{\pi^3}} \left(1+\sqrt{1+\tfrac{\F_0^2}{\delta_1^2}}\right) \tfrac{1}{\sqrt{\delta_0} \delta_1}=C\sqrt{\n},
\end{align}
for all $\n>\tilde{\n}_0$, $0<\t<1-\delta_0$, $z_0\in \Xdelta{A}{\delta_1}{\F}$, $\xi_1(\n),\xi_2(\n),\xi_3(\n)\in [0,\tfrac{1}{\n}]$ and $v,u,\xi_4,\xi_5,\xi_6 \in \cc{B}_\frac{A}{\sqrt{\n}}(0)$.

From Proposition \ref{prop_estim_f} we know that there exists $\epsilon>0$ so that
\begin{align}
\sup_{\substack{0<\t<1-\delta_0}} \sup_{\substack{z_0 \in \Xdelta{A}{\delta_1}{\F}\setminus \Edelta{\delta} }} \f(z_0) \le -4\epsilon.
\end{align}
We further know that there exists $R\in(0,\infty)$ so that $f(z_0)\le -\vert z_0 \vert$ for all $z_0\notin B_R(0)$ and $\t\in(0,1-\delta_0)$ (see \eqref{eq_f_large_z} in paragraph 4 of the proof of Proposition \ref{prop_unifrom_estimate_drho}). We can assume that $R>4\epsilon$.

We can find constants $C_1,C_2\in(0,\infty)$ so that for all $\t\in(0,1-\delta_0)$, $z_0\in \Xdelta{A}{\delta_1}{\F}$ and $u,v\in \cc{B}_\frac{A}{\sqrt{\n}}(0)$ we can estimate
\begin{align}
\vert \huu(\cc{z}_0)u^2 +\huu(z_0)v^2+\huv uv \vert &\le \tfrac{A^2}{\n } \left( \t + \tfrac{2\vert \U{\F}(z_0) \vert }{\F \vert \T{\F}(z_0)\vert } + 1 \right) < \tfrac{A^2}{\n } \left( 2 + \tfrac{1}{\delta_1^2}  \right)=\tfrac{C_1}{\n}, \\
\intertext{and}
\vert \hu(\cc{z}_0)u +\hu(z_0)v \vert &\le \tfrac{A}{\sqrt{\n} } \left(\tfrac{4 \vert \U{\F}(z_0) \vert}{\F}+2(1+\t)\vert z_0\vert  \right) \nn \\
&< \tfrac{A}{\sqrt{\n} } \left(\tfrac{2}{\delta_1}+4 \vert z_0\vert \right) \le \tfrac{C_2}{\sqrt{\n}}+  \left\{ \begin{array}{ll} \tfrac{4AR}{\sqrt{\n}}, & \text{if }\vert z_0\vert<R, \\ \tfrac{4A \vert z_0\vert }{\sqrt{\n}}, & \forall z_0\in \Xdelta{A}{\delta_1}{\F}, \end{array}\right.
\end{align}
where we have used Lemma \ref{lemma_estim_U_T} together with Lemma \ref{lemma_estim_b} and set $C_2=\tfrac{2A}{ \delta_1 }$. If $\n>\max\{ \tfrac{C_1}{\epsilon}, \tfrac{C_2^2}{\epsilon^2}, \tfrac{16 A^2 R^2}{\epsilon^2}  \} $, we find so that
\begin{align}
\left\vert \e^{\n \left( \f(z_0) +\hu(\cc{z}_0)u+\huu(\cc{z}_0)u^2 +\hu(z_0)v+\huu(z_0)v^2+\huv uv \right)} \right\vert \le \e^{-\n (4\epsilon - 3\epsilon) }=\e^{-\n \epsilon },
\end{align}
for all $0<\t<1-\delta_0$, $z_0\in (\Xdelta{A}{\delta_1}{\F}\setminus \Edelta{\delta}) \cap B_R(0)$ and $v,u \in \cc{B}_\frac{A}{\sqrt{\n}}(0)$.

If $\n>\max\{ \tfrac{C_1}{\epsilon}, \tfrac{C_2^2}{\epsilon^2}, \tfrac{16 A^2 R^2}{\epsilon^2}  \} $ it follows also that $\tfrac{C_1}{\n}<\tfrac{R}{4}$, $\tfrac{C_2}{\sqrt{\n}}<\tfrac{R}{4}$ and $\tfrac{4A}{\sqrt{\n}}<\tfrac{1}{4}$ since $R>4\epsilon$. We get additionally that
\begin{align}
\left\vert \e^{\n \left( \f(z_0) +\hu(\cc{z}_0)u+\huu(\cc{z}_0)u^2 +\hu(z_0)v+\huu(z_0)v^2+\huv uv \right)} \right\vert \le \e^{-\n (\vert z_0\vert - \frac{R}{2} - \frac{\vert z_0\vert}{4}) } \le \e^{-\n \frac{R}{4}} \le \e^{-\n \epsilon},
\end{align}
for all $0<\t<1-\delta_0$, $z_0\in \C \setminus B_R(0)$ and $v,u \in \cc{B}_\frac{A}{\sqrt{\n}}(0)$.

Therefore the Proposition is proven if we choose 
$\n_0>\max \{ \tilde{\n}_0, \tfrac{C_1}{\epsilon}, \tfrac{C_2^2}{\epsilon^2}, \tfrac{16 A^2 R^2}{\epsilon^2} \}$.
\end{proof}

\subsubsection{\sloppy Plancherel-Rotach Asymptotics on $\Bdelta$ and Maximum Principle on $\Cdelta$}
\fussy
In the following we will prove a more precise and generalized version of Proposition \ref{prop_H_B}.
\begin{prop}\label{prop_estim_B_C}
Let $0<\delta_0<1$ and $A\in (0,\infty)$. Then there exists $\n_0=\n_0(\delta_0)$, $\delta=\delta(\delta_0)>0$, $\epsilon=\epsilon(\delta_0)>0$, $C=C(\delta_0)>0$ so that for all $\n > \n_0$
\begin{align}
\sup_{\t \in(0,1-\delta_0)}\sup_{\substack{z_0 \in \Bdelta \cup \Cdelta \\ u,v \in \cc{B}_{\! \frac{A}{\sqrt{\n}} }(0) }} \left\vert  \left.\frac{\partial \HM}{\partial w}\right\vert_{\substack{\phantom{}_{w=\cc{z}_0+u} \\ \phantom{}_{z=z_0+v} \\ \phantom{}}}   \right\vert&< C\sqrt{\n} \e^{-\n \epsilon},\\
\shortintertext{and}
\sup_{\t \in(0,1-\delta_0)}\sup_{\substack{z_0 \in \Bdelta \cup \Cdelta \\ u,v \in \cc{B}_{\! \frac{A}{\sqrt{\n}} }(0) }} \left\vert  \left.\frac{\partial \HM}{\partial z}\right\vert_{\substack{\phantom{}_{w=\cc{z}_0+u} \\ \phantom{}_{z=z_0+v} \\ \phantom{}}}   \right\vert&< C\sqrt{\n} \e^{-\n \epsilon}. 
\end{align}
\end{prop}
\begin{proof}
Let $\epsilon=\min\left\{\tfrac{\delta_0^3}{24}, -\tfrac{\delta_0}{4-2\delta_0}-\tfrac{\log(1-\delta_0)}{2} \right\}>0$ and without loss of generality $\delta_0<\e^{-2}$. We define $\delta_1=\e^{-2}-\delta_0$. 

According to Proposition \ref{prop_using_maximum_principle} we get $\tilde{\delta}>0$ so that for all $\t \in [\delta_1,1-\delta_0]$ 
\begin{align}
\sup_{z_1,z_2\in \Xdelta{A}{2\tilde{\delta}}{\F}^C} \left\vert\op{\n}(z_1)\op{\n-1}(z_2)\right\vert & < g_1(\n,\t,\zmax)g_2(\n,\t,\zmax) \e^{\n\left(2+\log \t+\epsilon \right)},
\intertext{and for all $\t \in (0,\delta_1)$ }
\sup_{z_1,z_2\in \Xdelta{A}{2\tilde{\delta}}{\F}^C} \left\vert\op{\n}(z_1)\op{\n-1}(z_2)\right\vert & < g_1(\n,\t,\zmax)g_2(\n,\t,\zmax) \e^{-\n\delta_0} \nn \\
&< g_1(\n,\t,\zmax)g_2(\n,\t,\zmax) \e^{-2\n\epsilon},
\end{align}
where $g_1$, $g_2$ are as in Proposition \ref{prop_using_maximum_principle} and $\zmax=\F+2(1+i)\tilde{\delta}$. 

We set $w=\cc{z}_0+u$ and $z=z_0+v$ in identities \eqref{eq_id3} and \eqref{eq_id4}. Then for all $\t\in (0,1)$ and $\cc{z}_0+u,z_0+v \in \Xdelta{A}{2\tilde{\delta}}{\F}^C$
\begin{align}
\left\vert  \left.\frac{\partial \HM}{\partial w}\right\vert_{\substack{\phantom{}_{w=\cc{z}_0+u} \\ \phantom{}_{z=z_0+v} \\ \phantom{}}}   \right\vert& \le 2 \sup_{z_1,z_2\in \Xdelta{A}{2\tilde{\delta}}{\F}^C} \left\vert\op{\n}(z_1)\op{\n-1}(z_2)\right\vert \left. \tfrac{\e^{\n\Re\left(-wz+\t\frac{w^2+z^2}{2} \right) } }{\sqrt{1-\t^2}} \right\vert_{\substack{\phantom{}_{w=\cc{z}_0+u} \\ \phantom{}_{z=z_0+v} \\ \phantom{}}}, \label{eq_estim_B_C1}
\intertext{and}
\left\vert  \left.\frac{\partial \HM}{\partial z}\right\vert_{\substack{\phantom{}_{w=\cc{z}_0+u} \\ \phantom{}_{z=z_0+v} \\ \phantom{}}}   \right\vert& \le 2 \sup_{z_1,z_2\in \Xdelta{A}{2\tilde{\delta}}{\F}^C} \left\vert\op{\n}(z_1)\op{\n-1}(z_2)\right\vert \left. \tfrac{\e^{\n\Re\left(-wz+\t\frac{w^2+z^2}{2} \right) } }{\sqrt{1-\t^2}} \right\vert_{\substack{\phantom{}_{w=\cc{z}_0+u} \\ \phantom{}_{z=z_0+v} \\ \phantom{}}}. \label{eq_estim_B_C2}
\end{align}

If $\delta\le \tilde{\delta}$ and $\n>\tfrac{A^2}{\tilde{\delta}^2}$, then it follows for $z_0\in \Adelta^C$ and $u,v\in B_{\!\frac{A}{\sqrt{\n}}}(0)$ that $z_0+u$ and $\cc{z}_0+v$ lie in $\Xdelta{A}{2\tilde{\delta}}{\F}^C$.

Since $\tfrac{1}{1-\t^2}<\tfrac{1+ \t }{1- \t}$, we can define the finite constant $\tilde{C}$ as in the proof of Proposition \ref{prop_estim_max_princ} in \eqref{eq_proof_estim_mp4b} so that
\begin{align}
\frac{2g_1(\n,\t,\zmax) g_2(\n,\t,\zmax)}{\sqrt{1-\t^2}} < \frac{\sqrt{\n} \left( 1+\sqrt{1+\frac{1}{\delta_0 \tilde{\delta}^2}} \right) }{\sqrt{\pi^3 \delta_0}\tilde{\delta}}=\sqrt{\n} \tilde{C},
\end{align}
for all $\n>\tilde{\n}_0$ with $\tilde{\n}_0$ big enough.

Next we look at the exponential term
\begin{align}
\left. \e^{\n\Re\left(-wz+\t\frac{w^2+z^2}{2} \right) }  \right\vert_{\substack{\phantom{}_{w=\cc{z}_0+u} \\ \phantom{}_{z=z_0+v} \\ \phantom{}}}=\e^{\n\left(-\vert z_0\vert^2+\t\Re(z_0^2)+\Re\left( u(\t \cc{z}_0-z_0)+v(\t z_0-\cc{z}_0)+\frac{u^2\t}{2}+ \frac{v^2\t}{2}-uv \right) \right)}.
\end{align}
We already have seen that we need $\sqrt{\n}>\tfrac{A}{\tilde{\delta}}$. For $\t\in (0,1-\delta_0]$, $z_0\in \Xdelta{A}{\tilde{\delta}}{\F}^C$ and $u,v\in \cc{B}_{\!\frac{A}{\sqrt{\n}} }$ we find the estimation
\begin{align}
\vert u(\t \cc{z}_0-z_0)+v(\t z_0-\cc{z}_0)+\tfrac{u^2\t}{2}+ \tfrac{v^2\t}{2}-uv\vert &\le \tfrac{2A\vert \t z_0-\cc{z}_0\vert}{\sqrt{\n}} +\tfrac{A^2(1+\t)}{\n} \le \tfrac{ 2A(\F_0+2\tilde{\delta}) }{\sqrt{\n}}+\tfrac{2A^2}{\n}\nn \\
&<\tfrac{2A(\F_0+3\tilde{\delta})}{\sqrt{\n}}<\epsilon
\end{align} 
if $\n>\max\left\{\tfrac{4A^2 (\F_0+3\tilde{\delta})^2}{\epsilon^2},\tfrac{A^2}{\tilde{\delta}^2} \right\}$ with $\F_0=\sqrt{\tfrac{4(1-\delta_0)}{1-(1-\delta_0)^2}}< \tfrac{2}{\sqrt{\delta_0}}$.

Since $-\vert z_0\vert^2+\t\Re(z_0^2)\le 0$, thus for all $\t\in(0,\delta_1)$, $z_0\in \Adelta^C$ and $u,v\in \cc{B}_{\!\frac{A}{\sqrt{\n}} }(0)$ the right-hand side of \eqref{eq_estim_B_C1} and \eqref{eq_estim_B_C2} is smaller than $\sqrt{\n}C\e^{-\n\epsilon}$ if $\delta\le\tilde{\delta}$, $C\ge \tilde{C}$ and $\n> \max\left\{\tfrac{4A^2 (\F_0+3\tilde{\delta})^2}{\epsilon^2},\tfrac{A^2}{\tilde{\delta}^2}, \tilde{\n}_0 \right\}$.

We are left to consider $\t\in [\delta_1,1-\delta_0]$. We already have shown in the proof of Proposition \ref{prop_estim_max_princ} that there exists $\lambda>0$ so that for all $\t\in [\delta_1,1-\delta_0]$
\begin{align}
\sup_{z\in \Xdelta{C}{\tilde{\delta}}{\F,\lambda}} \left(-\vert z\vert^2 +\t \Re(z^2) + 2 +\log \t+\epsilon \right)<-\epsilon,
\end{align}
where we had set $\epsilon=\tfrac{\delta_0^3}{12}$. Thus for $\epsilon\le\tfrac{\delta_0^3}{24}$, we get that 
\begin{align}
\sup_{\substack{z_0\in \Xdelta{C}{\tilde{\delta}}{\F,\lambda} }} \sup_{ u,v \in \cc{B}_{\!\frac{A}{ \sqrt{\n} } }(0)} & \Bigg(-\vert z_0\vert^2 +\t \Re(z_0^2) + 2 +\log \t+\epsilon \nn \\
& \quad + \Re\left( u(\t \cc{z}_0-z_0)+v(\t z_0-\cc{z}_0)+\tfrac{u^2\t}{2}+ \tfrac{v^2\t}{2}-uv \right)  \Bigg) \nn \\
& < \sup_{z_0\in \Xdelta{C}{\tilde{\delta}}{\F,\lambda}} \left(-\vert z_0\vert^2 +\t \Re(z_0^2) + 2 +\log \t+2\epsilon \right)<-2\epsilon,
\end{align}
for all $\n>\max\left\{\tfrac{4A^2 (\F_0+3\tilde{\delta})^2}{\epsilon^2},\tfrac{A^2}{\tilde{\delta}^2} \right\}$ and $\t\in [\delta_1,1-\delta_0]$.

For the $\lambda$ we have found we get $\delta>0$ and $\n_0\in\N$ according to Proposition \ref{prop_pr_B} so that for all $\n>\n_0$, $\t \in [\delta_1,1-\delta_0]$ and $w,z\in \Xdelta{B}{2\delta}{\F,\frac{\lambda}{2}}$
\begin{align}
\left\vert\op{\n}(w)\op{\n-1}(z)\right\vert+\left\vert\op{\n-1}(w)\op{\n}(z)\right\vert & < 2g_1^r(\n,\t)g_2^r(\n,\t) \e^{\n\left(\frac{\Re(w^2+z^2)}{\F^2}+\log \t+\epsilon \right)},
\end{align}
where 
\begin{align}
g_1^r(\n,\t)&= \biggl( 2\sup_{z \in \Xdelta{B}{2\delta}{\F,\frac{\lambda}{2}} }\left\vert \tfrac{\F+z}{\F-z} \right\vert^{1/4} +1 \biggr)  g_0^r(\n,\t),\\
g_2^r(\n,\t)&= \biggl( 2\sup_{z \in \Xdelta{B}{2\delta}{\F,\frac{\lambda}{2}} }\left\vert \tfrac{\F\sqrt{1-\frac{1}{\n}}+z}{\F\sqrt{1-\frac{1}{\n}}-z} \right\vert^{1/4} +1 \biggr) g_0^r(\n,\t)\left(\tfrac{\n-1}{\n} \right)^{\frac{\n-1}{2}} \e^{\frac{1-\log 2}{2}},\\
\end{align}
and $g_0^r$ as in Proposition \ref{prop_pr_B}. We can assume that $\delta<\tilde{\delta}$ and
\begin{align}
\n_0>\max\left\{\tilde{\n}_0,\tfrac{4\F_0}{\lambda}, \tfrac{4A^2 (\F_0+3\tilde{\delta})^2}{\epsilon^2},\tfrac{A^2}{\delta^2},  \tfrac{ \left(2A(\F_0+2\delta)\left(1+\frac{2}{\F_1^2}\right) + 2\delta A\left(1+\frac{1}{\F_1^2}\right) \right)^2}{\epsilon^2}  \right\}, 
\end{align} 
with $\F_1=\sqrt{2\delta_1}$ so that $\F_1<\F$ for all $\t\in [\delta_1,1-\delta_0]$. 

Again we use identities \eqref{eq_id3} and \eqref{eq_id4} evaluated at $w=\cc{z}_0+u$ and $z=z_0+v$ and find that for all $\t\in (0,1)$ and $\cc{z}_0+u,z_0+v \in \Xdelta{A}{2\delta}{\F}^C$
\begin{align}
\left\vert  \left.\frac{\partial \HM}{\partial w}\right\vert_{\substack{\phantom{}_{w=\cc{z}_0+u} \\ \phantom{}_{z=z_0+v} \\ \phantom{}}}   \right\vert& \le  \left. \bigl( \left\vert\op{\n}(w)\op{\n-1}(z)\right\vert+\left\vert\op{\n-1}(w)\op{\n}(z)\right\vert \bigr) \tfrac{\e^{\n\Re\left(-wz+\t\frac{w^2+z^2}{2} \right) } }{\sqrt{1-\t^2}} \right\vert_{\substack{\phantom{}_{w=\cc{z}_0+u} \\ \phantom{}_{z=z_0+v} \\ \phantom{}}}, \label{eq_estim_B_C3}
\intertext{and}
\left\vert  \left.\frac{\partial \HM}{\partial z}\right\vert_{\substack{\phantom{}_{w=\cc{z}_0+u} \\ \phantom{}_{z=z_0+v} \\ \phantom{}}}   \right\vert& \le  \left. \bigl( \left\vert\op{\n}(w)\op{\n-1}(z)\right\vert+\left\vert\op{\n-1}(w)\op{\n}(z)\right\vert \bigr)  \tfrac{\e^{\n\Re\left(-wz+\t\frac{w^2+z^2}{2} \right) } }{\sqrt{1-\t^2}} \right\vert_{\substack{\phantom{}_{w=\cc{z}_0+u} \\ \phantom{}_{z=z_0+v} \\ \phantom{}}}. \label{eq_estim_B_C4}
\end{align}

If $\n>\n_0>\max\left\{\tfrac{A^2}{\delta^2},\tfrac{4A^2}{\lambda^2} \right\}$ then it follows for $z_0\in \Xdelta{B}{\delta}{\F,\lambda}$ and $u,v\in \cc{B}_{\!\frac{A}{\sqrt{\n}}}(0)$ that $\cc{z}_0+u$ and $z_0+v$ lie in $\Xdelta{B}{2\delta}{\F,\frac{\lambda}{2}}$.

As in the proof of Proposition \ref{prop_estim_pr_B} we have an estimation like \eqref{eq_proof_estim_pr_B2} so that we can find a finite constant $C>\tilde{C}$ so that for all $\n>\tfrac{4\F_0}{\lambda}$
\begin{align}
\tfrac{2}{\sqrt{1-\t^2}} g_1^r(\n,\t) g_2^r(\n,\t)& \le \sqrt{\n} \tfrac{2\sqrt{2}  \e^{\frac{1-\log 2}{2}}}{\sqrt{\pi^3}} \left( 2\left(\tfrac{4\F_0}{\lambda} \right)^{1/4} +1 \right)^2<\sqrt{\n} C. 
\end{align}

The remaining term is the exponential term
\begin{align}
\left. \e^{\n\Re\left(-wz+\t\frac{w^2+z^2}{2} + \frac{w^2+z^2}{\F^2} + \log \t+ \epsilon \right) }  \right\vert_{\substack{\phantom{}_{w=\cc{z}_0+u} \\ \phantom{}_{z=z_0+v} \\ \phantom{}}}=\e^{\n\left(-\vert z_0\vert^2+ \t\Re(z_0^2)+\frac{2 \Re(z_0^2)}{\F^2} + \log \t+ \epsilon \right)} \nn \\
\cdot \e^{\n \Re\left( u(\t\cc{z}_0-z_0)+v(\t z_0-\cc{z}_0)-uv+\frac{2u \cc{z}_0+2v z_0}{\F^2}+\left(\frac{\t}{2}+\frac{1}{\F^2}\right)(u^2+v^2)  \right) }. \label{eq_estim_B_C5}
\end{align}
We define
\begin{align}
\f^r(z_0)=-\vert z_0\vert^2 +\t \Re(z_0^2) + \tfrac{2\Re(z_0^2)}{\F^2} +\log \t,
\end{align}
for which we find, like in the proof of Proposition \ref{prop_estim_pr_B}, that
\begin{align}
\sup_{\t\in [\delta_1,1-\delta_0]} \sup_{ z_0\in \Xdelta{B}{\delta}{\F,\lambda} } \f^r(z_0)<-4\epsilon,
\end{align}
since $\epsilon\le -\tfrac{\delta_0}{4-2\delta_0}-\tfrac{\log(1-\delta_0)}{2}$ here is at most half as large as $\epsilon$ in that proof. 

For all $\t\in [\delta_1,1-\delta_0]$, $z_0\in \Xdelta{B}{\delta}{\F,\lambda}$, $u,v\in \cc{B}_{\!\frac{A}{\sqrt{\n}}}(0) $ and   
\begin{align}
\n > \max \left\{\tfrac{ \left(2A(\F_0+2\delta)\left(1+\frac{2}{\F_1^2}\right) + 2\delta A\left(1+\frac{1}{\F_1^2}\right) \right)^2}{\epsilon^2}, \tfrac{A^2}{\delta^2} \right\},
\end{align}
we have the estimation
\begin{align}
&\Re\left( u(\t\cc{z}_0-z_0)+v(\t z_0-\cc{z}_0)-uv+\tfrac{2u \cc{z}_0+2v z_0}{\F^2}+\left(\tfrac{\t}{2}+\tfrac{1}{\F^2}\right)(u^2+v^2) \right)\nn \\
&\qquad \le \tfrac{2A(\F_0+2\delta)}{\sqrt{\n}}+\tfrac{A^2}{\n}+\tfrac{4A (\F_0+\delta)}{\sqrt{\n}\F_1^2}+\tfrac{2A^2}{\n}\left(\tfrac{1}{2}+\tfrac{1}{\F_1^2}\right) \nn \\
&\qquad  < \tfrac{2A(\F_0+2\delta)\left(1+\frac{2}{\F_1^2}\right) + 2\delta A\left(1+\frac{1}{\F_1^2}\right)}{\sqrt{\n}} < \epsilon.
\end{align} 
Now it is easy to see that \eqref{eq_estim_B_C5} is smaller than $\e^{- 2\n\epsilon}$ and therefore the right-hand side of \eqref{eq_estim_B_C3} and \eqref{eq_estim_B_C4} is surely smaller than $\sqrt{\n}C\e^{- \n\epsilon}$ for all $\n>\n_0$, $\t\in [\delta_1,1-\delta_0]$, $z_0\in \Xdelta{B}{\delta}{\F,\lambda}$ and $u,v\in \cc{B}_{\!\frac{A}{\sqrt{\n}} }(0)$.

Since $\Bdelta \cup \Cdelta \subset \Xdelta{A}{\tilde{\delta}}{\F}^C$ and $\Bdelta \cup \Cdelta \subset \Xdelta{B}{\delta}{\F,\lambda} \cup \Xdelta{C}{\tilde{\delta}}{\F,\lambda}$, we have as well a uniform estimation for $\t\in [\delta_1,1-\delta_0]$ as for $\t\in (0,\delta_1)$ and so the proposition is proven.
\end{proof}

\subsection{Verification of Estimation of Kernel}

In the next proposition we will review Proposition \ref{prop_hm_wz}, Lemma \ref{lemma_hm_w_ellipse} and Proposition \ref{prop_hm_wz_ellipse}. We don't need to re-examine Corollary \ref{cor_hm_wz_ellipse} as there we haven't done any further approximations.
\begin{prop}\label{prop_veri_kernel}
Let $\delta>0$, $0<\delta_0<1$ and $A>0$. Then
\begin{align}
\lim_{\n\rightarrow\infty} \sup_{0<\t<1-\delta_0} \sup_{\substack{z_0 \in \Edelta{-,\delta} \\ a,b\in \cc{B}_{\!A}(0) }} & \left\vert \HM\left(\cc{z}_0+\tfrac{\cc{a}}{\sqrt{\n}},z_0+\tfrac{b}{\sqrt{\n}}\right) -\tfrac{1}{\pi}\right\vert = 0, \\
\lim_{\n\rightarrow\infty} \sup_{0<\t<1-\delta_0} \sup_{\substack{z_0\in \C\setminus \Edelta{+,\delta} \\ a,b\in \cc{B}_{\!A}(0) }} & \left\vert \HM\left(\cc{z}_0+\tfrac{\cc{a}}{\sqrt{\n}},z_0+\tfrac{b}{\sqrt{\n}}\right)  \right\vert = 0, \\
\lim_{\n\rightarrow\infty} \sup_{0<\t<1-\delta_0} \sup_{\substack{z_0\in \pEt \\ a,b\in \cc{B}_{\!A}(0) }} & \left\vert \HM\left(\cc{z}_0+\tfrac{\cc{a}}{\sqrt{\n}},z_0+\tfrac{b}{\sqrt{\n}}\right) - \frac{1}{2\pi}\erfc\left(\zeta(\cc{a},\phi)+ \zeta(b,-\phi) \right) \right\vert = 0,
\end{align}
with
\begin{equation}
\zeta(\cc{a},\phi)=\frac{\cc{a}}{\sqrt{2}}\frac{(1-\t)\cos\phi +(1+\t)i\sin\phi }{\sqrt{(1+\t)^2-4\t \cos^2\phi}},
\end{equation}
where $\phi$ is the parameter so that $z_0=\mas \cos \phi+i \mis \sin \phi$.
\end{prop}
\begin{proof}
We know from Proposition \ref{prop_veri_density} that $\HM(\cc{z}_0,z_0)$ converges uniformly for $\t\in (0,1-\delta_0]$ to $\tfrac{1}{\pi}$, $0$ or $\tfrac{1}{2\pi}$, uniformly in $z_0$ on $\Edelta{-,\delta}$, on $\C\setminus \Edelta{+,\delta}$ or on $\pEt$ respectively. Thus for the first and second claim we only have to show that 
\begin{align}
\int_0^{\frac{\cc{a}}{\sqrt{\n}}} \tfrac{\partial \HM}{\partial w}\left(\cc{z}_0+u,z_0\right)\dd u + \int_0^{\frac{b}{\sqrt{\n}}} \tfrac{\partial \HM}{\partial z}\left(\cc{z}_0+\tfrac{\cc{a} }{\sqrt{\n} }, z_0+v\right)\dd v
\end{align}
converges uniformly to zero as $\n\rightarrow \infty$. Combining Propositions \ref{prop_estim_A} and \ref{prop_estim_B_C} we get $\delta_1>0$, $\n_0\in\N$, $\C>0$ and $\epsilon>0$ so that for all $\n>\n_0$, $\t\in (0,1-\delta_0)$, $z_0\in (\Xdelta{C}{\delta_1}{\F}\cup \Xdelta{B}{\delta_1}{\F} \cup \Xdelta{A}{\delta_1}{\F}) \setminus \Edelta{\delta}=\C \setminus \Edelta{\delta}$ and $u,v\in \cc{B}_{\! \frac{A}{\sqrt{\n}} }$
\begin{align}
\left\vert \frac{\partial \HM}{\partial w}(\cc{z}_0+u,z_0+v) \right\vert &< C\sqrt{\n} \e^{-\n \epsilon}, \quad \text{and} \quad \left\vert \frac{\partial \HM}{\partial z}(\cc{z}_0+u,z_0+v) \right\vert < C\sqrt{\n} \e^{-\n \epsilon}.
\end{align}
Therefore
\begin{align}
\left \vert \int_0^{\frac{\cc{a}}{\sqrt{\n}}} \tfrac{\partial \HM}{\partial w}\left(\cc{z}_0+u,z_0\right)\dd u + \int_0^{\frac{b}{\sqrt{\n}}} \tfrac{\partial \HM}{\partial z}\left(\cc{z}_0+\tfrac{\cc{a} }{\sqrt{\n} }, z_0+v\right)\dd v \right \vert < 2 A C \e^{-\n \epsilon}, 
\end{align}
which converges uniformly to zero as $\n\rightarrow \infty$.

We are only left to consider $z_0\in \pEt$. On $\pEt$ we have that $\f=0$ according to Proposition \ref{prop_f_zero} and that $\hu=0$ according to Lemma \ref{lemma_hu_zero}. Thus Proposition \ref{prop_asym_wz_R} simplifies there to 
\begin{align}
\left.\frac{\partial \HM}{\partial w}\right\vert_{\substack{\phantom{}_{w=\cc{z}_0+u} \\ \phantom{}_{z=z_0+v} \\ \phantom{}}} =&\sqrt{\tfrac{\n}{2 \pi^3}} \gw(z_0) \e^{\n \left( \huu(\cc{z}_0)u^2 +\huu(z_0)v^2+\huv uv \right)} \left(1+\Rxi{+}{(\t,z_0,u,v)}\right), \label{eq_veri_kernel1} \\
\intertext{and}
\left.\frac{\partial \HM}{\partial z}\right\vert_{\substack{\phantom{}_{w=\cc{z}_0+u} \\ \phantom{}_{z=z_0+v} \\ \phantom{}}} =&\sqrt{\tfrac{\n}{2 \pi^3}} \gz(z_0) \e^{\n \left( \huu(\cc{z}_0)u^2 +\huu(z_0)v^2+\huv uv \right)} \left(1+\Rxi{-}{(\t,z_0,u,v)}\right), \label{eq_veri_kernel2} 
\end{align}
for all $\t\in (0,1-\delta_0]$, $z_0\in \pEt$, $u,v\in \cc{B}_{\!\frac{A}{\sqrt{\n_0}}}(0)$ and $\n\ge\n_0$ with $\n_0$ big enough. $\huu$, $\gw$, $\gz$, $\huv$, $\Rxi{\pm}{}$, $\ve{\xi}=\left(\xi_1(\n),\xi_2(\n),\xi_3(\n),\xi_4(v),\xi_5(v),\xi_6(v)\right)$ are as in Proposition \ref{prop_asym_wz_R} and we have set $\delta>0$ in Proposition \ref{prop_asym_wz_R} small enough so that $\pEt\subset \Adelta$.   

Then as in the proofs of Lemma \ref{lemma_hm_w_ellipse} and Proposition \ref{prop_hm_wz_ellipse} (but this time exactly) we get
\begin{align}
&\HM\left(\cc{z}_0+\tfrac{\cc{a}}{\sqrt{\n}},z_0+\tfrac{b}{\sqrt{\n}} \right) \nn \\
&\quad =\HM\left(\cc{z}_0,z_0 \right) + \sqrt{\tfrac{\n}{2 \pi^3}} \gw(z_0) \int_0^{\frac{\cc{a}}{\sqrt{\n}}} \e^{\n \huu(\cc{z}_0)u^2} \left(1+\Rxi{+}{(\t,z_0,u,0)}\right) \dd u \nn \\
&\qquad +\sqrt{\tfrac{\n}{2 \pi^3}} \gz(z_0) \e^{\huu(\cc{z}_0)\cc{a}^2} \int_0^{\frac{b}{\sqrt{\n}}} \e^{\n \left(\huu(z_0)v^2-\frac{\cc{a}}{\sqrt{\n}}v\right)}  \left(1+\Rxi{-}{(\t,z_0,\tfrac{\cc{a}}{\sqrt{\n}},v)}\right) \dd v. \label{eq_veri_kernel3}
\end{align} 
We can reuse the integrals we have calculated in \eqref{eq_integral_erfc1} and \eqref{eq_lemma_hm1} together with \eqref{eq_integral_erfc2}. So we know that
\begin{align}
&\tfrac{1}{2\pi}+\sqrt{\tfrac{\n}{2 \pi^3}} \gw(z_0) \int_0^{\frac{\cc{a}}{\sqrt{\n}}} \e^{\n \huu(\cc{z}_0)u^2} \dd u +\sqrt{\tfrac{\n}{2 \pi^3}} \gz(z_0) \e^{\huu(\cc{z}_0)\cc{a}^2} \int_0^{\frac{b}{\sqrt{\n}}} \e^{\n \left(\huu(z_0)v^2-\frac{\cc{a}}{\sqrt{\n}}v\right)} \dd v \nn \\
&\qquad =\tfrac{1}{2\pi }\erfc(\zeta(\cc{a},\phi))+\sqrt{\tfrac{\n}{2 \pi^3}} \gz(z_0) \e^{\huu(\cc{z}_0)\cc{a}^2} \int_0^{\frac{b}{\sqrt{\n}}} \e^{\n \left(\huu(z_0)v^2-\frac{\cc{a}}{\sqrt{\n}}v\right)} \dd v \nn \\
&\qquad = \tfrac{1}{2\pi} \erfc(\zeta(\cc{a},\phi)+\zeta(b,-\phi)). \label{eq_veri_kernel4}
\end{align}
Note that these results for the integrals are exact even for finite $\n$ since we haven't used any approximations in those calculations. With \eqref{eq_veri_kernel3} and \eqref{eq_veri_kernel4} we find that for all $\n>\n_0$, $0<\t<1-\delta_0$, $z_0\in \pEt$ and $a,b\in \cc{B}_{\!A}(0)$
\begin{align}
&\left\vert \HM\left(\cc{z}_0+\tfrac{\cc{a}}{\sqrt{\n}},z_0+\tfrac{b}{\sqrt{\n}}\right) - \tfrac{1}{2\pi}\erfc\left(\zeta(\cc{a},\phi)+ \zeta(b,-\phi) \right) \right\vert \nn \\
& \quad = \Bigg\vert \HM\left(\cc{z}_0,z_0 \right) + \sqrt{\tfrac{\n}{2 \pi^3}} \gw(z_0) \int_0^{\frac{\cc{a}}{\sqrt{\n}}} \e^{\n \huu(\cc{z}_0)u^2} \left(1+\Rxi{+}{(\t,z_0,u,0)}\right) \dd u \nn \\
&\qquad +\sqrt{\tfrac{\n}{2 \pi^3}} \gz(z_0) \e^{\huu(\cc{z}_0)\cc{a}^2} \int_0^{\frac{b}{\sqrt{\n}}} \e^{\n \left(\huu(z_0)v^2-\frac{\cc{a}}{\sqrt{\n}}v\right)}  \left(1+\Rxi{-}{(\t,z_0,\tfrac{\cc{a}}{\sqrt{\n}},v)}\right) \dd v \nn \\
&\qquad - \tfrac{1}{2\pi} - \sqrt{\tfrac{\n}{2 \pi^3}} \gw(z_0) \int_0^{\frac{\cc{a}}{\sqrt{\n}}} \e^{\n \huu(\cc{z}_0)u^2} \dd u \nn \\ 
&\qquad - \sqrt{\tfrac{\n}{2 \pi^3}} \gz(z_0) \e^{\huu(\cc{z}_0)\cc{a}^2} \int_0^{\frac{b}{\sqrt{\n}}} \e^{\n \left(\huu(z_0)v^2-\frac{\cc{a}}{\sqrt{\n}}v\right)} \dd v \Bigg\vert \nn \\
&\quad \le \left\vert \HM\left(\cc{z}_0,z_0 \right) - \tfrac{1}{2\pi} \right\vert + \left\vert \sqrt{\tfrac{\n}{2 \pi^3}} \gw(z_0) \int_0^{\frac{\cc{a}}{\sqrt{\n}}} \e^{\n \huu(\cc{z}_0)u^2} \Rxi{+}{(\t,z_0,u,0)} \dd u \right\vert \nn \\
&\qquad +\left\vert \sqrt{\tfrac{\n}{2 \pi^3}} \gz(z_0) \e^{\huu(\cc{z}_0)\cc{a}^2} \int_0^{\frac{b}{\sqrt{\n}}} \e^{\n \left(\huu(z_0)v^2-\frac{\cc{a}}{\sqrt{\n}}v\right)}  \Rxi{-}{(\t,z_0,\tfrac{\cc{a}}{\sqrt{\n}},v)} \dd v \right\vert .
\end{align}

As we have already mentioned $\left\vert \HM\left(\cc{z}_0,z_0 \right) - \tfrac{1}{2\pi} \right\vert$ converges uniformly to zero as $\n\rightarrow \infty$ for all $0<\t<1-\delta_0$ and $z_0\in \pEt$. 

From Lemma \ref{lemma_gwz} we know that $\vert \huu(z_0)\vert=\tfrac{1}{2}$ and $\vert \gw(z_0)\vert=\vert \gz(z_0)\vert=1$. According to Proposition \ref{prop_Rxi} for all $\epsilon>0$ we get $\n_0\in\N$ so that $\left\vert\Rxi{+}{(\t,z_0,u,0)}\right\vert<\epsilon$ for all $\n>\n_0$, $0<\t<1-\delta_0$, $z_0\in \pEt$, $u\in \cc{B}_{\tfrac{A}{\sqrt{\n}}}(0)$, $\xi_1,\xi_2,\xi_3\in [0,\frac{1}{\n}]$ and $\xi_4,\xi_5,\xi_6\in \cc{B}_{\!\frac{A}{\sqrt{\n}}}(0)$, 
and then if additionally $a\in \cc{B}_{\!A}(0)$
\begin{align}
&\left\vert \sqrt{\tfrac{\n}{2 \pi^3}} \gw(z_0) \int_0^{\frac{\cc{a}}{\sqrt{\n}}} \e^{\n \huu(\cc{z}_0)u^2} \Rxi{+}{(\t,z_0,u,0)} \dd u \right\vert \nn \\
&\qquad \le \sqrt{\tfrac{\n}{2 \pi^3}} \vert \gw(z_0)\vert \left( \sup_{u\in B_{\!\frac{A}{\sqrt{\n}}(0) } }\left\vert \e^{\n \huu(\cc{z}_0)u^2} \Rxi{+}{(\t,z_0,u,0)} \right\vert \right) \left\vert\int_0^{\frac{A}{\sqrt{\n}}} \dd u\right\vert \nn \\
&\qquad < \sqrt{\tfrac{\n}{2 \pi^3}} \e^{\frac{A^2}{2}}\epsilon \tfrac{A}{\sqrt{\n}}= \epsilon \tfrac{A}{\sqrt{2 \pi^3}}  \e^{\frac{A^2}{2}},
\end{align}
which converges uniformly to zero as $\epsilon \rightarrow 0$.

Analogously we can get with $\vert \huu(z_0)\vert=\vert \huu(\cc{z}_0)\vert=\tfrac{1}{2}$, $\vert \gw(z_0)\vert=\vert \gz(z_0)\vert=1$ and $\left\vert\Rxi{-}{(\t,z_0,u,v)}\right\vert<\epsilon$ for $\n>\n_0$, $0<\t<1-\delta_0$, $z_0\in \pEt$, $u,v\in \cc{B}_{\tfrac{A}{\sqrt\n}}(0)$, $\xi_1,\xi_2,\xi_3\in [0,\frac{1}{\n}]$ and $\xi_4,\xi_5,\xi_6\in \cc{B}_{\!\frac{A}{\sqrt{\n}}}(0)$ that then if additionally $a,b\in \cc{B}_{\!A}(0)$ 
\begin{align}
&\left\vert \sqrt{\tfrac{\n}{2 \pi^3}} \gz(z_0) \e^{\huu(\cc{z}_0)\cc{a}^2} \int_0^{\frac{b}{\sqrt{\n}}} \e^{\n \left(\huu(z_0)v^2-\frac{\cc{a}}{\sqrt{\n}}v\right)}  \Rxi{-}{(\t,z_0,\tfrac{\cc{a}}{\sqrt{\n}},v)} \dd v \right\vert \nn \\
&\qquad \le \sqrt{\tfrac{\n}{2 \pi^3}} \left\vert \gz(z_0) \e^{\huu(\cc{z}_0)\cc{a}^2} \right\vert
\left( \sup_{v\in B_{\!\frac{A}{\sqrt{\n}}(0) } }\left\vert  \e^{\n \left(\huu(z_0)v^2-\frac{\cc{a}}{\sqrt{\n}}v\right)}  \Rxi{-}{(\t,z_0,\tfrac{\cc{a}}{\sqrt{\n}},v)}  \right\vert \right) \left\vert \int_0^{\frac{A}{\sqrt{\n}}} \dd v\right\vert \nn \\
&\qquad < \sqrt{\tfrac{\n}{2 \pi^3}} \e^{\frac{A^2}{2}} \e^{\frac{A^2}{2}+A^2} \epsilon \tfrac{A}{\sqrt{\n}}= \epsilon \tfrac{A}{\sqrt{2 \pi^3}}  \e^{2A^2},
\end{align}
which also converges uniformly to zero as $\epsilon \rightarrow 0$.
\end{proof}


Now we have reviewed all the approximations of the previous chapters. When we computed the reproducing kernel $\tKn$ in Theorem \ref{th_km} and the correlation functions in Theorem \ref{th_corr} we hadn't done any further approximations.

Beside of estimations of the error terms, we have seen where the limit of the reproducing kernel and the correlation functions hold uniformly. This we want to summarize in the following corollary.

\begin{cor}\label{cor_unifrom_universality}
Let $\delta>0$, $0<\delta_0<1$ and $A>0$. Then the asymptotics of the reproducing kernel is
\begin{align}
\lim_{\n\rightarrow\infty} \sup_{0\le\t<1-\delta_0} \sup_{\substack{z_0 \in \Edelta{-,\delta} \\ a,b\in \cc{B}_{\!A}(0) }} & \left\vert \frac{\nK}{\n}\left(z_0+\tfrac{a}{\sqrt{\n}},z_0+\tfrac{b}{\sqrt{\n}}\right) - \tfrac{1}{\pi} \e^{i\nphase(z_0,a,b)}\e^{-\frac{\vert a-b\vert^2}{2}}\right\vert = 0, \\
\lim_{\n\rightarrow\infty} \sup_{0\le\t<1-\delta_0} \sup_{\substack{z_0\in \C\setminus \Edelta{+,\delta} \\ a,b\in \cc{B}_{\!A}(0) }} & \left\vert \frac{\nK}{\n}\left(z_0+\tfrac{a}{\sqrt{\n}},z_0+\tfrac{b}{\sqrt{\n}}\right)  \right\vert = 0, \\
\lim_{\n\rightarrow\infty} \sup_{0\le\t<1-\delta_0} \sup_{\substack{z_0\in \pEt \\ a,b\in \cc{B}_{\!A}(0) }} & \left\vert \frac{\nK}{\n}\left(z_0+\tfrac{a\e^{i\psi}}{\sqrt{\n}},z_0+\tfrac{b\e^{i\psi}}{\sqrt{\n}}\right) -   \e^{i\nphase(z_0,a\e^{i\psi},b\e^{i\psi})}\e^{-\frac{\vert a-b\vert^2}{2}} \tfrac{\erfc\left(\frac{\cc{a}+b}{\sqrt{2}}\right)}{2\pi}    \right\vert = 0,
\end{align}
and the asymptotics of the $m$-point correlation functions is
\begin{align} 
\lim_{\n\rightarrow\infty} \sup_{\substack{0 \le \t < 1-\delta_0 }} \sup_{\substack{z_0 \Edelta{-,\delta}  \\ a_1,\ldots,a_m\in \cc{B}_{\!A}(0) }} & \Big\vert \tfrac{\Rn{m}}{\n^m} \left(z_0+\tfrac{a_1}{\sqrt{\n}},\ldots,z_0+\tfrac{a_m}{\sqrt{\n}} \right)\nn \\
&\quad - \det \left( \tfrac{1}{\pi}\e^{-\frac{1}{2}\vert a_k-a_l \vert^2+i\Im \left(\cc{a}_k a_l\right) }  \right)_{k,l=1}^m \Big\vert=0, \\
\lim_{\n\rightarrow\infty} \sup_{\substack{0\le \t < 1-\delta_0 }} \sup_{\substack{z_0 \in \C\setminus \Edelta{+,\delta} \\ a_1,\ldots,a_m\in \cc{B}_{\!A}(0) }} & \left\vert \tfrac{\Rn{m}}{\n^m} \left(z_0+\tfrac{a_1}{\sqrt{\n}},\ldots,z_0+\tfrac{a_m}{\sqrt{\n}} \right)\right\vert=0, \\
\lim_{\n\rightarrow\infty} \sup_{\substack{0\le \t < 1-\delta_0 }} \sup_{\substack{z_0 \in \partial \pEt \\ a_1,\ldots,a_m\in \cc{B}_{\!A}(0) }} & \Big\vert \tfrac{\Rn{m}}{\n^m} \left(z_0+\tfrac{a_1\e^{i\psi(z_0)}}{\sqrt{\n}},\ldots,z_0+\tfrac{a_m\e^{i\psi(z_0)}}{\sqrt{\n}}\right) \nn \\
&\quad-\det \left( \tfrac{\erfc\left(\frac{\cc{a}_k+a_l}{\sqrt{2}}\right)}{2\pi}\e^{-\frac{1}{2}\vert a_k-a_l \vert^2+ i\Im \left(\cc{a}_k a_l\right) }  \right)_{k,l=1}^m \Big\vert=0, 
\end{align}
where
\begin{align}
\nphase(z_0,a,b)&=\sqrt{\n}\Im\left( (\cc{a}-\cc{b})(z_0-\t \cc{z}_0) \right) + \Im\left( \cc{a}b+\tfrac{\t}{2}(a^2-b^2) \right), 
\end{align}
and $\psi$ is defined as in Corollary \ref{cor_density_scaled_limit} if $z_0=\mas \cos \phi+i \mis \sin \phi\in\pEt$.
\end{cor}

\begin{proof}
The proof follows directly from Proposition \ref{prop_veri_kernel}, where we have to apply the same steps as in the proof of Proposition \ref{th_km} and \ref{th_corr} to get the reproducing kernel and the correlation functions. 
\end{proof}







\section{Uniform Convergence for Potentials with $t_0\neq 1$, Complex $t_2$ and Shift of Origin}\label{sec_t}
\sectionmark{Potentials with $\protect\chapt_0\neq 1$, $\protect\chapt_2\in\C$ and Shift of Origin}






Now we want to consider the most general potentials as in \eqref{def_potential} with $d=2$, i.e.\ we will discuss the cases when $t_0\neq 1$, $t_1\neq 0$ and $\t=2 t_2\in \C$. 

\subsection{Orthonormal Polynomials for Potential with $\protect\chapt_0\neq 1$ and $\protect\chapt_1\neq 0$}
In Proposition \ref{prop_op} we have already found the orthonormal polynomials for the potential $V$ when $t_0=1$, $t_1=0$ and any complex $t_2$ with $\vert t_2\vert<\tfrac{1}{2}$. Now we are looking for the orthonormal polynomials for arbitrary $t_0\in (0,\infty)$ and $t_1\in\C$.

\begin{lemma}\label{lemma_shift}
Let $t_1,t_2\in \C$ with $\vert t_2\vert <\tfrac{1}{2}$. Then there exists $v=v(t_1,t_2)\in\C$ and $C=C(t_1,t_2)\in \R$ so that
\begin{align}
\vert z \vert^2-2\Re\left(t_1 z+t_2 z^2\right)+C=\vert z-v\vert^2-2\Re\left(t_2 (z-v)^2\right). \label{eq_shift1} 
\end{align}
\end{lemma}
\begin{proof}
We find for the left-hand side of \eqref{eq_shift1}
\begin{align}
\vert z \vert^2-2\Re\left(t_1 z+t_2 z^2\right)+C &= \cc{z}z - \cc{t}_2\cc{z}^2 - t_2 z^2 -\cc{t}_1 \cc{z} - t_1 z + C, \\
\intertext{and for the right-hand side}
\vert z-v\vert^2-2\Re\left(t_2 (z-v)^2\right) &= \cc{z}z - \cc{t}_2\cc{z}^2 - t_2 z^2 -(v-2\cc{t}_2\cc{v}) \cc{z} - (\cc{v}-2t_2 v) z \nn \\
& \quad + \cc{v}v-\cc{t}_2\cc{v}^2-t_2 v^2.
\end{align}
Comparing the coefficients, we see that $C=\vert v \vert^2 - 2 \Re(t_2 v^2)\in \R$ and $t_1=\cc{v}-2 t_2 v$. Separating the last equation in real and imaginary part and solving for $\Re\, v$ and $\Im\, v$, we get
\begin{align}
\Re\, v&=\tfrac{\Re\, t_1+2\left( \Re\, t_1 \Re\, t_2 + \Im\, t_1 \Im\, t_2\right)}{1-4(\Re\, t_2)^2-4(\Im\, t_2)^2}, \nn \\
\Im\, v&=-\tfrac{\Im\, t_1+2\left( \Re\, t_1 \Im\, t_2 - \Im\, t_1 \Re\, t_2 \right)}{1-4(\Re\, t_2)^2-4(\Im\, t_2)^2}, \nn \\
\intertext{which is equivalent to}
v&=\tfrac{\cc{t}_1+2 t_1 \cc{t}_2}{1-4\vert t_2\vert^2}\in \C. \label{eq_shift2}
\end{align}
\end{proof}
\begin{remark}
If $t_2$ is real then \eqref{eq_shift2} can be simplified to $\Re\, v=\tfrac{\Re\, t_1}{1-2 t_2}$ and $\Im\, v=-\tfrac{\Im\, t_1}{1+2 t_2}$. 
\end{remark}

\begin{prop}\label{prop_op_t0t1}
Let $t_0>0$, $t_1,t_2\in\C$ with $\vert t_2\vert< \tfrac{1}{2}$ and $(\tildeop{k})_{k\in\N_0}$ be the orthonormal polynomials for the potential $\tilde{V}(z)=\vert z\vert^2-2\Re\left(t_2 z^2\right)$. Then the orthonormal polynomials $(\op{k})_{k\in\N_0}$ for the potential $V(z)=\tfrac{1}{t_0}\left(\vert z\vert^2-2\Re\left(t_1 z + t_2 z^2\right)+C\right)$ are related to $(\tildeop{k})_{k\in\N_0}$ as 
\begin{align}
\op{k}(z)=\tfrac{1}{\sqrt{t_0}}\tildeop{k}\left( \tfrac{z-v}{\sqrt{t_0}} \right), \qquad k\in \N_0, \label{eq_op_t0t1}
\end{align}
where $v=v(t_1,t_2)$ and $C$ are as in Lemma \ref{lemma_shift}.
\end{prop}
\begin{proof}
Let $k,l\in \N_0$. Substituting first $z$ by $\tfrac{z}{\sqrt{t}_0}$ and then $z$ by $z-v$ we get
\begin{align}
\delta_{kl}&=\int_C \cc{\tildeop{k}(z)} \tildeop{l}(z) \e^{-\n \left( \vert z\vert^2 -2\Re \left( t_2 z^2 \right) \right)} \dd^2 z \nn \\
&=\int_C \cc{\tildeop{k}\left(\tfrac{z}{\sqrt{t}_0}\right)} \tildeop{l}\left(\tfrac{z}{\sqrt{t}_0}\right) \e^{-\frac{\n}{t_0} \left( \vert z\vert^2 -2\Re \left( t_2 z^2 \right) \right)} \tfrac{\dd^2 z}{t_0} \nn \\
&=\int_C \cc{\tildeop{k}\left(\tfrac{z-v}{\sqrt{t_0}}\right)} \tildeop{l}\left(\tfrac{z-v}{\sqrt{t_0}}\right) \e^{-\frac{\n}{t_0} \left( \vert z-v\vert^2 -2\Re \left( t_2 (z-v)^2 \right) \right)} \tfrac{\dd^2 z}{t_0} \nn \\
&=\int_C \tfrac{1}{\sqrt{t_0}}\cc{\tildeop{k}\left(\tfrac{z-v}{\sqrt{t_0}}\right)} \tfrac{1}{\sqrt{t_0}}\tildeop{l}\left(\tfrac{z-v}{\sqrt{t_0}}\right) \e^{-\frac{\n}{t_0} \left( \vert z \vert^2-2\Re\left(t_1 z+t_2 z^2\right)+C \right)} \dd^2 z \nn \\
\end{align}
In the last step we have used Lemma \ref{lemma_shift}. Now it is obvious that the orthonormal polynomials regarding the potential $V$ are
\begin{align}
\op{k}(z)=\tfrac{1}{\sqrt{t_0}}\tildeop{k}\left(\tfrac{z-v}{\sqrt{t_0}}\right), \qquad k\in \N_0.
\end{align}
\end{proof}

\begin{remark}
For calculations of the normalized reproducing kernel or the correlation functions any constant $C\in\R$ in the potential $V$ can be omitted (see Proposition \ref{prop_kernel_const}). 
\end{remark}

From Proposition \ref{prop_op} we see that the orthonormal polynomials $(\op{k})_{k\in\N_0}$ for the potential $V(z)=\vert z\vert^2-2\Re\left(t_2 z^2\right)$ with $t_2\in \C$ and $\vert t_2\vert<\tfrac{1}{2}$ are 
\begin{align}
\op{k}(z)=H_k\left( \sqrt{\tfrac{\n(1-4\vert t_2\vert^2)}{4\cc{t}_2}} z\right)  \frac{\cc{t}_2^{\frac{k}{2}} \sqrt{\n} (1-4\vert t_2\vert^2)^{1/4}}{\sqrt{\pi k!} }, \qquad z\in\C.
\end{align}
From this we see that they are related to the orthonormal polynomials $(\tildeop{k})_{k\in\N_0}$ for the analogous potential $V(z)=\vert z\vert^2-2\Re\left(\tilde{t}_2 z^2\right)$ with real $\tilde{t}_2=\vert t_2\vert$ as
\begin{align}
\op{k}(z)=\tildeop{k}\left(\e^{i \frac{\theta}{2}} z\right) \e^{-i \frac{k \theta}{2}}, \label{eq_rel_complex_t2}
\end{align}
where $\theta\in (-\pi,\pi]$ is the phase of $t_2=\e^{i \theta} \vert t_2\vert$.


\subsection{Kernel and Correlation Functions for $\protect\chapt_1,\protect\chapt_2\in \C$}
In our calculations of the reproducing kernel we assumed, beginning with chapter \ref{sec_ident}, that $t_2$ is real and positive. Now we will show what happens if $t_2$ is complex and additionally we will also consider $t_1\neq 0$.


\begin{defin}
For $t_0>0$ and $t_1,t_2\in\C$ with $\vert t_2\vert<\tfrac{1}{2}$ we define the ellipse $\partial E_{t_0,t_1,t_2}$ as a generalization of $\pEt$ which we get when we first scale $\pEt$ with $\sqrt{t_0}$, then rotate it through the angle $-\tfrac{\arg t_2}{2}$ and finally shift it with $v(t_1,t_2)$ from Lemma \ref{lemma_shift}. 
\begin{align}
\partial E_{t_0,t_1,t_2}&=\left \{z\in \C \left\vert \tfrac{\left(\Re\left(\sqrt{t_2}(z-v)\right)\right)^2}{a_{t_0,t_2}^2} + \tfrac{\left(\Im\left(\sqrt{t_2}(z-v)\right)\right)^2}{b_{t_0,t_2}^2} =\vert t_2\vert \right. \right \}, \qquad \text{if $t_2\neq 0$},\\
\intertext{with the major semi-axis $a_{t_0,t_2}=\sqrt{t_0 \tfrac{1+2\vert t_2\vert }{1-2\vert t_2\vert}}$ and its minor semi-axis $b_{t_0,t_2}=\sqrt{t_0 \tfrac{1-2\vert t_2\vert }{1+2\vert t_2\vert}}$. And for $t_2=0$,}
\partial E_{t_0,t_1,0}&=\left \{z\in \C \left\vert \left(\Re\left(z-v\right)\right)^2 + \left(\Im\left(z-v\right)\right)^2 =t_0 \right. \right\}
\end{align}
\end{defin}

\begin{defin}
For $t_0>0$, $t_1,t_2\in\C$ with $\vert t_2\vert<\tfrac{1}{2}$ and $\delta>0$ we define analogously the generalization of $\Edelta{\pm \delta}$.
\begin{align}
E_{t_0,t_1,t_2,\pm \delta}&=\left\{ z\in \C  \left\vert \tfrac{\left(\Re\left(\sqrt{t_2}(z-v)\right)\right)^2}{\left(a_{t_0,t_2}\pm \delta\sqrt{\vert t_0\vert} \right)^2} + \tfrac{\left(\Im\left(\sqrt{t_2}(z-v)\right)\right)^2}{\left(b_{t_0,t_2}\pm\delta\sqrt{\vert t_0\vert} \right)^2} < \vert t_2\vert \right. \right\}, \qquad \text{if }t_2\neq 0, \\
E_{t_0,t_1,0,\pm \delta}&=\left \{z\in \C \left\vert \left(\Re\left(z-v\right)\right)^2 + \left(\Im\left(z-v\right)\right)^2 <t_0 (1\pm\delta)^2 \right. \right\}.
\end{align}
\end{defin}



\begin{theorem}[Uniform Universality of Kernel and Correlation Functions]\label{th_corr_uniform} 
Let $\delta>0$, $0<\delta_0<1$ and $A>0$. $\nK$ shall be the kernel for the potential 
$V(z)=\vert z\vert^2-2\Re\left(t_1 z + t_2 z^2\right)$ with $t_1,t_2\in\C$ and $2\vert \t_2\vert<1$ and $\Rn{m}$ the corresponding $m$-point correlation function as in Definition \ref{def_correlation}. Then
\begin{align}
\lim_{\n\rightarrow\infty} \sup_{\substack{0\le 2\vert t_2 \vert \le 1-\delta_0 \\ t_1\in \C}} \sup_{\substack{z_0 \in E_{1,t_1,t_2,-\delta} \\ a,b\in \cc{B}_{\!A}(0) }} & \Big\vert \tfrac{\nK}{\n}\left(z_0+\tfrac{a}{\sqrt{\n}},z_0+\tfrac{b}{\sqrt{\n}}\right) - \tfrac{1}{\pi} \e^{i\theta_{\n,t_2}(z_0-v,a,b)}\e^{-\frac{\vert a-b\vert^2}{2}}\Big\vert = 0, \\
\lim_{\n\rightarrow\infty} \sup_{\substack{0\le 2\vert t_2 \vert \le 1-\delta_0 \\ t_1\in \C}} \sup_{\substack{z_0\in \C\setminus E_{1,t_1,t_2,+\delta} \\ a,b\in \cc{B}_{\!A}(0) }} & \Big\vert \tfrac{\nK}{\n}\left(z_0+\tfrac{a}{\sqrt{\n}},z_0+\tfrac{b}{\sqrt{\n}}\right)  \Big\vert = 0, \\
\lim_{\n\rightarrow\infty} \sup_{\substack{0\le 2\vert t_2 \vert \le 1-\delta_0 \\ t_1\in \C}} \sup_{\substack{z_0\in \partial E_{1,t_1,t_2} \\ a,b\in \cc{B}_{\!A}(0) }} & \Big\vert \tfrac{\nK}{\n}\left(z_0+\tfrac{a\e^{i\psi(z_0)}}{\sqrt{\n}},z_0+\tfrac{b\e^{i\psi(z_0)}}{\sqrt{\n}}\right) \nn \\
&\quad - \e^{i\theta_{\n,t_2}(z_0-v,a\e^{i\psi(z_0)},b\e^{i\psi(z_0)})}\e^{-\frac{\vert a-b\vert^2}{2}} \tfrac{\erfc\left(\frac{\cc{a}+b}{\sqrt{2}}\right)}{2\pi} \Big\vert = 0, \label{eq_corr_uniform1a}
\end{align}
with the 
phase 
\begin{align}
\theta_{\n,t_2}(z_0,a,b)&=\sqrt{\n}\Im\left( (\cc{a}-\cc{b})(z_0- 2\cc{t}_2 \cc{z}_0) \right) + \Im\left( \cc{a}b+ t_2(a^2-b^2) \right), 
\end{align}
and 
\begin{align} 
\lim_{\n\rightarrow\infty} \sup_{\substack{0\le 2\vert t_2 \vert \le 1-\delta_0 \\ t_1\in \C}} \sup_{\substack{z_0 \in E_{1,t_1,t_2,-\delta} \\ a_1,\ldots,a_m\in \cc{B}_{\!A}(0) }} & \Big\vert \tfrac{\Rn{m}}{\n^m} \left(z_0+\tfrac{a_1}{\sqrt{\n}},\ldots,z_0+\tfrac{a_m}{\sqrt{\n}} \right)\nn \\
&\quad - \det \left( \tfrac{1}{\pi}\e^{-\frac{1}{2}\vert a_k-a_l \vert^2+i\Im \left(\cc{a}_k a_l\right) }  \right)_{k,l=1}^m \Big\vert=0, \\
\lim_{\n\rightarrow\infty} \sup_{\substack{0\le\vert 2 t_2 \vert \le 1-\delta_0 \\ t_1\in \C}} \sup_{\substack{z_0 \in \C\setminus E_{1,t_1,t_2,+\delta} \\ a_1,\ldots,a_m\in \cc{B}_{\!A}(0) }} & \left\vert \tfrac{\Rn{m}}{\n^m} \left(z_0+\tfrac{a_1}{\sqrt{\n}},\ldots,z_0+\tfrac{a_m}{\sqrt{\n}} \right)\right\vert=0, \\
\lim_{\n\rightarrow\infty} \sup_{\substack{0\le\vert 2 t_2 \vert \le 1-\delta_0 \\ t_1\in \C}} \sup_{\substack{z_0 \in \partial E_{1,t_1,t_2} \\ a_1,\ldots,a_m\in \cc{B}_{\!A}(0) }} & \Big\vert \tfrac{\Rn{m}}{\n^m} \left(z_0+\tfrac{a_1\e^{i\psi(z_0)}}{\sqrt{\n}},\ldots,z_0+\tfrac{a_m\e^{i\psi(z_0)}}{\sqrt{\n}}\right) \nn \\
&\quad-\det \left( \tfrac{\erfc\left(\frac{\cc{a}_k+a_l}{\sqrt{2}}\right)}{2\pi}\e^{-\frac{1}{2}\vert a_k-a_l \vert^2+ i\Im \left(\cc{a}_k a_l\right) }  \right)_{k,l=1}^m \Big\vert=0, \label{eq_corr_uniform1}
\end{align}
where in \eqref{eq_corr_uniform1a} and \eqref{eq_corr_uniform1} if $t_2\neq 0$,
\begin{align}
\psi(z_0)&=- \tfrac{\arg t_2}{2} + \left\{\begin{array}{ll} \arctan\left(\tfrac{1+2\vert t_2 \vert}{1-2 \vert t_2 \vert }\tan \phi \right) & \text{if } \phi\in (-\tfrac{\pi}{2},\tfrac{\pi}{2}), \\ \arctan\left(\tfrac{1+2\vert t_2 \vert}{1-2 \vert t_2 \vert }\tan \phi \right)+\sign(\phi)\pi & \text{if } \phi\notin [-\tfrac{\pi}{2},\tfrac{\pi}{2}],  \\  \phi & \text{if }\phi=\pm \tfrac{\pi}{2}, \end{array}\right.
\end{align}
with the parameter $\phi\in(-\pi,\pi]$ fulfilling $\sqrt{\tfrac{t_2}{\vert t_2\vert}}(z_0-v)=a_{1,t_2} \cos \phi+i b_{1,t_2} \sin\phi \in\partial E_{1,t_1,t_2}$ and if $t_2=0$,
\begin{align}
\psi(z_0)&=\phi, \nn 
\end{align}
with $z_0-v=\e^{i \phi} \in\partial E_{1,t_1,0}$. 
\end{theorem}
\begin{remark}
$\partial E_{1,t_1,t_2}$ is just the rotated and shifted ellipse $\pEt$. The shift $v$ depends only on $t_1$ and $t_2$ and is given by Lemma \ref{lemma_shift}. The same way we have defined for $z_0\in \partial E_{1,t_1,t_2}$ the angle $\psi(z_0)$ so that it is the angle between the real axis and the normal to the rotated ellipse in the point $z_0$.
\end{remark}
\begin{proof}
Let $(\op{k})_{k\in\N_0}$ be the orthonormal polynomials regarding the potential 
\begin{align}
V(z)=\vert z\vert^2-2\Re\left(t_1 z + t_2 z^2\right)+C= \vert z-v\vert^2-2\Re\left(t_2 (z-v)^2\right),
\end{align}
with $v$ and $C$ as in Lemma \ref{lemma_shift}. Analogously $(\tildeop{k})_{k\in\N_0}$ shall be the orthonormal polynomials regarding the potential $\tilde{V}(z)=\vert z\vert^2-2 \vert t_2\vert \Re\left(z^2\right)$ with $t_2=\e^{i \theta} \vert t_2\vert$.

In the definition of the normalized reproducing kernel (Definition \ref{def_norm_kernel}) we use Proposition \ref{prop_op_t0t1} and relation \eqref{eq_rel_complex_t2} to replace the orthonormal polynomials $(\op{k})_{k\in\N_0}$ by $(\tildeop{k})_{k\in\N_0}$. With
\begin{align}
\cc{\op{k}(w)}\op{k}(z)&=\cc{\tildeop{k}\left(\e^{i \frac{\theta}{2}} (w-v)\right) \e^{-i \frac{k \theta}{2}}} \tildeop{k}\left(\e^{i \frac{\theta}{2}} (z-v)\right) \e^{-i \frac{k \theta}{2}}\nn \\
&=\cc{\tildeop{k}\left(\e^{i \frac{\theta}{2}} (w-v)\right) } \tildeop{k}\left(\e^{i \frac{\theta}{2}} (z-v) \right),
\end{align}
we get therefore
\begin{align}
\tKn(w,z)&=\e^{-\frac{\n}{2} (V(w)+V(z))} \sum_{k=0}^{\n-1} \cc{\op{k}(w)}\op{k}(z) \nn \\
& =\e^{-\frac{\n \left(\vert w-v\vert^2+\vert z-v\vert^2-2\Re\left(t_2(w-v)^2+t_2(z-v)^2 \right) \right)}{2}} \sum_{k=0}^{\n-1} \cc{\tildeop{k}\left(\e^{i \frac{\theta}{2}} (w-v)\right) } \tildeop{k}\left(\e^{i \frac{\theta}{2}} (z-v)\right) \nn \\
&=\e^{-\frac{\n \left(\vert \tilde{w} \vert^2+\vert \tilde{z}\vert^2-2\vert t_2\vert \Re\left( \tilde{w}^2+\tilde{z}^2 \right) \right)}{2}} \sum_{k=0}^{\n-1} \cc{\tildeop{k}\left(\tilde{w}\right) } \tildeop{k}\left(\tilde{z}\right) \nn \\
&=\e^{-\frac{\n}{2} (\tilde{V}(\tilde{w})+\tilde{V}(\tilde{z}))} \sum_{k=0}^{\n-1} \cc{\tildeop{k}\left(\tilde{w}\right) } \tildeop{k}\left(\tilde{z}\right), \label{eq_kernel_transf_coord}
\end{align}
where we have introduced the shifted and rotated coordinates $\tilde{w}=\e^{i \frac{\theta}{2}} (w-v)$ and $\tilde{z}=\e^{i \frac{\theta}{2}} (z-v)$. We see that in the coordinates $\tilde{w}$ and $\tilde{z}$ the right-hand side of \eqref{eq_kernel_transf_coord} is exactly the normalized reproducing kernel for the potential $\tilde{V}$, which we have computed in Section \ref{subsec_asym_kernel}. 

Additionally we introduce the coordinates $\tilde{z}_0=\e^{i \frac{\theta}{2}} (z_0-v)$, $ \tilde{a}=\e^{i \frac{\theta}{2}} a$ and $\tilde{b}=\e^{i \frac{\theta}{2}} b$ so that for $w=z_0+\tfrac{a}{\sqrt{n}}$ and $z=z_0+\tfrac{b}{\sqrt{n}}$ we get
\begin{align}
\tilde{w}=\tilde{z}_0+\tfrac{\tilde{a}}{\sqrt{n}}=\e^{i \frac{\theta}{2}} \left( z_0-v+\tfrac{a}{\sqrt{n}} \right), \qquad \text{and} \quad \tilde{z}=\tilde{z}_0+\tfrac{\tilde{b}}{\sqrt{n}}=\e^{i \frac{\theta}{2}} \left( z_0-v+\tfrac{b}{\sqrt{n}} \right). \nn
\end{align}

Then in the coordinates $\tilde{w}$ and $\tilde{z}$ the asymptotics of the right-hand side of \eqref{eq_kernel_transf_coord} can be found by Corollary \ref{cor_unifrom_universality}. Writing it again in the coordinates $w$ and $z$ we find that for all $\epsilon>0$ there exists $\n_0\in\N$ so that for all $\n\ge \n_0$
\begin{align}
\sup_{0\le 2\vert t_2\vert <1-\delta_0} \sup_{\substack{z_0 \in E_{1,t_1,t_2,-\delta} \\ a,b\in \cc{B}_{\!A}(0) }} &\left\vert \tfrac{\nK}{\n}\left(z_0+\tfrac{a}{\sqrt{\n}},z_0+\tfrac{b}{\sqrt{\n}}\right) - \tfrac{ \e^{i\theta_{\n,t_2}(z_0-v,a,b)}\e^{-\frac{\vert a-b\vert^2}{2}}}{\pi}  \right\vert < \epsilon, \\
\sup_{0\le 2\vert t_2\vert <1-\delta_0} \sup_{\substack{z_0\in \C\setminus E_{1,t_1,t_2,+\delta} \\ a,b\in \cc{B}_{\!A}(0) }} & \left\vert \tfrac{\nK}{\n}\left(z_0+\tfrac{a}{\sqrt{\n}},z_0+\tfrac{b}{\sqrt{\n}}\right)  \right\vert < \epsilon, \\
\shortintertext{and}
\sup_{0\le 2\vert t_2\vert <1-\delta_0} \sup_{\substack{z_0\in \partial E_{1,t_1,t_2} \\ a,b\in \cc{B}_{\!A}(0) }} & \left\vert \tfrac{\nK}{\n}\left(z_0+\tfrac{a\e^{i\psi}}{\sqrt{\n}},z_0+\tfrac{b\e^{i\psi}}{\sqrt{\n}}\right) - \tfrac{ \e^{i\theta_{\n,t_2}(z_0-v,a\e^{i\psi},b\e^{i\psi})}\e^{-\frac{\vert a-b\vert^2}{2}} \erfc\left(\frac{\cc{a}+b}{\sqrt{2}}\right)}{2\pi} \right\vert < \epsilon, 
\end{align}
where
\begin{align}
\theta_{\n,t_2}(z_0,a,b)&=\sqrt{\n}\Im\left( (\cc{a}-\cc{b})(z_0- 2\cc{t}_2 \cc{z}_0) \right) + \Im\left( \cc{a}b+ t_2(a^2-b^2) \right),
\end{align}
and for $z_0\in \partial E_{1,t_1,t_2}$
\begin{align}
\psi&=- \tfrac{\arg t_2}{2} + \left\{\begin{array}{ll} \arctan\left(\tfrac{1+2\vert t_2 \vert}{1-2 \vert t_2 \vert }\tan \phi \right) & \text{if } \phi\in (-\tfrac{\pi}{2},\tfrac{\pi}{2}), \\ \arctan\left(\tfrac{1+2\vert t_2 \vert}{1-2 \vert t_2 \vert }\tan \phi \right)+\sign(\phi)\pi & \text{if } \phi\notin [-\tfrac{\pi}{2},\tfrac{\pi}{2}],  \\  \phi & \text{if }\phi=\pm \tfrac{\pi}{2}, \end{array}\right.
\end{align}
if $t_2\neq 0$ with the parameter $\phi\in(-\pi,\pi]$ fulfilling $\sqrt{\tfrac{t_2}{\vert t_2\vert}}(z_0-v)=a_{1,t_2} \cos \phi+i b_{1,t_2} \sin\phi \in\partial E_{1,t_1,t_2}$ and 
\begin{align}
\psi&=\phi, \qquad \text{if }t_2=0, 
\end{align}
with $z_0-v=\e^{i \phi} \in\partial E_{1,t_1,0}$.

Since this is true for all $t_1\in\C$ and $\n_0$ is independent of $t_1$, we can also take the supremum over $t_1\in \C$.

Note that we have defined $\partial E_{t_0,t_1,t_2}$ and $E_{t_0,t_1,t_2,\pm \delta}$ so that $\tilde{z}_0\in \pEt$ is equivalent to $z_0\in \partial E_{1,t_1,t_2}$ and $\tilde{z}_0\in \Edelta{\pm,\delta}$ is equivalent to $z_0\in E_{1,t_1,t_2,\pm \delta}$.


Additionally this corollary gives us the asymptotics of the $m$-point correlation functions, for which we get analogously that for all $\epsilon>0$ there exists $\n_0\in\N$ so that for all $\n\ge \n_0$
\begin{align} 
\sup_{\substack{0\le 2\vert t_2 \vert \le 1-\delta_0 \\ t_1\in \C}} \sup_{\substack{z_0 \in E_{1,t_1,t_2,-\delta} \\ a_1,\ldots,a_m\in \cc{B}_{\!A}(0) }} & \Big\vert \tfrac{\Rn{m}}{\n^m} \left(z_0+\tfrac{a_1}{\sqrt{\n}},\ldots,z_0+\tfrac{a_m}{\sqrt{\n}} \right)\nn \\
&\quad - \det \left( \tfrac{1}{\pi}\e^{-\frac{1}{2}\vert a_k-a_l \vert^2+i\Im \left(\cc{a}_k a_l\right) }  \right)_{k,l=1}^m \Big\vert<\epsilon, \\
\sup_{\substack{0\le 2\vert t_2 \vert \le 1-\delta_0 \\ t_1\in \C}} \sup_{\substack{z_0 \in \C\setminus E_{1,t_1,t_2,+\delta} \\ a_1,\ldots,a_m\in \cc{B}_{\!A}(0) }} & \left\vert \tfrac{\Rn{m}}{\n^m} \left(z_0+\tfrac{a_1}{\sqrt{\n}},\ldots,z_0+\tfrac{a_m}{\sqrt{\n}} \right)\right\vert<\epsilon, \\
\shortintertext{and}
\sup_{\substack{0\le 2\vert t_2 \vert \le 1-\delta_0 \\ t_1\in \C}} \sup_{\substack{z_0 \in \partial E_{1,t_1,t_2} \\ a_1,\ldots,a_m\in \cc{B}_{\!A}(0) }} & \Big\vert \tfrac{\Rn{m}}{\n^m} \left(z_0+\tfrac{a_1\e^{i\psi}}{\sqrt{\n}},\ldots,z_0+\tfrac{a_m\e^{i\psi}}{\sqrt{\n}}\right) \nn \\
&\quad-\det \left( \tfrac{\erfc\left(\frac{\cc{a}_k+a_l}{\sqrt{2}}\right)}{2\pi}\e^{-\frac{1}{2}\vert a_k-a_l \vert^2+ i\Im \left(\cc{a}_k a_l\right) }  \right)_{k,l=1}^m \Big\vert<\epsilon. 
\end{align}

\end{proof}

\subsection{Kernel and Correlation Functions for $\protect\chapt_0\neq 1$}


\begin{theorem}[Uniform Universality of Renormalized Correlation Functions]\label{th_corr_uniform_ren} 
Let $\delta>0$, $0<\delta_0<1$ and $A>0$. $\nK$ shall be the normalized reproducing kernel for the potential $V(z)=\tfrac{1}{t_0}\left(\vert z\vert^2-2\Re\left(t_1 z + t_2 z^2\right)\right)$ with $t_0>0$, $t_1,t_2\in\C$ and $2\vert \t_2\vert<1$ and $\Rn{m}$ the corresponding $m$-point correlation function as in Definition \ref{def_correlation}. Then
\begin{align} 
\lim_{\n\rightarrow\infty} \sup_{\substack{0\le\vert 2 t_2 \vert \le 1-\delta_0 \\ t_1\in \C \\ t_0>0}} \sup_{\substack{z_0 \in E_{t_0,t_1,t_2,-\delta} \\ a_1,\ldots,a_m\in \cc{B}_{\!A}(0) }} & \Big\vert \tfrac{\Rn{m}}{\nK^m(z_0,z_0)} \left(z_0+\tfrac{a_1}{\sqrt{\nK(z_0,z_0)}},\ldots,z_0+\tfrac{a_m}{\sqrt{\nK(z_0,z_0)}} \right)\nn \\
&\quad - \det \left( \e^{-\frac{\pi}{2}\vert a_k-a_l \vert^2+i\pi \Im \left(\cc{a}_k a_l\right) }  \right)_{k,l=1}^m \Big\vert=0, \\ 
\lim_{\n\rightarrow\infty} \sup_{\substack{0\le\vert 2 t_2 \vert \le 1-\delta_0 \\ t_1\in \C \\ t_0>0}} \sup_{\substack{z_0 \in \partial E_{t_0,t_1,t_2} \\ a_1,\ldots,a_m\in \cc{B}_{\!A}(0) }} & \Big\vert \tfrac{\Rn{m}}{\nK^m(z_0,z_0)} \left(z_0+\tfrac{a_1\e^{i\psi(z_0)}}{\sqrt{\nK(z_0,z_0)}},\ldots,z_0+\tfrac{a_m\e^{i\psi(z_0)}}{\sqrt{\nK(z_0,z_0)}}\right) \nn \\
&\quad-\det \left( \erfc\left(\sqrt{\pi}(\cc{a}_k+a_l)\right)\e^{-\pi\vert a_k-a_l \vert^2+ 2i\pi\Im \left(\cc{a}_k a_l\right) }  \right)_{k,l=1}^m \Big\vert=0, \label{eq_corr_uniform_ren1}
\end{align}
where in \eqref{eq_corr_uniform_ren1} if $t_2\neq 0$,
\begin{align}
\psi(z_0)&=- \tfrac{\arg t_2}{2} + \left\{\begin{array}{ll} \arctan\left(\tfrac{1+2\vert t_2 \vert}{1-2 \vert t_2 \vert }\tan \phi \right) & \text{if } \phi\in (-\tfrac{\pi}{2},\tfrac{\pi}{2}), \\ \arctan\left(\tfrac{1+2\vert t_2 \vert}{1-2 \vert t_2 \vert }\tan \phi \right)+\sign(\phi)\pi & \text{if } \phi\notin [-\tfrac{\pi}{2},\tfrac{\pi}{2}],  \\  \phi & \text{if }\phi=\pm \tfrac{\pi}{2}, \end{array}\right.
\end{align}
with the parameter $\phi\in(-\pi,\pi]$ fulfilling $\sqrt{\tfrac{t_2}{\vert t_2\vert}}(z_0-v)=a_{t_0,t_2} \cos \phi+i b_{t_0,t_2} \sin\phi \in\partial E_{t_0,t_1,t_2}$ and if $t_2=0$,
\begin{align}
\psi(z_0)&=\phi, \nn 
\end{align}
with $z_0-v=\sqrt{t_0} \e^{i \phi} \in\partial E_{t_0,t_1,0}$. 

\end{theorem}

\begin{proof}
Let $(\op{k})_{k\in\N_0}$ be the orthonormal polynomials regarding the potential 
\begin{align}
V(z)&=\tfrac{1}{t_0}\left(\vert z\vert^2-2\Re\left(t_1 z + t_2 z^2\right)\right),
\intertext{Analogously $(\tildeop{k})_{k\in\N_0}$ shall be the orthonormal polynomials regarding the potential}
\tilde{V}(z)&=\vert z\vert^2-2\Re\left(\tilde{t}_1 z + t_2 z^2\right),
\end{align}
with $\tilde{t}_1= \tfrac{t_1}{\sqrt{t_0}}$.
From Proposition \ref{prop_op_t0t1} it is obvious that the orthonormal polynomials $\op{k}$ are related to $\tildeop{k}$ as
\begin{align}
\op{k}(z)=\tfrac{1}{\sqrt{t_0}} \tildeop{k}\left( \tfrac{z}{\sqrt{t_0}} \right), \qquad k\in \N_0.
\end{align}

As in the proof of Theorem \ref{th_corr_uniform} we replace the orthonormal polynomials $(\op{k})_{k\in\N_0}$ by $(\tildeop{k})_{k\in\N_0}$ in the definition of the kernel and get
\begin{align}
t_0 \tKn(w,z)&=t_0 \e^{-\frac{\n}{2} (V(w)+V(z))} \sum_{k=0}^{\n-1} \cc{\op{k}(w)}\op{k}(z) \nn \\
&= \e^{-\frac{\n}{2} \left(\tilde{V}(\tilde{w})+\tilde{V}(\tilde{z})\right)} \sum_{k=0}^{\n-1} \cc{\tildeop{k}(\tilde{w})}\tildeop{k}(\tilde{z}), \label{eq_tzero_transf_coord}
\end{align}
where we have introduced the scaled coordinates $\tilde{w}=\tfrac{w}{\sqrt{t_0}}$ and $\tilde{z}=\tfrac{z}{\sqrt{t_0}}$. We see that in the coordinates $\tilde{w}$ and $\tilde{z}$ the right-hand side of \eqref{eq_tzero_transf_coord} is exactly the normalized reproducing kernel for the potential $\tilde{V}$. When we set $\tilde{a}=\tfrac{a}{\sqrt{t_0}}$ and $\tilde{b}=\tfrac{b}{\sqrt{t_0}}$, we therefore know the asymptotics of the kernel and the $m$-point correlation functions evaluated at $w=z_0+\tfrac{a}{\sqrt{n}}$ and $z=z_0+\tfrac{b}{\sqrt{n}}$ by Theorem \ref{th_corr_uniform}. Particularly, if we set $w=z=z_0$ we find that $t_0 \tfrac{\tKn}{\n}(z_0,z_0)$ converges to $\tfrac{1}{\pi}$ uniformly for $z_0\in E_{t_0,t_1,t_2,-\delta}$ and to $\tfrac{1}{2\pi}$ uniformly for $z_0\in \partial E_{t_0,t_1,t_2}$. Note that from \eqref{eq_shift2} in Lemma \ref{lemma_shift} follows that $\tilde{v}=\tfrac{v}{\sqrt{t_0}}$ scales as $\tilde{t}_1$ and therefore $\tilde{z}_0\in E_{1,\tilde{t}_1,t_2,\pm\delta}$ is equivalent to $z_0\in E_{t_0,t_1,t_2,\pm\delta}$ and $\tilde{z}_0\in \partial E_{1,\tilde{t}_1,t_2}$ to $z_0\in \partial E_{t_0,t_1,t_2}$ respectively.

Using that the convergence is uniform we can replace $t_0 \tfrac{\tKn}{\n}(z_0,z_0)$ by $\tfrac{1}{\pi}$. Since further the kernel and the correlation functions converge uniformly according to Theorem \ref{th_corr_uniform} we can replace $\left( t_0 \tfrac{\tKn}{\n}(z_0,z_0) \right)^{1/2}$ by $\tfrac{1}{\sqrt{\pi}}$ in the arguments too. From this follows the proof of the theorem. 
\end{proof}







\appendix
\section{Estimations and Approximations} 
\subsection{Estimations for $\U{\F}$, $\T{\F}$ and $\W{\F}$}

We already know from Corollary \ref{cor_U2} that $\vert \U{\F}(z) \vert \le 1$. Here we will give some additional estimations for $\U{\F}$, $\T{\F}$ and $\W{\F}$ from Definitions \ref{def_U}, \ref{def_T} and \ref{def_W}.


\begin{lemma}\label{lemma_estim_U_T}
Let $z\in \C\setminus [-\F,\F]$ be parametrized by $z=a \cos\phi+i b \sin\phi $ with $b>0$, $a=\sqrt{b^2+\F^2}>b$ and $\phi\in(-\pi,\pi]$. Then for all $\F\in(0,\infty)$ and $z\in \C\setminus [-\F,\F]$
\begin{center}
\begin{minipage}{50mm}
\begin{enumerate}[\normalfont (i)]
\item 
$\displaystyle\quad \frac{\vert \U{\F}(z)\vert}{\F} < \frac{1}{2b}$,
\item
$\displaystyle\quad \frac{1}{\vert\T{\F}(z)\vert}<\frac{1}{b}$.
\end{enumerate}
\end{minipage}
\end{center}
\end{lemma}
\begin{proof}
According to Proposition \ref{prop_U} $\tfrac{\vert \U{\F}(z)\vert}{\F} =\tfrac{\sqrt{b^2+\F^2}-b}{\F^2}$. Since
\begin{align}
\tfrac{\partial}{\partial \F}\tfrac{\vert \U{\F}(z)\vert}{\F}=\tfrac{\F^2-2(b^2+\F^2)+2b\sqrt{b^2+\F^2}}{F^3\sqrt{b^2+\F^2}}<0 \iff  2b \sqrt{b^2+\F^2} < 2b^2+\F^2 \iff 0<\F^4,
\end{align}
$\tfrac{\vert \U{\F}(z)\vert}{\F}$ is decreasing in $\F$ independently of $b$. Therefore we can find the supremum for $0<\F<\infty$ by the limit
\begin{align}
\lim_{\F\rightarrow 0} \tfrac{\vert \U{\F}(z)\vert}{\F}= \lim_{\F\rightarrow 0} \tfrac{b\left(\frac{\F^2}{2b^2}+\lO\left(\frac{\F^4}{b^4}\right)\right)}{\F^2}=\tfrac{1}{2b}.
\end{align}
The second claim follows directly from Lemma \ref{lemma_T} since
\begin{align}
\tfrac{1}{\vert\T{\F}(z)\vert^2}=\tfrac{1}{b^2+\F^2\sin^2\phi}<\tfrac{1}{b^2}.
\end{align}
\end{proof}

\begin{lemma}\label{lemma_estim_T2}
Let $\F_0\in (0,\infty)$ and $z\in \C\setminus [-\F,\F]$ be parametrized by $z=a \cos\phi+i b \sin\phi $ with $b>0$, $a=\sqrt{b^2+\F^2}>b$ and $\phi\in(-\pi,\pi]$. Then for all $\F\in(0,\F_0)$ and $z\in \C\setminus [-\F,\F]$
\begin{align}
\frac{\vert z\vert}{\vert\T{\F}(z)\vert^2}<\frac{b+\F_0}{b^2}.
\end{align}
\end{lemma}
\begin{proof}
From Lemma \ref{lemma_T} we get
\begin{align}
\tfrac{\vert z\vert}{\vert\T{\F}(z)\vert^2}=\tfrac{\vert a \cos\phi+i b\sin\phi\vert}{b^2+\F^2\sin^2\phi}<\tfrac{\sqrt{b^2+2b\F\vert \cos \phi\vert +\F^2\cos^2\phi} }{b^2}\le \tfrac{b+\F}{b^2}<\tfrac{b+\F_0}{b^2}.
\end{align}
\end{proof}

\begin{lemma}\label{lemma_estim_W}
Let $\delta>0$, $\F_0\in(0,\infty)$ and $z\in \C\setminus [-\F,\F]$ be parametrized by $z=a \cos\phi+i b \sin\phi $ with $a^2=b^2+\F^2$ and $\phi\in(-\pi,\pi]$. Then for all $0<\F<\F_0$ 
\begin{align}
\vert \W{\F}(z) \vert^2 \le 2\left(1+\sqrt{1+\tfrac{\F_0^2}{b^2} } \right),
\end{align}
\end{lemma}
\begin{proof}
According to Lemma \ref{lemma_W} we know that $\vert \W{\F}(z)\vert^2= \frac{2(a+b)}{\sqrt{b^2+\F^2 \sin^2\phi}}$. Therefore 
\begin{align}
\vert \W{\F}(z)\vert^2&\le \frac{2(a+b)}{b}=2\left(1+\sqrt{1+\tfrac{\F^2}{b^2}} \right)<2\left(1+\sqrt{1+\tfrac{\F_0^2}{b^2}} \right).
\end{align}
\end{proof}

\begin{lemma}\label{lemma_estim_U2}
Let $\delta_1>0$ and $z\in \C\setminus [-\F,\F]$ be parametrized by $z=a \cos\phi+i b \sin\phi $ with $b>0$, $a=\sqrt{b^2+\F^2}>b$ and $\phi\in(-\pi,\pi]$. Then for all $\F\in(0,\infty)$
\begin{center}
\begin{minipage}{90mm}
\begin{enumerate}[\normalfont (i)]
\item 
$\displaystyle\quad \vert \Re\, \U{\F}(z)\vert > \tfrac{\delta_1}{\F}\left(1-\tfrac{b}{a}\right), \qquad \forall z\in \C_{x,\delta_1}$,
\item
$\displaystyle\quad \vert \Im\, \U{\F}(z)\vert > \tfrac{\delta_1}{\F}(\tfrac{a}{b}-1), \qquad \forall z\in \C_{y,\delta_1}$.
\end{enumerate}
\end{minipage}
\end{center}
\end{lemma}
\begin{proof}
According to Proposition \ref{prop_U}, $\Re\, \U{\F}(z)=\tfrac{a-b}{\F}\cos\phi$ and $\Im\, \U{\F}(z)=-\tfrac{a-b}{\F}\sin\phi$. $\vert \Re\, z\vert>\delta_1$ is equivalent to $a \vert \cos \phi\vert>\delta_1$ and $\vert \Im\, z\vert>\delta_1$ to $b \vert \sin \phi\vert>\delta_1$.
From this follows that
\begin{align}
\vert \Re\, \U{\F}(z)\vert&>\tfrac{\delta_1}{a}\tfrac{a-b}{\F}=\tfrac{\delta_1}{\F}(1-\tfrac{b}{a}), \qquad \forall z\in \C_{x,\delta_1},
\shortintertext{and}
\vert \Im\, \U{\F}(z)\vert&>\tfrac{\delta_1}{b}\tfrac{a-b}{\F}=\tfrac{\delta_1}{\F}(\tfrac{a}{b}-1), \qquad \forall z\in \C_{y,\delta_1}.
\end{align}
\end{proof}

\subsection{Taylor Polynomials}
\begin{lemma}\label{lemma_R2_R3_new}
Let $\n_0\in \N$. Then there exists functions $\xi_2$ and $\xi_3$
\begin{align*}
\begin{array}{rcl}
\xi_k: \N\setminus\{1,\ldots,\n_0-1\} & \longrightarrow & \left[0,\tfrac{1}{\n_0}\right]\\
\n & \longmapsto & \xi_k(\n)\le \frac{1}{\n}
\end{array}, \qquad k=2,3,
\end{align*}
so that for all $0<\F<\infty$, $z\in \C\setminus [-\F,\F]$ and $\n>\n_0$
\begin{enumerate}[\normalfont (i)]\item \begin{align} \MoveEqLeft \left(\tfrac{\n-1}{\n}\right)^{\n-1} \e^{\frac{\n}{2} \Bigl( \sqrt{1-\frac{1}{\n}} \U{\F}\left(  z \cramped[\scriptscriptstyle]{\left(1-\frac{1}{\n}\right)^{-1/2}} \right) \Bigr)^2 } \left( \sqrt{1-\tfrac{1}{\n}}  \U{\F}\left(  z \left(1-\tfrac{1}{\n}\right)^{-1/2} \right) \right)^{-(\n-1)}\nn\\
&=\e^{\frac{\n}{2}\left( \U{\F}(z) \right)^2} \left( \U{\F}(z) \right)^{-(\n-1)} \e^{\Rtwo}, \label{eq_f2}
\shortintertext{with} 
\MoveEqLeft \Rtwo=-\tfrac{\F\U{\F}\left(\frac{z}{\sqrt{1-\xi_2(\n) } }  \right)  }{4\n\left(1-\xi_2(\n) \right) \T{\F}\left(\frac{z}{\sqrt{1-\xi_2(\n) } }  \right) }, 
\end{align}

\item \begin{align} \MoveEqLeft  \W{\F}\left(  z \left(1-\tfrac{1}{\n}\right)^{-1/2} \right) =   
\W{\F}(z) \e^{\Rthree}, \label{eq_f3}
\shortintertext{with} 
\MoveEqLeft\Rthree=-\tfrac{z\F\U{\F}\left(\frac{z}{\sqrt{1-\xi_3(\n) } }  \right)  }{4\n\left(1-\xi_3(\n) \right)^{3/2} \left(\T{\F}\left(\frac{z}{\sqrt{1-\xi_3(\n) } }  \right)\right)^2  }.
\end{align}
\end{enumerate}
\end{lemma}
\begin{remark}
Note that $\xi_2$ and $\xi_3$ also depend on $z$ and $\t$. 
\end{remark}
\begin{proof}
We set $\alpha=\tfrac{1}{\n}<1$ and define the left-hand side of \eqref{eq_f2} and \eqref{eq_f3} as $\e^{\n\ftwo(\alpha)}$ and $\e^{\fthree{}(\alpha)}$, i.e.\  
\begin{align}\label{eq_error_f2_new}
\ftwo(\alpha)&=(1-\alpha)\log(1-\alpha)+\tfrac{1}{2} \left( \sqrt{1-\alpha}\U{\F}\bigl(\tfrac{z}{\sqrt{1-\alpha}} \bigr)  \right)^2   -(1-\alpha)\log\left( \sqrt{1-\alpha}\U{\F}\bigl(\tfrac{z}{\sqrt{1-\alpha}} \bigr)  \right),
\end{align}
and
\begin{align}
\fthree{}(\alpha)&=\log\left( \W{\F}\bigl(\tfrac{z}{\sqrt{1-\alpha}} \bigr) \right). \label{eq_error_f3_new}
\end{align}
According to Proposition \ref{prop_U} we see that $\sqrt{1-\alpha} \U{\F}\left(z (1-\alpha)^{-1/2} \right)\neq 0$ and according to Lemma \ref{lemma_W} that $ \W{\F}\left(z (1-\alpha)^{-1/2} \right)\neq 0$ for all $z\in \C\setminus [-\F,\F]$ and $\alpha\in [0,1)$. Therefore it is easy to see that $\ftwo,\fthree{}\in C^{\infty}\left([0,1)\right)$ for all $z\in \C\setminus [-\F,\F]$. For our need it is sufficient that we can continuously differentiate $\ftwo$ twice and $\fthree{}$ once regarding $\alpha$. 
We will use Taylor's theorem to expand $\ftwo$ and $\fthree{}$ at $\alpha=0$ and we will write the remainder in Lagrange form.

For $\ftwo$ we need the Taylor polynomial up to order one. We compute
\begin{align}
\ftwo'(\alpha)&=-\log(1-\alpha)+\log\left( \sqrt{1-\alpha}\U{\F}\bigl(\tfrac{z}{\sqrt{1-\alpha}} \bigr) \right),\\
\ftwo''(\alpha)&= -\frac{ \F \U{\F}\bigl(\tfrac{z}{\sqrt{1-\alpha}} \bigr) }{2(1-\alpha) \T{\F}\bigl(\tfrac{z}{\sqrt{1-\alpha}} \bigr) }.
\end{align}
If $\n>\n_0$ then by Taylor's theorem there exists $\xi_2(\n)\in [0,\tfrac{1}{\n}]$ so that
\begin{align}
\ftwo(\tfrac{1}{\n})&=\tfrac{\left(\U{\F}(z)\right)^2}{2}-\log\left(\U{\F}(z) \right)+\tfrac{1}{\n} \log\left(\U{\F}(z) \right) 
-\tfrac{1}{\n^2} \tfrac{\F\U{\F}\left(\frac{z}{\sqrt{1-\xi_2(\n) } }  \right)  }{4\left(1-\xi_2(\n) \right) \T{\F}\left(\frac{z}{\sqrt{1-\xi_2(\n) } }  \right) },   
\shortintertext{and}
\e^{\n \ftwo\left(\tfrac{1}{\n}\right)}&=\e^{\frac{\n}{2} \left( \U{\F}(z) \right)^2} \left( \U{\F}(z) \right)^{-\n} \U{\F}(z) \e^{\Rtwo },
\end{align}
where we have defined the error term
\begin{align}
\Rtwo 
&=-\tfrac{\F\U{\F}\left(\frac{z}{\sqrt{1-\xi_2(\n) } }  \right)  }{4\n\left(1-\xi_2(\n) \right) \T{\F}\left(\frac{z}{\sqrt{1-\xi_2(\n) } }  \right) }.
\end{align}

The Taylor polynomial of $\fthree{}$ we only need up to order zero.
\begin{align}
\fthree{}'(\alpha)
&=\frac{\F \U{\F}\bigl(\tfrac{z}{\sqrt{1-\alpha}} \bigr) }{4 (1-\alpha)\T{\F}\bigl(\tfrac{z}{\sqrt{1-\alpha}} \bigr)} -\tfrac{\F^2}{  4 (1-\alpha)\left(\T{\F}\bigl(\tfrac{z}{\sqrt{1-\alpha}} \bigr)\right)^2} = -\tfrac{z \F \U{\F}\bigl(\tfrac{z}{\sqrt{1-\alpha}} \bigr) }{4 (1-\alpha)^{3/2}\left(\T{\F}\bigl(\tfrac{z}{\sqrt{1-\alpha}} \bigr) \right)^2 }.
\end{align}
Therefore for $\n>\n_0$ there exists $\xi_3(\n)\in [0,\tfrac{1}{\n}]$ so that
\begin{align}
\e^{\fthree{}\left(\frac{1}{\n}\right)}&=\W{\F}(z) \e^{\Rthree}
\end{align}
where the error term is
\begin{align}
\Rthree
&=-\tfrac{z\F\U{\F}\left(\frac{z}{\sqrt{1-\xi_3(\n) } }  \right)  }{4\n\left(1-\xi_3(\n) \right)^{3/2} \left(\T{\F}\left(\frac{z}{\sqrt{1-\xi_3(\n) } }  \right)\right)^2  }.
\end{align}
\end{proof}
\begin{remark}
Because of (skew)-symmetries of the left-hand side of \eqref{eq_f2} and \eqref{eq_f3} it is easy to see that $\xi_2$ and $\xi_3$ must have the symmetries $\xi_k(\n,-z)=\xi_k(\n,z)$ and $\xi_k(\n,\cc{z})=\xi_k(\n,z)$ for $k=2,3$.
\end{remark}

\begin{lemma}\label{lemma_R4_R5_R6}
Let $\delta>0$ and $\n_0\in \N$. 
Then there exists functions $\xi_4$, $\xi_5$ and $\xi_6$
\begin{align*}
\begin{array}{rcl}
\xi_k: \cc{B}_{\delta}(0) & \longrightarrow & \cc{B}_{\delta}(0) \\
v & \longmapsto & \xi_k(v)\in \cc{B}_{\vert v\vert}(0)
\end{array}, \qquad k=4,5,6
\end{align*}
so that for all $0<\F<\infty$, $z_0\in \Adelta$, $v\in B_{\delta}(0)$, $\xi_2,\xi_3\in\left[0,\tfrac{1}{2}\right] $ and $\n>\n_0$ 
\begin{enumerate}[\normalfont (i)]\item \begin{align}  
\MoveEqLeft\e^{\frac{\left(\U{\F}(z_0+v)\right)^2}{2}} \tfrac{1}{\U{\F}(z_0+v)}=\e^{\frac{\left(\U{\F}(z_0)\right)^2}{2}} \tfrac{1}{\U{\F}(z_0)} \e^{\frac{2v \U{\F}(z_0)}{\F} -\frac{v^2 \U{\F}(z_0)}{\F \T{\F}(z_0)} +  \Rfour}, \label{eq_f4}
\shortintertext{with}
\MoveEqLeft \Rfour=\frac{v^3}{3 \bigl(\T{\F}\left(z_0+\xi_4(v) \right)\bigr)^3 } ,
\end{align} 

\item \begin{align} \MoveEqLeft \W{F}(z_0+v)=  \W{\F}(z_0)\e^{ \Rfive} \label{eq_f5}, 
\shortintertext{with}
\MoveEqLeft \Rfive=-\frac{v \F \U{\F}(z_0+\xi_5(v))}{2 \bigl(\T{\F}(z_0+\xi_5(v))\bigr)^2 } ,
\end{align}

\item \begin{align} \MoveEqLeft \U{\F}(z_0+v)\e^{\Rtwoz{z_0+v}}\e^{\Rthreez{z_0+v}} \nn \\ 
&\qquad =  \U{\F}(z_0)   \e^{\Rtwoz{z_0} + \Rthreez{z_0} + \Rsix}, \label{eq_f6}
\shortintertext{with} 
\Rsix=& -\frac{v}{\T{\F}(z_0+\xi_6(v)) } +\frac{v \F^2}{4\n\left(1-\xi_2(\n) \right)^{3/2}\left(\T{\F}\Bigl(\frac{z_0+\xi_6(v)}{\sqrt{1-\xi_2(\n) } } \Bigr)\right)^3 } \nn \\
&  + \frac{v\F^2(z_0+\xi_6(v))^2 - v ( 1-\xi_3(\n) )\F^4  \left(\U{\F}\Bigl(\frac{z_0+\xi_6(v)}{\sqrt{1-\xi_3(\n) } }  \Bigr) \right)^2  }{4\n\left(1-\xi_3(\n) \right)^{5/2} \left(\T{\F}\Bigl(\frac{z_0+\xi_6(v)}{\sqrt{1-\xi_3(\n) } }  \Bigr)\right)^5  }, 
\end{align}
and $\Rtwotz{}$ and $\Rthreetz{}$ as in Lemma \ref{lemma_R2_R3_new}.
\end{enumerate}
\end{lemma}
\begin{remark}
Note that $\xi_4$ and $\xi_5$ depend also on $z_0$ and $\t$. $\xi_6$ additionally depends on $\n, \xi_2, \xi_3$. 
\end{remark}
\begin{proof}
We define the left hand side of \eqref{eq_f4}, \eqref{eq_f5} and \eqref{eq_f6} as $\e^{\ffour(v)}$, $\e^{\ffive(v)}$ and $\e^{\fsix(v)}$ respectively, i.e.\  
\begin{align}
\ffour(v)&=\tfrac{\left(\U{\F}(z_0+v)\right)^2}{2}-\log\left(\U{\F}(z_0+v)\right),\\
\ffive(v)&=\log\left(\W{F}(z_0+v) \right) ,\\
\shortintertext{and}
\fsix(v)& = \log\left( \U{\F}(z_0+v)\right) +  \Rtwoz{z_0+v} + \Rthreez{z_0+v} \nn \\
&=\log\left(\U{\F}(z_0+v)\right) \nn \\ 
&\phantom{=} -\tfrac{\F\U{\F}\left(\frac{z_0+v}{\sqrt{1-\xi_2(\n) } }  \right)  }{4\n\left(1-\xi_2(\n) \right) \T{\F}\left(\frac{z_0+v}{\sqrt{1-\xi_2(\n) } }  \right) }   -\tfrac{(z_0+v)\F\U{\F}\left(\frac{z_0+v}{\sqrt{1-\xi_3(\n) } }  \right)  }{4\n\left(1-\xi_3(\n) \right)^{3/2} \left(\T{\F}\left(\frac{z_0+v}{\sqrt{1-\xi_3(\n) } }  \right)\right)^2  } .
\end{align}
Obviously under the map $(z_0,v)\mapsto (-z_0,-v)$, $\e^{\ffive(v)}$ is symmetric whereas $\e^{\ffour(v)}$ and $\e^{\fsix(v)}$ are skew symmetric. Therefore it is sufficient to consider only $z_0\in \Cp$ but all $v\in B_{\delta}(0)$. For $z_0\notin \Cp$ we can get the relations by $\e^{\ffourz{-z_0}(v)}=-\e^{\ffourz{z_0}(-v)}$, $\e^{\ffivez{-z_0}(v)}=\e^{\ffivez{z_0}(-v)}$ and $\e^{\fsixz{-z_0}(v)}=-\e^{\fsixz{z_0}(-v)}$.

We choose $\F\in (0,\infty)$ arbitrarily. $\{z_0+v\in \C \vert z_0\in \Adelta, v\in B_\delta\}=\C\setminus[-\F,\F]$. If $z_0\in \Adelta\cap\Cp$ and $v\in B_\delta$, then $z_0+v \notin (-\infty,\F]$ and so we see that according to Proposition \ref{prop_U}, $U(z_0+v)\notin (-\infty,0]$. So we know that not only $\U{\F}(z_0+v)$ (according to Lemma \ref{lemma_holomorh_U_W}) but also $\log\left(\U{\F}(z_0+v) \right)$ are holomorphic as functions of $v\in B_\delta(0)$ for all $z_0\in\Adelta\cap\Cp$. We see, also from Lemma \ref{lemma_holomorh_U_W}, that $\W{\F}(z_0+v)$ is holomorphic as function of $v\in B_\delta(0)$ for all $z_0\in\Adelta$ and from Lemma \ref{lemma_W} that $\W{\F}(z_0+v)\notin(-\infty,0]$. Therefore also $\log\left(\W{\F}(z_0+v)\right)$ is a holomorphic function of $v$ for $v\in B_\delta(0)$ and $z_0\in\Adelta$. Again by Lemma \ref{lemma_holomorh_U_W} we know that $\T{\F}(z_0+v)$ too is a holomorphic function. $\T{\F}(z_0+v)\neq 0$ for all $v\in B_\delta(0)$ and $z_0\in\Adelta$. Further $\tfrac{z_0+v}{\sqrt{1-\xi}}\notin [-\F,\F]$ for all $\xi\in [0,1)$. Therefore $\ffour$, $\ffive$ and $\fsix$ are holomorphic on $B_\delta(0)$ for all $z_0\in\Adelta\cap \Cp$, $\xi_2,\xi_3\in [0,1)$, $\n\in\N$.

This time we will use the complex version of Taylor's theorem. For $l\in\N_0$ and a holomorphic function $f_{z_0}$ on $B_\delta(0)$ there exists a function $\xi$ so that for all $v\in B_\delta(0)$
\begin{align}
f_{z_0}(v)&=\sum_{k=0}^{l} \frac{f_{z_0}^{(k)}(0)}{k!} v^k + R_{z_0,l}(v),\\
\intertext{where the remainder in Lagrange form is}
R_{z_0,l}(v)&=v^{l+1} \frac{f_{z_0}^{(l+1)}(\xi(v))}{(l+1)!},
\end{align}
with $\xi(v)\in B_{\vert v\vert}(0)$. 

For $\ffour$ we will compute the Taylor polynomial up to order two. 
\begin{align}
\ffour'(v)&=\tfrac{2 \U{\F}(z_0+v) }{\F} ,\\
\ffour''(v)&=-\tfrac{2 \U{\F}(z_0+v) }{\F \T{\F}(z_0+v) } ,\\
\ffour'''(v)&=\tfrac{2}{ \left( \T{\F}(z_0+v) \right)^3 }.
\end{align}
If $v\in B_\delta(0)$ then by Taylor's theorem there exists $\xi_4(v)\in B_{\vert v\vert}(0)$ so that
\begin{align}
\ffour(v)&=\tfrac{\left(\U{\F}(z_0)\right)^2}{2}-\log\left(\U{\F}(z_0)\right) +\tfrac{2v \U{\F}(z_0) }{\F} -\tfrac{v^2 \U{\F}(z_0) }{\F \T{\F}(z_0+v)} + \tfrac{v^3}{ 3 \left( \T{\F}(z_0+\xi_4(v)) \right)^3 },\\
\intertext{and we find}
\e^{\ffour(v)}&= \e^{\frac{\left(\U{\F}(z_0)\right)^2}{2}} \tfrac{1}{\U{\F}(z_0)} \e^{\frac{2v \U{\F}(z_0)}{\F} -\frac{v^2 \U{\F}(z_0)}{\F \T{\F}(z_0)} + \Rfour}, 
\end{align}
where we have defined the error term
\begin{align}
\Rfour&=\tfrac{v^3}{3 \left(\T{\F}\left(z_0+\xi_4(v) \right)\right)^3}.
\end{align}

The Taylor polynomial of $\ffive$ we only need up to order zero.
\begin{align}
\ffive'(v)&=-\frac{\F \U{\F}(z_0+v)}{2 \left(\T{\F}(z_0+v)\right)^2 } .\\
\intertext{Thus for $v\in B_\delta(0)$ there exists $\xi_5(v)\in B_{\vert v \vert}(0)$ so that}
\ffive(v)&=\log\left(\W{\F}(z_0)\right)-\frac{v \F \U{\F}(z_0+\xi_5(v))}{2 \left(\T{\F}(z_0+\xi_5(v))\right)^2 }, \\
\intertext{and we get}
\e^{\ffive(v)}&=\W{\F}(z_0)\e^{ \Rfive} , \\
\intertext{where the error term is}
\Rfive&=-\frac{v \F \U{\F}(z_0+\xi_5(v))}{2 \left(\T{\F}(z_0+\xi_5(v))\right)^2 }.
\end{align}

For $\fsix$ we also need the Taylor polynomial up to order zero.
\begin{align}
\fsix'(v)&=-\tfrac{1}{\T{\F}(z_0+v) } +  \tfrac{\F^2}{4\n\left(1-\xi_2(\n) \right)^{3/2}\left(\T{\F}\left(\frac{z_0+v}{\sqrt{1-\xi_2(\n) } } \right)\right)^3 } \nn \\ 
&\phantom{=}+\tfrac{\F^2(z_0+v)^2 - ( 1-\xi_3(\n) )\F^4  \left(\U{\F}\left(\frac{z_0+v}{\sqrt{1-\xi_3(\n) } }  \right) \right)^2  }{4\n\left(1-\xi_3(\n) \right)^{5/2} \left(\T{\F}\left(\frac{z_0+v}{\sqrt{1-\xi_3(\n) } }  \right)\right)^5  }     .\\
\intertext{So for $v\in B_\delta(0)$ there exists $\xi_6(v)\in B_{\vert v \vert}(0)$ so that}
\fsix(v)&=\log\left(\U{\F}(z_0)\right) +  \Rtwoz{z_0} + \Rthreez{z_0} \nn \\
& \phantom{=} -\tfrac{v}{\T{\F}(z_0+\xi_6(v)) } +  \tfrac{v \F^2}{4\n\left(1-\xi_2(\n) \right)^{3/2}\left(\T{\F}\left(\frac{z_0+\xi_6(v)}{\sqrt{1-\xi_2(\n) } } \right)\right)^3 } \nn \\
& \phantom{=} +\tfrac{v \F^2(z_0+\xi_6(v))^2 - v( 1-\xi_3(\n) )\F^4  \left(\U{\F}\left(\frac{z_0+\xi_6(v)}{\sqrt{1-\xi_3(\n) } }  \right) \right)^2  }{4\n\left(1-\xi_3(\n) \right)^{5/2} \left(\T{\F}\left(\frac{z_0+\xi_6(v)}{\sqrt{1-\xi_3(\n) } }  \right)\right)^5  }  ,         
\intertext{and we get}
\e^{\fsix(v)}&= \U{\F}(z_0) \e^{\Rtwoz{z_0} + \Rthreez{z_0} + \Rsix}, 
\end{align}
where we have defined the error term
\begin{align}
\Rsix&= -\tfrac{v}{\T{\F}(z_0+\xi_6(v)) } +\tfrac{v \F^2}{4\n\left(1-\xi_2(\n) \right)^{3/2}\left(\T{\F}\left(\frac{z_0+\xi_6(v)}{\sqrt{1-\xi_2(\n) } } \right)\right)^3 } \nn \\
& \phantom{=} +\tfrac{v\F^2(z_0+\xi_6(v))^2 - v( 1-\xi_3(\n) )\F^4  \left(\U{\F}\left(\frac{z_0+\xi_6(v)}{\sqrt{1-\xi_3(\n) } }  \right) \right)^2  }{4\n\left(1-\xi_3(\n) \right)^{5/2} \left(\T{\F}\left(\frac{z_0+\xi_6(v)}{\sqrt{1-\xi_3(\n) } }  \right)\right)^5  }. 
\end{align}
\end{proof}
\begin{remark}
Because of (skew)-symmetries of the left-hand side of \eqref{eq_f4}, \eqref{eq_f5} and \eqref{eq_f6} under the maps $(z_0,v)\mapsto (-z_0,-v)$ and $(z_0,v)\mapsto (\cc{z}_0,\cc{v})$, $\xi_4$, $\xi_5$ and $\xi_6$ must have the symmetries $\xi_k(-z_0,-v)=-\xi_k(z_0,v)$ and $\xi_k(\cc{z}_0,\cc{v})=\cc{\xi_k(z_0,v)}$ for $k=4,5,6$.
\end{remark}

\subsection{Estimations of Error Terms}
Now we want to find some estimations for the error terms on $\Adelta$. We already have estimated the basic functions where we always used the parametrization $z=a \cos \phi+i b\sin \phi$. It is handy to have an estimation for $b$ when $z\in\Adelta$. 
\begin{lemma}\label{lemma_estim_b}
Let $\delta>0$ and $z\in \C\setminus [-\F,\F]$ be parametrized by $z=a \cos\phi+i b \sin\phi $ with $b>0$, $a=\sqrt{b^2+\F^2}>b$ and $\phi\in(-\pi,\pi]$. Then it follows that $b\ge\delta$ for all $\F\in(0,\infty)$ and $z\in \Adelta$.
\end{lemma}
\begin{proof}
$z\in \Adelta$ is equivalent to $\vert \Im\, z\vert \ge \delta \lor \vert \Re\, z\vert \ge \F+\delta$.
\begin{enumerate}
\item We assume that $\vert \Im\, z\vert \ge \delta$. Then $b \vert\sin \phi\vert \ge \delta$ which implies that $b\ge\delta$. 
\item We assume that $\vert \Re\, z\vert \ge \F+\delta$. Then 
\begin{align*}
a \vert\cos \phi\vert\ge \F+\delta &\Longrightarrow \sqrt{b^2+\F^2}\ge F+\delta \iff b^2+\F^2\ge \F^2+2\F\delta+\delta^2\nn \\
&\Longrightarrow b^2\ge \delta^2.
\end{align*}
\end{enumerate}
\end{proof}

\begin{lemma}\label{lemma_estim_R2_R3}
Let $\epsilon>0$, $\delta>0$ and $0<\delta_0<1$. Then there exists $\n_0=\n_0(\epsilon,\delta,\delta_0)\in\N$ so that for all $\n> \n_0$
\begin{align}
\sup_{0<\t<1-\delta_0} \sup_{\substack{z\in \Adelta \\ \xi_2(\n)\in [0,\frac{1}{2}] } } \vert \Rtwo \vert <\epsilon, \qquad \text{and}\qquad \sup_{0<\t<1-\delta_0} \sup_{\substack{z\in \Adelta \\ \xi_3(\n)\in [0,\frac{1}{2}] } } \vert \Rthree \vert <\epsilon,
\end{align}
where $\Rtwotz{}$ and $\Rthreetz{}$ are as in Lemma \ref{lemma_R2_R3_new}.
\end{lemma}
\begin{proof}
The largest values of $\F$ is $\F_0=\sqrt{\tfrac{4(1-\delta_0)}{2\delta_0-\delta_0^2}}$. 
If $z\in \Adelta$ then also $\tfrac{z}{\sqrt{1-\xi}}\in \Adelta$ for all $\xi\in [0,\tfrac{1}{2}]$. We estimate $\vert \U{\F}(z)\vert<1$ according to Corollary \ref{cor_U2}, $1-\xi_2(\n)>\tfrac{1}{2}$ and $\F\le \F_0$. When we use the parametrization $z=a \cos\phi+i b \sin \phi$ with $a^2=b^2+\F^2$, we know from Lemma \ref{lemma_estim_U_T} that $\tfrac{1}{\vert\T{\F}(z)\vert}<\tfrac{1}{b}$. Thus we get
\begin{align}
\vert \Rtwo \vert &< \frac{\F_0}{2 \n b},
\intertext{and analogously with $\tfrac{\vert z\vert}{\vert \T{\F}(z)\vert^2}<\tfrac{b+\F_0}{b^2}$, according to Lemma \ref{lemma_estim_T2}, and since $(1-\xi_3(\n))^{3/2}>\tfrac{1}{4}$, }
\vert \Rthree \vert &< \frac{\F_0}{ \n b}+\frac{\F_0^2}{ \n b^2}.
\end{align}
Both estimations have their maximum when $b$ is minimal. But we know from Lemma \ref{lemma_estim_b} that $b\ge\delta$ when $z\in \Adelta$. Therefore
\begin{align}
\vert \Rtwo \vert &< \frac{\F_0}{2 \n \delta}, \qquad \text{and}\qquad \vert \Rthree \vert < \frac{\F_0}{ \n \delta}+\frac{\F_0^2}{ \n \delta^2}.
\end{align}
Then we see that there exists $\n_0$ so that for all $\n>\n_0$ both estimations are smaller than $\epsilon$, uniformly for all $z\in \Adelta$, $0<\F<\F_0$, $\xi_2,\xi_3\in[0,\tfrac{1}{2}]$.  
\end{proof}

\begin{lemma}\label{lemma_estim_R4_R6}
Let $\epsilon>0$, $A\in(0,\infty)$, $\delta>0$ and $0<\delta_0<1$. Then there exists $\n_0=\n_0(\epsilon,A,\delta,\delta_0)\in\N$ so that for all $\n\ge \n_0$
\begin{enumerate}[\normalfont (i)]
\item
\begin{align}
\sup_{\substack{ 0<\t<1-\delta_0 \\ v\in \cc{B}_{\!\frac{A}{\sqrt{\n}}}(0) }} \sup_{\substack{z_0\in \Adelta \\ \xi_4(v)\in \cc{B}_{\!\frac{A}{\sqrt{\n}}}(0) } } \vert \n \Rfour \vert <\epsilon, 
\end{align}
\item
\begin{align}
\sup_{\substack{ 0<\t<1-\delta_0 \\ v\in \cc{B}_{\!\frac{A}{\sqrt{\n}}}(0) }} \sup_{\substack{z_0\in \Adelta \\ \xi_5(v)\in \cc{B}_{\!\frac{A}{\sqrt{\n}}}(0) } } \vert \Rfive \vert <\epsilon, 
\end{align}
\item
\begin{align}
\sup_{\substack{ 0<\t<1-\delta_0 \\ v\in \cc{B}_{\!\frac{A}{\sqrt{\n}}}(0) }} \sup_{\substack{z_0\in \Adelta \\ \xi_2(\n),\xi_3(\n)\in [0,\frac{1}{2}] \\ \xi_6(v)\in \cc{B}_{\!\frac{A}{\sqrt{\n}}}(0) } } \vert \Rsix \vert <\epsilon, 
\end{align}
\end{enumerate}
where $\Rfourtz{}$, $\Rfive{}$ and $\Rsix{}$ are as in Lemma \ref{lemma_R4_R5_R6}.
\end{lemma}
\begin{proof}
Let $\F_0=\tfrac{1}{\sqrt{\delta_0}}$ so that $\F_0>\F$ for all $\t \in (0,1-\delta_0]$. Assume that $\n> \tfrac{4 A^2}{\delta^2}$. Then $\cc{B}_{\!\frac{A}{\sqrt{\n}}}(0)\subset B_{\!\frac{\delta}{2}}(0)$ and thus $z_0+\xi_4(v)$, $z_0+\xi_5(v)$, $z_0+\xi_6(v)$, $\tfrac{z_0+\xi_6(v)}{\sqrt{1-\xi_2(\n)}}$, $\tfrac{z_0+\xi_6(v)}{\sqrt{1-\xi_3(\n)}}\in \Xdelta{A}{\frac{\delta}{2}}{\F}$ if $z_0\in \Adelta$, $\xi_2(\n), \xi_3(\n) \in [0,\tfrac{1}{2}]$ and $\xi_4(v)$, $\xi_5(v)$, $\xi_6(v) \in \cc{B}_{\!\frac{A}{\sqrt{\n}}}(0)$.

Using Lemma \ref{lemma_estim_U_T} and \ref{lemma_estim_b}, we easily find
\begin{align}
\sup_{\substack{ 0<\t<1-\delta_0 \\ v\in \cc{B}_{\!\frac{A}{\sqrt{\n}}}(0) }} \sup_{\substack{z_0\in \Adelta \\ \xi_4(v)\in \cc{B}_{\!\frac{A}{\sqrt{\n}}}(0) } } \vert \n \Rfour \vert & \le \tfrac{A^3}{3 \sqrt{\n} \left(\frac{\delta}{2}\right)^3} < \epsilon,
\end{align}
for all $\n> \tfrac{64 A^6}{9\epsilon^2 \delta^6}$. 

The same way we get
\begin{align}
\sup_{\substack{ 0<\t<1-\delta_0 \\ v\in \cc{B}_{\!\frac{A}{\sqrt{\n}}}(0) }} \sup_{\substack{z_0\in \Adelta \\ \xi_5(v)\in \cc{B}_{\!\frac{A}{\sqrt{\n}}}(0) } } \vert \Rfive \vert & \le \tfrac{A \F_0^2}{4 \sqrt{\n} \left(\frac{\delta}{2}\right)^3} < \epsilon,
\end{align}
for all $\n> \tfrac{4 A^2 \F_0^4}{\epsilon^2 \delta^6}$. 

For the last error term we find, using additionally Lemma \eqref{lemma_estim_T2}, 
\begin{align}
\sup_{\substack{ 0<\t<1-\delta_0 \\ v\in \cc{B}_{\!\frac{A}{\sqrt{\n}}}(0) }} \sup_{\substack{z_0\in \Adelta \\ \xi_2,\xi_3\in [0,\frac{1}{2}] \\ \xi_6(v)\in \cc{B}_{\!\frac{A}{\sqrt{\n}}}(0) } } \!\!\!\!\!\! \vert \Rsix \vert & \le \tfrac{A}{\sqrt{\n}\frac{\delta}{2}}+\tfrac{A \F_0^2}{\sqrt{2}\n^{3/2}\left( \frac{\delta}{2} \right)^3}+\tfrac{\sqrt{2} A \F_0^2}{\n^{3/2}\frac{\delta}{2}}\left(1+\tfrac{2\F_0}{\delta} \right)^2+\tfrac{ A \F_0^6}{2\sqrt{2}\n^{3/2}\left(\frac{\delta}{2}\right)^7} \nn \\ 
& \le \tfrac{2A}{\sqrt{\n}\delta}+\tfrac{4\sqrt{2} A \F_0^2}{\n^{3/2}\delta^3}+\tfrac{2\sqrt{2} A \F_0^2}{\n^{3/2}\delta}\left(1+\tfrac{2\F_0}{\delta} \right)^2+\tfrac{32 \sqrt{2} A \F_0^6}{\n^{3/2}\delta^7} \nn \\
&\le \tfrac{2A}{\sqrt{\n}\delta} + \tfrac{\sqrt{2} \F_0^2}{\sqrt{\n} A \delta} + \tfrac{\F_0^2 \delta}{\sqrt{2}\sqrt{\n} A}\left(1+\tfrac{2\F_0}{\delta} \right)^2 +\tfrac{8\sqrt{2}\F_0^6}{\sqrt{\n}A\delta^5}<\epsilon,
\end{align}
if $\n> \tfrac{1}{\epsilon^2} \left( \tfrac{2A}{\delta} + \tfrac{\sqrt{2} \F_0^2}{ A \delta} + \tfrac{\F_0^2 \delta}{\sqrt{2} A}\left(1+\tfrac{2\F_0}{\delta} \right)^2 +\tfrac{8\sqrt{2}\F_0^6}{A\delta^5} \right)^2$. 

So we can choose 
\begin{align}
\n_0>\max\left\{\tfrac{4A^2}{\delta^2}, \tfrac{64A^6}{9 \epsilon^2 \delta^6}, \tfrac{4A^2 \F_0^4 }{\epsilon^2 \delta^6}, \tfrac{1}{\epsilon^2} \left( \tfrac{2A}{\delta} + \tfrac{\sqrt{2} \F_0^2}{ A \delta} + \tfrac{\F_0^2 \delta}{\sqrt{2} A}\left(1+\tfrac{2\F_0}{\delta} \right)^2 +\tfrac{8\sqrt{2}\F_0^6}{A\delta^5} \right)^2 \right\}
\end{align}
and the lemma is proven.
\end{proof}









\nocite{*} 

\clearpage
\addcontentsline{toc}{section}{Bibliography}
\bibliography{dissertation}
\bibliographystyle{hunsrt}

\end{document}